%% file: simcount.tex
\documentclass[11pt]{article}

\usepackage{etex}
\usepackage{amssymb,amsmath,amsthm,amsfonts}
\usepackage[numbers,comma,sort&compress]{natbib}
\usepackage[margin=1in]{geometry}
\usepackage[colorlinks]{hyperref}
\usepackage[all]{hypcap}
\usepackage{url}
\usepackage{microtype}
\usepackage{graphicx}
\usepackage[font=small,skip=6pt]{caption}
\usepackage[labelformat=simple]{subcaption}
\usepackage{booktabs}
\usepackage{pgfplots}
\usepackage{tikz-qtree}
\usetikzlibrary{positioning}
\usetikzlibrary{calc}
\usetikzlibrary{shapes.geometric}
\usetikzlibrary{pgfplots.groupplots}
\usepackage{enumitem}
\usepackage{qcircuit}
\usepackage{etoc}

\makeatletter
\g@addto@macro\bfseries{\boldmath}
\makeatother

\pgfdeclareplotmark{hexagon*}
{  \pgfpathmoveto{\pgfqpointpolar{0}{1.1547\pgfplotmarksize}}
  \pgfpathlineto{\pgfqpointpolar{60}{1.1547\pgfplotmarksize}}
  \pgfpathlineto{\pgfqpointpolar{120}{1.1547\pgfplotmarksize}}
  \pgfpathlineto{\pgfqpointpolar{180}{1.1547\pgfplotmarksize}}
  \pgfpathlineto{\pgfqpointpolar{240}{1.1547\pgfplotmarksize}}
  \pgfpathlineto{\pgfqpointpolar{300}{1.1547\pgfplotmarksize}}
  \pgfpathclose
\pgfusepathqfillstroke
}
\pgfdeclareplotmark{octagon*}
{  \pgfpathmoveto{\pgfqpointpolar{22.5}{1.1547\pgfplotmarksize}}
  \pgfpathlineto{\pgfqpointpolar{67.5}{1.1547\pgfplotmarksize}}
  \pgfpathlineto{\pgfqpointpolar{112.5}{1.1547\pgfplotmarksize}}
  \pgfpathlineto{\pgfqpointpolar{157.5}{1.1547\pgfplotmarksize}}
  \pgfpathlineto{\pgfqpointpolar{202.5}{1.1547\pgfplotmarksize}}
  \pgfpathlineto{\pgfqpointpolar{247.5}{1.1547\pgfplotmarksize}}
  \pgfpathlineto{\pgfqpointpolar{292.5}{1.1547\pgfplotmarksize}}
  \pgfpathlineto{\pgfqpointpolar{337.5}{1.1547\pgfplotmarksize}}
  \pgfpathclose
\pgfusepathqfillstroke
}

\usepackage{mathtools}
\mathtoolsset{centercolon}
\DeclarePairedDelimiter\bra{\langle}{\rvert}
\DeclarePairedDelimiter\ket{\lvert}{\rangle}
\DeclarePairedDelimiterX\braket[2]{\langle}{\rangle}{#1 \delimsize\vert #2}

\let\originalleft\left
\let\originalright\right
\renewcommand{\left}{\mathopen{}\mathclose\bgroup\originalleft}
\renewcommand{\right}{\aftergroup\egroup\originalright}

\newcommand{\eq}[1]{(\ref{eq:#1})}
\renewcommand{\sec}[1]{\hyperref[sec:#1]{Section~\ref*{sec:#1}}}
\newcommand{\app}[1]{\hyperref[app:#1]{Appendix~\ref*{app:#1}}}
\newcommand{\thm}[1]{\hyperref[thm:#1]{Theorem~\ref*{thm:#1}}}
\newcommand{\prop}[1]{\hyperref[prop:#1]{Proposition~\ref*{prop:#1}}}
\newcommand{\lem}[1]{\hyperref[lem:#1]{Lemma~\ref*{lem:#1}}}
\newcommand{\fig}[1]{\hyperref[fig:#1]{Figure~\ref*{fig:#1}}}
\newcommand{\tab}[1]{\hyperref[tab:#1]{Table~\ref*{tab:#1}}}

\DeclarePairedDelimiter\ceil{\lceil}{\rceil}

\DeclarePairedDelimiter\norm{\|}{\|}

\newcommand{\C}{\mathbb{C}}
\newcommand{\N}{\mathbb{N}}
\newcommand{\R}{\mathbb{R}}
\newcommand{\Z}{\mathbb{Z}}
\renewcommand{\Re}{\operatorname{Re}}

\newcommand{\Sym}[1]{\operatorname{Sym}(#1)}

\DeclareMathOperator{\poly}{poly}

\DeclareMathOperator{\select}{select}
\newcommand{\rem}{\mathcal{R}}

\newcommand{\comment}[1]{}

\newcommand{\ys}[1]{\textbf{\color{purple}[Yuan: #1]}}
\newcommand{\nam}[1]{\textbf{\color{teal}[Nam: #1]}}

\newtheorem{theorem}{Theorem}[section]
\newtheorem{lemma}[theorem]{Lemma}
\newtheorem{proposition}[theorem]{Proposition}

\newtheorem{definition}{Definition}
\newtheorem{example}{Example}

\usepackage{stmaryrd}

\pgfplotsset{
	log x ticks with fixed point/.style={
		xticklabel={
			\pgfkeys{/pgf/fpu=true}
			\pgfmathparse{exp(\tick)}			\pgfmathprintnumber[fixed relative, precision=3]{\pgfmathresult}
			\pgfkeys{/pgf/fpu=false}
		}
	},
	log y ticks with fixed point/.style={
		yticklabel={
			\pgfkeys{/pgf/fpu=true}
			\pgfmathparse{exp(\tick)}			\pgfmathprintnumber[fixed relative, precision=3]{\pgfmathresult}
			\pgfkeys{/pgf/fpu=false}
		}
	}
}

\usepackage{pbox}

\newcommand{\PF}{PF}
\newcommand{\TS}{TS}
\newcommand{\QSP}{QSP}

\newcommand{\NOT}{\textsc{not}}
\newcommand{\CNOT}{\textsc{cnot}}
\newcommand{\TOF}{\textsc{tof}}
\newcommand{\OR}{\textsc{or}}
\newcommand{\XOR}{\textsc{xor}}
\newcommand{\T}{T}
\newcommand{\Rz}{R_z}

\newcommand{\numselect}{\Gamma}
\newcommand{\lenstring}{w}

\title{Toward the first quantum simulation with quantum speedup}

\author
{Andrew M.\ Childs,$^{1,2,3,\ast}$ Dmitri Maslov,$^{2,3,4}$ Yunseong Nam$^{2,3,5}$, \\
Neil J.\ Ross$^{2,3,6}$, and Yuan Su$^{1,2,3}$\\
\\
\normalsize{$^{1}$Department of Computer Science, University of Maryland}\\
\normalsize{$^{2}$Institute for Advanced Computer Studies, University of Maryland}\\
\normalsize{$^{3}$Joint Center for Quantum Information and Computer Science, University of Maryland}\\
\normalsize{$^{4}$National Science Foundation}\\
\normalsize{$^{5}$IonQ, Inc.}\\
\normalsize{$^{6}$Department of Mathematics and Statistics, Dalhousie University}
}

\date{}

\begin{document}

\etocdepthtag.toc{mtchapter}
\etocsettagdepth{mtchapter}{subsection}
\etocsettagdepth{mtappendix}{none}

\maketitle

\begin{abstract}
With quantum computers of significant size now on the horizon, we should understand how to best exploit their initially limited abilities. To this end, we aim to identify a practical problem that is beyond the reach of current classical computers, but that requires the fewest resources for a quantum computer. We consider quantum simulation of spin systems, which could be applied to understand condensed matter phenomena. We synthesize explicit circuits for three leading quantum simulation algorithms, employing diverse techniques to tighten error bounds and optimize circuit implementations. Quantum signal processing appears to be preferred among algorithms with rigorous performance guarantees, whereas higher-order product formulas prevail if empirical error estimates suffice. Our circuits are orders of magnitude smaller than those for the simplest classically-infeasible instances of factoring and quantum chemistry.
\end{abstract}

\section{Introduction}
\label{sec:intro}

While a scalable quantum computer remains a long-term goal, recent experimental progress suggests that devices capable of outperforming classical computers will soon be available \cite{Che14,www:IBM,DLFLWM16,Son17,Ber17,Zha17}.  Multiple groups have already developed programmable devices with several qubits and two-qubit gate fidelities around 98\% \cite{Lin17}, and similar devices with around 50 qubits are under active development. While the error rates of these early machines severely limit the total number of gates that can be reliably performed, future improvements should lead to machines with more qubits and more reliable gates.  This raises the exciting possibility of solving practical problems that are beyond the reach of classical computation.
Such an outcome would be a landmark in the development of quantum computers and would begin an era in which they serve not only as testbeds for science, but as practical computing machines.

Reaching this goal will require not only significant experimental advances, but also careful quantum algorithm design and implementation.  Here we address the latter issue by developing explicit circuits, and thereby producing concrete resource estimates, for practical quantum computations that can outperform classical computers. Through this work, we aim to identify applications for small quantum computers that help to motivate the significant investment required to develop scalable, fault-tolerant quantum computers.

There has been considerable previous research on compiling quantum algorithms into explicit circuits (see \app{related} for more detail). However, to the best of our knowledge, none of these studies aimed to identify minimal examples of super-classical quantum computation, and typical resource counts were high.  Our work is also distinct from recent work on quantum computational supremacy \cite{HM17}, where the goal is merely to accomplish a super-classical task, regardless of its practicality.  Instead, we aim to pave the way toward practical quantum computations (which may not be far beyond the threshold for supremacy).

Arguably, the most natural application of quantum computers is to the problem of simulating quantum dynamics \cite{Fey82}.  Quantum computers can simulate a wide variety of quantum systems, including fermionic lattice models \cite{WHWCNT15}, quantum chemistry \cite{WBCHT14}, and quantum field theories \cite{JLP12}.  However, simulations of spin systems with local interactions likely have less overhead, so we focus on them as an early candidate for practical quantum simulation.  While analog simulation may be easier to realize in the short term, we focus on digital simulation for its greater flexibility and the prospect of invoking fault tolerance.

Efficient quantum algorithms for simulating quantum dynamics have been known for over two decades \cite{Llo96}.  Recent work has led to algorithms with significantly improved asymptotic performance as a function of various parameters such as the evolution time and the allowed simulation error \cite{BACS05,BC12,BCCKS14,LC16,LC17}.  Our work investigates whether these alternative algorithms can be advantageous for simulations of relatively small systems, and aims to lay the groundwork for the first practical application of quantum computers.

\section{Target system}
\label{sec:targetsystem}

To produce concrete benchmarks, we focus on a specific simulation task. Specifically, we consider a one-dimensional nearest-neighbor Heisenberg model with a random magnetic field in the $z$ direction.  This model is described by the Hamiltonian
\begin{align}
  \sum_{j=1}^n (\vec \sigma_j \cdot \vec \sigma_{j+1} + h_j \sigma_j^z)
	\label{eq:heisenberg}
\end{align}
where $\vec \sigma_j = (\sigma^x_j,\sigma^y_j,\sigma^z_j)$ denotes a vector of Pauli $x$, $y$, and $z$ matrices on qubit $j$.  We impose periodic boundary conditions (i.e., $\vec\sigma_{n+1} = \vec\sigma_1$), and $h_j \in [-h,h]$ is chosen uniformly at random.  The parameter $h$ characterizes the strength of the disorder.

This Hamiltonian has been considered in recent studies of self-thermalization and many-body localization (see \app{targetsystem} for more detail).  Despite intensive investigation, the details of a transition between thermal and localized phases remain poorly understood. A major challenge is the difficulty of simulating quantum systems with classical computers; indeed, the most extensive numerical study we are aware of was restricted to at most 22 spins \cite{LLA15}.

Hamiltonian simulation can efficiently access any feature that could be observed experimentally (and more), and there are several proposals for exploring self-thermalization by simulating dynamics \cite{Ser14,Sch15,Smi16}.  Since all of these approaches involve only very simple state preparations and measurements, we focus on the cost of simulating dynamics.  We consider evolution times comparable to the number of spins, since the system must evolve for this long for self-thermalization to take place (or even for information to propagate across the system, owing to the Lieb-Robinson bound).

Specifically, we produce gate counts for simulations with $h=1$, evolution time $t=n$ (the number of spins in the system), and overall accuracy $\epsilon = 10^{-3}$.  These explicit choices help us to focus on the system-size dependence of quantum simulation algorithms. This is a key consideration for practical applications, yet it has been deemphasized in the literature on sparse Hamiltonian simulation.

\section{Implementations}
\label{sec:implementations}

\begin{table}
  \begin{center}
    \renewcommand{\arraystretch}{1.15}
    \small
    \begin{tabular}{l|c|l|l}
      Algorithm & Ref. & Gate complexity ($t,\epsilon$) & Gate complexity ($n$)\\ \hline
      Product formula (\PF), 1st order & \cite{Llo96} & $O(t^2/\epsilon)$ & $O(n^5)$ \\
      Product formula (\PF), ($2k$)th order & \cite{BACS05} & $O(5^{2k} t^{1+1/2k} / \epsilon^{1/2k})$ & $O(5^{2k} n^{3+1/k})$ \\
      Quantum walk & \cite{BC12} & $O(t/\sqrt{\epsilon})$ & $O(n^4 \log n)$ \\
      Fractional-query simulation & \cite{BCCKS13} & $O\bigl(t \frac{\log^2(t/\epsilon)}{\log\log(t/\epsilon)}\bigr)$ & $O\bigl(n^4\frac{\log n}{\log\log n}\bigr)$ \\
      Taylor series (\TS) & \cite{BCCKS14} & $O\bigl(t \frac{\log^2(t/\epsilon)}{\log\log(t/\epsilon)}\bigr)$ & $O\bigl(n^3\frac{\log^2 n}{\log\log n}\bigr)$ \\
      Linear combination of quantum walk steps & \cite{BCK15} & $O\bigl(t \frac{\log^{3.5}(t/\epsilon)}{\log\log(t/\epsilon)}\bigr)$ & $O\bigl(n^4\frac{\log n}{\log\log n}\bigr)$ \\
      Quantum signal processing (\QSP) & \cite{LC16} & $O(t + \log(1/\epsilon))$ & $O(n^3 \log n)$
    \end{tabular}
    \renewcommand{\arraystretch}{1}
  \end{center}
\caption{Previously-established asymptotic gate complexities of quantum simulation algorithms as a function of the simulation time $t$, allowed error $\epsilon$, and the system size $n$ for a one-dimensional nearest-neighbor spin model as in \eq{heisenberg} with $t=n$ and fixed $\epsilon$.
\label{tab:algsummary}}
\end{table}

There are many distinct quantum algorithms for Hamiltonian simulation, some of which are summarized in \tab{algsummary}.  We implement algorithms based on high-order product formulas (\PF, introduced in \sec{algpf}) \cite{BACS05}, direct application of the Taylor series (\TS, \sec{alglcu}) \cite{BCCKS14}, and the recent quantum signal processing method (\QSP, \sec{algqsp}) \cite{LC16}.  We expect these to be among the most efficient approaches to digital quantum simulation.  In particular, approaches based on quantum walk \cite{BC12,BCK15} appear to incur greater overhead (as discussed in \app{algother}).

To produce concrete circuits, we implement quantum simulation algorithms in a quantum circuit description language called Quipper \cite{GLRSV13} (see \app{circuit} for more details). Wherever possible, we tighten the analysis of algorithm parameters and manually optimize the implementation.  We also process all circuits using an automated tool we developed for large-scale quantum circuit optimization \cite{Optimizer}. Our implementation is available in a public repository \cite{SourceCode}.

We express our circuits over the set of two-qubit $\CNOT$ gates, single-qubit Clifford gates, and single-qubit $z$ rotations $R_z(\theta) := \exp(-i \sigma^z \theta/2)$ for $\theta \in \R$.  Such gates can be directly implemented at the physical level with both trapped ions \cite{DLFLWM16} and superconducting circuits \cite{Che14, www:IBM}.  In both technologies, two-qubit gates take longer to perform and incur more error than single-qubit gates.  Thus, the $\CNOT$ count is a useful figure of merit for assessing the cost of physical-level circuits on a universal device.  We also produce Clifford+$\T$ circuits using optimal circuit synthesis \cite{RS16} so that we can count $\T$ gates, which are typically the most expensive gates for fault-tolerant computation.

Our analysis ignores many practical details, such as architectural constraints, instead aiming to give a broad overview of potential implementation costs that can be refined for specific systems.  When counting qubits, we assume that measured ancillas can be reused later.

\subsection{Product formula algorithm}
\label{sec:pf}

The product formula (\PF) approach approximates the exponential of a sum of operators by a product of exponentials of the individual operators. The asymptotic complexity of this approach can be improved with higher-order Suzuki formulas \cite{Suz91}.  By splitting the evolution into $r$ segments and making $r$ sufficiently large, we can ensure that the simulation is arbitrarily precise.  The main challenge in making these algorithms concrete is to choose an explicit $r$ that ensures some desired upper bound on the error.  \app{pf} gives a detailed description of these implementation details.

We present two bounds, which we call the \emph{analytic} and \emph{minimized} bounds, that slightly strengthen previous analysis \cite{BACS05}.  However, bounds of this type are far from tight \cite{RWS12,BMWAW15,RWSWT17}.  Thus, we develop an improved bound that exploits commutation relations among terms in the target Hamiltonian.  For a one-dimensional system of $n$ spins with nearest-neighbor couplings, evolving for time $t=n$, this \emph{commutator bound} improves the asymptotic complexity of the $(2k)$th-order \PF\ algorithm from $O(n^{3+1/k})$ to $O(n^{3+2/(2k+1)})$ while also significantly improving the leading constant.

Naive computation of the commutator bound takes time $O(n^{2k+1})$, which can be prohibitive even for small $n$.  To make this approach practical, we develop techniques that exploit the structure of the Hamiltonian to compute the commutator bound in closed form.  We explicitly evaluate this bound for the first-, second-, and fourth-order product formulas.

Unfortunately, even the commutator bound can be very loose.  To address this, we report empirical error estimates by extrapolating the error seen in direct classical simulations of small instances (as also explored in previous work on simulating many-body dynamics \cite{RWS12} and quantum chemistry \cite{BMWAW15,RWSWT17}).  While these \emph{empirical bounds} do not provide rigorous guarantees on the simulation error, they may nevertheless be useful in practice, and they improve the cost of \PF\ algorithms by several orders of magnitude.

\subsection{Taylor series algorithm}
\label{sec:ts}

The Taylor series (\TS) algorithm directly implements the (truncated) Taylor series of the evolution operator for a carefully-chosen constant time using a procedure for implementing linear combinations of unitary operations \cite{BCCKS14}.  This segment is then simply repeated until the entire evolution time has been simulated.  The circuit for a segment is built using three subroutines: a state preparation procedure, a reflection about the $\ket{0}$ state, and an operation denoted $\select(V)$ (discussed further below).  Our implementation of the \TS\ algorithm (described in detail in \app{lcu}) also includes a concrete error analysis that establishes rigorous, non-asymptotic bounds on the simulation parameters.

The aforementioned $\select(V)$ operation applies a unitary $V_j$ conditioned on a control register being in the state $\ket{j}$, for $j \in \{1,\ldots,\numselect\}$.  We develop an improved implementation of this operation by designing an optimized walk on a binary tree, saving a factor of about $\log_2 \numselect$ in the gate count.  For our simulations of systems with 10 to 100 spins, this reduces $\CNOT$ and $\T$ gate counts over a naive implementation by a factor of between $5$ and $9$, significantly improving the overall complexity. Furthermore, the cost of our $\select(V)$ implementation meets a previously-established asymptotic lower bound \cite{ar:m}.  This improvement may be more broadly applied to any algorithm using the $\select(V)$ procedure, such as others based on linear combinations of unitaries.

\subsection{Quantum signal processing algorithm}
\label{sec:qsp}

The quantum signal processing (\QSP) algorithm of Low and Chuang \cite{LC17,LC16} effectively implements a linear combination of unitary operators by a different mechanism. This algorithm applies a sequence of operations called \emph{phased iterates} that manifest each eigenvalue of the Hamiltonian as a rotation acting on an ancilla qubit.  By carefully choosing a sequence of rotation angles for that qubit, we induce the desired evolution.

The circuit for each phased iterate is built from similar subroutines as the \TS\ algorithm.  However, computing the $M$ rotation angles for the phased iterates requires finding the roots of a polynomial of degree $2M$, and these roots must be determined to high precision.  Because of these challenges, we were unable to compute the parameters of the algorithm explicitly except in very small instances.  Instead, we produced estimates of the gate count (but not a complete implementation) by synthesizing a version of the algorithm with placeholder values of the parameters.

One way to alleviate this problem is to consider a segmented implementation of the \QSP\ algorithm.  In this approach, we divide the evolution time into $r$ segments, each of which is sufficiently short that the classical preprocessing is tractable.  Since the optimality of the \QSP\ approach to Hamiltonian simulation relies essentially on simulating the entire evolution as a single segment, the segmented approach has higher asymptotic complexity.  However, it allows us to develop a complete implementation, and the overhead for moderate values of $n$ is not too high.

For the full version of the algorithm, we consider an empirical error bound on the Jacobi-Anger expansion, giving a modest improvement. Numerical evidence suggests that the additional savings from an empirical error bound for the overall algorithm would not be significant. For the segmented version of the algorithm, we instead used an analytic error bound so that the algorithm remains rigorous (and because an empirical Jacobi-Anger error bound did not give much improvement in that case).

\app{lowchuang} discusses our implementation of \QSP\ algorithms in detail.

\section{Results}
\label{sec:results}

\input{preoptimcounts.tex}
\input{postoptimcounts.tex}

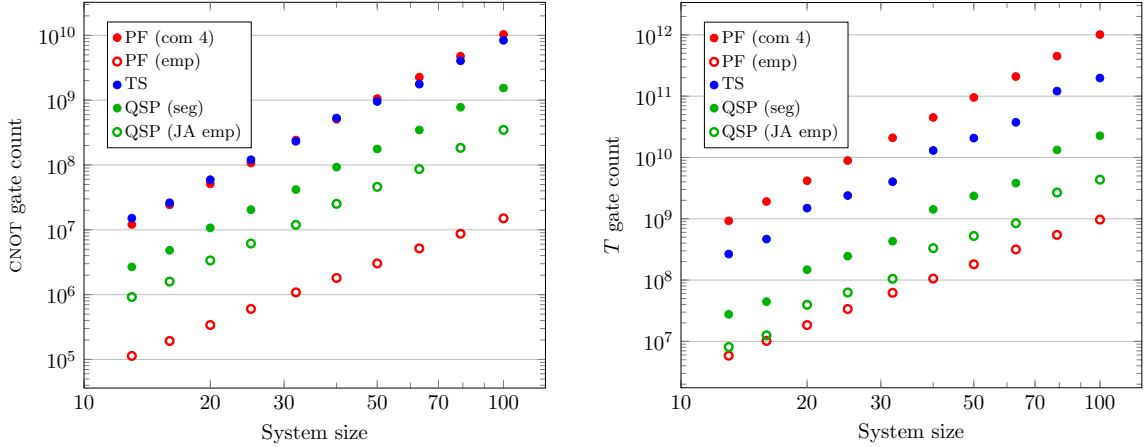
\begin{figure}
  \begin{subfigure}{.5\linewidth}
  \centering
  \resizebox{.9\textwidth}{!}{
  \begin{tikzpicture}
    \begin{loglogaxis}[
      width=10cm,
      ymajorgrids=true,
      legend style={at={(0.05,0.95)},anchor=north west,font=\footnotesize},
      xlabel={System size},
      ylabel={$\CNOT$ gate count},
      xmin=10,
      xtick={10,100},
      xticklabels={10,100},
      extra x ticks={20,30,50,70},
      extra x tick labels={20,30,50,70},
      every axis legend/.append style={nodes={right}}
      ]

      \addplot[only marks, red] coordinates {
        \postoptimcznormalcnotfrthcomavg
      };
      \addlegendentry{\PF~(com 4)}

      \addplot[only marks,red,mark=o,mark options={fill=white,line width=1.25pt}] coordinates {
        \postoptimcznormalcnotbestfitavg
      };
      \addlegendentry{\PF~(emp)}

      \addplot[only marks, blue] coordinates {
        \postoptimcznormalcnotlcuavg
      };
      \addlegendentry{\TS}

      \addplot[only marks, black!30!green] coordinates {
        \postoptimcznormalcnotspsegmentavg
      };
      \addlegendentry{\QSP~(seg)}

      \addplot[only marks,black!30!green,mark=o,mark options={fill=white,line width=1.25pt}] coordinates {
        \postoptimcznormalcnotspjaavg
      };
      \addlegendentry{\QSP~(JA emp)}

    \end{loglogaxis}
  \end{tikzpicture}
  }
 \end{subfigure}
 \begin{subfigure}{.5\linewidth}
  \resizebox{.9\textwidth}{!}{
  \begin{tikzpicture}
    \begin{loglogaxis}[
      width=10cm,
      ymajorgrids=true,
      legend style={at={(0.05,0.95)},anchor=north west,font=\footnotesize},
      xlabel={System size},
      ylabel={$\T$ gate count},
      xmin=10,
      xtick={10,100},
      xticklabels={10,100},
      extra x ticks={20,30,50,70},
      extra x tick labels={20,30,50,70},
      every axis legend/.append style={nodes={right}}
      ]

      \addplot[only marks, red] coordinates {
        \postoptimctnormaltfrthcomavg
      };
      \addlegendentry{\PF~(com 4)}

      \addplot[only marks,red,mark=o,mark options={fill=white,line width=1.25pt}] coordinates {
        \postoptimctnormaltbestfitavg
      };
      \addlegendentry{\PF~(emp)}

      \addplot[only marks, blue] coordinates {
        \postoptimctnormaltlcuavg
      };
      \addlegendentry{\TS}

      \addplot[only marks, black!30!green] coordinates {
        \postoptimctnormaltspsegmentavg
      };
      \addlegendentry{\QSP~(seg)}

      \addplot[only marks,black!30!green,mark=o,mark options={fill=white,line width=1.25pt}] coordinates {
        \postoptimctnormaltspjaavg
      };
      \addlegendentry{\QSP~(JA emp)}

    \end{loglogaxis}
  \end{tikzpicture}
  }
  \end{subfigure}
 \caption{Gate counts for optimized implementations of the \PF\ algorithm (using the fourth-order formula with commutator bound and the better of the fourth- or sixth-order formula with empirical error bound), the \TS\ algorithm, and the \QSP\ algorithm (using the segmented version with analytic error bound and the non-segmented version with empirical Jacobi-Anger error bound) for system sizes between 10 and 100.  Left: $\CNOT$ gates for Clifford+$\Rz$ circuits.  Right: $\T$ gates for Clifford+$\T$ circuits.
 \label{fig:pflcu_gates}}
\end{figure}

\begin{figure}
\centering \resizebox{.5\linewidth}{!}{\input{PlotQubits.tex}}
  \caption{Number of qubits used by the \PF, \TS, and \QSP\ algorithms.\label{fig:qubitcounts}}
\end{figure}
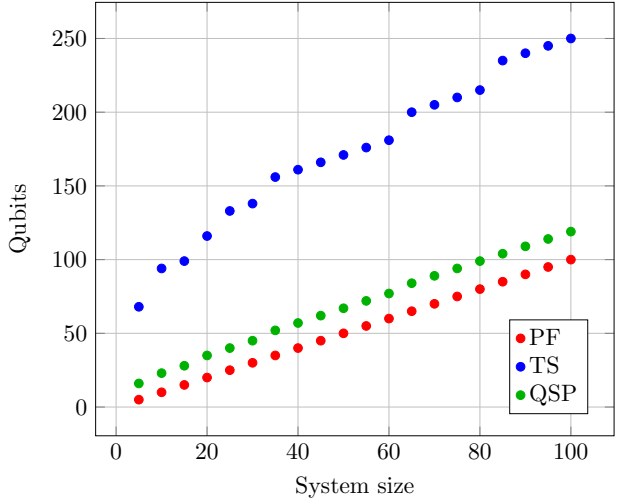

\fig{pflcu_gates} compares gate counts for the \PF\ algorithm (with commutator and empirical error bounds), the \TS\ algorithm, and the \QSP\ algorithm (in both segmented and non-segmented versions).  The \TS\ algorithm uses significantly more qubits than the \QSP\ algorithm (as shown in \fig{qubitcounts}) while also requiring more gates, so the latter is clearly preferred.  In contrast, the \QSP\ algorithm has only slightly greater space requirements than the \PF\ algorithm.

Surprisingly, despite being more involved, the \QSP\ algorithm outperforms the rigorously-bounded \PF\ algorithm even for small system sizes.  In particular, among the rigorously-analyzed algorithms, the segmented \QSP\ algorithm has the best performance, improving over the \PF\ algorithm by about an order of magnitude for $\CNOT$ count and by almost two orders of magnitude for $\T$ count.

Empirical error bounds improve the performance of the \PF\ algorithm by two to three orders of magnitude, making it the preferred approach if rigorous performance guarantees are not required.  For the $\CNOT$ count, the empirical \PF\ algorithm improves over the full \QSP\ algorithm by about an order of magnitude.  The advantage in the $\T$ count is less significant, but still indicates that the \PF\ algorithm is dominant, especially considering its lower qubit count.

Although we expected that higher-order product formulas would not be advantageous for small system sizes, we find that the fourth- and sixth-order formulas had the best performance for our benchmark system with tens to hundreds of qubits, as shown in \fig{pf}.  The fourth-order formula with commutator bound gives the best available \PF\ algorithm with a rigorous performance guarantee.  Using empirical error bounds, the sixth-order formula outperforms the fourth-order formula for systems of about 30 or more qubits, making the former the method of choice for simulations just beyond the reach of classical computers (again, provided a heuristic error bound can be tolerated).  These results suggest that higher-order formulas may be advantageous for other quantum simulations, such as those for quantum chemistry, even though they have not usually been considered \cite{Pou15,RWSWT17}.

\begin{figure}
\centering
  \begin{subfigure}{.32\linewidth}
    \resizebox{.95\textwidth}{!}{\input{min12468.tex}}
  \end{subfigure}
  \begin{subfigure}{.32\linewidth}
    \resizebox{.95\textwidth}{!}{\input{com124.tex}}
  \end{subfigure}
  \begin{subfigure}{.32\linewidth}
    \resizebox{.95\textwidth}{!}{\input{fit12468.tex}}
  \end{subfigure}
\caption{Total gate counts in the Clifford+$\Rz$ basis for product formula algorithms using the minimized (left), commutator (center), and empirical (right) bounds, for system sizes between 13 and 500.}\label{fig:pf}
\end{figure}
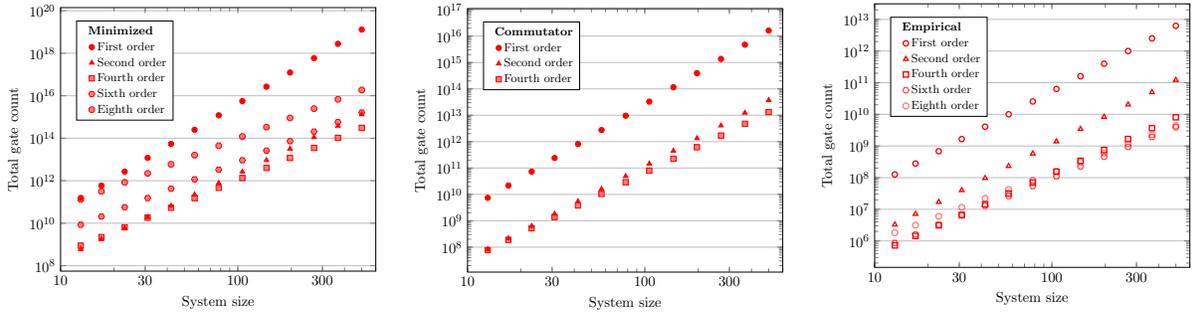

For a system of $50$ qubits---which is presumably close to the limits of direct classical simulation for circuits such as ours \cite{HS17}\footnote{Recent work has demonstrated simulation of 56-qubit computations, but only for circuits of much smaller depth than those considered in our work \cite{Ped17}.}---the \TS\ algorithm uses $171$ qubits and the \QSP\ algorithm uses $67$, whereas the \PF\ algorithm uses only $50$. At this size, the segmented \QSP\ algorithm is the best rigorously-analyzed approach, using about $1.8 \times 10^8$ $\CNOT$ gates (over the set of Clifford+$\Rz$ gates) and $2.4 \times 10^9$ $\T$ gates (over the set of Clifford+$\T$ gates).  Using the empirical error bound, the \PF\ algorithm uses about $3 \times 10^6$ $\CNOT$s and $1.8 \times 10^8$ $\T$s (over Clifford+$\Rz$ and Clifford+$\T$, respectively).

For comparison, previous estimates of gate counts for factoring, discrete logarithms, and quantum chemistry simulations are significantly larger.
First consider factoring a 1024-bit number, which is beyond the factorization of RSA-768 that was achieved classically in 2009 \cite{Kle10}. The best implementation we are aware of uses 3132 qubits and about $5.7 \times 10^9$ $\T$ gates (when realized over the set of Clifford+$\T$ gates) \cite{Kut06}.\footnote{Reference \cite{Kut06} does not give explicit resource counts; we estimate them as described in \app{related}.}
Quantum algorithms for classically-hard instances of the elliptic curve discrete logarithm problem have roughly comparable cost \cite{RNSL17}.
For quantum chemistry, a natural target for a problem just beyond the reach of classical computing is a simulation of FeMoco, the primary cofactor of nitrogenase, an enzyme that catalyzes the nitrogen fixation process.  Even for a fairly low-precision simulation, and using non-rigorous estimates of the product formula error, the best implementation we are aware of uses $111$ qubits and $1.0 \times 10^{14}$ $\T$ gates \cite{RWSWT17}.  Thus it appears that simulation of spin systems is indeed a significantly easier task for near-term practical quantum computation.

For a more detailed discussion of the results, see \app{results}.

\section{Discussion}
\label{sec:discussion}

The work described in this paper represents progress toward the first genuine application of quantum computers, solving a practical problem that is beyond the reach of classical computation.  Of course, our results only represent upper bounds. While we attempted to optimize the implementation wherever possible, it is likely that further improvements can be found, and it is conceivable that another algorithm (or computational task) may offer better performance.  Our work establishes a concrete set of benchmarks that we hope can be improved through future studies.

Demonstrations of digital quantum simulation performed to date \cite{BCC06,Lan11,Bar15} have been limited in scope, primarily using the first-order formula (except for some limited applications of the second-order formula \cite{BCC06,Lan11}).  Our results show that higher-order formulas are useful even for simulations of small systems.  In the near term, it could be fruitful to demonstrate the utility of these formulas experimentally. Even relatively small experiments might be able to probe the validity of our empirical error bounds.

We have also identified some avenues for future improvement of quantum simulation algorithms.  We saw that rigorous error bounds for product formulas are very loose, even with our newly-developed commutator bound.  This motivates attempting to prove stronger rigorous error bounds for product formulas.  Also, the difficulty of computing the angles needed to perform the \QSP\ algorithm prevents us from taking full advantage of the algorithm in practice, so it would be useful to develop a more efficient classical procedure for specifying these angles.

Further reduction of the gate count could be especially significant if it led to a simulation with sufficiently few gates to be performed without invoking fault tolerance.  With our current estimate of millions of $\CNOT$ gates for a superclassical simulation, this is likely out of reach at present.  However, further improvement could obviate the need for error correction in a system with highly accurate gates, making an early demonstration of superclassical simulation more accessible.

Finally, our work has considered an idealized system, and we hope future work will take more realistic constraints into account.  Practical devices will come with architectural constraints, may employ different basic operations than those considered here, may allow parallelization of gates, and will likely require fault tolerance.  By incorporating such features, we hope the work begun here will lead to a blueprint for the first practical quantum computation.

\section*{Acknowledgments}

We thank Zhexuan Gong, Alexey Gorshkov, Guang Hao Low, Chris Monroe, and Nathan Wiebe for helpful discussions.

This work was supported in part by the Army Research Office (MURI award W911NF-16-1-0349), the Canadian Institute for Advanced Research, and the National Science Foundation (grant 1526380).

This material was partially based on work supported by the National Science Foundation during DM's assignment at the Foundation. Any opinion, finding, and conclusions or recommendations expressed in this material are those of the authors and do not necessarily reflect the views of the National Science Foundation.

\clearpage

\appendix

\renewcommand{\contentsname}{Appendices}
\etocdepthtag.toc{mtappendix}
\etocsettagdepth{mtchapter}{none}
\etocsettagdepth{mtappendix}{subsubsection}
\tableofcontents‎‎

\section{Related work}
\label{app:related}
\input{related.tex}

\section{Self-thermalization in spin models}
\label{app:targetsystem}
\input{targetsystem.tex}

\section{Simulation algorithms}
\label{app:algorithms}
\input{simalg.tex}

\section{System-size dependence for other simulation algorithms}
\label{app:algother}

\input{algother.tex}

\section{Circuit synthesis and optimization}
\label{app:circuit}
\input{quipper.tex}

\section{Product formula implementation details}
\label{app:pf}
\input{pfappendix.tex}

\section{Taylor series implementation details}
\label{app:lcu}
\input{lcu.tex}

\input{selectV.tex}

\section{Quantum signal processing implementation details}
\label{app:lowchuang}
\input{low_chuang.tex}

\section{Detailed results}
\label{app:results}
\input{results.tex}

\input{bib.tex}

\end{document}

%% file: preoptimcounts.tex
\def\preoptimcznormaltotalfstanaavg{
(13,419860070266)
(16,1185735236192)
(20,3618576770120)
(25,11043019928200)
(32,37943527552320)
(40,115794456641760)
(50,353376637699800)
(63,1122252443826696)
(79,3479551501731006)
(100,11308052406391000)
}
 
\def\preoptimcznormaltotalfstminavg{
(13,154458116488)
(16,436208041984)
(20,1331200832000)
(25,4062501625000)
(32,13958647119872)
(40,42598406656000)
(50,130000013000000)
(63,412853627892888)
(79,1280055513260056)
(100,4160000104000000)
}
 
\def\preoptimcznormaltotalfstcomavg{
(13,7427444206)
(16,17042994592)
(20,41608873280)
(25,101584162550)
(32,272687910144)
(40,665741966240)
(50,1625346595600)
(63,4096643444712)
(79,10129181042292)
(100,26005545506200)
}
 
\def\preoptimcznormaltotalfstfitavg{
(13,125901958)
(16,232979552)
(20,451398480)
(25,874585400)
(32,1817916672)
(40,3522220000)
(50,6824312300)
(63,13538124054)
(79,26477746334)
(100,53249510600)
}
 
\def\preoptimcznormaltotalsndanaavg{
(13,992122482)
(16,2276521680)
(20,5557913700)
(25,13569125025)
(32,36424340352)
(40,88926611040)
(50,217105982550)
(63,547209930582)
(79,1353007289292)
(100,3473695695300)
}
 
\def\preoptimcznormaltotalsndminavg{
(13,602200911)
(16,1381615296)
(20,3372676920)
(25,8233276800)
(32,22099163328)
(40,53949770640)
(50,131706931350)
(63,331950606624)
(79,820740978411)
(100,2107106932200)
}
 
\def\preoptimcznormaltotalsndcomavg{
(13,82434105)
(16,173313504)
(20,384499200)
(25,851962650)
(32,2052509280)
(40,4542355800)
(50,10051437000)
(63,22884463980)
(79,51234198324)
(100,118666983600)
}
 
\def\preoptimcznormaltotalsndfitavg{
(13,3273231)
(16,5955168)
(20,11331180)
(25,21557700)
(32,43918752)
(40,83560440)
(50,158987400)
(63,309521142)
(79,594313761)
(100,1172520600)
}
 
\def\preoptimcznormaltotalfrthanaavg{
(13,1114834500)
(16,2305836480)
(20,5035163700)
(25,10995083625)
(32,26087560800)
(40,56966357400)
(50,124395171000)
(63,279318958065)
(79,616741693125)
(100,1407370627500)
}
 
\def\preoptimcznormaltotalfrthminavg{
(13,873820740)
(16,1806203760)
(20,3941739000)
(25,8602731000)
(32,20400399840)
(40,44528344800)
(50,97197279000)
(63,218170507590)
(79,481574524740)
(100,1098608391000)
}
 
\def\preoptimcznormaltotalfrthcomavg{
(13,77133420)
(16,154709520)
(20,325686000)
(25,683482875)
(32,1547519520)
(40,3232839000)
(50,6744176250)
(63,14428747620)
(79,30365888070)
(100,65890827000)
}
 
\def\preoptimcznormaltotalfrthfitavg{
(13,722670)
(16,1228080)
(20,2167500)
(25,3837750)
(32,7205280)
(40,12739800)
(50,22516500)
(63,40628385)
(79,72441420)
(100,132268500)
}
 
\def\preoptimcznormaltotalsxthanaavg{
(13,9909512925)
(16,19798893600)
(20,41655678000)
(25,87641045625)
(32,199560225600)
(40,419862804000)
(50,883366155000)
(63,1908574934175)
(79,4058177413650)
(100,8903772525000)
}
 
\def\preoptimcznormaltotalsxthminavg{
(13,8481079425)
(16,16932591600)
(20,35599147500)
(25,74847663750)
(32,170311072800)
(40,358115880000)
(50,753047130000)
(63,1626166210725)
(79,3456063810825)
(100,7579261725000)
}
 
\def\preoptimcznormaltotalsxthfitavg{
(13,861900)
(16,1387200)
(20,2320500)
(25,3888750)
(32,6895200)
(40,11526000)
(50,19316250)
(63,32933250)
(79,55499475)
(100,95752500)
}
 
\def\preoptimcznormaltotaleigthanaavg{
(13,147721206750)
(16,290079432000)
(20,599065635000)
(25,1237177125000)
(32,2759715264000)
(40,5699304060000)
(50,11770078668750)
(63,24945007066875)
(79,52049299774125)
(100,111976488675000)
}
 
\def\preoptimcznormaltotaleigthminavg{
(13,132102667125)
(16,259237692000)
(20,535010655000)
(25,1104186337500)
(32,2461414224000)
(40,5080362960000)
(50,10486195743750)
(63,22212256586625)
(79,46324574757750)
(100,99612600112500)
}
 
\def\preoptimcznormaltotaleigthfitavg{
(13,1823250)
(16,2856000)
(20,4462500)
(25,7331250)
(32,12240000)
(40,19635000)
(50,31875000)
(63,52211250)
(79,84609000)
(100,139612500)
}
 
\def\preoptimcznormaltotallcuavg{
(13,54427655)
(16,90990931)
(20,218737242)
(25,430586721)
(32,798264070)
(40,1937060895)
(50,3381345692)
(63,6067043170)
(79,14779269571)
(100,29556188245)
}
 
\def\preoptimcznormaltotalspavg{
(13,3917202)
(16,6687180)
(20,15000188)
(25,27072874)
(32,51744930)
(40,116017194)
(50,209895913)
(63,389642362)
(79,882150419)
(100,1649670860)
}
 
\def\preoptimcznormaltotalspjaavg{
(13,3018703)
(16,5094574)
(20,11298856)
(25,20122302)
(32,37917779)
(40,83904742)
(50,149754920)
(63,274139268)
(79,612236028)
(100,1128617998)
}
 
\def\preoptimcznormaltotalspsegmentavg{
(13,8884651)
(16,15586061)
(20,36361703)
(25,67568438)
(32,133553108)
(40,311278768)
(50,579672151)
(63,1109067960)
(79,2610928166)
(100,5032905446)
}
 
\def\preoptimcznormaltotalspsegmentfitavg{
(13,8401132)
(16,14701262)
(20,34292338)
(25,63820945)
(32,126280414)
(40,294463218)
(50,548791595)
(63,1050682425)
(79,2475944448)
(100,4776142259)
}
 
\def\preoptimcznormalcnotfstanaavg{
(13,96890785446)
(16,273631208352)
(20,835056177720)
(25,2548389214200)
(32,8756198665920)
(40,26721797686560)
(50,81548454853800)
(63,258981333190776)
(79,802973423476386)
(100,2609550555321000)
}
 
\def\preoptimcznormalcnotfstminavg{
(13,35644180728)
(16,100663394304)
(20,307200192000)
(25,937500375000)
(32,3221226258432)
(40,9830401536000)
(50,30000003000000)
(63,95273914129128)
(79,295397426136936)
(100,960000024000000)
}
 
\def\preoptimcznormalcnotfstcomavg{
(13,1714025586)
(16,3932998752)
(20,9602047680)
(25,23442499050)
(32,62927979264)
(40,153632761440)
(50,375079983600)
(63,945379256472)
(79,2337503317452)
(100,6001279732200)
}
 
\def\preoptimcznormalcnotfstfitavg{
(13,29054298)
(16,53764512)
(20,104168880)
(25,201827400)
(32,419519232)
(40,812820000)
(50,1574841300)
(63,3124182474)
(79,6110249154)
(100,12288348600)
}
 
\def\preoptimcznormalcnotsndanaavg{
(13,233440584)
(16,535652160)
(20,1307744400)
(25,3192735300)
(32,8570433024)
(40,20923908480)
(50,51083760600)
(63,128755277784)
(79,318354656304)
(100,817340163600)
}
 
\def\preoptimcznormalcnotsndminavg{
(13,141694332)
(16,325085952)
(20,793571040)
(25,1937241600)
(32,5199803136)
(40,12694063680)
(50,30989866200)
(63,78106025088)
(79,193115524332)
(100,495789866400)
}
 
\def\preoptimcznormalcnotsndcomavg{
(13,19396260)
(16,40779648)
(20,90470400)
(25,200461800)
(32,482943360)
(40,1068789600)
(50,2365044000)
(63,5384579760)
(79,12055105488)
(100,27921643200)
}
 
\def\preoptimcznormalcnotsndfitavg{
(13,770172)
(16,1401216)
(20,2666160)
(25,5072400)
(32,10333824)
(40,19661280)
(50,37408800)
(63,72828504)
(79,139838532)
(100,275887200)
}
 
\def\preoptimcznormalcnotfrthanaavg{
(13,262314000)
(16,542549760)
(20,1184744400)
(25,2587078500)
(32,6138249600)
(40,13403848800)
(50,29269452000)
(63,65722107780)
(79,145115692500)
(100,331146030000)
}
 
\def\preoptimcznormalcnotfrthminavg{
(13,205604880)
(16,424989120)
(20,927468000)
(25,2024172000)
(32,4800094080)
(40,10477257600)
(50,22869948000)
(63,51334237080)
(79,113311652880)
(100,258496092000)
}
 
\def\preoptimcznormalcnotfrthcomavg{
(13,18149040)
(16,36402240)
(20,76632000)
(25,160819500)
(32,364122240)
(40,760668000)
(50,1586865000)
(63,3394999440)
(79,7144914840)
(100,15503724000)
}
 
\def\preoptimcznormalcnotfrthfitavg{
(13,170040)
(16,288960)
(20,510000)
(25,903000)
(32,1695360)
(40,2997600)
(50,5298000)
(63,9559620)
(79,17045040)
(100,31122000)
}
 
\def\preoptimcznormalcnotsxthanaavg{
(13,2331650100)
(16,4658563200)
(20,9801336000)
(25,20621422500)
(32,46955347200)
(40,98791248000)
(50,207850860000)
(63,449076455100)
(79,954865273800)
(100,2095005300000)
}
 
\def\preoptimcznormalcnotsxthminavg{
(13,1995548100)
(16,3984139200)
(20,8376270000)
(25,17611215000)
(32,40073193600)
(40,84262560000)
(50,177187560000)
(63,382627343700)
(79,813191484900)
(100,1783355700000)
}
 
\def\preoptimcznormalcnotsxthfitavg{
(13,202800)
(16,326400)
(20,546000)
(25,915000)
(32,1622400)
(40,2712000)
(50,4545000)
(63,7749000)
(79,13058700)
(100,22530000)
}
 
\def\preoptimcznormalcnoteigthanaavg{
(13,34757931000)
(16,68253984000)
(20,140956620000)
(25,291100500000)
(32,649344768000)
(40,1341012720000)
(50,2769430275000)
(63,5869413427500)
(79,12246894064500)
(100,26347409100000)
}
 
\def\preoptimcznormalcnoteigthminavg{
(13,31082980500)
(16,60997104000)
(20,125884860000)
(25,259808550000)
(32,579156288000)
(40,1195379520000)
(50,2467340175000)
(63,5226413314500)
(79,10899899943000)
(100,23438258850000)
}
 
\def\preoptimcznormalcnoteigthfitavg{
(13,429000)
(16,672000)
(20,1050000)
(25,1725000)
(32,2880000)
(40,4620000)
(50,7500000)
(63,12285000)
(79,19908000)
(100,32850000)
}
 
\def\preoptimcznormalcnotlcuavg{
(13,15293875)
(16,26273434)
(20,59722877)
(25,121398913)
(32,232537682)
(40,531572580)
(50,957949080)
(63,1773102110)
(79,4061154075)
(100,8397372240)
}
 
\def\preoptimcznormalcnotspavg{
(13,1203248)
(16,2097200)
(20,4499940)
(25,8326604)
(32,16315236)
(40,34917300)
(50,64776972)
(63,123067130)
(79,265569240)
(100,509992356)
}
 
\def\preoptimcznormalcnotspjaavg{
(13,927224)
(16,1597700)
(20,3389508)
(25,6188800)
(32,11955444)
(40,25252392)
(50,46216428)
(63,86585734)
(79,184311882)
(100,348909616)
}
 
\def\preoptimcznormalcnotspsegmentavg{
(13,2708916)
(16,4854800)
(20,10817604)
(25,20625080)
(32,41826120)
(40,92905128)
(50,177547680)
(63,347943904)
(79,779462640)
(100,1544213904)
}
 
\def\preoptimcznormalcnotspsegmentfitavg{
(13,2561492)
(16,4579200)
(20,10201968)
(25,19481168)
(32,39548460)
(40,87886312)
(50,168089280)
(63,329626820)
(79,739164800)
(100,1465432912)
}
 
\def\preoptimcznormalrzfstanaavg{
(13,64593856964)
(16,182420805568)
(20,556704118480)
(25,1698926142800)
(32,5837465777280)
(40,17814531791040)
(50,54365636569200)
(63,172654222127184)
(79,535315615650924)
(100,1739700370214000)
}
 
\def\preoptimcznormalrzfstminavg{
(13,23762787152)
(16,67108929536)
(20,204800128000)
(25,625000250000)
(32,2147484172288)
(40,6553601024000)
(50,20000002000000)
(63,63515942752752)
(79,196931617424624)
(100,640000016000000)
}
 
\def\preoptimcznormalrzfstcomavg{
(13,1142683724)
(16,2621999168)
(20,6401365120)
(25,15628332700)
(32,41951986176)
(40,102421840960)
(50,250053322400)
(63,630252837648)
(79,1558335544968)
(100,4000853154800)
}
 
\def\preoptimcznormalrzfstfitavg{
(13,19369532)
(16,35843008)
(20,69445920)
(25,134551600)
(32,279679488)
(40,541880000)
(50,1049894200)
(63,2082788316)
(79,4073499436)
(100,8192232400)
}
 
\def\preoptimcznormalrzsndanaavg{
(13,136173674)
(16,312463760)
(20,762850900)
(25,1862428925)
(32,4999419264)
(40,12205613280)
(50,29798860350)
(63,75107245374)
(79,185706882844)
(100,476781762100)
}
 
\def\preoptimcznormalrzsndminavg{
(13,82655027)
(16,189633472)
(20,462916440)
(25,1130057600)
(32,3033218496)
(40,7404870480)
(50,18077421950)
(63,45561847968)
(79,112650722527)
(100,289210755400)
}
 
\def\preoptimcznormalrzsndcomavg{
(13,11314485)
(16,23788128)
(20,52774400)
(25,116936050)
(32,281716960)
(40,623460600)
(50,1379609000)
(63,3141004860)
(79,7032144868)
(100,16287625200)
}
 
\def\preoptimcznormalrzsndfitavg{
(13,449267)
(16,817376)
(20,1555260)
(25,2958900)
(32,6028064)
(40,11469080)
(50,21821800)
(63,42483294)
(79,81572477)
(100,160934200)
}
 
\def\preoptimcznormalrzfrthanaavg{
(13,153016500)
(16,316487360)
(20,691100900)
(25,1509129125)
(32,3580645600)
(40,7818911800)
(50,17073847000)
(63,38337896205)
(79,84650820625)
(100,193168517500)
}
 
\def\preoptimcznormalrzfrthminavg{
(13,119936180)
(16,247910320)
(20,541023000)
(25,1180767000)
(32,2800054880)
(40,6111733600)
(50,13340803000)
(63,29944971630)
(79,66098464180)
(100,150789387000)
}
 
\def\preoptimcznormalrzfrthcomavg{
(13,10586940)
(16,21234640)
(20,44702000)
(25,93811375)
(32,212404640)
(40,443723000)
(50,925671250)
(63,1980416340)
(79,4167866990)
(100,9043839000)
}
 
\def\preoptimcznormalrzfrthfitavg{
(13,99190)
(16,168560)
(20,297500)
(25,526750)
(32,988960)
(40,1748600)
(50,3090500)
(63,5576445)
(79,9942940)
(100,18154500)
}
 
\def\preoptimcznormalrzsxthanaavg{
(13,1360129225)
(16,2717495200)
(20,5717446000)
(25,12029163125)
(32,27390619200)
(40,57628228000)
(50,121246335000)
(63,261961265475)
(79,557004743050)
(100,1222086425000)
}
 
\def\preoptimcznormalrzsxthminavg{
(13,1164069725)
(16,2324081200)
(20,4886157500)
(25,10273208750)
(32,23376029600)
(40,49153160000)
(50,103359410000)
(63,223199283825)
(79,474361699525)
(100,1040290825000)
}
 
\def\preoptimcznormalrzsxthfitavg{
(13,118300)
(16,190400)
(20,318500)
(25,533750)
(32,946400)
(40,1582000)
(50,2651250)
(63,4520250)
(79,7617575)
(100,13142500)
}
 
\def\preoptimcznormalrzeigthanaavg{
(13,20275459750)
(16,39814824000)
(20,82224695000)
(25,169808625000)
(32,378784448000)
(40,782257420000)
(50,1615500993750)
(63,3423824499375)
(79,7144021537625)
(100,15369321975000)
}
 
\def\preoptimcznormalrzeigthminavg{
(13,18131738625)
(16,35581644000)
(20,73432835000)
(25,151554987500)
(32,337841168000)
(40,697304720000)
(50,1439281768750)
(63,3048741100125)
(79,6358274966750)
(100,13672317662500)
}
 
\def\preoptimcznormalrzeigthfitavg{
(13,250250)
(16,392000)
(20,612500)
(25,1006250)
(32,1680000)
(40,2695000)
(50,4375000)
(63,7166250)
(79,11613000)
(100,19162500)
}
 
\def\preoptimcznormalrzlcuavg{
(13,2641346)
(16,4009709)
(20,12410183)
(25,21861035)
(32,35842739)
(40,112179375)
(50,175119150)
(63,277788600)
(79,870496473)
(100,1548544755)
}
 
\def\preoptimcznormalrzspavg{
(13,119322)
(16,180468)
(20,560514)
(25,875082)
(32,1432258)
(40,4465149)
(50,6975771)
(63,11072589)
(79,34783372)
(100,55730272)
}
 
\def\preoptimcznormalrzspjaavg{
(13,91974)
(16,137511)
(20,422248)
(25,650464)
(32,1049585)
(40,3229332)
(50,4977123)
(63,7790415)
(79,24140797)
(100,38127947)
}
 
\def\preoptimcznormalrzspsegmentavg{
(13,284004)
(16,442428)
(20,1421388)
(25,2291940)
(32,3890520)
(40,12530268)
(50,20214576)
(63,33167904)
(79,107634744)
(100,178376688)
}
 
\def\preoptimcznormalrzspsegmentfitavg{
(13,268548)
(16,417312)
(20,1340496)
(25,2164824)
(32,3678660)
(40,11853372)
(50,19137696)
(63,31421820)
(79,102070080)
(100,169276464)
}
 
\def\preoptimctnormaltotalfstanaavg{
(13,47260343885622)
(16,141265531701824)
(20,446434950212074)
(25,1395715396242198)
(32,5056431098360490)
(40,15575913189848436)
(50,48590157533907608)
(63,158560978678151552)
(79,511523208186702144)
(100,1730565203569807360)
}
 
\def\preoptimctnormaltotalfstminavg{
(13,17653909889616)
(16,50955365885338)
(20,153870384084480)
(25,500912600182500)
(32,1745589007247770)
(40,5626544567573760)
(50,17740040887002000)
(63,57235928134176792)
(79,178765295338061696)
(100,619908487748856064)
}
 
\def\preoptimctnormaltotalfstcomavg{
(13,751814494133)
(16,1762565712089)
(20,4465084206251)
(25,11417405141426)
(32,30311273281114)
(40,78516917426896)
(50,191753449474848)
(63,492872194722907)
(79,1284882274304587)
(100,3217791158686036)
}
 
\def\preoptimctnormaltotalfstfitavg{
(13,5561216132)
(16,10208648725)
(20,20080808612)
(25,40005689022)
(32,84460583381)
(40,167953673950)
(50,347077174319)
(63,678341012429)
(79,1367875954280)
(100,2799900228510)
}
 
\def\preoptimctnormaltotalsndanaavg{
(13,59080593082)
(16,136855035060)
(20,350217251353)
(25,859060730304)
(32,2384560540080)
(40,5993508411632)
(50,15664452643297)
(63,39995414151540)
(79,102213286969661)
(100,251852057483364)
}
 
\def\preoptimctnormaltotalsndminavg{
(13,34961799406)
(16,83513945112)
(20,207676314582)
(25,540854292763)
(32,1380319136250)
(40,3625748997980)
(50,8955671979901)
(63,23988774844990)
(79,60513524477754)
(100,158740221777242)
}
 
\def\preoptimctnormaltotalsndcomavg{
(13,4343749830)
(16,8790832645)
(20,22063215439)
(25,45985482629)
(32,119732803942)
(40,274407911743)
(50,590356506191)
(63,1475290076839)
(79,3275047277747)
(100,7787183791464)
}
 
\def\preoptimctnormaltotalsndfitavg{
(13,105857179)
(16,194023168)
(20,390816842)
(25,758807369)
(32,1619315603)
(40,3240506632)
(50,6265762017)
(63,12421517895)
(79,24141234435)
(100,49322837729)
}
 
\def\preoptimctnormaltotalfrthanaavg{
(13,55583770145)
(16,117139830531)
(20,242561386591)
(25,590476618930)
(32,1492420269566)
(40,3280992214612)
(50,7059461661616)
(63,16279713951272)
(79,37443319864546)
(100,87778932917643)
}
 
\def\preoptimctnormaltotalfrthminavg{
(13,42781761523)
(16,89235194646)
(20,194179134049)
(25,450828244337)
(32,1115866766699)
(40,2508172926758)
(50,5573653510828)
(63,12610001473503)
(79,28687755908510)
(100,67733961074872)
}
 
\def\preoptimctnormaltotalfrthcomavg{
(13,3506641108)
(16,7242540015)
(20,15277543114)
(25,33307621453)
(32,77423801835)
(40,167148848350)
(50,348913988924)
(63,777499138046)
(79,1697220157624)
(100,3668684155163)
}
 
\def\preoptimctnormaltotalfrthfitavg{
(13,21535697)
(16,38175890)
(20,67141075)
(25,122982339)
(32,247109846)
(40,425700667)
(50,806741648)
(63,1495811444)
(79,2652236273)
(100,5003015035)
}
 
\def\preoptimctnormaltotalsxthanaavg{
(13,517453753142)
(16,1029268089018)
(20,2232543055884)
(25,4828695504885)
(32,10892601131340)
(40,23785983887270)
(50,50597130300512)
(63,110851168242978)
(79,244598695952240)
(100,554646372164712)
}
 
\def\preoptimctnormaltotalsxthminavg{
(13,408873018526)
(16,852646769713)
(20,1813388883103)
(25,3971989950183)
(32,9632942194890)
(40,19872066461442)
(50,43280460938841)
(63,98438276643426)
(79,212194011163776)
(100,473110165586685)
}
 
\def\preoptimctnormaltotalsxthfitavg{
(13,26253396)
(16,42506256)
(20,74603529)
(25,124180750)
(32,223451969)
(40,403030094)
(50,679782318)
(63,1228908990)
(79,2029116457)
(100,3610901124)
}
 
\def\preoptimctnormaltotaleigthanaavg{
(13,7712672501160)
(16,15831841109120)
(20,32348783835825)
(25,70966162086860)
(32,160148800244650)
(40,338881949896655)
(50,717333828438080)
(63,1540400212311360)
(79,3249654475114170)
(100,7294278201404250)
}
 
\def\preoptimctnormaltotaleigthminavg{
(13,6920033926920)
(16,14097904068040)
(20,29782300872495)
(25,62618877733300)
(32,138433045104780)
(40,301212578041520)
(50,625388959783280)
(63,1338919176698230)
(79,2885666353735580)
(100,6466291422996410)
}
 
\def\preoptimctnormaltotaleigthfitavg{
(13,58941850)
(16,91924140)
(20,146588400)
(25,258129460)
(32,424756800)
(40,698339565)
(50,1159474000)
(63,1827286500)
(79,3140732280)
(100,5303369700)
}
 
\def\preoptimctnormaltotallcuavg{
(13,714339441)
(16,1258744526)
(20,3972645353)
(25,6393946285)
(32,10767711506)
(40,34478477298)
(50,55034890434)
(63,99737485361)
(79,319332877327)
(100,522927657872)
}
 
\def\preoptimctnormaltotalspavg{
(13,29569510)
(16,46490865)
(20,145287646)
(25,234458468)
(32,398747037)
(40,1255287617)
(50,2025062258)
(63,3320912804)
(79,10535421877)
(100,17382712665)
}
 
\def\preoptimctnormaltotalspjaavg{
(13,22579337)
(16,35034919)
(20,108185228)
(25,172383737)
(32,289860939)
(40,897046335)
(50,1428286995)
(63,2310358735)
(79,7226204699)
(100,11767713971)
}
 
\def\preoptimctnormaltotalspsegmentavg{
(13,76303828)
(16,123402508)
(20,401997876)
(25,670269092)
(32,1178407289)
(40,3844462406)
(50,6386089186)
(63,10382683530)
(79,35532475124)
(100,60712341531)
}
 
\def\preoptimctnormaltotalspsegmentfitavg{
(13,72209585)
(16,116028054)
(20,378303822)
(25,629145180)
(32,1106373419)
(40,3600332442)
(50,5978755482)
(63,10118821509)
(79,33179512838)
(100,56604377934)
}
 
\def\preoptimctnormalcnotfstanaavg{
(13,193781570892)
(16,547262416704)
(20,1670112355440)
(25,5096778428400)
(32,17512397331840)
(40,53443595373120)
(50,163096909707600)
(63,517962666381552)
(79,1605946846952772)
(100,5219101110641400)
}
 
\def\preoptimctnormalcnotfstminavg{
(13,71288308728)
(16,201326690304)
(20,614400192000)
(25,1875000375000)
(32,6442451730432)
(40,19660801536000)
(50,60000003000000)
(63,190547822257128)
(79,590794840440936)
(100,1920000024000000)
}
 
\def\preoptimctnormalcnotfstcomavg{
(13,3427685586)
(16,7865158848)
(20,19202047920)
(25,46879999650)
(32,125842540608)
(40,307232764800)
(50,750079992000)
(63,1890556937262)
(79,4674508229592)
(100,12001279866600)
}
 
\def\preoptimctnormalcnotfstfitavg{
(13,29054298)
(16,53764512)
(20,104168880)
(25,201827400)
(32,419519232)
(40,812820000)
(50,1574841300)
(63,3124182474)
(79,6110249154)
(100,12288348600)
}
 
\def\preoptimctnormalcnotsndanaavg{
(13,330134844)
(16,757526400)
(20,1849429680)
(25,4515209400)
(32,12120422400)
(40,29590874880)
(50,72243346800)
(63,182087459568)
(79,450221472420)
(100,1155893544000)
}
 
\def\preoptimctnormalcnotsndminavg{
(13,200342376)
(16,459659712)
(20,1122119760)
(25,2739362700)
(32,7352980608)
(40,17950844640)
(50,43823804400)
(63,110453628348)
(79,273096791232)
(100,701132872800)
}
 
\def\preoptimctnormalcnotsndcomavg{
(13,26861952)
(16,56414400)
(20,124985760)
(25,276515700)
(32,664914816)
(40,1468773600)
(50,3243663600)
(63,7368759972)
(79,16459949568)
(100,38029857600)
}
 
\def\preoptimctnormalcnotsndfitavg{
(13,770172)
(16,1401216)
(20,2666160)
(25,5072400)
(32,10333824)
(40,19661280)
(50,37408800)
(63,72828504)
(79,139838532)
(100,275887200)
}
 
\def\preoptimctnormalcnotfrthanaavg{
(13,311946180)
(16,645204480)
(20,1408906800)
(25,3076572000)
(32,7299649920)
(40,15939952800)
(50,34807440000)
(63,78157196460)
(79,172572611760)
(100,393801210000)
}
 
\def\preoptimctnormalcnotfrthminavg{
(13,244258560)
(16,504936960)
(20,1102047600)
(25,2405394000)
(32,5704600320)
(40,12452390400)
(50,27182958000)
(63,61018729380)
(79,134695186440)
(100,307292106000)
}
 
\def\preoptimctnormalcnotfrthcomavg{
(13,21532680)
(16,43185600)
(20,90898800)
(25,190723500)
(32,431702400)
(40,901560000)
(50,1880067000)
(63,4020411780)
(79,8456643480)
(100,18338448000)
}
 
\def\preoptimctnormalcnotfrthfitavg{
(13,170040)
(16,288960)
(20,510000)
(25,903000)
(32,1695360)
(40,2997600)
(50,5298000)
(63,9559620)
(79,17045040)
(100,31122000)
}
 
\def\preoptimctnormalcnotsxthanaavg{
(13,2617188600)
(16,5229057600)
(20,11001624000)
(25,23146762500)
(32,52705603200)
(40,110889420000)
(50,233304690000)
(63,504071278200)
(79,1071800031000)
(100,2351563920000)
}
 
\def\preoptimctnormalcnotsxthminavg{
(13,2237266200)
(16,4467076800)
(20,9392334000)
(25,19748947500)
(32,44940816000)
(40,94503708000)
(50,198734235000)
(63,429180330600)
(79,912176217600)
(100,2000530950000)
}
 
\def\preoptimctnormalcnotsxthfitavg{
(13,202800)
(16,326400)
(20,546000)
(25,915000)
(32,1622400)
(40,2712000)
(50,4545000)
(63,7749000)
(79,13058700)
(100,22530000)
}
 
\def\preoptimctnormalcnoteigthanaavg{
(13,37903788000)
(16,74431488000)
(20,153714270000)
(25,317447362500)
(32,708115488000)
(40,1462384740000)
(50,3020085075000)
(63,6400640736000)
(79,13355332636500)
(100,28732053300000)
}
 
\def\preoptimctnormalcnoteigthminavg{
(13,33859429500)
(16,66449136000)
(20,137144190000)
(25,283060875000)
(32,631023936000)
(40,1302494880000)
(50,2688551400000)
(63,5695236981000)
(79,11878125738000)
(100,25542761550000)
}
 
\def\preoptimctnormalcnoteigthfitavg{
(13,429000)
(16,672000)
(20,1050000)
(25,1725000)
(32,2880000)
(40,4620000)
(50,7500000)
(63,12285000)
(79,19908000)
(100,32850000)
}
 
\def\preoptimctnormalcnotlcuavg{
(13,15293875)
(16,29561472)
(20,67194252)
(25,121398913)
(32,232537682)
(40,531572580)
(50,957949080)
(63,1970166845)
(79,4512477222)
(100,8397372240)
}
 
\def\preoptimctnormalcnotspavg{
(13,1204550)
(16,2097200)
(20,4499940)
(25,8326604)
(32,16315236)
(40,34925324)
(50,64776972)
(63,123067130)
(79,265569240)
(100,509992356)
}
 
\def\preoptimctnormalcnotspjaavg{
(13,929828)
(16,1602200)
(20,3395700)
(25,6201030)
(32,11981796)
(40,25300536)
(50,46306944)
(63,86756794)
(79,184664052)
(100,349566216)
}
 
\def\preoptimctnormalcnotspsegmentavg{
(13,2837912)
(16,5109200)
(20,11374608)
(25,21699664)
(32,43938132)
(40,97638784)
(50,186533160)
(63,351010200)
(79,819092560)
(100,1622595668)
}
 
\def\preoptimctnormalcnotspsegmentfitavg{
(13,2690488)
(16,4812400)
(20,10729656)
(25,20417096)
(32,41329176)
(40,91650424)
(50,174980400)
(63,342456848)
(79,766326880)
(100,1516267944)
}
 
\def\preoptimctnormaltfstanaavg{
(13,18705387225129)
(16,54384202659960)
(20,173992305189739)
(25,550723898450048)
(32,1961497953649404)
(40,6139956527099846)
(50,18707215543461720)
(63,61934084074550104)
(79,199943770519136896)
(100,664683841046846208)
}
 
\def\preoptimctnormaltfstminavg{
(13,6753013532429)
(16,19427186753126)
(20,61030419072000)
(25,193415038683000)
(32,703113075730022)
(40,2182643882519040)
(50,6886720344336000)
(63,22417900879074324)
(79,70045334268142872)
(100,238936322986703968)
}
 
\def\preoptimctnormaltfstcomavg{
(13,288892372338)
(16,674961715139)
(20,1718263254708)
(25,4319460661085)
(32,11936164976669)
(40,30568635988384)
(50,75101009119008)
(63,190576141633416)
(79,497773983094958)
(100,1297938417572790)
}
 
\def\preoptimctnormaltfstfitavg{
(13,2084459636)
(16,3875525240)
(20,7771692907)
(25,15162350701)
(32,32155275134)
(40,65131266600)
(50,130094490110)
(63,262837967440)
(79,529219765334)
(100,1082365936920)
}
 
\def\preoptimctnormaltsndanaavg{
(13,22670105788)
(16,52266165240)
(20,135532371716)
(25,332367574074)
(32,921101600640)
(40,2347542740480)
(50,6042915148664)
(63,15109501783867)
(79,38862091678218)
(100,96702439188888)
}
 
\def\preoptimctnormaltsndminavg{
(13,13525165174)
(16,32082332649)
(20,78681167372)
(25,203980251609)
(32,535289328908)
(40,1402409737500)
(50,3482473224581)
(63,9255897173416)
(79,23523543566445)
(100,62032263499065)
}
 
\def\preoptimctnormaltsndcomavg{
(13,1660206387)
(16,3443864060)
(20,8291347009)
(25,17594140960)
(32,47175013576)
(40,104610951704)
(50,232008445096)
(63,558567601137)
(79,1269166288589)
(100,3020052404970)
}
 
\def\preoptimctnormaltsndfitavg{
(13,41192353)
(16,73852841)
(20,143812670)
(25,289978963)
(32,614012140)
(40,1211904915)
(50,2353437486)
(63,4670927391)
(79,9140483690)
(100,18655952276)
}
 
\def\preoptimctnormaltfrthanaavg{
(13,21358555014)
(16,44599490845)
(20,94987557185)
(25,229244814541)
(32,554864639544)
(40,1257720722414)
(50,2796833495904)
(63,6305697798584)
(79,14318968529035)
(100,33658031898216)
}
 
\def\preoptimctnormaltfrthminavg{
(13,16506241920)
(16,34850328198)
(20,76391000838)
(25,175588630493)
(32,426286955038)
(40,958941981517)
(50,2117269725437)
(63,4867699723449)
(79,11054462367837)
(100,25926316927782)
}
 
\def\preoptimctnormaltfrthcomavg{
(13,1328268211)
(16,2726612826)
(20,5856033992)
(25,12490558304)
(32,29685206094)
(40,63334289480)
(50,134481443186)
(63,299111403945)
(79,648315548450)
(100,1429352429901)
}
 
\def\preoptimctnormaltfrthfitavg{
(13,8348441)
(16,14252832)
(20,26282680)
(25,48511809)
(32,92494603)
(40,162923557)
(50,303965333)
(63,560273648)
(79,1036453504)
(100,1905153978)
}
 
\def\preoptimctnormaltsxthanaavg{
(13,194806071460)
(16,394667479908)
(20,846629974920)
(25,1815255529830)
(32,4201756569108)
(40,9181847273270)
(50,19848225411644)
(63,43218271279728)
(79,93060824024960)
(100,213281360561520)
}
 
\def\preoptimctnormaltsxthminavg{
(13,161929885608)
(16,326213867166)
(20,707164480750)
(25,1504395824760)
(32,3652452468360)
(40,7667054328804)
(50,16694762153818)
(63,37034447594508)
(79,81157357267968)
(100,181775310659690)
}
 
\def\preoptimctnormaltsxthfitavg{
(13,9976928)
(16,16053848)
(20,28143206)
(25,47892320)
(32,86868704)
(40,152200604)
(50,258710490)
(63,458056100)
(79,769292976)
(100,1381025916)
}
 
\def\preoptimctnormalteigthanaavg{
(13,3009871772640)
(16,6062196592640)
(20,12746704601660)
(25,27099931094530)
(32,61945647842120)
(40,131246593107100)
(50,275427329381890)
(63,594388559454720)
(79,1276121756482650)
(100,2827923613999200)
}
 
\def\preoptimctnormalteigthminavg{
(13,2694724401520)
(16,5422775553260)
(20,11368019053290)
(25,24103199628000)
(32,53536333656880)
(40,116054030467840)
(50,242926750298400)
(63,510485485412620)
(79,1120031076708920)
(100,2492668716992170)
}
 
\def\preoptimctnormalteigthfitavg{
(13,22413820)
(16,34297200)
(20,55443150)
(25,95494160)
(32,159468600)
(40,266620200)
(50,444783000)
(63,708758700)
(79,1203671280)
(100,2031774690)
}
 
\def\preoptimctnormaltlcuavg{
(13,267604832)
(16,470285430)
(20,1499261106)
(25,2406388900)
(32,4055118440)
(40,13067087389)
(50,20820441486)
(63,37694248872)
(79,121388793869)
(100,198622026823)
}
 
\def\preoptimctnormaltspavg{
(13,10891518)
(16,17058610)
(20,53923117)
(25,86875904)
(32,147581816)
(40,468634093)
(50,754038193)
(63,1234384663)
(79,3950463032)
(100,6511333865)
}
 
\def\preoptimctnormaltspjaavg{
(13,8301092)
(16,12861281)
(20,40127812)
(25,63814928)
(32,107003393)
(40,335262063)
(50,531291883)
(63,858123333)
(79,2708033579)
(100,4401254761)
}
 
\def\preoptimctnormaltspsegmentavg{
(13,28204114)
(16,45525382)
(20,149546374)
(25,249171538)
(32,437115024)
(40,1438001248)
(50,2388303654)
(63,3874299660)
(79,13383018834)
(100,22819433606)
}
 
\def\preoptimctnormaltspsegmentfitavg{
(13,26679631)
(16,42780783)
(20,140743690)
(25,233796482)
(32,410119717)
(40,1346826700)
(50,2235355444)
(63,3776805202)
(79,12492177266)
(100,21275296099)
}
 
\def\preoptimczlargetotalfstanaavg{
(13,419860070266)
(17,1605581736772)
(23,7278250395234)
(31,32373993986646)
(42,147786330051360)
(57,680396531306544)
(78,3264831905974704)
(106,15132724963071974)
(145,72481608689772880)
(197,335520409626236288)
(269,1592755863473543168)
(367,7528667788363251712)
(500,35337663769967583232)
}
 
\def\preoptimczlargetotalfstminavg{
(13,154458116488)
(17,590661022952)
(23,2677519953368)
(31,11909729914264)
(42,54367560217152)
(57,250303914972072)
(78,1201064586441408)
(106,5567018526681664)
(145,26664494017057000)
(197,123431061590030800)
(269,585942139001555328)
(367,2769642103889321472)
(500,13000000012999999488)
}
 
\def\preoptimczlargetotalfstcomavg{
(13,7427444206)
(17,21720092160)
(23,72774179296)
(31,240166674852)
(42,809213521116)
(57,2745145645554)
(78,9625967234676)
(106,32831402032312)
(145,114957676416050)
(197,391679527900570)
(269,1361680091770362)
(367,4717698961817712)
(500,16253465930208000)
}
 
\def\preoptimczlargetotalfstfitavg{
(13,125901958)
(17,278841004)
(23,683075666)
(31,1654647072)
(42,4070256372)
(57,10062956358)
(78,25496622372)
(106,63288169412)
(145,160181528130)
(197,397298804006)
(269,1000246042744)
(367,2511846032306)
(500,6281619162000)
}
 
\def\preoptimczlargetotalsndanaavg{
(13,992122482)
(17,2901265509)
(23,9720825777)
(31,32080309731)
(42,108090850644)
(57,366683354955)
(78,1285790406498)
(106,4385460780060)
(145,15355493527635)
(197,52318667554611)
(269,181886677658076)
(367,630167537591535)
(500,2171059808160000)
}
 
\def\preoptimczlargetotalsndminavg{
(13,602200911)
(17,1760708802)
(23,5898459825)
(31,19463768469)
(42,65575528578)
(57,222442474293)
(78,779968106892)
(106,2660159379942)
(145,9314199517590)
(197,31734435562815)
(269,110323817399829)
(367,382226016072171)
(500,1316839837470000)
}
 
\def\preoptimczlargetotalsndcomavg{
(13,82434105)
(17,215237952)
(23,632933205)
(31,1833111003)
(42,5403861162)
(57,16024854942)
(78,48961323450)
(106,146061610956)
(145,446523618060)
(197,1333897218981)
(269,4062759558834)
(367,12354869229738)
(500,37440126502500)
}
 
\def\preoptimczlargetotalsndfitavg{
(13,3273231)
(17,7092927)
(23,16952196)
(31,40078350)
(42,96180084)
(57,231952437)
(78,572887692)
(106,1386974172)
(145,3422051040)
(197,8278677765)
(269,20320884618)
(367,49756617705)
(500,121339225500)
}
 
\def\preoptimczlargetotalfrthanaavg{
(13,1114834500)
(17,2850886740)
(23,8212131945)
(31,23343955110)
(42,67574220210)
(57,196775251515)
(78,589847298300)
(106,1725754487790)
(145,5166500628900)
(197,15102201801450)
(269,44930521250760)
(367,133272241272165)
(500,393372043320000)
}
 
\def\preoptimczlargetotalfrthminavg{
(13,873820740)
(17,2232767760)
(23,6426567885)
(31,18256099485)
(42,52815443940)
(57,153720241380)
(78,460582049070)
(106,1347052618950)
(145,4031441228475)
(197,11781086178270)
(269,35041526025345)
(367,103918498196445)
(500,306677039790000)
}
 
\def\preoptimczlargetotalfrthcomavg{
(13,77133420)
(17,189452505)
(23,518290050)
(31,1393327395)
(42,3797112690)
(57,10380940605)
(78,29120630760)
(106,79796019840)
(145,223425171975)
(197,612043799055)
(269,1706371072065)
(367,4750691946735)
(500,13188039637500)
}
 
\def\preoptimczlargetotalfrthfitavg{
(13,722670)
(17,1434885)
(23,3102585)
(31,6640200)
(42,14426370)
(57,31468275)
(78,70112250)
(106,153476340)
(145,341685975)
(197,747496800)
(269,1656569250)
(367,3663104070)
(500,8071005000)
}
 
\def\preoptimczlargetotalsxthanaavg{
(13,9909512925)
(17,24232823400)
(23,66374351625)
(31,179520099450)
(42,494012994300)
(57,1367173364625)
(78,3889460758950)
(106,10812520327500)
(145,30723266815125)
(197,85335185887725)
(269,241035855603375)
(367,678881955208950)
(500,1903154584687500)
}
 
\def\preoptimczlargetotalsxthminavg{
(13,8481079425)
(17,20720389650)
(23,56699418300)
(31,153221389725)
(42,421309980000)
(57,1165131414000)
(78,3312464688000)
(106,9203080065300)
(145,26135929792875)
(197,72558877145025)
(269,204858045926175)
(367,576758158737150)
(500,1616289354375000)
}
 
\def\preoptimczlargetotalsxthfitavg{
(13,861900)
(17,1603950)
(23,3225750)
(31,6403050)
(42,12905550)
(57,26090325)
(78,53901900)
(106,109471500)
(145,225732375)
(197,458394375)
(269,941809350)
(367,1930658550)
(500,3945487500)
}
 
\def\preoptimczlargetotaleigthanaavg{
(13,147721206750)
(17,353253081000)
(23,943500937125)
(31,2489157427125)
(42,6678624802500)
(57,18018562025250)
(78,49938363013500)
(106,135322778463000)
(145,374604379346250)
(197,1014243136670625)
(269,2791397933310000)
(367,7661107413222374)
(500,20930488118437500)
}
 
\def\preoptimczlargetotaleigthminavg{
(13,132102667125)
(17,315635576625)
(23,842276608875)
(31,2220287033625)
(42,5952612645000)
(57,16048224315000)
(78,44447005512000)
(106,120367291648500)
(145,333007737288750)
(197,901138285669125)
(269,2478862859724750)
(367,6800179917126375)
(500,18570433640437500)
}
 
\def\preoptimczlargetotaleigthfitavg{
(13,1823250)
(17,3142875)
(23,6011625)
(31,11462250)
(42,21955500)
(57,42151500)
(78,82046250)
(106,158125500)
(145,309665625)
(197,596540625)
(269,1160970375)
(367,2257738125)
(500,4373250000)
}
 
\def\preoptimczlargecnotfstanaavg{
(13,96890785446)
(17,370518862332)
(23,1679596245054)
(31,7470921689226)
(42,34104537704160)
(57,157014584147664)
(78,753422747532624)
(106,3492167299170456)
(145,16726525082255280)
(197,77427786836823776)
(269,367559045416971520)
(367,1737384874237673728)
(500,8154845485377134592)
}
 
\def\preoptimczlargecnotfstminavg{
(13,35644180728)
(17,136306389912)
(23,617889220008)
(31,2748399210984)
(42,12546360050112)
(57,57762441916632)
(78,277168750717248)
(106,1284696583080384)
(145,6153344773167000)
(197,28484091136160952)
(269,135217416692666624)
(367,639148177820612736)
(500,3000000002999999488)
}
 
\def\preoptimczlargecnotfstcomavg{
(13,1714025586)
(17,5012328960)
(23,16794041376)
(31,55423078812)
(42,186741581796)
(57,633495148974)
(78,2221377054156)
(106,7576477392072)
(145,26528694557550)
(197,90387583361670)
(269,314233867331622)
(367,1088699760419472)
(500,3750799830048000)
}
 
\def\preoptimczlargecnotfstfitavg{
(13,29054298)
(17,64347924)
(23,157632846)
(31,381841632)
(42,939289932)
(57,2322220698)
(78,5883835932)
(106,14604962172)
(145,36964968030)
(197,91684339386)
(269,230826009864)
(367,579656776686)
(500,1449604422000)
}
 
\def\preoptimczlargecnotsndanaavg{
(13,233440584)
(17,682650708)
(23,2287253124)
(31,7548308172)
(42,25433141328)
(57,86278436460)
(78,302538919176)
(106,1031873124720)
(145,3613057300620)
(197,12310274718732)
(269,42796865331312)
(367,148274714727420)
(500,510837601920000)
}
 
\def\preoptimczlargecnotsndminavg{
(13,141694332)
(17,414284424)
(23,1387872900)
(31,4579710228)
(42,15429536136)
(57,52339405716)
(78,183521907504)
(106,625919854104)
(145,2191576357080)
(197,7466926014780)
(269,25958545270548)
(367,89935533193452)
(500,309844667640000)
}
 
\def\preoptimczlargecnotsndcomavg{
(13,19396260)
(17,50644224)
(23,148925460)
(31,431320236)
(42,1271496744)
(57,3770554104)
(78,11520311400)
(106,34367437872)
(145,105064380720)
(197,313858169172)
(269,955943425608)
(367,2907028054056)
(500,8809441530000)
}
 
\def\preoptimczlargecnotsndfitavg{
(13,770172)
(17,1668924)
(23,3988752)
(31,9430200)
(42,22630608)
(57,54577044)
(78,134797104)
(106,326346864)
(145,805188480)
(197,1947924180)
(269,4781384616)
(367,11707439460)
(500,28550406000)
}
 
\def\preoptimczlargecnotfrthanaavg{
(13,262314000)
(17,670796880)
(23,1932266340)
(31,5492695320)
(42,15899816520)
(57,46300059180)
(78,138787599600)
(106,406059879480)
(145,1215647206800)
(197,3553459247400)
(269,10571887353120)
(367,31358174416980)
(500,92558127840000)
}
 
\def\preoptimczlargecnotfrthminavg{
(13,205604880)
(17,525357120)
(23,1512133620)
(31,4295552820)
(42,12427163280)
(57,36169468560)
(78,108372246840)
(106,316953557400)
(145,948574406700)
(197,2772020277240)
(269,8245064947140)
(367,24451411340340)
(500,72159303480000)
}
 
\def\preoptimczlargecnotfrthcomavg{
(13,18149040)
(17,44577060)
(23,121950600)
(31,327841740)
(42,893438280)
(57,2442574260)
(78,6851913120)
(106,18775534080)
(145,52570628700)
(197,144010305660)
(269,401499075780)
(367,1117809869820)
(500,3103068150000)
}
 
\def\preoptimczlargecnotfrthfitavg{
(13,170040)
(17,337620)
(23,730020)
(31,1562400)
(42,3394440)
(57,7404300)
(78,16497000)
(106,36112080)
(145,80396700)
(197,175881600)
(269,389781000)
(367,861906840)
(500,1899060000)
}
 
\def\preoptimczlargecnotsxthanaavg{
(13,2331650100)
(17,5701840800)
(23,15617494500)
(31,42240023400)
(42,116238351600)
(57,321687850500)
(78,915167237400)
(106,2544122430000)
(145,7229003956500)
(197,20078867267700)
(269,56714318965500)
(367,159736930637400)
(500,447801078750000)
}
 
\def\preoptimczlargecnotsxthminavg{
(13,1995548100)
(17,4875385800)
(23,13341039600)
(31,36052091700)
(42,99131760000)
(57,274148568000)
(78,779403456000)
(106,2165430603600)
(145,6149630539500)
(197,17072676975300)
(269,48201893159100)
(367,135707802055800)
(500,380303377500000)
}
 
\def\preoptimczlargecnotsxthfitavg{
(13,202800)
(17,377400)
(23,759000)
(31,1506600)
(42,3036600)
(57,6138900)
(78,12682800)
(106,25758000)
(145,53113500)
(197,107857500)
(269,221602200)
(367,454272600)
(500,928350000)
}
 
\def\preoptimczlargecnoteigthanaavg{
(13,34757931000)
(17,83118372000)
(23,222000220500)
(31,585684100500)
(42,1571441130000)
(57,4239661653000)
(78,11750203062000)
(106,31840653756000)
(145,88142206905000)
(197,238645443922500)
(269,656799513720000)
(367,1802613508993500)
(500,4924820733750000)
}
 
\def\preoptimczlargecnoteigthminavg{
(13,31082980500)
(17,74267194500)
(23,198182731500)
(31,522420478500)
(42,1400614740000)
(57,3776052780000)
(78,10458118944000)
(106,28321715682000)
(145,78354761715000)
(197,212032537804500)
(269,583261849347000)
(367,1600042333441500)
(500,4369513797750000)
}
 
\def\preoptimczlargecnoteigthfitavg{
(13,429000)
(17,739500)
(23,1414500)
(31,2697000)
(42,5166000)
(57,9918000)
(78,19305000)
(106,37206000)
(145,72862500)
(197,140362500)
(269,273169500)
(367,531232500)
(500,1029000000)
}
 
\def\preoptimczlargerzfstanaavg{
(13,64593856964)
(17,247012574888)
(23,1119730830036)
(31,4980614459484)
(42,22736358469440)
(57,104676389431776)
(78,502281831688416)
(106,2328111532780304)
(145,11151016721503520)
(197,51618524557882512)
(269,245039363611314336)
(367,1158256582825115904)
(500,5436563656918089728)
}
 
\def\preoptimczlargerzfstminavg{
(13,23762787152)
(17,90870926608)
(23,411926146672)
(31,1832266140656)
(42,8364240033408)
(57,38508294611088)
(78,184779167144832)
(106,856464388720256)
(145,4102229848778000)
(197,18989394090773968)
(269,90144944461777744)
(367,426098785213741824)
(500,2000000001999999744)
}
 
\def\preoptimczlargerzfstcomavg{
(13,1142683724)
(17,3341552640)
(23,11196027584)
(31,36948719208)
(42,124494387864)
(57,422330099316)
(78,1480918036104)
(106,5050984928048)
(145,17685796371700)
(197,60258388907780)
(269,209489244887748)
(367,725799840279648)
(500,2500533220032000)
}
 
\def\preoptimczlargerzfstfitavg{
(13,19369532)
(17,42898616)
(23,105088564)
(31,254561088)
(42,626193288)
(57,1548147132)
(78,3922557288)
(106,9736641448)
(145,24643312020)
(197,61122892924)
(269,153884006576)
(367,386437851124)
(500,966402948000)
}
 
\def\preoptimczlargerzsndanaavg{
(13,136173674)
(17,398212913)
(23,1334230989)
(31,4403179767)
(42,14835999108)
(57,50329087935)
(78,176481036186)
(106,601925989420)
(145,2107616758695)
(197,7180993585927)
(269,24964838109932)
(367,86493583590995)
(500,297988601120000)
}
 
\def\preoptimczlargerzsndminavg{
(13,82655027)
(17,241665914)
(23,809592525)
(31,2671497633)
(42,9000562746)
(57,30531320001)
(78,107054446044)
(106,365119914894)
(145,1278419541630)
(197,4355706841955)
(269,15142484741153)
(367,52462394362847)
(500,180742722790000)
}
 
\def\preoptimczlargerzsndcomavg{
(13,11314485)
(17,29542464)
(23,86873185)
(31,251603471)
(42,741706434)
(57,2199489894)
(78,6720181650)
(106,20047672092)
(145,61287555420)
(197,183083932017)
(269,557633664938)
(367,1695766364866)
(500,5138840892500)
}
 
\def\preoptimczlargerzsndfitavg{
(13,449267)
(17,973539)
(23,2326772)
(31,5500950)
(42,13201188)
(57,31836609)
(78,78631644)
(106,190369004)
(145,469693280)
(197,1136289105)
(269,2789141026)
(367,6829339685)
(500,16654403500)
}
 
\def\preoptimczlargerzfrthanaavg{
(13,153016500)
(17,391298180)
(23,1127155365)
(31,3204072270)
(42,9274892970)
(57,27008367855)
(78,80959433100)
(106,236868263030)
(145,709127537300)
(197,2072851227650)
(269,6166934289320)
(367,18292268409905)
(500,53992241240000)
}
 
\def\preoptimczlargerzfrthminavg{
(13,119936180)
(17,306458320)
(23,882077945)
(31,2505739145)
(42,7249178580)
(57,21098856660)
(78,63217143990)
(106,184889575150)
(145,553335070575)
(197,1617011828390)
(269,4809621219165)
(367,14263323281865)
(500,42092927030000)
}
 
\def\preoptimczlargerzfrthcomavg{
(13,10586940)
(17,26003285)
(23,71137850)
(31,191241015)
(42,521172330)
(57,1424834985)
(78,3996949320)
(106,10952394880)
(145,30666200075)
(197,84006011635)
(269,234207794205)
(367,652055757395)
(500,1810123087500)
}
 
\def\preoptimczlargerzfrthfitavg{
(13,99190)
(17,196945)
(23,425845)
(31,911400)
(42,1980090)
(57,4319175)
(78,9623250)
(106,21065380)
(145,46898075)
(197,102597600)
(269,227372250)
(367,502778990)
(500,1107785000)
}
 
\def\preoptimczlargerzsxthanaavg{
(13,1360129225)
(17,3326073800)
(23,9110205125)
(31,24640013650)
(42,67805705100)
(57,187651246125)
(78,533847555150)
(106,1484071417500)
(145,4216918974625)
(197,11712672572825)
(269,33083352729875)
(367,93179876205150)
(500,261217295937500)
}
 
\def\preoptimczlargerzsxthminavg{
(13,1164069725)
(17,2843975050)
(23,7782273100)
(31,21030386825)
(42,57826860000)
(57,159919998000)
(78,454652016000)
(106,1263167852100)
(145,3587284481375)
(197,9959061568925)
(269,28117771009475)
(367,79162884532550)
(500,221843636875000)
}
 
\def\preoptimczlargerzsxthfitavg{
(13,118300)
(17,220150)
(23,442750)
(31,878850)
(42,1771350)
(57,3581025)
(78,7398300)
(106,15025500)
(145,30982875)
(197,62916875)
(269,129267950)
(367,264992350)
(500,541537500)
}
 
\def\preoptimczlargerzeigthanaavg{
(13,20275459750)
(17,48485717000)
(23,129500128625)
(31,341649058625)
(42,916673992500)
(57,2473135964250)
(78,6854285119500)
(106,18573714691000)
(145,51416287361250)
(197,139209842288125)
(269,383133049670000)
(367,1051524546912875)
(500,2872812094687500)
}
 
\def\preoptimczlargerzeigthminavg{
(13,18131738625)
(17,43322530125)
(23,115606593375)
(31,304745279125)
(42,817025265000)
(57,2202697455000)
(78,6100569384000)
(106,16521000814500)
(145,45706944333750)
(197,123685647052625)
(269,340236078785750)
(367,933358027840875)
(500,2548883048687500)
}
 
\def\preoptimczlargerzeigthfitavg{
(13,250250)
(17,431375)
(23,825125)
(31,1573250)
(42,3013500)
(57,5785500)
(78,11261250)
(106,21703500)
(145,42503125)
(197,81878125)
(269,159348875)
(367,309885625)
(500,600250000)
}

\def\preoptimcznormalcnotbestfitavg{
(13,170040)
(16,288960)
(20,510000)
(25,903000)
(32,1622400)
(40,2712000)
(50,4545000)
(63,7749000)
(79,13058700)
(100,22530000)
}

 \def\preoptimctnormaltbestfitavg{
(13,8348441)
(16,14252832)
(20,26282680)
(25,48511809)
(32,86868704)
(40,152200604)
(50,258710490)
(63,458056100)
(79,769292976)
(100,1381025916)
 }

%% file: postoptimcounts.tex
 \def\postoptimctnormaltotalfstanaavg{
 (                   13 ,       46922468326118 )
 (                   16 ,      140429816386316 )
 (                   20 ,      443033488048161 )
 (                   25 ,     1385848033204816 )
 (                   32 ,     5022300165643956 )
 (                   40 ,    15478022337656672 )
 (                   50 ,    48306151448470104 )
 (                   63 ,   157634675073723200 )
 (                   79 ,   508304538345763136 )
 (                  100 ,  1721824948909853184 )
 }
 \def\postoptimctnormaltfstanaavg{
 (                   13 ,       18706380976774 )
 (                   16 ,       54386482920030 )
 (                   20 ,      173919933654337 )
 (                   25 ,      550737489859190 )
 (                   32 ,     1960804754588352 )
 (                   40 ,     6140134672417757 )
 (                   50 ,    18708955243831936 )
 (                   63 ,    61951075442441984 )
 (                   79 ,   199915310701595968 )
 (                  100 ,   664868249286088832 )
 }
 \def\postoptimctnormalcnotfstanaavg{
 (                   13 ,         193781570892 )
 (                   16 ,         547262416704 )
 (                   20 ,        1670112355440 )
 (                   25 ,        5096778428400 )
 (                   32 ,       17512397331840 )
 (                   40 ,       53443595373120 )
 (                   50 ,      163096909707600 )
 (                   63 ,      517962666381552 )
 (                   79 ,     1605946846952772 )
 (                  100 ,     5219101110641400 )
 }
 \def\postoptimctnormaltotalfstminavg{
 (                   13 ,       17515720245005 )
 (                   16 ,       50605980191539 )
 (                   20 ,      152848431765120 )
 (                   25 ,      497280099456000 )
 (                   32 ,     1734059703005184 )
 (                   40 ,     5587567028528640 )
 (                   50 ,    17639200881960000 )
 (                   63 ,    56861788997459224 )
 (                   79 ,   177641538738606944 )
 (                  100 ,   616110727701383936 )
 }
 \def\postoptimctnormaltfstminavg{
 (                   13 ,        6753013532429 )
 (                   16 ,       19423831308288 )
 (                   20 ,       61036563073920 )
 (                   25 ,      193415038683000 )
 (                   32 ,      702952014436762 )
 (                   40 ,     2181464234426880 )
 (                   50 ,     6888560344428000 )
 (                   63 ,    22417295965352872 )
 (                   79 ,    70046331390236440 )
 (                  100 ,   238918402986480032 )
 }
 \def\postoptimctnormalcnotfstminavg{
 (                   13 ,          71288308728 )
 (                   16 ,         201326690304 )
 (                   20 ,         614400192000 )
 (                   25 ,        1875000375000 )
 (                   32 ,        6442451730432 )
 (                   40 ,       19660801536000 )
 (                   50 ,       60000003000000 )
 (                   63 ,      190547822257128 )
 (                   79 ,      590794840440936 )
 (                  100 ,     1920000024000000 )
 }
 \def\postoptimctnormaltotalfstfitavg{
 (                   13 ,           5507204937 )
 (                   16 ,          10112992698 )
 (                   20 ,          19871081934 )
 (                   25 ,          39624907994 )
 (                   32 ,          83717247742 )
 (                   40 ,         166501435550 )
 (                   50 ,         343986285794 )
 (                   63 ,         672684423844 )
 (                   79 ,        1356302575187 )
 (                  100 ,        2777633740847 )
 }
 \def\postoptimctnormaltfstfitavg{
 (                   13 ,           2084459636 )
 (                   16 ,           3875525240 )
 (                   20 ,           7771692907 )
 (                   25 ,          15162350701 )
 (                   32 ,          32155275134 )
 (                   40 ,          65131266600 )
 (                   50 ,         130094490110 )
 (                   63 ,         262841273453 )
 (                   79 ,         529219765334 )
 (                  100 ,        1082357744688 )
 }
 \def\postoptimctnormalcnotfstfitavg{
 (                   13 ,             29054298 )
 (                   16 ,             53764512 )
 (                   20 ,            104168880 )
 (                   25 ,            201827400 )
 (                   32 ,            419519232 )
 (                   40 ,            812820000 )
 (                   50 ,           1574841300 )
 (                   63 ,           3124182474 )
 (                   79 ,           6110249154 )
 (                  100 ,          12288348600 )
 }
 \def\postoptimctnormaltotalfstcomavg{
 (                   13 ,         746233518884 )
 (                   16 ,        1748211797192 )
 (                   20 ,        4429592421012 )
 (                   25 ,       11339896875338 )
 (                   32 ,       30081086300585 )
 (                   40 ,       77936759556032 )
 (                   50 ,      190252289384192 )
 (                   63 ,      489413175733990 )
 (                   79 ,     1276257116292771 )
 (                  100 ,     3196704909960420 )
 }
 \def\postoptimctnormaltfstcomavg{
 (                   13 ,         288892372338 )
 (                   16 ,         674928943644 )
 (                   20 ,        1718263254708 )
 (                   25 ,        4319585674417 )
 (                   32 ,       11936164976669 )
 (                   40 ,       30569660097600 )
 (                   50 ,       75106009652288 )
 (                   63 ,      190588145169526 )
 (                   79 ,      497730591035443 )
 (                  100 ,     1298442471327187 )
 }
 \def\postoptimctnormalcnotfstcomavg{
 (                   13 ,           3427685586 )
 (                   16 ,           7865158848 )
 (                   20 ,          19202047920 )
 (                   25 ,          46879999650 )
 (                   32 ,         125842540608 )
 (                   40 ,         307232764800 )
 (                   50 ,         750079992000 )
 (                   63 ,        1890556937262 )
 (                   79 ,        4674508229592 )
 (                  100 ,       12001279866600 )
 }
 \def\postoptimctnormaltotalsndanaavg{
 (                   13 ,          40882548530 )
 (                   16 ,          96781893396 )
 (                   20 ,         248342965202 )
 (                   25 ,         607295672400 )
 (                   32 ,        1717141915568 )
 (                   40 ,        4337529090360 )
 (                   50 ,       10791349945250 )
 (                   63 ,       27535044069568 )
 (                   79 ,       70328678308537 )
 (                  100 ,      181513045632504 )
 }
 \def\postoptimctnormaltsndanaavg{
 (                   13 ,          15928584481 )
 (                   16 ,          37497558656 )
 (                   20 ,          94040419285 )
 (                   25 ,         237108699442 )
 (                   32 ,         665070306912 )
 (                   40 ,        1673314661779 )
 (                   50 ,        4205092575303 )
 (                   63 ,       10838828296534 )
 (                   79 ,       27705527616596 )
 (                  100 ,       70746464373920 )
 }
 \def\postoptimctnormalcnotsndanaavg{
 (                   13 ,            220089922 )
 (                   16 ,            505017632 )
 (                   20 ,           1232953160 )
 (                   25 ,           3010139650 )
 (                   32 ,           8080281664 )
 (                   40 ,          19727250000 )
 (                   50 ,          48162231300 )
 (                   63 ,         121391639838 )
 (                   79 ,         300147648438 )
 (                  100 ,         770595696200 )
 }
 \def\postoptimctnormaltotalsndminavg{
 (                   13 ,          24168486262 )
 (                   16 ,          58044971481 )
 (                   20 ,         147408204192 )
 (                   25 ,         363277852673 )
 (                   32 ,        1002885289597 )
 (                   40 ,        2577240176878 )
 (                   50 ,        6417119695292 )
 (                   63 ,       16486168331201 )
 (                   79 ,       42165337981246 )
 (                  100 ,      111439344251599 )
 }
 \def\postoptimctnormaltsndminavg{
 (                   13 ,           9425339751 )
 (                   16 ,          22378726428 )
 (                   20 ,          56476290001 )
 (                   25 ,         140653493952 )
 (                   32 ,         387325916721 )
 (                   40 ,         987326378314 )
 (                   50 ,        2483699513235 )
 (                   63 ,        6406602658118 )
 (                   79 ,       16480757593144 )
 (                  100 ,       43717036899826 )
 }
 \def\postoptimctnormalcnotsndminavg{
 (                   13 ,            133561610 )
 (                   16 ,            306439840 )
 (                   20 ,            748079880 )
 (                   25 ,           1826241850 )
 (                   32 ,           4901987136 )
 (                   40 ,          11967229840 )
 (                   50 ,          29215869700 )
 (                   63 ,          73635752358 )
 (                   79 ,         182064527646 )
 (                  100 ,         467421915400 )
 }
 \def\postoptimctnormaltotalsndfitavg{
 (                   13 ,             73518181 )
 (                   16 ,            137128817 )
 (                   20 ,            273088740 )
 (                   25 ,            526281626 )
 (                   32 ,           1124204453 )
 (                   40 ,           2140787512 )
 (                   50 ,           4031284745 )
 (                   63 ,           8575517068 )
 (                   79 ,          16760232768 )
 (                  100 ,          34016278123 )
 }
 \def\postoptimctnormaltsndfitavg{
 (                   13 ,             28288257 )
 (                   16 ,             52652035 )
 (                   20 ,            104368593 )
 (                   25 ,            206570868 )
 (                   32 ,            430288561 )
 (                   40 ,            836312769 )
 (                   50 ,           1550175694 )
 (                   63 ,           3302991647 )
 (                   79 ,           6513475256 )
 (                  100 ,          13198915260 )
 }
 \def\postoptimctnormalcnotsndfitavg{
 (                   13 ,               513474 )
 (                   16 ,               934176 )
 (                   20 ,              1777480 )
 (                   25 ,              3381650 )
 (                   32 ,              6889280 )
 (                   40 ,             13107600 )
 (                   50 ,             24939300 )
 (                   63 ,             48552462 )
 (                   79 ,             93225846 )
 (                  100 ,            183925000 )
 }
 \def\postoptimctnormaltotalsndcomavg{
 (                   13 ,           2990255305 )
 (                   16 ,           6234147950 )
 (                   20 ,          14634171019 )
 (                   25 ,          33333974235 )
 (                   32 ,          85350484043 )
 (                   40 ,         192603967022 )
 (                   50 ,         426834788945 )
 (                   63 ,        1011508533321 )
 (                   79 ,        2357592634999 )
 (                  100 ,        5588696773037 )
 }
 \def\postoptimctnormaltsndcomavg{
 (                   13 ,           1173593121 )
 (                   16 ,           2476828852 )
 (                   20 ,           5653524584 )
 (                   25 ,          12863144226 )
 (                   32 ,          33127306688 )
 (                   40 ,          74570864530 )
 (                   50 ,         166743777022 )
 (                   63 ,         393168188940 )
 (                   79 ,         910596367953 )
 (                  100 ,        2166129995819 )
 }
 \def\postoptimctnormalcnotsndcomavg{
 (                   13 ,             17907994 )
 (                   16 ,             37609632 )
 (                   20 ,             83323880 )
 (                   25 ,            184343850 )
 (                   32 ,            443276608 )
 (                   40 ,            979182480 )
 (                   50 ,           2162442500 )
 (                   63 ,           4912506774 )
 (                   79 ,          10973299870 )
 (                  100 ,          25353238600 )
 }
 \def\postoptimctnormaltotalfrthanaavg{
 (                   13 ,          38194374130 )
 (                   16 ,          82440066693 )
 (                   20 ,         178438990491 )
 (                   25 ,         410413482046 )
 (                   32 ,        1017030579854 )
 (                   40 ,        2198746476197 )
 (                   50 ,        5090019594880 )
 (                   63 ,       11586490115032 )
 (                   79 ,       26193835603182 )
 (                  100 ,       60348578406023 )
 }
 \def\postoptimctnormaltfrthanaavg{
 (                   13 ,          14992134919 )
 (                   16 ,          31567975344 )
 (                   20 ,          69399463559 )
 (                   25 ,         160106450768 )
 (                   32 ,         389183881102 )
 (                   40 ,         870796970165 )
 (                   50 ,        1967822383428 )
 (                   63 ,        4516046879069 )
 (                   79 ,       10223711240016 )
 (                  100 ,       23565484475824 )
 }
 \def\postoptimctnormalcnotfrthanaavg{
 (                   13 ,            207964146 )
 (                   16 ,            430136352 )
 (                   20 ,            939271240 )
 (                   25 ,           2051048050 )
 (                   32 ,           4866433344 )
 (                   40 ,          10626635280 )
 (                   50 ,          23204960100 )
 (                   63 ,          52104797766 )
 (                   79 ,         115048407998 )
 (                  100 ,         262534140200 )
 }
 \def\postoptimctnormaltotalfrthminavg{
 (                   13 ,          29322054141 )
 (                   16 ,          64340229422 )
 (                   20 ,         139185678866 )
 (                   25 ,         318834341612 )
 (                   32 ,         778490771466 )
 (                   40 ,        1745817882826 )
 (                   50 ,        3840906677820 )
 (                   63 ,        8979571122119 )
 (                   79 ,       20059307328374 )
 (                  100 ,       47838487355255 )
 }
 \def\postoptimctnormaltfrthminavg{
 (                   13 ,          11410258328 )
 (                   16 ,          24472192307 )
 (                   20 ,          54632910082 )
 (                   25 ,         123229299928 )
 (                   32 ,         305269805573 )
 (                   40 ,         676141974298 )
 (                   50 ,        1519338890091 )
 (                   63 ,        3487094480845 )
 (                   79 ,        7917826370459 )
 (                  100 ,       18484132343810 )
 }
 \def\postoptimctnormalcnotfrthminavg{
 (                   13 ,            162839066 )
 (                   16 ,            336624672 )
 (                   20 ,            734698440 )
 (                   25 ,           1603596050 )
 (                   32 ,           3803066944 )
 (                   40 ,           8301593680 )
 (                   50 ,          18121972100 )
 (                   63 ,          40679153046 )
 (                   79 ,          89796791118 )
 (                  100 ,         204861404200 )
 }
 \def\postoptimctnormaltotalfrthfitavg{
 (                   13 ,             15227637 )
 (                   16 ,             26946945 )
 (                   20 ,             48794500 )
 (                   25 ,             87152677 )
 (                   32 ,            172272872 )
 (                   40 ,            303496806 )
 (                   50 ,            552604899 )
 (                   63 ,           1024111912 )
 (                   79 ,           1881296359 )
 (                  100 ,           3523789515 )
 }
 \def\postoptimctnormaltfrthfitavg{
 (                   13 ,              5846221 )
 (                   16 ,             10172704 )
 (                   20 ,             18432690 )
 (                   25 ,             33860997 )
 (                   32 ,             65697727 )
 (                   40 ,            116119713 )
 (                   50 ,            212882578 )
 (                   63 ,            392674266 )
 (                   79 ,            722884766 )
 (                  100 ,           1352640988 )
 }
 \def\postoptimctnormalcnotfrthfitavg{
 (                   13 ,               113386 )
 (                   16 ,               192672 )
 (                   20 ,               340040 )
 (                   25 ,               602050 )
 (                   32 ,              1130304 )
 (                   40 ,              1998480 )
 (                   50 ,              3532100 )
 (                   63 ,              6373206 )
 (                   79 ,             11363518 )
 (                  100 ,             20748200 )
 }
 \def\postoptimctnormaltotalfrthcomavg{
 (                   13 ,           2382185944 )
 (                   16 ,           4909001778 )
 (                   20 ,          10828340800 )
 (                   25 ,          22917953325 )
 (                   32 ,          54421133904 )
 (                   40 ,         116240247440 )
 (                   50 ,         244488439941 )
 (                   63 ,         540420156896 )
 (                   79 ,        1154054253002 )
 (                  100 ,        2616455830166 )
 }
 \def\postoptimctnormaltfrthcomavg{
 (                   13 ,            924482148 )
 (                   16 ,           1911792316 )
 (                   20 ,           4172226900 )
 (                   25 ,           8924743689 )
 (                   32 ,          20966440338 )
 (                   40 ,          44910915680 )
 (                   50 ,          95302356859 )
 (                   63 ,         209182246076 )
 (                   79 ,         450249549691 )
 (                  100 ,        1010829937718 )
 }
 \def\postoptimctnormalcnotfrthcomavg{
 (                   13 ,             14355146 )
 (                   16 ,             28790432 )
 (                   20 ,             60599240 )
 (                   25 ,            127149050 )
 (                   32 ,            287801664 )
 (                   40 ,            601040080 )
 (                   50 ,           1253378100 )
 (                   63 ,           2680274646 )
 (                   79 ,           5637762478 )
 (                  100 ,          12225632200 )
 }
 \def\postoptimctnormaltotalsxthanaavg{
 (                   13 ,         342349747317 )
 (                   16 ,         711110442334 )
 (                   20 ,        1538187299230 )
 (                   25 ,        3289287667805 )
 (                   32 ,        7634342948623 )
 (                   40 ,       16709152336737 )
 (                   50 ,       35398838470024 )
 (                   63 ,       79258759609801 )
 (                   79 ,      171531383495600 )
 (                  100 ,      381675394941790 )
 }
 \def\postoptimctnormaltsxthanaavg{
 (                   13 ,         132271371282 )
 (                   16 ,         277911776599 )
 (                   20 ,         595619858102 )
 (                   25 ,        1284705812256 )
 (                   32 ,        2989435465118 )
 (                   40 ,        6519318378550 )
 (                   50 ,       13833033250938 )
 (                   63 ,       30568247845420 )
 (                   79 ,       66626074698074 )
 (                  100 ,      150033509259086 )
 }
 \def\postoptimctnormalcnotsxthanaavg{
 (                   13 ,           1744792426 )
 (                   16 ,           3486038432 )
 (                   20 ,           7334416040 )
 (                   25 ,          15431175050 )
 (                   32 ,          35137068864 )
 (                   40 ,          73926280080 )
 (                   50 ,         155536460100 )
 (                   63 ,         336047518926 )
 (                   79 ,         714533354158 )
 (                  100 ,        1567709280200 )
 }
 \def\postoptimctnormaltotalsxthminavg{
 (                   13 ,         294962097678 )
 (                   16 ,         606199264546 )
 (                   20 ,        1291179815090 )
 (                   25 ,        2817865153078 )
 (                   32 ,        6477814169520 )
 (                   40 ,       14162229617150 )
 (                   50 ,       30612992430462 )
 (                   63 ,       67071679892105 )
 (                   79 ,      145041091405471 )
 (                  100 ,      325401696459050 )
 }
 \def\postoptimctnormaltsxthminavg{
 (                   13 ,         113541605457 )
 (                   16 ,         234434425922 )
 (                   20 ,         507264934086 )
 (                   25 ,        1081710421214 )
 (                   32 ,        2537920235784 )
 (                   40 ,        5436048597423 )
 (                   50 ,       11786010657579 )
 (                   63 ,       26224053606432 )
 (                   79 ,       57184812047797 )
 (                  100 ,      128363668315160 )
 }
 \def\postoptimctnormalcnotsxthminavg{
 (                   13 ,           1491510826 )
 (                   16 ,           2978051232 )
 (                   20 ,           6261556040 )
 (                   25 ,          13165965050 )
 (                   32 ,          29960544064 )
 (                   40 ,          63002472080 )
 (                   50 ,         132489490100 )
 (                   63 ,         286120220526 )
 (                   79 ,         608117478558 )
 (                  100 ,        1333687300200 )
 }
 \def\postoptimctnormaltotalsxthfitavg{
 (                   13 ,             18248701 )
 (                   16 ,             30182853 )
 (                   20 ,             51856837 )
 (                   25 ,             87471392 )
 (                   32 ,            159949414 )
 (                   40 ,            276336825 )
 (                   50 ,            470562834 )
 (                   63 ,            822471933 )
 (                   79 ,           1413883016 )
 (                  100 ,           2500801172 )
 }
 \def\postoptimctnormaltsxthfitavg{
 (                   13 ,              7053478 )
 (                   16 ,             11652101 )
 (                   20 ,             19726079 )
 (                   25 ,             33668686 )
 (                   32 ,             61971256 )
 (                   40 ,            105675143 )
 (                   50 ,            180892328 )
 (                   63 ,            316551594 )
 (                   79 ,            546220371 )
 (                  100 ,            970090617 )
 }
 \def\postoptimctnormalcnotsxthfitavg{
 (                   13 ,               135226 )
 (                   16 ,               217632 )
 (                   20 ,               364040 )
 (                   25 ,               610050 )
 (                   32 ,              1081664 )
 (                   40 ,              1808080 )
 (                   50 ,              3030100 )
 (                   63 ,              5166126 )
 (                   79 ,              8705958 )
 (                  100 ,             15020200 )
 }
 \def\postoptimctnormaltotallcuavg{
 (                   13 ,            690241422 )
 (                   16 ,           1215441743 )
 (                   20 ,           3858762347 )
 (                   25 ,           6201621022 )
 (                   32 ,          10429570468 )
 (                   40 ,          33554077124 )
 (                   50 ,          53522184143 )
 (                   63 ,          96861283642 )
 (                   79 ,         311490601682 )
 (                  100 ,         509515774312 )
 }
 \def\postoptimctnormaltlcuavg{
 (                   13 ,            265427934 )
 (                   16 ,            466238131 )
 (                   20 ,           1491110405 )
 (                   25 ,           2391284775 )
 (                   32 ,           4025983222 )
 (                   40 ,          13009333548 )
 (                   50 ,          20709208054 )
 (                   63 ,          37457403781 )
 (                   79 ,         120918056368 )
 (                  100 ,         197676206082 )
 }
 \def\postoptimctnormalcnotlcuavg{
 (                   13 ,             15113318 )
 (                   16 ,             29399669 )
 (                   20 ,             66727835 )
 (                   25 ,            120512824 )
 (                   32 ,            231936676 )
 (                   40 ,            529760688 )
 (                   50 ,            954465726 )
 (                   63 ,           1966464875 )
 (                   79 ,           4502905836 )
 (                  100 ,           8382111144 )
 }
 \def\postoptimctnormaltotalspavg{
 (                   13 ,             27930105 )
 (                   16 ,             43760894 )
 (                   20 ,            138802175 )
 (                   25 ,            223326102 )
 (                   32 ,            378452547 )
 (                   40 ,           1205915291 )
 (                   50 ,           1940548688 )
 (                   63 ,           3170653537 )
 (                   79 ,          10162087705 )
 (                  100 ,          16725591055 )
 }
 \def\postoptimctnormaltspavg{
 (                   13 ,             10630753 )
 (                   16 ,             16626990 )
 (                   20 ,             53100914 )
 (                   25 ,             85342826 )
 (                   32 ,            144724584 )
 (                   40 ,            463069944 )
 (                   50 ,            743682818 )
 (                   63 ,           1214839210 )
 (                   79 ,           3912565658 )
 (                  100 ,           6436868830 )
 }
 \def\postoptimctnormalcnotspavg{
 (                   13 ,              1191600 )
 (                   16 ,              2086016 )
 (                   20 ,              4469420 )
 (                   25 ,              8265332 )
 (                   32 ,             16270660 )
 (                   40 ,             34803454 )
 (                   50 ,             64532226 )
 (                   63 ,            122808134 )
 (                   79 ,            264958428 )
 (                  100 ,            509013696 )
 }
 \def\postoptimctnormaltotalspjaavg{
 (                   13 ,             21276620 )
 (                   16 ,             32972754 )
 (                   20 ,            103248285 )
 (                   25 ,            164084902 )
 (                   32 ,            275146498 )
 (                   40 ,            861451238 )
 (                   50 ,           1367087054 )
 (                   63 ,           2204629973 )
 (                   79 ,           6966119007 )
 (                  100 ,          11315696692 )
 }
 \def\postoptimctnormaltspjaavg{
 (                   13 ,              8098966 )
 (                   16 ,             12538406 )
 (                   20 ,             39521888 )
 (                   25 ,             62688464 )
 (                   32 ,            104942752 )
 (                   40 ,            331309558 )
 (                   50 ,            523865896 )
 (                   63 ,            844484712 )
 (                   79 ,           2681464588 )
 (                  100 ,           4350199087 )
 }
 \def\postoptimctnormalcnotspjaavg{
 (                   13 ,               919832 )
 (                   16 ,              1593656 )
 (                   20 ,              3372670 )
 (                   25 ,              6155400 )
 (                   32 ,             11949060 )
 (                   40 ,             25212252 )
 (                   50 ,             46131984 )
 (                   63 ,             86574214 )
 (                   79 ,            184239324 )
 (                  100 ,            348895410 )
 }
 \def\postoptimctnormaltotalspsegmentavg{
 (                   13 ,             72274110 )
 (                   16 ,            116607561 )
 (                   20 ,            384906398 )
 (                   25 ,            639344192 )
 (                   32 ,           1121527928 )
 (                   40 ,           3700699261 )
 (                   50 ,           6130090371 )
 (                   63 ,           9939016590 )
 (                   79 ,          34339097291 )
 (                  100 ,          58528770319 )
 }
 \def\postoptimctnormaltspsegmentavg{
 (                   13 ,             27611276 )
 (                   16 ,             44486865 )
 (                   20 ,            147529706 )
 (                   25 ,            245196688 )
 (                   32 ,            429638369 )
 (                   40 ,           1422960301 )
 (                   50 ,           2358998400 )
 (                   63 ,           3819496620 )
 (                   79 ,          13269217062 )
 (                  100 ,          22586142498 )
 }
 \def\postoptimctnormalcnotspsegmentavg{
 (                   13 ,              2803416 )
 (                   16 ,              5082208 )
 (                   20 ,             11287696 )
 (                   25 ,             21524384 )
 (                   32 ,             43819300 )
 (                   40 ,             97255296 )
 (                   50 ,            185760080 )
 (                   63 ,            350279400 )
 (                   79 ,            817032320 )
 (                  100 ,           1619181628 )
 }
 \def\postoptimctnormaltotalspsegmentfitavg{
 (                   13 ,             68317459 )
 (                   16 ,            109531223 )
 (                   20 ,            362329166 )
 (                   25 ,            600468301 )
 (                   32 ,           1052773034 )
 (                   40 ,           3465338978 )
 (                   50 ,           5740876414 )
 (                   63 ,           9687080238 )
 (                   79 ,          32064248389 )
 (                  100 ,          54563122325 )
 }
 \def\postoptimctnormaltspsegmentfitavg{
 (                   13 ,             26111866 )
 (                   16 ,             41782074 )
 (                   20 ,            138801547 )
 (                   25 ,            230094499 )
 (                   32 ,            403094595 )
 (                   40 ,           1332889510 )
 (                   50 ,           2208074456 )
 (                   63 ,           3723626184 )
 (                   79 ,          12385395411 )
 (                  100 ,          21057397243 )
 }
 \def\postoptimctnormalcnotspsegmentfitavg{
 (                   13 ,              2657784 )
 (                   16 ,              4786976 )
 (                   20 ,             10647672 )
 (                   25 ,             20252176 )
 (                   32 ,             41217400 )
 (                   40 ,             91290456 )
 (                   50 ,            174255200 )
 (                   63 ,            341743856 )
 (                   79 ,            764399360 )
 (                  100 ,           1513077624 )
 }
  \def\postoptimczlargetotalsxthanaavg{
 (                   17 ,           9027914719 )
 (                   23 ,          24727699786 )
 (                   31 ,          66880037267 )
 (                   42 ,         184044056994 )
 (                   57 ,         509339097024 )
 (                   78 ,        1449014793096 )
 (                  106 ,        4028193848242 )
 (                  145 ,       11445922932140 )
 (                  197 ,       31791539841904 )
 (                  269 ,       89797671697258 )
 (                  367 ,      252916806845119 )
 (                  500 ,      709018374691000 )
 }
 \def\postoptimczlargerzsxthanaavg{
 (                   17 ,           2375767017 )
 (                   23 ,           6507289398 )
 (                   31 ,          17600009781 )
 (                   42 ,          48432646542 )
 (                   57 ,         134036604432 )
 (                   78 ,         381319682328 )
 (                  106 ,        1060051012606 )
 (                  145 ,        3012084982020 )
 (                  197 ,        8366194695072 )
 (                  269 ,       23630966235894 )
 (                  367 ,       66557054432617 )
 (                  500 ,      186583782813000 )
 }
 \def\postoptimczlargecnotsxthanaavg{
 (                   17 ,           3801227234 )
 (                   23 ,          10411663046 )
 (                   31 ,          28160015662 )
 (                   42 ,          77492234484 )
 (                   57 ,         214458567114 )
 (                   78 ,         610111491756 )
 (                  106 ,        1696081620212 )
 (                  145 ,        4819335971290 )
 (                  197 ,       13385911512194 )
 (                  269 ,       37809545977538 )
 (                  367 ,      106491287092334 )
 (                  500 ,      298534052501000 )
 }
 \def\postoptimczlargetotalsxthminavg{
 (                   17 ,           7719360969 )
 (                   23 ,          21123312861 )
 (                   31 ,          57082478742 )
 (                   42 ,         156958620294 )
 (                   57 ,         434068566399 )
 (                   78 ,        1234055472546 )
 (                  106 ,        3428598456442 )
 (                  145 ,        9736915021890 )
 (                  197 ,       27031738545604 )
 (                  269 ,       76319664170458 )
 (                  367 ,      214870686590919 )
 (                  500 ,      602147014378500 )
 }
 \def\postoptimczlargerzsxthminavg{
 (                   17 ,           2031410767 )
 (                   23 ,           5558766523 )
 (                   31 ,          15021704906 )
 (                   42 ,          41304900042 )
 (                   57 ,         114228570057 )
 (                   78 ,         324751440078 )
 (                  106 ,         902262751606 )
 (                  145 ,        2562346058270 )
 (                  197 ,        7113615406572 )
 (                  269 ,       20084122149894 )
 (                  367 ,       56544917523617 )
 (                  500 ,      158459740625500 )
 }
 \def\postoptimczlargecnotsxthminavg{
 (                   17 ,           3250257234 )
 (                   23 ,           8894026446 )
 (                   31 ,          24034727862 )
 (                   42 ,          66087840084 )
 (                   57 ,         182765712114 )
 (                   78 ,         519602304156 )
 (                  106 ,        1443620402612 )
 (                  145 ,        4099753693290 )
 (                  197 ,       11381784650594 )
 (                  269 ,       32134595439938 )
 (                  367 ,       90471868037934 )
 (                  500 ,      253535585001000 )
 }
 \def\postoptimczlargetotalsxthfitavg{
 (                   17 ,               597669 )
 (                   23 ,              1201911 )
 (                   31 ,              2385667 )
 (                   42 ,              4808244 )
 (                   57 ,              9720324 )
 (                   78 ,             20081646 )
 (                  106 ,             40784242 )
 (                  145 ,             84097390 )
 (                  197 ,            170775754 )
 (                  269 ,            350872033 )
 (                  367 ,            719267519 )
 (                  500 ,           1469891000 )
 }
 \def\postoptimczlargerzsxthfitavg{
 (                   17 ,               157267 )
 (                   23 ,               316273 )
 (                   31 ,               627781 )
 (                   42 ,              1265292 )
 (                   57 ,              2557932 )
 (                   78 ,              5284578 )
 (                  106 ,             10732606 )
 (                  145 ,             22130770 )
 (                  197 ,             44940822 )
 (                  269 ,             92334519 )
 (                  367 ,            189280617 )
 (                  500 ,            386813000 )
 }
 \def\postoptimczlargecnotsxthfitavg{
 (                   17 ,               251634 )
 (                   23 ,               506046 )
 (                   31 ,              1004462 )
 (                   42 ,              2024484 )
 (                   57 ,              4092714 )
 (                   78 ,              8455356 )
 (                  106 ,             17172212 )
 (                  145 ,             35409290 )
 (                  197 ,             71905394 )
 (                  269 ,            147735338 )
 (                  367 ,            302849134 )
 (                  500 ,            618901000 )
 }
 \def\postoptimczlargetotaleigthanaavg{
 (                   17 ,         131604089119 )
 (                   23 ,         351500349286 )
 (                   31 ,         927333159342 )
 (                   42 ,        2488115122794 )
 (                   57 ,        6712797617649 )
 (                   78 ,       18604488182046 )
 (                  106 ,       50414368447742 )
 (                  145 ,      139558494267265 )
 (                  197 ,      377855286212004 )
 (                  269 ,     1039932563391883 )
 (                  367 ,     2854138055908944 )
 (                  500 ,     7797632828441000 )
 }
 \def\postoptimczlargerzeigthanaavg{
 (                   17 ,          34632655017 )
 (                   23 ,          92500091898 )
 (                   31 ,         244035041906 )
 (                   42 ,         654767137542 )
 (                   57 ,        1766525688807 )
 (                   78 ,        4895917942578 )
 (                  106 ,       13266939065106 )
 (                  145 ,       36725919543895 )
 (                  197 ,       99435601634572 )
 (                  269 ,      273666464050269 )
 (                  367 ,      751088962080992 )
 (                  500 ,     2052008639063000 )
 }
 \def\postoptimczlargecnoteigthanaavg{
 (                   17 ,          55412248034 )
 (                   23 ,         148000147046 )
 (                   31 ,         390456067062 )
 (                   42 ,        1047627420084 )
 (                   57 ,        2826441102114 )
 (                   78 ,        7833468708156 )
 (                  106 ,       21227102504212 )
 (                  145 ,       58761471270290 )
 (                  197 ,      159096962615394 )
 (                  269 ,      437866342480538 )
 (                  367 ,     1201742339329734 )
 (                  500 ,     3283213822501000 )
 }
 \def\postoptimczlargetotaleigthminavg{
 (                   17 ,         117589724744 )
 (                   23 ,         313789325036 )
 (                   31 ,         827165757842 )
 (                   42 ,        2217640005294 )
 (                   57 ,        5978750235399 )
 (                   78 ,       16558688328546 )
 (                  106 ,       44842716497242 )
 (                  145 ,      124061706049765 )
 (                  197 ,      335718184858504 )
 (                  269 ,      923497928134633 )
 (                  367 ,     2533400361284944 )
 (                  500 ,     6918396846441000 )
 }
 \def\postoptimczlargerzeigthminavg{
 (                   17 ,          30944664392 )
 (                   23 ,          82576138148 )
 (                   31 ,         217675199406 )
 (                   42 ,         583589475042 )
 (                   57 ,        1573355325057 )
 (                   78 ,        4357549560078 )
 (                  106 ,       11800714867606 )
 (                  145 ,       32647817381395 )
 (                  197 ,       88346890752072 )
 (                  269 ,      243025770561519 )
 (                  367 ,      666684305600992 )
 (                  500 ,     1820630749063000 )
 }
 \def\postoptimczlargecnoteigthminavg{
 (                   17 ,          49511463034 )
 (                   23 ,         132121821046 )
 (                   31 ,         348280319062 )
 (                   42 ,         933743160084 )
 (                   57 ,        2517368520114 )
 (                   78 ,        6972079296156 )
 (                  106 ,       18881143788212 )
 (                  145 ,       52236507810290 )
 (                  197 ,      141355025203394 )
 (                  269 ,      388841232898538 )
 (                  367 ,     1066694888961734 )
 (                  500 ,     2913009198501000 )
 }
 \def\postoptimczlargetotaleigthfitavg{
 (                   17 ,              1170994 )
 (                   23 ,              2239786 )
 (                   31 ,              4270467 )
 (                   42 ,              8179794 )
 (                   57 ,             15703899 )
 (                   78 ,             30566796 )
 (                  106 ,             58910242 )
 (                  145 ,            115366640 )
 (                  197 ,            222242004 )
 (                  269 ,            432520258 )
 (                  367 ,            841120694 )
 (                  500 ,           1629253500 )
 }
 \def\postoptimczlargerzeigthfitavg{
 (                   17 ,               308142 )
 (                   23 ,               589398 )
 (                   31 ,              1123781 )
 (                   42 ,              2152542 )
 (                   57 ,              4132557 )
 (                   78 ,              8043828 )
 (                  106 ,             15502606 )
 (                  145 ,             30359520 )
 (                  197 ,             58484572 )
 (                  269 ,            113820894 )
 (                  367 ,            221347242 )
 (                  500 ,            428750500 )
 }
 \def\postoptimczlargecnoteigthfitavg{
 (                   17 ,               493034 )
 (                   23 ,               943046 )
 (                   31 ,              1798062 )
 (                   42 ,              3444084 )
 (                   57 ,              6612114 )
 (                   78 ,             12870156 )
 (                  106 ,             24804212 )
 (                  145 ,             48575290 )
 (                  197 ,             93575394 )
 (                  269 ,            182113538 )
 (                  367 ,            354155734 )
 (                  500 ,            686001000 )
 }

\def\postoptimcznormaltotalfstanaavg{
(13,242226963615)
(16,684078020880)
(20,2087640444300)
(25,6370973035500)
(32,21890496664800)
(40,66804494216400)
(50,203871137134500)
(63,647453332976940)
(79,2007433558690965)
(100,6523876388302500)
}

\def\postoptimcznormaltotalfstminavg{
(13,89110451820)
(16,251658485760)
(20,768000480000)
(25,2343750937500)
(32,8053065646080)
(40,24576003840000)
(50,75000007500000)
(63,238184785322820)
(79,738493565342340)
(100,2400000060000000)
}

\def\postoptimcznormaltotalfstcomavg{
(13,4285063965)
(16,9832496880)
(20,24005119200)
(25,58606247625)
(32,157319948160)
(40,384081903600)
(50,937699959000)
(63,2363448141180)
(79,5843758293630)
(100,15003199330500)
}

\def\postoptimcznormaltotalfstfitavg{
(13,72635745)
(16,134411280)
(20,260422200)
(25,504568500)
(32,1048798080)
(40,2032050000)
(50,3937103250)
(63,7810456185)
(79,15275622885)
(100,30720871500)
}

\def\postoptimcznormaltotalsndanaavg{
(13,369614349)
(16,848116032)
(20,2070595440)
(25,5055164400)
(32,13569852512)
(40,33129522040)
(50,80882621300)
(63,203862523599)
(79,504061539701)
(100,1294121926400)
}

\def\postoptimcznormaltotalsndminavg{
(13,224349450)
(16,514719536)
(20,1256487620)
(25,3067299375)
(32,8233021856)
(40,20098934440)
(50,49067288500)
(63,123667873497)
(79,305766247412)
(100,785000622500)
}

\def\postoptimcznormaltotalsndcomavg{
(13,30710836)
(16,64567888)
(20,143244940)
(25,317398025)
(32,764660544)
(40,1692250480)
(50,3744653350)
(63,8525585061)
(79,19087250909)
(100,44209269100)
}

\def\postoptimcznormaltotalsndfitavg{
(13,1219530)
(16,2218704)
(20,4221560)
(25,8031475)
(32,16362112)
(40,31130640)
(50,59230950)
(63,115312239)
(79,221411562)
(100,436822100)
}

\def\postoptimcznormaltotalfrthanaavg{
(13,415330591)
(16,859037232)
(20,1875845440)
(25,4096207800)
(32,9718895424)
(40,21222760880)
(50,46343299350)
(63,104060004426)
(79,229766513678)
(100,524314548200)
}

\def\postoptimcznormaltotalfrthminavg{
(13,325541151)
(16,672899552)
(20,1468491140)
(25,3204939175)
(32,7600149184)
(40,16588991480)
(50,36210751350)
(63,81279209151)
(79,179410117613)
(100,409285479700)
}

\def\postoptimcznormaltotalfrthcomavg{
(13,28736071)
(16,57636992)
(20,121334140)
(25,254631050)
(32,576527104)
(40,1204391280)
(50,2512536600)
(63,5375416221)
(79,11312782383)
(100,24547563700)
}

\def\postoptimcznormaltotalfrthfitavg{
(13,269321)
(16,457632)
(20,807640)
(25,1429925)
(32,2684544)
(40,4746480)
(50,8388850)
(63,15136506)
(79,26988533)
(100,49277200)
}

\def\postoptimcznormaltotalsxthanaavg{
(13,3691779416)
(16,7376058512)
(20,15518782140)
(25,32650585800)
(32,74345966624)
(40,156419476280)
(50,329097195350)
(63,711037721016)
(79,1511870017403)
(100,3317091725700)
}

\def\postoptimcznormaltotalsxthminavg{
(13,3159617916)
(16,6308220512)
(20,13262427640)
(25,27884423925)
(32,63449223424)
(40,133415720280)
(50,280546970350)
(63,605826627966)
(79,1287553184978)
(100,2823646525700)
}

\def\postoptimcznormaltotalsxthfitavg{
(13,321191)
(16,516912)
(20,864640)
(25,1448925)
(32,2569024)
(40,4294280)
(50,7196600)
(63,12269691)
(79,20676828)
(100,35673200)
}

\def\postoptimcznormaltotaleigthanaavg{
(13,55033390841)
(16,108068808112)
(20,223181315140)
(25,460909125175)
(32,1028129216224)
(40,2123270140280)
(50,4384931269100)
(63,9293237927316)
(79,19390915602678)
(100,41716731075700)
}

\def\postoptimcznormaltotaleigthminavg{
(13,49214719216)
(16,96578748112)
(20,199317695140)
(25,411363537675)
(32,916997456224)
(40,1892684240280)
(50,3906621944100)
(63,8275154415066)
(79,17258174910303)
(100,37110576513200)
}

\def\postoptimcznormaltotaleigthfitavg{
(13,679341)
(16,1064112)
(20,1662640)
(25,2731425)
(32,4560224)
(40,7315280)
(50,11875350)
(63,19451691)
(79,31521553)
(100,52013200)
}

\def\postoptimcznormaltotallcuavg{
(13,31565650)
(16,54077683)
(20,118703830)
(25,242534100)
(32,465842357)
(40,1031508555)
(50,1873844322)
(63,3493409371)
(79,7762720416)
(100,16222059150)
}

\def\postoptimcznormaltotalspavg{
(13,2486152)
(16,4329453)
(20,9108371)
(25,16862821)
(32,33132993)
(40,69522783)
(50,129453077)
(63,246285095)
(79,521324152)
(100,1011143533)
}

\def\postoptimcznormaltotalspjaavg{
(13,1915829)
(16,3298285)
(20,6860715)
(25,12533375)
(32,24279102)
(40,50279209)
(50,92360820)
(63,173277517)
(79,361812200)
(100,691770383)
}

\def\postoptimcznormaltotalspsegmentavg{
(13,5526289)
(16,9867931)
(20,21574766)
(25,41183758)
(32,83529828)
(40,181955716)
(50,349233343)
(63,687069768)
(79,1509877670)
(100,3009991730)
}

\def\postoptimcznormaltotalspsegmentfitavg{
(13,5225538)
(16,9307742)
(20,20346934)
(25,38899617)
(32,78981174)
(40,172126310)
(50,330628827)
(63,650899815)
(79,1431817728)
(100,2856431311)
}

\def\postoptimcznormalcnotfstanaavg{
(13,96890785446)
(16,273631208352)
(20,835056177720)
(25,2548389214200)
(32,8756198665920)
(40,26721797686560)
(50,81548454853800)
(63,258981333190776)
(79,802973423476386)
(100,2609550555321000)
}

\def\postoptimcznormalcnotfstminavg{
(13,35644180728)
(16,100663394304)
(20,307200192000)
(25,937500375000)
(32,3221226258432)
(40,9830401536000)
(50,30000003000000)
(63,95273914129128)
(79,295397426136936)
(100,960000024000000)
}

\def\postoptimcznormalcnotfstcomavg{
(13,1714025586)
(16,3932998752)
(20,9602047680)
(25,23442499050)
(32,62927979264)
(40,153632761440)
(50,375079983600)
(63,945379256472)
(79,2337503317452)
(100,6001279732200)
}

\def\postoptimcznormalcnotfstfitavg{
(13,29054298)
(16,53764512)
(20,104168880)
(25,201827400)
(32,419519232)
(40,812820000)
(50,1574841300)
(63,3124182474)
(79,6110249154)
(100,12288348600)
}

\def\postoptimcznormalcnotsndanaavg{
(13,155627082)
(16,357101472)
(20,871829640)
(25,2128490250)
(32,5713622080)
(40,13949272400)
(50,34055840500)
(63,85836851982)
(79,212236437694)
(100,544893442600)
}

\def\postoptimcznormalcnotsndminavg{
(13,94462914)
(16,216724000)
(20,529047400)
(25,1291494450)
(32,3466535488)
(40,8462709200)
(50,20659910900)
(63,52070683518)
(79,128743683046)
(100,330526577800)
}

\def\postoptimcznormalcnotsndcomavg{
(13,12930866)
(16,27186464)
(20,60313640)
(25,133641250)
(32,321962304)
(40,712526480)
(50,1576696100)
(63,3589719966)
(79,8036737150)
(100,18614429000)
}

\def\postoptimcznormalcnotsndfitavg{
(13,513474)
(16,934176)
(20,1777480)
(25,3381650)
(32,6889280)
(40,13107600)
(50,24939300)
(63,48552462)
(79,93225846)
(100,183925000)
}

\def\postoptimcznormalcnotfrthanaavg{
(13,174876026)
(16,361699872)
(20,789829640)
(25,1724719050)
(32,4092166464)
(40,8935899280)
(50,19512968100)
(63,43814738646)
(79,96743795158)
(100,220764020200)
}

\def\postoptimcznormalcnotfrthminavg{
(13,137069946)
(16,283326112)
(20,618312040)
(25,1349448050)
(32,3200062784)
(40,6984838480)
(50,15246632100)
(63,34222824846)
(79,75541102078)
(100,172330728200)
}

\def\postoptimcznormalcnotfrthcomavg{
(13,12099386)
(16,24268192)
(20,51088040)
(25,107213050)
(32,242748224)
(40,507112080)
(50,1057910100)
(63,2263333086)
(79,4763276718)
(100,10335816200)
}

\def\postoptimcznormalcnotfrthfitavg{
(13,113386)
(16,192672)
(20,340040)
(25,602050)
(32,1130304)
(40,1998480)
(50,3532100)
(63,6373206)
(79,11363518)
(100,20748200)
}

\def\postoptimcznormalcnotsxthanaavg{
(13,1554433426)
(16,3105708832)
(20,6534224040)
(25,13747615050)
(32,31303564864)
(40,65860832080)
(50,138567240100)
(63,299384303526)
(79,636576849358)
(100,1396670200200)
}

\def\postoptimcznormalcnotsxthminavg{
(13,1330365426)
(16,2656092832)
(20,5584180040)
(25,11740810050)
(32,26715462464)
(40,56175040080)
(50,118125040100)
(63,255084895926)
(79,542127656758)
(100,1188903800200)
}

\def\postoptimcznormalcnotsxthfitavg{
(13,135226)
(16,217632)
(20,364040)
(25,610050)
(32,1081664)
(40,1808080)
(50,3030100)
(63,5166126)
(79,8705958)
(100,15020200)
}

\def\postoptimcznormalcnoteigthanaavg{
(13,23171954026)
(16,45502656032)
(20,93971080040)
(25,194067000050)
(32,432896512064)
(40,894008480080)
(50,1846286850100)
(63,3912942285126)
(79,8164596043158)
(100,17564939400200)
}

\def\postoptimcznormalcnoteigthminavg{
(13,20721987026)
(16,40664736032)
(20,83923240040)
(25,173205700050)
(32,386104192064)
(40,796919680080)
(50,1644893450100)
(63,3484275543126)
(79,7266599962158)
(100,15625505900200)
}

\def\postoptimcznormalcnoteigthfitavg{
(13,286026)
(16,448032)
(20,700040)
(25,1150050)
(32,1920064)
(40,3080080)
(50,5000100)
(63,8190126)
(79,13272158)
(100,21900200)
}

\def\postoptimcznormalcnotlcuavg{
(13,15131174)
(16,26149939)
(20,59337878)
(25,120551221)
(32,231981742)
(40,529826220)
(50,954541356)
(63,1769858741)
(79,4052668785)
(100,8382277740)
}

\def\postoptimcznormalcnotspavg{
(13,1190312)
(16,2086016)
(20,4469420)
(25,8265332)
(32,16270660)
(40,34795458)
(50,64532226)
(63,122808134)
(79,264958428)
(100,509013696)
}

\def\postoptimcznormalcnotspjaavg{
(13,917256)
(16,1589180)
(20,3366520)
(25,6143260)
(32,11922780)
(40,25164276)
(50,46041810)
(63,86403514)
(79,183887964)
(100,348240070)
}

\def\postoptimcznormalcnotspsegmentavg{
(13,2675988)
(16,4829152)
(20,10734948)
(25,20458480)
(32,41713000)
(40,92540232)
(50,176811840)
(63,347219488)
(79,777502080)
(100,1540964784)
}

\def\postoptimcznormalcnotspsegmentfitavg{
(13,2530356)
(16,4555008)
(20,10124016)
(25,19323808)
(32,39441500)
(40,87541128)
(50,167392640)
(63,328940540)
(79,737305600)
(100,1462349552)
}

\def\postoptimcznormalrzfstanaavg{
(13,64593856964)
(16,182420805568)
(20,556704118480)
(25,1698926142800)
(32,5837465777280)
(40,17814531791040)
(50,54365636569200)
(63,172654222127184)
(79,535315615650924)
(100,1739700370214000)
}

\def\postoptimcznormalrzfstminavg{
(13,23762787152)
(16,67108929536)
(20,204800128000)
(25,625000250000)
(32,2147484172288)
(40,6553601024000)
(50,20000002000000)
(63,63515942752752)
(79,196931617424624)
(100,640000016000000)
}

\def\postoptimcznormalrzfstcomavg{
(13,1142683724)
(16,2621999168)
(20,6401365120)
(25,15628332700)
(32,41951986176)
(40,102421840960)
(50,250053322400)
(63,630252837648)
(79,1558335544968)
(100,4000853154800)
}

\def\postoptimcznormalrzfstfitavg{
(13,19369532)
(16,35843008)
(20,69445920)
(25,134551600)
(32,279679488)
(40,541880000)
(50,1049894200)
(63,2082788316)
(79,4073499436)
(100,8192232400)
}

\def\postoptimcznormalrzsndanaavg{
(13,97266923)
(16,223188416)
(20,544893520)
(25,1330306400)
(32,3571013792)
(40,8718295240)
(50,21284900300)
(63,53648032473)
(79,132647773539)
(100,340558401600)
}

\def\postoptimcznormalrzsndminavg{
(13,59039318)
(16,135452496)
(20,330654620)
(25,807184025)
(32,2166584672)
(40,5289193240)
(50,12912444300)
(63,32544177183)
(79,80464801884)
(100,206579111100)
}

\def\postoptimcznormalrzsndcomavg{
(13,8081788)
(16,16991536)
(20,37696020)
(25,83525775)
(32,201226432)
(40,445329040)
(50,985435050)
(63,2243574963)
(79,5022960699)
(100,11634018100)
}

\def\postoptimcznormalrzsndfitavg{
(13,320918)
(16,583856)
(20,1110920)
(25,2113525)
(32,4305792)
(40,8192240)
(50,15587050)
(63,30345273)
(79,58266134)
(100,114953100)
}

\def\postoptimcznormalrzfrthanaavg{
(13,109297513)
(16,226062416)
(20,493643520)
(25,1077949400)
(32,2557604032)
(40,5584937040)
(50,12195605050)
(63,27384211638)
(79,60464871954)
(100,137977512600)
}

\def\postoptimcznormalrzfrthminavg{
(13,85668713)
(16,177078816)
(20,386445020)
(25,843405025)
(32,2000039232)
(40,4365524040)
(50,9529145050)
(63,21389265513)
(79,47213188779)
(100,107706705100)
}

\def\postoptimcznormalrzfrthcomavg{
(13,7562113)
(16,15167616)
(20,31930020)
(25,67008150)
(32,151717632)
(40,316945040)
(50,661193800)
(63,1414583163)
(79,2977047929)
(100,6459885100)
}

\def\postoptimcznormalrzfrthfitavg{
(13,70863)
(16,120416)
(20,212520)
(25,376275)
(32,706432)
(40,1249040)
(50,2207550)
(63,3983238)
(79,7102179)
(100,12967600)
}

\def\postoptimcznormalrzsxthanaavg{
(13,971520888)
(16,1941068016)
(20,4083890020)
(25,8592259400)
(32,19564728032)
(40,41163020040)
(50,86604525050)
(63,187115189688)
(79,397860530829)
(100,872918875100)
}

\def\postoptimcznormalrzsxthminavg{
(13,831478388)
(16,1660058016)
(20,3490112520)
(25,7338006275)
(32,16697164032)
(40,35109400040)
(50,73828150050)
(63,159428059938)
(79,338829785454)
(100,743064875100)
}

\def\postoptimcznormalrzsxthfitavg{
(13,84513)
(16,136016)
(20,227520)
(25,381275)
(32,676032)
(40,1130040)
(50,1893800)
(63,3228813)
(79,5441204)
(100,9387600)
}

\def\postoptimcznormalrzeigthanaavg{
(13,14482471263)
(16,28439160016)
(20,58731925020)
(25,121291875025)
(32,270560320032)
(40,558755300040)
(50,1153929281300)
(63,2445588928188)
(79,5102872526954)
(100,10978087125100)
}

\def\postoptimcznormalrzeigthminavg{
(13,12951241888)
(16,25415460016)
(20,52452025020)
(25,108253562525)
(32,241315120032)
(40,498074800040)
(50,1028058406300)
(63,2177672214438)
(79,4541624976329)
(100,9765941187600)
}

\def\postoptimcznormalrzeigthfitavg{
(13,178763)
(16,280016)
(20,437520)
(25,718775)
(32,1200032)
(40,1925040)
(50,3125050)
(63,5118813)
(79,8295079)
(100,13687600)
}

\def\postoptimcznormalrzlcuavg{
(13,2641346)
(16,4009709)
(20,12410183)
(25,21861035)
(32,35842739)
(40,112179375)
(50,175119150)
(63,277788600)
(79,870496473)
(100,1548544755)
}

\def\postoptimcznormalrzspavg{
(13,118399)
(16,179071)
(20,558335)
(25,871679)
(32,1426687)
(40,4456447)
(50,6962175)
(63,11051007)
(79,34749439)
(100,55675903)
}

\def\postoptimcznormalrzspjaavg{
(13,91263)
(16,136447)
(20,420607)
(25,647935)
(32,1045503)
(40,3223039)
(50,4967423)
(63,7775231)
(79,24117247)
(100,38090751)
}

\def\postoptimcznormalrzspsegmentavg{
(13,282093)
(16,439451)
(20,1416591)
(25,2284205)
(32,3877390)
(40,12509091)
(50,20180412)
(63,33111848)
(79,107543718)
(100,178225836)
}

\def\postoptimcznormalrzspsegmentfitavg{
(13,266741)
(16,414504)
(20,1335972)
(25,2157518)
(32,3666245)
(40,11833339)
(50,19105352)
(63,31368715)
(79,101983760)
(100,169133308)
}

\def\postoptimczlargetotalfstanaavg{
(13,242226963615)
(17,926297155830)
(23,4198990612635)
(31,18677304223065)
(42,85261344260400)
(57,392536460369160)
(78,1883556868831560)
(106,8730418247926141)
(145,41816312705638200)
(197,193569467092059424)
(269,918897613542428800)
(367,4343462185594183680)
(500,20387113713442836480)
}

\def\postoptimczlargetotalfstminavg{
(13,89110451820)
(17,340765974780)
(23,1544723050020)
(31,6870998027460)
(42,31365900125280)
(57,144406104791580)
(78,692921876793120)
(106,3211741457700960)
(145,15383361932917500)
(197,71210227840402384)
(269,338043541731666560)
(367,1597870444551531776)
(500,7500000007499999232)
}

\def\postoptimczlargetotalfstcomavg{
(13,4285063965)
(17,12530822400)
(23,41985103440)
(31,138557697030)
(42,466853954490)
(57,1583737872435)
(78,5553442635390)
(106,18941193480180)
(145,66321736393875)
(197,225968958404175)
(269,785584668329055)
(367,2721749401048680)
(500,9376999575120000)
}

\def\postoptimczlargetotalfstfitavg{
(13,72635745)
(17,160869810)
(23,394082115)
(31,954604080)
(42,2348224830)
(57,5805551745)
(78,14709589830)
(106,36512405430)
(145,92412420075)
(197,229210848465)
(269,577065024660)
(367,1449141941715)
(500,3624011055000)
}

\def\postoptimczlargetotalsndanaavg{
(13,369614349)
(17,1080863740)
(23,3621484274)
(31,11951488156)
(42,40269140730)
(57,136607524794)
(78,479019955908)
(106,1633799114882)
(145,5720674060330)
(197,19491268306038)
(269,67761703443127)
(367,234768298320984)
(500,808826203043500)
}

\def\postoptimczlargetotalsndminavg{
(13,224349450)
(17,655950457)
(23,2197465586)
(31,7251208078)
(42,24430099176)
(57,82870726116)
(78,290576354094)
(106,991039769740)
(145,3469995899725)
(197,11822632858114)
(269,41101030013584)
(367,142397927558868)
(500,490587390433500)
}

\def\postoptimczlargetotalsndcomavg{
(13,30710836)
(17,80186807)
(23,235798806)
(31,682923924)
(42,2013203472)
(57,5970044397)
(78,18240493596)
(106,54415110706)
(145,166351937155)
(197,496942102568)
(269,1513577092429)
(367,4602794421491)
(500,13948282426000)
}

\def\postoptimczlargetotalsndfitavg{
(13,1219530)
(17,2642582)
(23,6315685)
(31,14931367)
(42,35832090)
(57,86414052)
(78,213429294)
(106,516716610)
(145,1274882775)
(197,3084214664)
(269,7570527525)
(367,18536781714)
(500,45204813000)
}

\def\postoptimczlargetotalfrthanaavg{
(13,415330591)
(17,1062095179)
(23,3059421866)
(31,8696767807)
(42,25174709784)
(57,73308427434)
(78,219747033246)
(106,642928143252)
(145,1924774745115)
(197,5626310476429)
(269,16738821644323)
(367,49650442829454)
(500,146550369083500)
}

\def\postoptimczlargetotalfrthminavg{
(13,325541151)
(17,831815559)
(23,2394211726)
(31,6801292182)
(42,19676342154)
(57,57268325619)
(78,171589391376)
(106,501843133292)
(145,1501909478290)
(197,4389032107009)
(269,13054686168188)
(367,38714734624774)
(500,114252230513500)
}

\def\postoptimczlargetotalfrthcomavg{
(13,28736071)
(17,70580464)
(23,193088611)
(31,519082972)
(42,1414610904)
(57,3867409644)
(78,10848862986)
(106,29727929702)
(145,83236829790)
(197,228016318674)
(269,635706871868)
(367,1769865629784)
(500,4913191241000)
}

\def\postoptimczlargetotalfrthfitavg{
(13,269321)
(17,534684)
(23,1156026)
(31,2474017)
(42,5374824)
(57,11723874)
(78,26120796)
(106,57178202)
(145,127295790)
(197,278480579)
(269,617155133)
(367,1364688399)
(500,3006848500)
}

\def\postoptimczlargecnotfstanaavg{
(13,96890785446)
(17,370518862332)
(23,1679596245054)
(31,7470921689226)
(42,34104537704160)
(57,157014584147664)
(78,753422747532624)
(106,3492167299170456)
(145,16726525082255280)
(197,77427786836823776)
(269,367559045416971520)
(367,1737384874237673728)
(500,8154845485377134592)
}

\def\postoptimczlargecnotfstminavg{
(13,35644180728)
(17,136306389912)
(23,617889220008)
(31,2748399210984)
(42,12546360050112)
(57,57762441916632)
(78,277168750717248)
(106,1284696583080384)
(145,6153344773167000)
(197,28484091136160952)
(269,135217416692666624)
(367,639148177820612736)
(500,3000000002999999488)
}

\def\postoptimczlargecnotfstcomavg{
(13,1714025586)
(17,5012328960)
(23,16794041376)
(31,55423078812)
(42,186741581796)
(57,633495148974)
(78,2221377054156)
(106,7576477392072)
(145,26528694557550)
(197,90387583361670)
(269,314233867331622)
(367,1088699760419472)
(500,3750799830048000)
}

\def\postoptimczlargecnotfstfitavg{
(13,29054298)
(17,64347924)
(23,157632846)
(31,381841632)
(42,939289932)
(57,2322220698)
(78,5883835932)
(106,14604962172)
(145,36964968030)
(197,91684339386)
(269,230826009864)
(367,579656776686)
(500,1449604422000)
}

\def\postoptimczlargecnotsndanaavg{
(13,155627082)
(17,455100506)
(23,1524835462)
(31,5032205510)
(42,16955427636)
(57,57518957754)
(78,201692612940)
(106,687915416692)
(145,2408704867370)
(197,8206849812882)
(269,28531243554746)
(367,98849809819014)
(500,340558401281000)
}

\def\postoptimczlargecnotsndminavg{
(13,94462914)
(17,276189650)
(23,925248646)
(31,3053140214)
(42,10286357508)
(57,34892937258)
(78,122347938492)
(106,417279902948)
(145,1461050905010)
(197,4977950676914)
(269,17305696847570)
(367,59957022129702)
(500,206563111761000)
}

\def\postoptimczlargecnotsndcomavg{
(13,12930866)
(17,33762850)
(23,99283686)
(31,287546886)
(42,847664580)
(57,2513702850)
(78,7680207756)
(106,22911625460)
(145,70042920770)
(197,209238779842)
(269,637295617610)
(367,1938018703438)
(500,5872961021000)
}

\def\postoptimczlargecnotsndfitavg{
(13,513474)
(17,1112650)
(23,2659214)
(31,6286862)
(42,15087156)
(57,36384810)
(78,89864892)
(106,217564788)
(145,536792610)
(197,1298616514)
(269,3187590282)
(367,7804960374)
(500,19033605000)
}

\def\postoptimczlargecnotfrthanaavg{
(13,174876026)
(17,447197954)
(23,1288177606)
(31,3661796942)
(42,10599877764)
(57,30866706234)
(78,92525066556)
(106,270706586532)
(145,810431471490)
(197,2368972831994)
(269,7047924902618)
(367,20905449612054)
(500,61705418561000)
}

\def\postoptimczlargecnotfrthminavg{
(13,137069946)
(17,350238114)
(23,1008089126)
(31,2863701942)
(42,8284775604)
(57,24112979154)
(78,72248164716)
(106,211302371812)
(145,632382938090)
(197,1848013518554)
(269,5496709965298)
(367,16300940894294)
(500,48106202321000)
}

\def\postoptimczlargecnotfrthcomavg{
(13,12099386)
(17,29718074)
(23,81300446)
(31,218561222)
(42,595625604)
(57,1628382954)
(78,4567942236)
(106,12517022932)
(145,35047086090)
(197,96006870834)
(269,267666051058)
(367,745206580614)
(500,2068712101000)
}

\def\postoptimczlargecnotfrthfitavg{
(13,113386)
(17,225114)
(23,486726)
(31,1041662)
(42,2263044)
(57,4936314)
(78,10998156)
(106,24074932)
(145,53598090)
(197,117254794)
(269,259854538)
(367,574605294)
(500,1266041000)
}

\def\postoptimczlargerzfstanaavg{
(13,64593856964)
(17,247012574888)
(23,1119730830036)
(31,4980614459484)
(42,22736358469440)
(57,104676389431776)
(78,502281831688416)
(106,2328111532780304)
(145,11151016721503520)
(197,51618524557882512)
(269,245039363611314336)
(367,1158256582825115904)
(500,5436563656918089728)
}

\def\postoptimczlargerzfstminavg{
(13,23762787152)
(17,90870926608)
(23,411926146672)
(31,1832266140656)
(42,8364240033408)
(57,38508294611088)
(78,184779167144832)
(106,856464388720256)
(145,4102229848778000)
(197,18989394090773968)
(269,90144944461777744)
(367,426098785213741824)
(500,2000000001999999744)
}

\def\postoptimczlargerzfstcomavg{
(13,1142683724)
(17,3341552640)
(23,11196027584)
(31,36948719208)
(42,124494387864)
(57,422330099316)
(78,1480918036104)
(106,5050984928048)
(145,17685796371700)
(197,60258388907780)
(269,209489244887748)
(367,725799840279648)
(500,2500533220032000)
}

\def\postoptimczlargerzfstfitavg{
(13,19369532)
(17,42898616)
(23,105088564)
(31,254561088)
(42,626193288)
(57,1548147132)
(78,3922557288)
(106,9736641448)
(145,24643312020)
(197,61122892924)
(269,153884006576)
(367,386437851124)
(500,966402948000)
}

\def\postoptimczlargerzsndanaavg{
(13,97266923)
(17,284437812)
(23,953022158)
(31,3145128436)
(42,10597142262)
(57,35949348582)
(78,126057883068)
(106,429947135406)
(145,1505440542070)
(197,5129281133002)
(269,17832027221649)
(367,61781131136792)
(500,212849000800500)
}

\def\postoptimczlargerzsndminavg{
(13,59039318)
(17,172618527)
(23,578280398)
(31,1908212626)
(42,6428973432)
(57,21808085772)
(78,76467461538)
(106,260799939316)
(145,913156815595)
(197,3111219173022)
(269,10816060529664)
(367,37473138830972)
(500,129101944850500)
}

\def\postoptimczlargerzsndcomavg{
(13,8081788)
(17,21101777)
(23,62052298)
(31,179716796)
(42,529790352)
(57,1571064267)
(78,4800129828)
(106,14319765886)
(145,43776825445)
(197,130774237352)
(269,398309760939)
(367,1211261689557)
(500,3670600638000)
}

\def\postoptimczlargerzsndfitavg{
(13,320918)
(17,695402)
(23,1662003)
(31,3929281)
(42,9429462)
(57,22740492)
(78,56165538)
(106,135977966)
(145,335495345)
(197,811635272)
(269,1992243859)
(367,4878100142)
(500,11896003000)
}

\def\postoptimczlargerzfrthanaavg{
(13,109297513)
(17,279498717)
(23,805110998)
(31,2288623081)
(42,6624923592)
(57,19291691382)
(78,57828166578)
(106,169191616556)
(145,506519669645)
(197,1480608019947)
(269,4404953064069)
(367,13065906007442)
(500,38565886600500)
}

\def\postoptimczlargerzfrthminavg{
(13,85668713)
(17,218898817)
(23,630055698)
(31,1789813706)
(42,5177984742)
(57,15070611957)
(78,45155102928)
(106,132063982356)
(145,395239336270)
(197,1155008449047)
(269,3435443728244)
(367,10188088058842)
(500,30066376450500)
}

\def\postoptimczlargerzfrthcomavg{
(13,7562113)
(17,18573792)
(23,50812773)
(31,136600756)
(42,372265992)
(57,1017739332)
(78,2854963878)
(106,7823139306)
(145,21904428770)
(197,60004294222)
(269,167291281844)
(367,465754112792)
(500,1292945063000)
}

\def\postoptimczlargerzfrthfitavg{
(13,70863)
(17,140692)
(23,304198)
(31,651031)
(42,1414392)
(57,3085182)
(78,6873828)
(106,15046806)
(145,33498770)
(197,73284197)
(269,162409019)
(367,359128217)
(500,791275500)
}

\def\postoptimcznormalcnotbestfitavg{
(13,113386)
(16,192672)
(20,340040)
(25,602050)
(32,1081664)
(40,1808080)
(50,3030100)
(63,5166126)
(79,8705958)
(100,15020200)
}

 \def\postoptimctnormaltbestfitavg{
 (                   13 ,              5846221 )
 (                   16 ,             10172704 )
 (                   20 ,             18432690 )
 (                   25 ,             33668686 )
 (                   32 ,             61971256 )
 (                   40 ,            105675143 )
 (                   50 ,            180892328 )
 (                   63 ,            316551594 )
 (                   79 ,            546220371 )
 (                  100 ,            970090617 )
 }

%% file: PlotQubits.tex
\begin{tikzpicture}
  \begin{axis}[
    width=10cm,
    grid=major,
    legend style={at={(0.95,0.05)},anchor=south east},
    xlabel={System size},
    ylabel={Qubits},
    every axis legend/.append style={nodes={right}}
    ]
    
    \addplot[only marks, red] coordinates {
      (5,5)
      (10,10)
      (15,15)
      (20,20)
      (25,25)
      (30,30)
      (35,35)
      (40,40)
      (45,45)
      (50,50)
      (55,55)
      (60,60)
      (65,65)
      (70,70)
      (75,75)
      (80,80)
      (85,85)
      (90,90)
      (95,95)
      (100,100)
    };
    \addlegendentry{PF}

    \addplot[only marks, blue] coordinates {
      (5,68)
      (10,94)
      (15,99)
      (20,116)
      (25,133)
      (30,138)
      (35,156)
      (40,161)
      (45,166)
      (50,171)
      (55,176)
      (60,181)
      (65,200)
      (70,205)
      (75,210)
      (80,215)
      (85,235)
      (90,240)
      (95,245)
      (100,250)
    };
    \addlegendentry{TS}
 
    \addplot[only marks, black!30!green] coordinates {
      (5,16)
      (10,23)
      (15,28)
      (20,35)
      (25,40)
      (30,45)
      (35,52)
      (40,57)
      (45,62)
      (50,67)
      (55,72)
      (60,77)
      (65,84)
      (70,89)
      (75,94)
      (80,99)
      (85,104)
      (90,109)
      (95,114)
      (100,119)
      };
      \addlegendentry{QSP}
      \end{axis}
\end{tikzpicture}

%% file: min12468.tex
\begin{tikzpicture}
  \begin{loglogaxis}[
    width=10cm,
    ymajorgrids=true,
    legend style={at={(0.05,0.95)},anchor=north west,font=\footnotesize},
    xlabel={System size},
    ylabel={Total gate count},
    xmin=10,
    xmax=600,
    xtick={10,100},
    xticklabels={10,100},
    extra x ticks={30,300},
    extra x tick labels={30,300},
    every axis legend/.append style={nodes={right}}
    ]

    \addlegendimage{empty legend}
    \addlegendentry[yshift=0pt]{\hspace{-.25cm}\textbf{Minimized}}
    
    \addplot[only marks, red] coordinates {
      \preoptimczlargetotalfstminavg
    };
    \addlegendentry{First order}
    
    \addplot[only marks, red, mark=triangle*] coordinates {
      \preoptimczlargetotalsndminavg
    };
    \addlegendentry{Second order}
    
    \addplot[only marks, red, fill opacity=.5, mark=square*] coordinates {
      \preoptimczlargetotalfrthminavg
    };
    \addlegendentry{Fourth order}

    \addplot[only marks, red, fill opacity=.5, mark=hexagon*] coordinates {
      \preoptimczlargetotalsxthminavg
    };
    \addlegendentry{Sixth order}

    \addplot[only marks, red, fill opacity=.5, mark=octagon*] coordinates {
      \preoptimczlargetotaleigthminavg
    };        
    \addlegendentry{Eighth order}    
  \end{loglogaxis}
\end{tikzpicture}

%% file: com124.tex
\begin{tikzpicture}
  \begin{loglogaxis}[
    width=10cm,
    ymajorgrids=true,
    legend style={at={(0.05,0.95)},anchor=north west,font=\footnotesize},
    xlabel={System size},
    ylabel={Total gate count},
    xmin=10,
    xmax=600,
    xtick={10,100},
    xticklabels={10,100},
    extra x ticks={30,300},
    extra x tick labels={30,300},
    every axis legend/.append style={nodes={right}}
    ]

    \addlegendimage{empty legend}
    \addlegendentry[yshift=0pt]{\hspace{-.25cm}\textbf{Commutator}}
    
    \addplot[only marks, red] coordinates {
      \preoptimczlargetotalfstcomavg
    };
    \addlegendentry{First order}

    \addplot[only marks, mark=triangle*, red] coordinates {
      \preoptimczlargetotalsndcomavg
    };
    \addlegendentry{Second order}
    
    \addplot[only marks, fill opacity=.5, red, mark=square*] coordinates {
      \preoptimczlargetotalfrthcomavg
    };
    \addlegendentry{Fourth order}

  \end{loglogaxis}
\end{tikzpicture}

%% file: fit12468.tex
\begin{tikzpicture}
  \begin{loglogaxis}[
    width=10cm,
    ymajorgrids=true,
    legend style={at={(0.05,0.95)},anchor=north west,font=\footnotesize},
    xlabel={System size},
    ylabel={Total gate count},
    xmin=10,
    xmax=600,
    xtick={10,100},
    xticklabels={10,100},
    extra x ticks={30,300},
    extra x tick labels={30,300},
    every axis legend/.append style={nodes={right}}
    ]

    \addlegendimage{empty legend}
    \addlegendentry[yshift=0pt]{\hspace{-.25cm}\textbf{Empirical}}
    
    \addplot[only marks, red, mark options={fill=white,line width=1pt}] coordinates {
      \preoptimczlargetotalfstfitavg
    };
    \addlegendentry{First order}
    
    \addplot[only marks, red, mark=triangle*, mark options={fill=white,line width=1pt}] coordinates {
      \preoptimczlargetotalsndfitavg
    };
    \addlegendentry{Second order}
    
    \addplot[only marks, red, fill opacity=0, mark=square*, mark options={fill=white,line width=1pt}] coordinates {
      \preoptimczlargetotalfrthfitavg
    };
    \addlegendentry{Fourth order}

    \addplot[only marks, red, fill opacity=0, draw opacity=.66, mark=hexagon*, mark options={fill=white,line width=1pt}] coordinates {
      \preoptimczlargetotalsxthfitavg
    };
    \addlegendentry{Sixth order}    

    \addplot[only marks, red, fill opacity=0, draw opacity=.5, mark=octagon*, mark options={fill=white,line width=1pt}] coordinates {
      \preoptimczlargetotaleigthfitavg
    };        
    \addlegendentry{Eighth order}    

  \end{loglogaxis}
\end{tikzpicture}

%% file: related.tex
In this appendix, we discuss related work on quantum algorithms for simulating physics and resource estimates for practical quantum computation.

The idea of simulating physics with quantum computers was suggested by Feynman \cite{Fey82} and others in the 1980s. Since that time, there has been significant progress on the development of algorithms for the basic problem of simulating Hamiltonian dynamics \cite{Llo96,AT03,Chi04,BACS05,Chi10,CK10,BC12,BCCKS13,BCCKS14,BCK15,LC16,LC17}. In addition, many authors have developed methods for using quantum computers to simulate the behavior of specific types of systems. Quantum computers presumably offer an advantage for simulations of any system in which quantum mechanics plays a significant role, including fermionic lattice models \cite{SOGKL02a,SOGKL02b,RWS12,WHWCNT15}, quantum chemistry \cite{KWPYA11,HWBT15,WBCHT14,Pou15,Bab16,Bab15,RWSWT17,BMWAW15}, and quantum field theories \cite{JLP12,JLP14a,JLP14b}. Much of this work has focused on asymptotic analysis, although some work on concrete resource estimates for quantum chemistry simulation is mentioned below.

There has been considerable previous research on compiling quantum algorithms into explicit circuits, including work on the IARPA Quantum Computer Science (QCS) program \cite{QCS}.  The QCS program focused on synthesizing a selection of quantum algorithms
and developing optimized implementations over certain models of physical machines that were chosen to describe realistic devices.  To the best of our knowledge, none of these studies aimed to construct a minimal example of a super-classical quantum computation, and typical resource counts were high.

In addition, many researchers have developed optimized implementations of Shor's integer factoring algorithm \cite{BCDP96,Bea03,FDH04,Kun05,Kut06,TK06,PG14,HRS16} and estimated resource requirements for
simulating quantum chemistry \cite{WBCHT14,HWBT15,RWSWT17}. As discussed in \sec{results}, our results suggest that quantum simulation of spin models will be accessible with dramatically fewer computational resources, making this a promising candidate for an early demonstration of practical quantum computation.

The best implementation of the factoring algorithm that we are aware of appears in \cite{Kut06}. That paper does not give explicit resource counts, so we estimate them as follows. The implementation of Shor's algorithm described there uses $4n^3 + O(n^2 \log(n))$ gates and $3n + 6\log(n) + O(1)$ qubits \cite[page 10]{Kut06}.  However, this count assumes we can perform arbitrary 2-qubit gates.  The dominant contribution to the gate count comes from modular exponentiation (in particular, the cost of the QFT is small), which relies mainly on so-called psuedo-Toffoli gates.
Each of these gates is realized with two controlled-$H$ gates and one controlled-$Z$ gate; the former needs two $T$ gates and the latter needs none, so the total number of $T$ gates is approximately $\frac{16}{3}n^3$.

Note that it is possible to factor an $n$-bit number using only about $2n$ qubits \cite{Bea03,HRS16}, but at the expense of a significantly higher gate count. Similarly, gate counts for quantum algorithms the elliptic curve discrete logarithm problem can also be reduced at the cost of using significantly more qubits \cite{CMMP09}.

%% file: targetsystem.tex
In this appendix we motivate our choice of a candidate spin system to simulate with an early universal quantum computer, elaborating on the brief discussion in \sec{targetsystem}.

Recently, there has been considerable interest among the condensed matter community in understanding the equilibration of closed quantum systems \cite{NH15}.  While equilibration is normally viewed as a consequence of coupling a system to a bath, a large but isolated quantum system may effectively display features of thermalization through its own unitary dynamics---or it may fail to self-thermalize through a phenomenon known as \emph{many-body localization}.  Despite intense study, the details of phase transitions between localized and thermal phases in various model systems remain poorly understood.  A major challenge is the difficulty of simulating quantum systems with classical computers, which has restricted numerical investigations to systems with fewer than 25 qubits.

To produce concrete benchmarks, we focus on a specific simulation task with the potential for practical applications. As discussed in \sec{targetsystem}, we consider a one-dimensional nearest-neighbor Heisenberg model with a random magnetic field in the $z$ direction, described by the Hamiltonian \eq{heisenberg}.  This Hamiltonian has been studied as a model of self-thermalization and many-body localization \cite{PH10,NH15,LLA15}, where the most extensive numerical study we are aware of was restricted to at most 22 spins \cite{LLA15}.  The nature of the phase transition as a function of $h$ remains unclear \cite[Section 6.2]{NH15} and could be illuminated by larger-scale simulations.

Classical numerics typically involve performing exact diagonalization to evaluate properties of the full spectrum of the Hamiltonian (which is inefficient, and even more costly than simulating the dynamics classically).  To apply efficient quantum simulation, we must choose an initial state and final measurement so that the outcomes are informative (say, so they provide useful information about the phase diagram).  While more limited than a calculation of the full spectrum, a quantum simulation performed on a universal quantum computer should be able to efficiently extract any information that can be probed by an experiment (in addition to other quantities that could be hard to extract experimentally).  Experimental probes of thermalization and many-body localization typically involve preparing a product initial state and performing a final standard-basis measurement to see how the outcomes evolve in time---whether they retain memory of the initial state or approach a thermal distribution.  Specific observables include the Hamming distance from an initial classical configuration \cite{Smi16} and the imbalance between occupation of even and odd sites \cite{Sch15}. Another proposal considers performing a spin-echo sequence (which involves simulating evolution backward in time, something that is easy to accomplish using digital simulation) \cite{Ser14}.

One might also perform similar tests on a (randomly sampled) eigenstate of the Hamiltonian, using phase estimation \cite{BN14}.  However, this introduces additional overhead.  Since we expect that it will be informative (if not decisive) to study a global quench, we focus here on the cost of simulating dynamics.

While we have considered the Heisenberg model \eq{heisenberg} for concreteness, there is also interest in exploring many-body localization and thermalization (among other phenomena) for diverse spin systems including Ising \cite{KBP14,Smi16} and XXZ \cite{Ser14} models.  Our basic approach to estimating simulation resource requirements would apply to these models essentially unchanged and we expect that similar conclusions about the relative performance of different quantum simulation algorithms would hold.

%% file: simalg.tex
\newcommand{\iso}{\mathcal{V}}
\newcommand{\tseg}{t_{\mathrm{seg}}}
\newcommand{\trem}{t_{\mathrm{rem}}}

As discussed in \sec{implementations}, there are many different quantum algorithms for Hamiltonian simulation, including algorithms based on product formulas \cite{Llo96,Chi04,BACS05}, discrete-time quantum walks \cite{Chi10,BC12,BCK15,LC17}, and linear combinations of unitaries \cite{CW12,BCCKS13,BCCKS14,BCK15,LC17,LC16}. The asymptotic performance of these algorithms is summarized in \tab{algsummary}, both as a function of the simulation time $t$ and allowed error $\epsilon$ (as commonly emphasized in the literature on sparse Hamiltonian simulation) and as a function of the number of qubits $n$ (for the target system described in \sec{targetsystem}).

In this paper, we focus on the algorithms based on product formulas (\PF) \cite{BACS05}, on direct implementation of the Taylor series (\TS) \cite{BCCKS14}, and on Low and Chuang's recent quantum signal processing (\QSP) approach \cite{LC16}.  We expect these three algorithms to be among the best for simulations of spin systems.  Considering the dependence on system size, the algorithm based on quantum walk \cite{BC12} has worse asymptotic performance than the three methods we consider, and requires a costly computation of trigonometric functions performed in quantum superposition. The algorithm based on the fractional-query model \cite{BCCKS13} is conceptually similar to the one that implements the Taylor series \cite{BCCKS14}, but the latter has a streamlined form, resulting in improved asymptotic complexity as a function of $n$. For sparse Hamiltonians, the algorithm based on a linear combination of quantum walk steps \cite{BCK15} has improved query complexity as a function of sparsity over \cite{BCCKS14}, but this improvement is not relevant to local Hamiltonians, and again the gate complexity is higher.

As discussed in \sec{targetsystem}, we focus on a specific type of Hamiltonian that can be described as a sum of $L$ terms, of the form
\begin{equation}\label{eq:hamsum}
  H = \sum_{\ell=1}^{L} \alpha_\ell H_\ell,
\end{equation}
where each $H_\ell$ is a tensor product of Pauli operators acting nontrivially on at most two out of $n$ qubits. We also assume that all the coefficients are positive real numbers $\alpha_\ell>0$, since if $\alpha_\ell < 0$, we can absorb the negative sign into the definition of $H_\ell$. In each of the algorithms described below, our goal is to implement an approximation of the unitary operation $\exp(-iHt)$ for a given time $t \in \R$, up to an allowed error at most $\epsilon > 0$. Although for some applications we might want to simulate evolutions for negative times (e.g., to implement the spin-echo sequence proposed in \cite{Ser14}), we can absorb this into the sign of the Hamiltonian, so we can consider $t>0$ without loss of generality.

In the remainder of this appendix, we review the three algorithms we consider, emphasizing aspects relevant to our implementations.  We consider the \PF\ algorithm in \sec{algpf}, the \TS\ algorithm in \sec{alglcu}, and the \QSP\ algorithm in \sec{algqsp}.  We briefly discuss the system size dependence of other quantum simulation algorithms in \app{algother}.

\subsection{Product formula algorithm}
\label{sec:algpf}

The product formula (\PF) approach approximates the exponential of a sum of operators by a product of exponentials of the individual operators.  The first-order formula
\begin{equation}
\norm*{\exp\biggl(-it\sum_{j=1}^{L}\alpha_jH_{j}\biggr)-\biggl[\prod_{j=1}^{L}\exp\biggl(-\frac{it}{r}\alpha_jH_{j}\biggr)\biggr]^{r}}=O\biggl(\frac{(L \Lambda t)^2}{r}\biggr),
\end{equation}
where $\Lambda := \max_j \alpha_j$, underlies the first explicit quantum simulation algorithm \cite{Llo96}.  The complexity of quantum simulation can be improved \cite{Chi04,BACS05} using the $(2k)$th-order Suzuki formula $S_{2k}$, defined recursively by \cite{Suz91}
\begin{align}
S_{2}(\lambda)&:=\prod_{j=1}^{L}\exp(\alpha_jH_{j}\lambda/2)\prod_{j=L}^{1}\exp(\alpha_jH_{j}\lambda/2)\\
S_{2k}(\lambda)&:=[S_{2k-2}(p_{k}\lambda)]^{2}S_{2k-2}((1-4p_{k})\lambda)[S_{2k-2}(p_{k}\lambda)]^{2},
\label{eq:recursive_def}
\end{align}
with $p_{k}:=1/(4-4^{1/(2k-1)})$ for $k>1$.
Using this improved formula, we have \cite{BACS05}
\begin{equation}
\norm*{\exp\biggl(-it\sum_{j=1}^{L}\alpha_jH_{j}\biggr)- \biggl[S_{2k}\biggl(-\frac{it}{r}\biggr)\biggr]^r}=O\biggl(\frac{(L \Lambda t)^{2k+1}}{r^{2k}}\biggr).
\end{equation}

For the $n$-qubit system described in \sec{targetsystem}, we have $L = O(n)$ and $\Lambda = O(1)$.  Using the first-order formula to simulate that system for time $t=n$ within a fixed error $\epsilon$, it suffices to take $r=O(n^4)$. The circuit for each segment has size $O(n)$, giving an overall gate complexity of $O(n^5)$. Similarly, the complexity of simulation using the $(2k)$th-order formula is $O(n^{3+1/k})$.

The main challenge in making these algorithms concrete is to choose an explicit $r$ that ensures some desired upper bound on the error.  We derive such bounds in \app{pf}.  We consider four bounds, which we call the \emph{analytic bound}, the \emph{minimized bound}, the \emph{commutator bound}, and the \emph{empirical bound}.  The analytic and minimized bounds involve only minor improvements over standard approaches \cite{BACS05}, whereas the commutator bound can be significantly tighter, and the empirical bound attempts to capture the true performance.

The commutator bound, derived in \sec{pfcom}, exploits the fact that the error is smaller if many pairs of terms in the Hamiltonian commute.
This bound not only improves the overall gate count of the \PF\ algorithm, but also tightens the asymptotic complexity with respect to the system size $n$. Specifically, for the ($2k$)th-order formula, this bound gives an asymptotic gate count $O(n^{3+2/(2k+1)})$, improving over the above $O(n^{3+1/k})$.

To evaluate the commutator bound, we must compute the number of $(2k+1)$-tuples of terms from the Hamiltonian satisfying certain commutation relations. A brute-force approach to this counting problem takes time $O(n^{2k+1})$, which can be prohibitive for large $n$, even for modest $k$. However, for certain Hamiltonians (including our target system), we show how to use the combinatorial structure of the bound to compute it in closed form. We evaluate the bound for the first-, second- and fourth-order product formulas, and study its asymptotic behavior for higher-order formulas. We compare the commutator bound to other \PF\ error bounds in \sec{results-pf}.

While the analytic, minimized, and commutator bounds all provide rigorous performance guarantees, even the strongest of these---the commutator bound---is likely to be loose in practical applications. To overcome this, we consider a non-rigorous bound based on extrapolating the actual error seen in small instances using classical simulation.  The details of how we compute this empirical bound are described in \sec{empiricalbounds}.

Although we focus on improving the gate count for small instances, note that the commutator and empirical bounds actually improve the asymptotic performance of the \PF\ algorithm as a function of system size.  We discuss the nature of this improvement in \sec{results-pf}.

\subsection{Taylor series algorithm}
\label{sec:alglcu}

We now summarize the Taylor series (\TS) algorithm \cite{BCCKS14}.
This algorithm directly implements the (truncated) Taylor series of the evolution operator for a carefully-chosen constant time, and repeats that procedure until the entire evolution time has been simulated.

Denote the Taylor series for the evolution up to time $t$, truncated at order $K$, by
\begin{equation}
	\tilde{U}(t):=\sum_{k=0}^{K}\frac{(-itH)^k}{k!}.
  \label{eq:utwiddle}
\end{equation}
For sufficiently large $K$, the operator $\tilde U(t)$ is a good approximation of $\exp(-iHt)$.
Using \eq{hamsum}, we can rewrite $\tilde U(t)$ as a linear combination of unitaries, namely
\begin{align}\label{eq:defnofm}
	\tilde{U}(t)&=\sum_{k=0}^{K}\frac{(-itH)^k}{k!} \\
	&=\sum_{k=0}^{K}\sum_{\ell_1,\ldots,\ell_k=1}^{L}\frac{t^k}{k!}\alpha_{\ell_1}\cdots\alpha_{\ell_k}(-i)^k H_{\ell_1}\cdots H_{\ell_k}\\
	&=\sum_{j=0}^{\numselect-1}\beta_{j}\tilde{V}_{j}, \label{eq:tslcu}
\end{align}
for $\numselect=\sum_{k=0}^{K}L^k$, where the $\tilde V_j$ are products of the form $(-i)^k H_{\ell_1} \cdots H_{\ell_k}$, and the $\beta_j$ are the corresponding coefficients such that $\beta_j>0$. (For notational convenience, we omit the dependence of $\beta_j$ on $t$.) The \TS\ algorithm effectively implements this linear combination.

To do this, for any $t>0$ we define an isometry $\iso(t)\colon \mathcal{H}\rightarrow \C^{\numselect}\otimes\mathcal{H}$ as follows. Let
$B$ be a unitary operator on $\C^\numselect$ satisfying
\begin{align}
  B|0\rangle = \frac{1}{\sqrt{s}}\sum_{j=0}^{\numselect-1}\sqrt{\beta_{j}}|j\rangle,
\end{align}
where
\begin{align}
  s
  &:= \sum_{j=0}^{\numselect-1} \beta_j
  = \sum_{k=0}^K \frac{(t(\alpha_1 + \cdots + \alpha_L))^k}{k!}.
\end{align}
We also define
\begin{align}
	W &:= (B^{\dagger} \otimes I) \select(\tilde V) (B \otimes I)
\end{align}
where
\begin{align}
	\select(\tilde V) &:= \sum_{j=0}^{\numselect-1}|j\rangle\langle j|\otimes \tilde{V}_{j}. \label{eq:defnofselectV}
\end{align}
It is easy to see that $(\langle 0| \otimes I)W(|0\rangle \otimes I) \propto \tilde{U}(t)$. More precisely, we have
\begin{equation}
	W|0\rangle|\psi\rangle=\frac{1}{s}|0\rangle \tilde{U}(t)|\psi\rangle+\sqrt{1-\frac{1}{s^2}}|\Phi\rangle
\end{equation}
for $|\psi\rangle\in\mathcal{H}$ and some $|\Phi\rangle$ whose ancillary state is supported in the subspace orthogonal to $|0\rangle$. To boost the amplitude to perform the desired operation, we consider the isometry
\begin{align}
  \iso(t) := -WRW^{\dagger}RW(|0\rangle\otimes I)
  \label{eq:iso}
\end{align}
where
\begin{align}
	R&:=(I-2|0\rangle\langle 0|)\otimes I.
\end{align}

To ensure that $\iso(t)$ implements evolution according to $H$ nearly deterministically, we consider evolution for time
\begin{align}
  \tseg := \frac{\ln 2}{\alpha_1 + \cdots + \alpha_L}.
\label{eq:tseg}
\end{align}
The overall evolution is realized as a sequence of $r:=\ceil{t/\tseg}$ segments, where the first $r-1$ segments each evolve the state for time $\tseg$ and the final segment evolves the state for time $\trem := t-(r-1)\tseg$. It can be shown that there is a choice of $K$ with
\begin{equation}
	K=O\biggl(\frac{\log\frac{(\alpha_1+\cdots+\alpha_L)\tseg}{\epsilon}}{\log\log\frac{(\alpha_1+\cdots+\alpha_L)\tseg}{\epsilon}}\biggr)
\end{equation}
such that
\begin{align}
	\norm{(\langle 0|\otimes I) \iso(\tseg) - \exp(-i \tseg H)} = O(\epsilon/r).
\end{align}

The evolution for the remaining time $\trem$ can be performed by rotating an ancilla qubit to artificially increase the duration of the segment. Specifically, provided $s<2$, we can introduce an ancilla register in the state $|0\rangle$ and apply the rotation
\begin{equation}
	|0\rangle\mapsto \frac{s}{2}|0\rangle+\sqrt{1-\frac{s^2}{4}}|1\rangle.
  \label{eq:isorot}
\end{equation}
Together with the the unitary operator $W$, this implements the transformation
\begin{equation}
	|0\rangle|\psi\rangle\mapsto \frac{1}{2}|00\rangle \tilde{U}(t)|\psi\rangle+\frac{\sqrt{3}}{2}|\Phi'\rangle
\end{equation}
for some normalized state $|\Phi'\rangle$ with $(\langle 00| \otimes I)|\Phi'\rangle = 0$.  Then we can proceed as before, but with $s=2$.  Indeed, we also perform a similar rotation for the initial $r-1$ segments to ensure that they have $s=2$ instead of a slightly smaller value.
By an abuse of notation, we incorporate this rotation into the definition of $\iso$, so that $\iso(\tseg)$ and $\iso(\trem)$ are the corresponding evolution operators for the first $r-1$ segments and the final segment.

The asymptotic gate complexity of this simulation algorithm is \cite{BCCKS14}
\begin{equation}
	O\biggl(TL(n+\log L)\frac{\log(T/\epsilon)}{\log\log(T/\epsilon)}\biggr),
\end{equation}
where $T=(\alpha_1+\cdots+\alpha_L)t$. For the Hamiltonian studied in this paper, we have $T=O(n^2)$ and $L=O(n)$, which gives a bound of $O(n^4\frac{\log n}{\log\log n})$. However, the analysis in \cite{BCCKS14} assumes only that each term in the Hamiltonian is a tensor product of Pauli gates, possibly acting nontrivially on all of the qubits. Since our Hamiltonian is 2-local, we have a tighter bound of $O(n^3\frac{\log^2 n}{\log\log n})$ for the asymptotic gate count, as indicated in \tab{algsummary}.

To concretely implement the \TS\ algorithm, we must replace the asymptotic statements above with an explicit error analysis. We present the details of such an analysis in \app{lcu}. In particular, \sec{lcuerror} shows how to choose $K$ to ensure that the overall error is at most some allowed $\epsilon$.  In addition, unlike the \PF\ algorithm, the \TS\ algorithm requires a measurement on the ancilla register and succeeds probabilistically. We discuss in \sec{tsfailure} how this fact can be taken into account to make a fair comparison between the \PF\ and \TS\ algorithms.

Another crucial aspect is the implementation of the $\select(\tilde V)$ gate.  In \sec{controlencoding}, we discuss how we encode the control register for this operation.  Then, in \sec{selectV}, we present a novel method to implement $\select(V)$ gates by walking on a binary tree, improving the gate complexity to $O(\numselect)$ from the naive complexity of $O(\numselect\log\numselect)$.

It is also natural to ask whether an empirical error bound could be established for the \TS\ algorithm. Unfortunately, the number of ancilla qubits used by the algorithm (as shown in \fig{qubitcounts}) makes direct classical simulation infeasible even for very small sizes.  An alternative is to only give an empirical bound on the remainder of the Taylor series.  However, as we discuss in \sec{lcuerror}, such a bound does not give a significant improvement.

\subsection{Quantum signal processing algorithm}
\label{sec:algqsp}

Now we summarize the quantum signal processing (\QSP) algorithm of Low and Chuang \cite{LC17,LC16}.
Again consider a Hamiltonian of the form \eq{hamsum}.
We have
\begin{equation}
	H = \alpha (\langle G| \otimes I)\select(H)(|G\rangle \otimes I),
\end{equation}
where the first register holds an $L$-dimensional ancilla,
\begin{align}
	\select(H)
  &:=\sum_{\ell=1}^{L}|\ell\rangle\langle \ell| \otimes H_\ell,
\end{align}
(similarly to \eq{defnofselectV}), and
\begin{align}
	|G\rangle
  &:=\frac{1}{\sqrt{\alpha}}\sum_{\ell=1}^{L}\sqrt{\alpha_\ell}|\ell\rangle, \qquad
  \alpha
  :=\sum_{\ell=1}^{L} \alpha_\ell.
\end{align}
Low and Chuang's concept of \emph{qubitization} \cite{LC16} relates the spectral decompositions of $H/\alpha$ and
\begin{equation}
	-iQ := -i\bigl((2|G\rangle\langle G|-I)\otimes I\bigr)\select(H).
\end{equation}
Specifically, let $H/\alpha=\sum_{\lambda 
}\lambda|\lambda\rangle\langle\lambda|$ be a spectral decomposition of $H/\alpha$,
where the sum runs over all eigenvalues of $H/\alpha$. By the triangle inequality, $\norm{H} \le \alpha$, i.e., $|\lambda|\leq1$.
For each eigenvalue $\lambda \in (-1,1)$ of $H/\alpha$, the qubitization theorem (Theorem 2 of \cite{LC16}) asserts that $-iQ$ has two corresponding eigenvalues
\begin{equation}
	\lambda_\pm
  =\mp\sqrt{1-\lambda^2}-i\lambda
  =\mp e^{\pm i\arcsin\lambda}
\end{equation}
with eigenvectors $|\lambda_\pm\rangle=(|G_\lambda\rangle \pm i|G_\lambda^\bot\rangle)/\sqrt{2}$, where
\begin{equation}
	|G_\lambda\rangle:=|G\rangle\otimes|\lambda\rangle
\end{equation}
and
\begin{equation}
	|G_\lambda^\bot\rangle
  :=\frac{\lambda |G_\lambda\rangle - \select(H)|G_\lambda\rangle}
         {\sqrt{1-\lambda^2}}.
\end{equation}
(Eigenvalues $\lambda_\pm
=\pm 1$ correspond to degenerate cases that can be analyzed separately.)

The signal processing algorithm applies a sequence of operations called \emph{phased iterates}. We introduce an additional ancilla qubit and define the operator
\begin{equation}
\label{eq:qspiter}
	V_\phi := (e^{-i\phi\sigma^z/2}\otimes I)\big(|+\rangle\langle+| \otimes I+|-\rangle\langle-|\otimes (-iQ)\big)(e^{i\phi\sigma^z/2}\otimes I)
\end{equation}
for any $\phi \in \R$. Let
$-iQ=\sum_{\nu}e^{i\theta_\nu}|\nu\rangle\langle \nu|$
be a spectral decomposition of $-iQ$, where the sum runs over $\nu$ labeling all eigenvectors of $-iQ$.
As described above, each eigenvalue $\lambda \in (-1,1)$ of $H/\alpha$
corresponds to two eigenvalues $e^{i\theta_{\lambda_\pm}}$ of $-iQ$,
where $\theta_{\lambda_+}=\arcsin(\lambda)+\pi$ and $\theta_{\lambda_-}=-\arcsin(\lambda)$. Eigenvalues $\pm1$ of $-iQ$ correspond to degenerate cases that can be handled separately. The remaining eigenspaces cannot be reached during any execution of the quantum signal processing algorithm, so we can neglect them. Then one can show that
\begin{equation}
	V_\phi = \sum_{\nu}e^{i\theta_\nu/2}
	R_\phi(\theta_\nu)
	\otimes |\nu\rangle\langle \nu|
\end{equation}
where
\begin{equation}
	R_\phi(\theta):=e^{-i\theta\sigma^\phi/2},\qquad
	\sigma^\phi:=\cos(\phi)\sigma^x+\sin(\phi)\sigma^y.
\end{equation}
Thus each eigenvalue $e^{i\theta_\nu}$ of $-iQ$ is manifested in $V_\phi$ as an $\mathrm{SU}(2)$ operator $R_\phi(\theta_\nu)$ acting on the ancilla qubit.

For any positive even integer $M$, composing gates with the same rotation amplitude $\theta$ but with varying phases $\phi_1,\ldots,\phi_M$ yields
\begin{align}
R_{\phi_M}(\theta)\cdots R_{\phi_1}(\theta)
&= A(\cos\tfrac{\theta}{2})\,I+iB(\cos\tfrac{\theta}{2})\,\sigma^z+i\cos\tfrac{\theta}{2}C(\sin\tfrac{\theta}{2})\,\sigma^x+i\cos\tfrac{\theta}{2}D(\sin\tfrac{\theta}{2})\,\sigma^y
\end{align}
for polynomials $A,B,C,D$ of degree at most $M$. In the \QSP\ algorithm, only the polynomials $A$ and $C$ are used. This component can be extracted by preparing the ancilla qubit in the state $|+\rangle$, composing the primitive rotations, and postselecting the ancilla qubit in the state $|+\rangle$.
The unwanted factor $e^{i\theta_\nu/2}$ may be canceled by alternating between $V_\phi$ and $V_{\phi+\pi}^\dagger$, giving
\begin{equation}
  V := V_{\phi_M+\pi}^\dagger V_{\phi_{M-1}}\cdots V_{\phi_2+\pi}^\dagger V_{\phi_1}.
\label{eq:pisequence}
\end{equation}

To perform Hamiltonian simulation by qubitization, we implement a function of $\theta$ that converts the eigenvalue $e^{i \theta_{\lambda_\pm}}$ of $-iQ$ to the desired phase $e^{-i\lambda t}$, namely the Jacobi-Anger expansion
\begin{equation}
	e^{i\sin(\theta)t}=\sum_{k=-\infty}^{\infty}J_k(t)e^{ik\theta}.
\end{equation}
To do this with a polynomial of degree $M$, we truncate the expansion at order $q:=\frac{M}{2}+1$, giving an approximation with error at most \cite{BCK15}
\begin{equation}
	2\sum_{k=q}^{\infty}|J_k(t)|\leq\frac{4t^q}{2^q q!}.
	\label{eq:jaerror}
\end{equation}
The angles $\phi_1,\ldots,\phi_M$ that realize this expansion can be computed by an efficient classical procedure (see Lemmas 1 and 3 of \cite{LYC16}).

To simulate evolution of an initial state $|\psi\rangle$, the \QSP\ algorithm applies $V$ to the state $|+\rangle \otimes |G\rangle \otimes |\psi\rangle$ and postselects the ancilla register of the output on the state $|+\rangle \otimes |G\rangle$. This procedure simulates the desired evolution with error at most
\begin{equation}
	\label{eq:qsp_errbd}
	8\frac{4(\alpha t)^{q}}{2^{q} \, q!}\leq \epsilon.
\end{equation}
To achieve simulation for time $t$ and error $\epsilon$, the \QSP\ algorithm uses
\begin{equation}
M = O\biggl(\alpha t+\frac{\log(1/\epsilon)}{\log\log(1/\epsilon)}\biggr)
\end{equation}
phased iterates \cite{LC17}. For each phased iterate, the dominant part is the $\select(H)$ subroutine, which is straightforward to implement with $O(n \log n)$ elementary gates. Overall, we see that the asymptotic gate count in terms of the system size is $O(n^3 \log n)$.  (Note that by using our improved $\select(\cdot)$ implementation described in \sec{selectV}, the asymptotic complexity is reduced to $O(n^3)$.)

Operationally, post-selection of the ancilla is achieved by measurement. If the outcome is $|+\rangle\otimes|G\rangle$, then a state close to $e^{-iHt}|\psi\rangle$ is produced; otherwise, the algorithm fails. In \sec{lc_failure}, we compute the success probability of the \QSP\ algorithm and discuss how it can be fairly compared with the deterministic \PF\ algorithm.

When implementing the \QSP\ algorithm, we eliminate unnecessary gates wherever possible to reduce the gate count. In particular, note that in place of \eq{qspiter}, Low and Chuang define the phased iterate as
\begin{equation}
  V_\phi' :=
(e^{-i\phi\sigma^z/2}\otimes I)\Bigl(|+\rangle\langle+| \otimes I+|-\rangle\langle-|\otimes \bigl(Z_{-\pi/2}(-iQ)Z_{\pi/2}\bigr)\Bigr)(e^{i\phi\sigma^z/2}\otimes I),
\label{eq:lciter}
\end{equation}
where $Z_{\varphi}:=(1+e^{-i\varphi})|G\rangle\langle G|-I$ is a partial reflection about $|G\rangle$. Our modified definition \eq{qspiter} saves two partial reflections in each phased iterate---saving $O(n^2 \log n)$ gates overall---but has the same behavior.  This and other optimizations of the implementation are detailed in \sec{lc_opt}.

The \QSP\ algorithm requires more substantial classical preprocessing than the \PF\ and \TS\ approaches. Computing the angles $\phi_1,\ldots,\phi_M$ requires finding the roots of a polynomial of degree $2M$, and these roots must be determined to high precision.  Thus we were unable to compute the parameters of the algorithm explicitly except in very small instances.

As discussed in \sec{qsp}, we address this issue by considering a segmented version of the algorithm. We discuss this approach further in \sec{lc_phase_seg}. Here we briefly consider how the segmented implementation impacts the asymptotic performance of the \QSP\ algorithm.

Suppose we fix a positive even integer $M$, the maximal number of phased iterates for which the angles $\phi_1,\ldots,\phi_M$ can practically be computed.
As shown in \sec{lc_phase_seg}, it suffices to use $r=O(t (\tfrac{t}{\epsilon})^{2/M})$ segments to ensure overall error at most $\epsilon$.
In the instance of Hamiltonian simulation considered in this paper, with $t=n$ and $\alpha=O(n)$, we have $r=O(n^{2+4/M})$ segments. Within each segment, the number of phased iterates is $M$, which is independent of the system size. The circuit size of each phased iterate is $O(n)$ using the improved $\select(V)$ implementation described in \sec{selectV}. Thus the segmented algorithm has gate complexity $O(n^{3+4/M})$.  In our implementation of the segmented algorithm, we use $M=28$, so the exponent is about $3.14$.

We also consider empirical bounds for the \QSP\ algorithm. Specifically, we find an improved empirical estimate of the truncation error of the Jacobi-Anger expansion. This partial empirical bound leads to a small reduction in the gate count, as discussed further in \sec{lc_bound}. As with the \TS\ algorithm, the need for ancilla qubits in the \QSP\ algorithm makes it difficult to establish a comprehensive empirical bound by performing a full simulation on a classical computer. Fixing the target error $\epsilon=10^{-3}$, we investigate the empirical performance of \QSP\ for systems of size $5$ to $9$. Our preliminary data suggest that the gate count is not significantly improved even with such an empirical bound, so we do not consider such a bound in our study.

%% file: algother.tex
In this appendix, we discuss the asymptotic gate complexity of the Hamiltonian simulation algorithms we chose not to implement.

For a local Hamiltonian, the gate complexity of the algorithm based on fractional queries is \cite{BCCKS13}
\begin{equation}
O\biggl(\tilde{\tau}  \frac{\log(\tilde{\tau}/\epsilon)}{\log\log(\tilde{\tau}/\epsilon)}(\log(\tilde\tau/\epsilon)+n)\biggr)
\end{equation}
where $\tilde{\tau}=L\norm{H}_{\max}t$, with $\norm{H}_{\max}$ denoting the largest entry of $H$ in absolute value (note that this expression slightly tightens the one given in the main statement of Theorem 1.1 of \cite{BCCKS13}; see the end of the proof of that theorem for details).  We have $L=\Theta(n)$ and $\norm{H}_{\max}=\Theta(n)$, so $\tilde\tau = O(n^3)$, and the gate complexity is $O(n^4\frac{\log n}{\log\log n})$.

Similarly, the gate complexity of the algorithm based on a linear combination of quantum walk steps is \cite[Theorem 1]{BCK15}
\begin{equation}
O\biggl(\tau \bigl(n+\log^{5/2}(\tau/\epsilon)\bigr)\frac{\log(\tau/\epsilon)}{\log\log(\tau/\epsilon)}\biggr)
\end{equation}
where $\tau=d\norm{H}_{\max}t$ with $d$ denoting the sparsity of $H$. Since $d = \Theta(n)$, we find a gate complexity of $O(n^4\frac{\log n}{\log\log n})$.

We now turn to the simulation algorithm based on quantum walk \cite[Theorem 1]{BC12}.  Previous work on this approach has not focused on the gate complexity, so we evaluate it here.  The algorithm proceeds as follows. First we perform phase estimation on the quantum walk with one bit of precision. We then apply $O(d\norm{H}_{\max}t)$ steps of a lazy quantum walk. Finally, we invoke phase estimation again with $O(d\norm{H}_{\max}t)$ applications of the walk operator and use this estimate to correct the phase, further improving the accuracy.

The above procedure uses $O(d\norm{H}_{\max}t)=O(n^3)$ quantum walk steps, which dominates the cost of the Fourier transform in the phase estimation procedure (since the estimated phase has $\log(d\norm{H}_{\max}t) = O(\log n)$ bits, the Fourier transform has complexity $O(\log n \log\log n)$) and the coherent computation of the sine function to correct the phase (with complexity $O(M(\log n)\log\log n) = \poly(\log n)$, where $M(k)$ is the complexity of multiplying $k$-bit numbers \cite{BB87}). Thus, to understand the asymptotic gate complexity of the algorithm, it suffices to study the gate complexity of performing a quantum walk step.

The quantum walk operator is the product of a swap (with complexity $O(n)$) and the reflection
\begin{equation}
	\sum_{x=1}^{2^n}|x\rangle\langle x|\otimes (2|\phi_x\rangle\langle\phi_x|-I),
\end{equation}
where
\begin{equation}
	|\phi_x\rangle:=\frac{1}{\sqrt{2^n}}\sum_{y=1}^{2^n}|y\rangle\bigg[\sqrt{\frac{H_{xy}^*}{\mathcal{X}}}|0\rangle+\sqrt{1-\frac{|H_{xy}|}{\mathcal{X}}}|1\rangle\bigg]
\end{equation}
with
\begin{equation}
	\mathcal{X}:=\frac{1}{dt}\max\bigl\{\lceil \norm{H}t/\sqrt{\epsilon}\rceil,\lceil \norm{H}_{\max}dt\rceil\bigr\}.
\end{equation}
To implement the walk operator, it suffices to give procedures for preparing $|\phi_x\rangle$ and reflecting about $|0\rangle$. The reflection operator can be implemented by performing an $X$ gate on each qubit, applying a controlled-$Z$ gate with $n-1$ control qubits, and again performing an $X$ gate on each qubit, for a total cost of $O(n)$.  Thus it remains to understand the cost of preparing $|\phi_x\rangle$.

The algorithm uses oracles $O_F$ and $O_H$ acting as
\begin{equation}
	O_F|x,\ell\rangle=|x,f(x,\ell)\rangle
  \label{eq:OFdef}
\end{equation}
and
\begin{equation}
	O_H|x,y,z\rangle=|x,y,z\oplus H_{xy}\rangle,
\end{equation}
where $f(x,\ell)$ is the column index of the $\ell$th nonzero element
in row $x$. Reference \cite{BC12} observes that $|\phi_x\rangle$ can be prepared with cost
\begin{equation}
	O( n+\text{cost}(O_F)+\text{cost}(O_H) ).
\end{equation}
Therefore, it suffices to separately understand the cost of implementing $O_F$ and $O_H$ for our Hamiltonian.

Recall that the Hamiltonian takes the form
\begin{equation}
H = \sum_{j=1}^n ( \vec\sigma_j \cdot \vec\sigma_{j+1} + h_j \sigma_j^z).
\end{equation}
For any $j \in \{1,\ldots,n\}$, the term $\vec \sigma_j \cdot \vec \sigma_{j+1}$ acts on the $j$th and $(j+1)$st qubits according to the matrix
\begin{equation}
\begin{bmatrix}
	1 & 0 & 0 & 0\\
	0 & -1 & 2 & 0\\
	0 & 2 & -1 & 0\\
	0 & 0 & 0 & 1
\end{bmatrix}.
\end{equation}
Thus the nonzero elements of row $x$ of $H$, aside from the diagonal, correspond to the ways of swapping two adjacent bits in the binary representation of $x$.  There are at most $n+1$ such nonzero elements, so we can suppose that $\ell \in \{1,\ldots,n+1\}$ in \eq{OFdef}, where (say) $\ell=n+1$ corresponds to the diagonal element.

To implement $O_F$, we begin by making a copy of the first register using $O(n)$ CNOT gates, performing
\begin{equation}
	|x,\ell\rangle\mapsto|x,x,\ell\rangle.
\end{equation}
Then, conditioned on the value of $\ell \in \{1,\ldots,n\}$, we swap the $\ell$th and $(\ell+1)$st bits of the second register, giving
\begin{equation}
	|x,x,\ell\rangle\mapsto|x,f(x,\ell),\ell\rangle.
\end{equation}
This can be done with complexity $O(n \log n)$, since it is a $\select(\cdot)$ gate where the selection register has $n+1$ possible states and the target gates act only on two qubits.
To uncompute the third register, note that there is a classical algorithm with running time $O(n)$ that computes $\ell$ from $x$ and $f(x,\ell)$: we just compare the two strings and detect which pair of bits are swapped. This procedure can be made reversible by standard techniques. Therefore, we can erase $|\ell\rangle$ and thereby implement $O_F$ with gate complexity $O(n \log n)$.

To implement $O_H$, we use a classical algorithm that, given the row and column indices $(x,y)$, outputs the corresponding matrix element $H_{xy}$.
The $(x,x)$ diagonal element is easy to compute from the binary representation of $x$: every adjacent pair of bits gives a contribution of $+1$ if the bits agree and $-1$ if they disagree, and the $j$th bit gives an additional contribution of $h_j$ if the bit is $0$ and $-h_j$ if the bit is $1$. For $x \ne y$, the off-diagonal element $H_{xy}$ is $2$ if $x$ and $y$ differ by swapping some adjacent pair of distinct bits, and is $0$ otherwise.  Since all these calculations can be performed in time $O(n)$, the complexity of implementing $O_H$ is $O(n)$.

Altogether, we see that the gate complexity of the quantum walk simulation is $O(n^4 \log n)$, as shown in \tab{algsummary}.

In fact, our improved implementation of $\select(\cdot)$ gates in \lem{sV2} shows that $O_F$ can be implemented in time $O(n)$.  With this improvement, the quantum walk simulation can be performed slightly faster, in time $O(n^4)$.

%% file: quipper.tex
We implemented quantum simulation algorithms in the Quipper programming language
\cite{GLRSV13}. Quipper is a circuit description language equipped
with many high-level circuit combinators (e.g., circuit iteration and
circuit reversal) that allow for a concise specification of complex
quantum circuits. Quipper also supports a hierarchical circuit
structure that allows for the efficient manipulation of the very large
quantum circuits considered here. We made full use of these features. The Quipper source code for our implementations, together with sample output circuits and optimized versions thereof, are available in a public repository \cite{SourceCode}.

\subsection{Overview of algorithm implementations}

\subsubsection{Product formulas}

We implement simulation algorithms using product formulas of order $1$, $2$, $4$, $6$, and $8$.  For each order, we implement various time-slicing methods as discussed in \sec{algpf} and \app{pf}. Specifically, we implement the analytic, minimized, and empirical bounds for all orders.  For first-, second-, and fourth-order algorithms, we also implement the commutator bound. Overall, we consider $18$ different types of product formula algorithm.

The product formula algorithms are the simplest of the simulation algorithms we consider.  To simulate the evolution for a short time, these algorithms simply exponentiate each term of the Hamiltonian (which in our case is always proportional to a 1- or 2-qubit Pauli operator) in a carefully-chosen sequence.  Simulation for a longer time is then obtained by iterating this sequence.  The duration of the individual evolutions and the total number of repetitions is determined using an appropriate error bound.  The iteration is easily realized using Quipper's built-in iteration combinator.

\subsubsection{Taylor series}

The \TS\ algorithm is more involved than the \PF\ algorithm, using more complicated subroutines. Our implementation is based on a mixed unary-binary encoding of the control register, as explained in \sec{controlencoding}. The implementation consists of several subroutines, including two distinct state preparation subroutines, the $\select(V)$ procedure discussed in \sec{selectV}, and the reflection about $\ket{0}$. The first state preparation subroutine maps $|0\rangle$ into the state $\sum_{l=1}^{L}\sqrt{\alpha_l}|l\rangle$, where $\{\alpha_l\}_{l=1}^{L}$ are the coefficients of terms in the target Hamiltonian.  This is accomplished through the generic state preparation algorithm described in \cite{Shende2006}.  The second state preparation subroutine generates a state proportional to $\sum_{k=0}^{K}\sqrt{t^k/k!}|1^k0^{K-k}\rangle$ starting with the basis state $|0\rangle$, for which the generic method \cite{Shende2006} is suboptimal.  Instead, we use the state preparation procedure described in \cite{BCCKS14}, applying a rotation on the first qubit, followed by rotations on qubits $k=2$ to $K$ controlled by the qubit $k{-}1$.  We develop our own implementation of the $\select(V)$ operation, as described in \sec{selectV}.  The reflection about $\ket{0}$ is implemented as a multiply-controlled $Z$ gate, following the construction in \cite{M16} that uses ancillas in the state $\ket{0}$.

\subsubsection{Quantum signal processing}

We implement two versions of the \QSP\ algorithm. First, we consider the original version of the algorithm as described in \sec{algqsp}. In choosing the parameters for this implementation, we use the empirical estimate of the remainder of the Jacobi-Anger expansion described in \sec{lc_bound}. Unfortunately, our implementation of the classical computation of the rotation angles \cite{LYC16} is only able to handle very small system sizes. Instead, to generate gate counts, we use randomly-selected angles as placeholder values. Thus the circuits constructed by our implementation have the correct structure (and in particular, give the correct gate count over the Clifford+$R_z$ basis, and a good approximation over the Clifford+$T$ basis), but do not actually implement the desired unitary.

We also implement a segmented version of the \QSP\ algorithm as discussed near the end of \sec{algqsp} and detailed in \sec{lc_phase_seg}.  In this case, we invoke a rigorous error bound, and we are able to correctly compute all parameters of the algorithm.  However, the rotation angle computation is involved and uses high-precision arithmetic.  For this reason, we compute these parameters off-line using Mathematica and store the results in a look-up table that is accessed by the Quipper program to construct the quantum circuit.  Thus our Quipper implementation can only produce circuits for the system sizes for which these values have been precomputed. The Mathematica scripts are available as part of our implementation \cite{SourceCode} so that the interested reader can compute the required parameters and add them to the look-up tables if additional system sizes are of interest.

\subsection{Gate counts and synthesis}
\label{sec:synthesis}

We express all algorithms using Clifford gates (including $\CNOT$) and single-qubit $z$-rotations by arbitrary angles.  We use Quipper's standard gate-counting feature to produce the first set of resource estimates (in the Clifford+$R_z$ basis).  To obtain our second set of estimates (in the Clifford+$T$ basis), we must approximate $z$-rotations by Clifford+$T$ circuits.  These approximations are obtained using the optimal algorithm of Ross and Selinger \cite{RS16} (which builds upon the exact synthesis algorithm of Kliuchnikov, Maslov, and Mosca \cite{kmm12}).

Note that a better approach to decomposing $\Rz$ gates into Clifford+$T$ circuits might be to rely on the repeat-until-success (RUS) strategy of \cite{brs14}.  This reduces the $T$-count by a factor of about $2.5$ on average, at the cost of using additional resources in the form of measurements, classical feedback, and additional $\CNOT$ gates.  One may combine RUS decomposition with Campbell's unitary mixing approach \cite{c16} to further reduce the $T$-count by a factor of about $2$.  Finally, rather than using only the $T=\Rz(\pi/4)$ gate, one may also distill $\Rz(\pi/6)$ with comparable or better efficiency as the $T$ gate \cite{BK05}.  Using $\Rz(\pi/6)$ along with the $T$ gate in the RUS circuits \cite{brs14} together with unitary mixing \cite{c16}, we expect an improvement in the resource count for a fault-tolerant implementation by a factor of $5$ or more. Alternatively, one might employ the gate distillation techniques of \cite{DP15}, which could also significantly reduce the cost of implementing fault-tolerant gates. We leave a detailed investigation of such possible improvements as a subject for future work.

When producing Clifford+$T$ gate counts for the non-segmented \QSP\ algorithm, we synthesize rotations that are incorrect, since we are unable to compute the exact rotation angles.  Since the number of Clifford+$T$ gates required to synthesize a given rotation depends on its angle, the resulting gate counts are not, strictly speaking, precise.  Nevertheless, we expect the produced gate counts to accurately represent the true counts.  Since the rotation angles are randomly chosen, their cost is that of a typical angle (roughly $3\log (1/\tau)$ where $\tau$ is the approximation precision). Furthermore, because we synthesize many rotations, occurrences of over/underestimation of the specific gate counts average out over the entire circuit by the law of large numbers.  The cost of the approximation therefore depends primarily on the precision---which, in turn, depends on the number of rotations in the overall circuit---but not on the individual rotation angles.

To approximate a given Hamiltonian $e^{-iHt}$ to precision $\epsilon$ with a Clifford+$T$ circuit, we divide the allowed error $\epsilon$ between gate synthesis and simulation algorithm errors.  In our implementation, we allocate $\epsilon/2$ to the simulation error and $\epsilon/2$ to the gate synthesis error.  The latter is then divided by the total number of $z$-rotations appearing in the circuit.  While anecdotal evidence suggests that such an even division of $\epsilon$ between simulation and synthesis is a good choice, we leave it as an avenue for future research to determine whether a more subtle partition should be preferred.

\subsection{Optimization}
\label{sec:circuit-opt}

We also employ a circuit optimizer \cite{Optimizer} that uses a variety of techniques to reduce gate counts in an automated way.  The optimizer uses various circuit equivalences to induce quantum gate commutations, mergers, and cancellations.

Many of the subcircuits found in the \TS\ and \QSP\ algorithms are based on Toffoli gates.  In our unoptimized circuits, we simply replace each Toffoli gate with an optimal Clifford+$T$ implementation \cite{ar:ammr}.  However, there are many ways to write the Toffoli gate as a Clifford+$T$ circuit (e.g., some gates in an implementation may commute, controls can be interchanged, and the circuit can be reversed and/or complex conjugated since it is real and self-inverse).  Carefully choosing when to use which decomposition can enable additional gate cancellations.  For this reason, we outsource Toffoli gate decomposition to the optimizer.

Once the circuits are expressed over the Clifford+$\Rz$ basis, our optimizer first performs a sequence of rewrites to reduce the number of Hadamard gates. The resulting circuit is more amenable to further optimization since it typically contains larger chunks that can be expressed using so-called phase polynomials \cite{ar:ammr,Optimizer}.  Then we identify pairs of gates that can be canceled or merged (allowing the gates in each pair to be separated by a subcircuit through which one of the gates commutes, according to a fixed set of commutation rules).  Finally, we use the phase polynomial representation over $\{\CNOT,\Rz\}$ to further lower the $\Rz$ count of the circuit.  This representation can be used to identify $\Rz$ gates that are applied to the same linear Boolean function of the input, and thereby merge them, even if they originally correspond to distant gates.
Merging such rotations can enable additional gate commutations, leading to more optimizations.  We repeat the entire sequence of optimization procedures until no more gate count reduction is achieved.

Our optimizer comes with ``light'' and ``heavy'' options \cite{Optimizer}.  In general, the heavy version finds more reductions, but is more computationally intensive.  We used the light version of the optimizer to obtain the results reported in this paper since we do not expect the heavy version to change our conclusions qualitatively.

\subsection{Correctness}

We carry out a number of tests to verify correctness of our circuit-level implementations.  We simulate entire \PF\ circuits for systems with up to $15$ qubits and check for proximity to the ideal evolution operator by evaluating the spectral norm distance.  We also simulate the entire segmented \QSP\ circuit for a system of size $5$ and compare the outcome of the circuit to the ideal operator in the same sense.  Unfortunately, the number of ancillas makes full simulation for larger instances of the \TS\ algorithm prohibitively expensive.  Thus, in lieu of full simulation, we test individual subroutines for correctness: we verify the reflection and state preparation subroutines for systems of up to $15$ qubits, and test various instances of the $\select(V)$ circuit, including its controlled versions, for systems of size $5$.  Finally, we test the output of the optimizer by simulating and comparing circuits before and after optimization \cite{Optimizer}.

%% file: pfappendix.tex
\newcommand{\support}{\operatorname{\mathsf{supp}}}
\newcommand{\interval}{\operatorname{\mathsf{ival}}}
\newcommand{\numterms}{\mu}   
\newcommand{\numcycles}{\nu}  
\newcommand{\tuplesize}{\tau} 
\newcommand{\range}[1]{\left\llbracket #1 \right\rrbracket}
\newcommand{\boundary}{\operatorname{\mathsf{bndry}}}
\newcommand{\infix}{\operatorname{\mathsf{infix}}}
\newcommand{\prefix}{\operatorname{\mathsf{prefix}}}
\newcommand{\suffix}{\operatorname{\mathsf{suffix}}}

The main issue involved in implementing product formula (\PF) algorithms---as introduced in \sec{algpf}---is to establish error bounds that determine how finely to split the evolution so that some target error is achieved. In this section, we present explicit error analysis for PF algorithms that yields effective methods for computing values of $r$ that ensure the error is at most $\epsilon$.

First, we establish error bounds that provide a proven guarantee that the simulated evolution is $\epsilon$-close to the ideal evolution. We present three such bounds, which we call the \emph{analytic}, \emph{minimized}, and \emph{commutator} bounds.  The analytic and minimized bounds, presented in \sec{pfanamin}, follow essentially the same strategy as in previous work \cite{BACS05}.  The commutator bound, presented in \sec{pfcom}, is substantially more involved, using the structure of commutators between terms in the Hamiltonian to improve the result.

Second, in \sec{empiricalbounds}, we present \emph{empirical} error bounds, which are obtained by extrapolating numerical data. Since they ostensibly describe the true performance of \PF\ algorithms, the empirical bounds result in smaller gate counts than the rigorous bounds. However, while the empirical bounds are plausible, they do not provide a guarantee about the actual distance between the simulated evolution and the ideal evolution, unlike with rigorous bounds.

Throughout, we quantify the simulation error with respect to the spectral norm.
Note that for a unitary process, the diamond norm distance is at most twice the spectral norm distance \cite[Lemma 7]{BCK15}.

\subsection{Analytic and minimized bounds}
\label{sec:pfanamin}

Reference \cite{BACS05} gives an explicit error bound for Suzuki product formulas.  Using similar techniques, we can also give an explicit bound for the first-order case.  We present these bounds here, tightening the analysis wherever possible.

First, recall some useful properties of Taylor expansions. For any $k\in\N$ and any analytic function $f\colon \C \to \C$ with $f(x) = \sum_{j=0}^\infty a_j x^j$, let $\rem_k(f) := \sum_{j=k+1}^\infty a_j x^j$ denote the remainder of the Taylor series expansion of $f$ up to order $k$.

\begin{lemma}
  \label{lem:prodexp}
  If $\lambda\in\C$ and $H_1,\ldots, H_L$ are Hermitian, then
  \begin{equation}
    \norm*{\rem_{k}\Biggl(\prod_{j=1}^L\exp(\lambda H_{j})\Biggr)}
    \leq
    \rem_{k}\Biggl(\exp\Biggl(\sum_{j=1}^L|\lambda|\cdot \norm{H_{j}}\Biggr)\Biggr).
  \end{equation}
\end{lemma}

\begin{proof}
  We have
  \begin{equation}
    \label{eq:ts1}
    \prod_{j=1}^L\exp(\lambda H_j)=
    \prod_{j=1}^L\sum_{r_j=0}^{\infty}\lambda^{r_j}\frac{H_j^{r_j}}{r_j!}=
    \sum_{r_{1},\ldots,r_{L}=0}^{\infty}\prod_{j=1}^L\lambda^{r_j}\frac{H_j^{r_j}}{r_j!},
  \end{equation}
  so
  \begin{equation}
    \label{eq:ts2}
    \rem_{k}\left(\prod_{j=1}^L\exp(\lambda H_j)\right)=
    \sum_{\substack{r_{1},\ldots,r_L=0\\\sum_{l}r_{l}\geq
        k+1}}^{\infty}\prod_{j=1}^L\lambda^{r_j}\frac{H_{j}^{r_{j}}}{r_{j}!}.
  \end{equation}
  Using the triangle inequality and submultiplicativity of the norm, we find
  \begin{equation}
    \norm*{\rem_{k}\Biggl(\prod_{j=1}^L\exp(\lambda H_{j})\Biggr)}
    =\norm*{\sum_{\substack{r_{1},\ldots,r_L=0\\\sum_{l}r_{l}\geq
          k+1}}^{\infty}\prod_{j=1}^L\lambda^{r_j}\frac{H_{j}^{r_j}}{r_j!}}
    \leq\sum_{\substack{r_{1},\ldots,r_L=0\\\sum_{l}r_{l} \geq
        k+1}}^{\infty}\prod_{j=1}^L|\lambda|^{r_j}\frac{\norm{H_j}^{r_j}}{r_j!}.
  \end{equation}
  Finally, similarly as in \eq{ts1} and \eq{ts2}, we have
  \begin{equation}
    \sum_{\substack{r_{1},\ldots,r_L=0\\\sum_{l}r_{l}\geq
        k+1}}^{\infty}\prod_{j=1}^L|\lambda|^{r_j}\frac{\norm{H_{j}}^{r_j}}{r_j!}
    =\rem_{k}\Biggl(\prod_{j=1}^L\exp(|\lambda|\cdot \norm{H_j})\Biggr)
    =\rem_{k}\Biggl(\exp\Biggl(\sum_{j=1}^L|\lambda|\cdot
    \norm{H_{j}}\Biggr)\Biggr),
  \end{equation}
  which completes the proof.
\end{proof}

\begin{lemma}
  \label{lem:prodtail}
  If $\lambda \in \C$, then
  $|\rem_{k}(\exp(\lambda))|\leq
  \frac{|\lambda|^{k+1}}{(k+1)!}\exp(|\lambda|)$.
\end{lemma}

\begin{proof}
Using the Taylor expansion of the exponential function, we have
  \begin{align}
    |\rem_{k}(\exp(\lambda))|
    &=\Biggl|\sum_{l=k+1}^{\infty}\frac{\lambda^{l}}{l!}\Biggr|
      \leq \sum_{l=k+1}^{\infty}\frac{|\lambda|^{l}}{l!}
      =|\lambda|^{k+1}\sum_{l=k+1}^{\infty}\frac{|\lambda|^{l-(k+1)}}{l!}\\
    &=|\lambda|^{k+1}\sum_{l=0}^{\infty}\frac{|\lambda|^{l}}{(l+(k+1))!}
      \leq|\lambda|^{k+1}\sum_{l=0}^{\infty}\frac{|\lambda|^{l}}{l!}\frac{1}{(k+1)!}\\
    &=\frac{|\lambda|^{k+1}}{(k+1)!}\exp(|\lambda|),
  \end{align}
which completes the proof.
\end{proof}

Now we give an explicit error bound for the first-order product formula.

\begin{proposition}
  \label{prop:firstana}
  Let $r\in \N$ and $t\in\R$. Let $H_1,\ldots,H_L$ be Hermitian operators,
  and $\Lambda := \max_{j}\norm{H_{j}}$. Then
  \begin{equation}
    \norm*{\exp\Biggl(-it\sum_{j=1}^L H_{j}\Biggr)-\Biggl[\prod_{j=1}^L
        \exp\biggl(-\frac{it}{r}H_{j}\biggr)\Biggr]^{r}}
    \leq\frac{(L\Lambda t)^2}{r}\exp\biggl(\frac{L\Lambda |t|}{r}\biggr).
  \end{equation}
\end{proposition}

\begin{proof}
  Let $\lambda = -it$. We start by considering
  $\norm{\exp(\lambda\sum_{j=1}^L H_{j})-\prod_{j=1}^L\exp(\lambda H_{j})}$. Using \lem{prodexp} and \lem{prodtail}, we find
  \begin{align}
    \norm*{\exp\biggl(\lambda\sum_{j=1}^L H_{j}\biggr)-\prod_{j=1}^L\exp\left(\lambda H_{j}\right)}
    &= \norm*{\rem_1 \biggl( \exp\biggl(\lambda\sum_{j=1}^L H_{j}\biggr)-\prod_{j=1}^L \exp(\lambda H_{j}) \biggr)} \\
    & \leq \norm*{\rem_1 \biggl( \exp(\lambda\sum_{j=1}^L H_{j}) \biggr)}
      +\norm*{\rem_1 \Biggl(\prod_{j=1}^L \exp(\lambda H_{j}) \Biggr)}\\
    &\leq \rem_1\biggl(\exp\biggl(|\lambda|\cdot\biggl|\sum_{j=1}^L H_{j}\biggr|\biggl)\biggr)
      +\rem_1 \biggl(\exp\biggl(\sum_{j=1}^L|\lambda|\cdot \norm{H_{j}}\biggr)\biggr)\\
    &\leq 2\rem_{1}(\exp(|\lambda|L\Lambda))\\
    &\leq (|\lambda|L\Lambda)^2 \exp(|\lambda|L\Lambda).
  \end{align}
  We then divide the evolution into $r$ segments and apply the above
  inequality to each segment, giving
  \begin{align}
    \norm*{\exp\left(-it\sum_{j=1}^LH_{j}\right)-
    \left[\prod_{j=1}^L\exp\left(-\frac{it}{r}H_{j}\right)\right]^{r}}
    &\leq r \norm*{\exp\left(-\frac{it}{r}\sum_{j=1}^LH_{j}\right)
      -\prod_{j=1}^L\exp\left(-\frac{it}{r}H_{j}\right)}\\
    &\leq r \left(\frac{t L\Lambda}{r}\right)^2 \exp\left(\frac{|t| L\Lambda}{r}\right)\\
    &=\frac{(L\Lambda t)^2}{r}\exp\left(\frac{L\Lambda |t|}{r}\right),
  \end{align}
  which completes the proof.
\end{proof}

A similar bound holds for the higher-order case, as follows.

\begin{proposition}
  \label{prop:higherana}
  Let $r\in \N$ and $t\in\R$. Let $H_1,\ldots,H_L$ be Hermitian operators,
  and $\Lambda := \max_{j}\norm{H_{j}}$. Then
  \begin{equation}
    \norm*{\exp\Biggl(-it\sum_{j=1}^L H_{j}\Biggr)
      -\left[S_{2k}\left(-\frac{it}{r}\right)\right]^r}
    \leq\frac{(2L5^{k-1}\Lambda |t|)^{2k+1}}{3r^{2k}}
    \exp\left(\frac{2L5^{k-1}\Lambda |t|}{r}\right).
  \end{equation}
\end{proposition}

We omit the proof, which follows along the same lines as the proof of \prop{firstana}.

When performing quantum simulation, our goal is to ensure that the error is at most some given $\epsilon$.  Thus, to apply \prop{firstana} and \prop{higherana}, we must choose a value of $r$ that ensures the right-hand side is at most $\epsilon$.  One approach \cite{BACS05} is to give a closed-form expression for some suitable $r$ as a function of $\epsilon$ (and other simulation parameters).  We call the resulting error bound an \emph{analytic bound}.

\begin{definition}
  \label{def:anabound}
  Let $t\in\R$ and $\epsilon > 0$.  Let $H_1,\ldots,H_L$ be Hermitian operators, and $\Lambda := \max_{j}\norm{H_{j}}$. Then the \emph{first-order
    analytic bound} is given by
  \begin{equation}
    r_{1}= \ceil*{\max\bigg\{L|t|\Lambda,
      \frac{e(Lt\Lambda)^{2}}{\epsilon}\bigg\}}
  \end{equation}
  and the \emph{$(2k)$th order analytic bound} is given by
  \begin{equation}
    r_{2k}= \ceil*{\max\bigg\{2L5^{k-1}\Lambda |t|,
      \sqrt[2k]{\frac{e(2L5^{k-1}\Lambda
          |t|)^{2k+1}}{3\epsilon}}\bigg\}}.
  \end{equation}
\end{definition}

Analytic bounds can be computed efficiently and result in circuits approximating the time evolution up to the required precision. However, because the error estimates are pessimistic, the corresponding circuits are unnecessarily large.

A slight improvement can be obtained by searching for the smallest possible $r$ that satisfies the constraints given by \prop{firstana} and \prop{higherana}.  We call the resulting error bound a \emph{minimized bound}.

\begin{definition}
  \label{def:minbound}
  Let $t\in\R$ and $\epsilon > 0$.  Let $H_1,\ldots,H_L$ be Hermitian operators, and $\Lambda := \max_{j}\norm{H_{j}}$. Then the \emph{first-order
    minimized bound} is given by
  \begin{equation}
    r_1 = \min\left\{r \in \N : \frac{(L\Lambda t)^2}{r}\exp\left(\frac{L\Lambda|t|}{r}\right) \le \epsilon \right\}
  \end{equation}
  and the \emph{$(2k)$th-order minimized bound} is given by
  \begin{equation}
    r_{2k} = \min\left\{r \in \N : \frac{(2L5^{k-1}\Lambda |t|)^{2k+1}}{3r^{2k}} \exp\left(\frac{2L5^{k-1}\Lambda |t|}{r}\right) \le \epsilon \right\}.
  \end{equation}
\end{definition}

Minimized bounds can be computed efficiently using binary search. First, we compute the error bound with $r=1$. If the result is within our desired accuracy $\epsilon$, we stop; otherwise, we double $r$ and repeat the procedure. This allows us to bracket the value of $r$ in logarithmic time. Then we perform binary search to determine the precise value of $r$.

\subsection{Commutator bounds}
\label{sec:pfcom}

If all terms in the Hamiltonian commute, then even the lowest-order product formula has no error.  Intuitively, if many pairs of terms commute, then the error should be small.  Indeed, the second-order contribution to the error in the first-order formula involves commutators between all pairs of terms, and this fact can be used to strengthen the bound \cite{Llo96}.

In this section, we present improved error bounds that take advantage of commuting pairs of terms in the Hamiltonian. In particular, we give commutator bounds for the second- and fourth-order formulas, whose constructions can in principle be adapted to higher-order cases. To the best of our knowledge, these are the first error bounds for higher-order product formulas that take advantage of commutator structure.

In \sec{pfcomabstract}, we present abstract versions of the commutator bounds, stated in terms of coefficients that count the number of terms in the Hamiltonian with particular commutation relations.  We specialize these bounds to our particular system in \sec{pfcomconcrete}, using the symmetry of our Hamiltonian to substantially simplify application of the bounds.

\subsubsection{Abstract commutator bounds}
\label{sec:pfcomabstract}

We begin by presenting a bound for the first-order formula as described in \cite{Llo96}, giving an explicit bound for the higher-order terms.

\begin{theorem}[First-order commutator bound]
	\label{thm:first_comm}
Let $H_1,\ldots,H_L$ be Hermitian operators with norm at most $\Lambda := \max_{j}\norm{H_{j}}$.  Let $C:=|\{(H_i,H_j):[H_i,H_j]\neq 0,\,i<j\}|$ be the number of non-commuting pairs of operators, and let $t\in\R$.  Then
	\begin{equation}
	\norm*{\exp\left(-it\sum_{j=1}^{L}H_{j}\right)-\left[\prod_{j=1}^{L}\exp\left(-\frac{it}{r}H_{j}\right)\right]^{r}}
	\leq C\frac{(\Lambda t)^2}{r}+\frac{(L\Lambda |t|)^3}{3r^2}\exp\left(\frac{L\Lambda |t|}{r}\right).
	\end{equation}
\end{theorem}
\begin{proof}
We show that
\begin{equation}
	\norm*{\exp\left(\lambda\sum_{j=1}^{L}H_{j}\right)-\prod_{j=1}^{L}\exp(\lambda H_{j})}
	\leq C(|\lambda|\Lambda)^2+\frac{(|\lambda|L\Lambda)^3}{3} \exp\left(|\lambda|L\Lambda\right),
\label{eq:first_comm_one_segment}
	\end{equation}
which implies the claimed result by the triangle inequality.
	The upper bound \eq{first_comm_one_segment} can be established by explicitly computing the second-order error and bounding the higher-order errors by the norm of the remainder $\rem_{2}$. The second-order error is bounded as follows:
	\begin{align}
	&\norm*{\frac{\lambda^2}{2}\left(\sum_{j=1}^{L}H_j\right)^2-\left[\frac{\lambda^2}{2}\sum_{j=1}^{L}H_j^2+\lambda^2\sum_{j<k}H_jH_k\right]} \nonumber\\
	&\quad=\norm*{\left[\frac{\lambda^2}{2}\sum_{j=1}^{L}H_j^2+\frac{\lambda^2}{2}\sum_{j<k}H_jH_k+\frac{\lambda^2}{2}\sum_{j>k}H_jH_k\right]-\left[\frac{\lambda^2}{2}\sum_{j=1}^{L}H_j^2+\lambda^2\sum_{j<k}H_jH_k\right]}\\
	&\quad=\norm*{\frac{\lambda^2}{2}\left[\sum_{j>k}H_jH_k-\sum_{j<k}H_jH_k\right]}
	=\norm*{\frac{\lambda^2}{2}\left[\sum_{j>k}H_jH_k-H_kH_j\right]}\\
	&\quad\le
  C |\lambda|^2 \Lambda^2.
	\end{align}
	The rest of the proof proceeds similarly to the second half of the proof of \prop{firstana}; we omit the details.
\end{proof}

\begin{theorem}[Second-order commutator bound]
	\label{thm:second_comm}
Let $H_1,\ldots,H_L$ be Hermitian operators with norm at most $\Lambda := \max_{j}\norm{H_{j}}$, where each $H_j$ is a tensor product of Pauli operators. Define the augmented set of Hamiltonians
	\begin{equation}
	\tilde{H}_{j}=
	\begin{cases}
	H_{j},   & 1\leq j\leq L\\
	H_{2L-j+1},       & L+1\leq j \leq 2L.
	\end{cases}
	\end{equation}
	Let $f(i,j)=1$ if $\tilde{H}_{i}, \tilde{H}_{j}$ commute and $f(i,j)=-1$ otherwise. Finally, let
	\begin{align}
	D&:=|\{(i,j):f(i,j)=-1,i\neq j\}|, \label{eq:countD}\\
	T_{1}&:=|\{(i,j,k):f(i,j)=f(j,k)=f(i,k)=1,\,i<j<k\}|,\\
	T_{2}&:=|\{(i,j,k):f(i,j)=1,f(j,k)=f(i,k)=-1,\,i<j<k\}|\nonumber\\
	&\quad+|\{(i,j,k):f(i,j)=f(i,k)=-1,\,f(j,k)=1,\,i<j<k\}|,\\
	T_{3}&:=|\{(i,j,k):f(i,j)=f(j,k)=-1,\,f(i,k)=1,\,i<j<k\}|,\\
	T_{4}&:=|\{(i,j,k):\text{all other cases}\}| \label{eq:countT4}
	\end{align}
where $i,j,k \in \{1,\ldots,2L\}$, and let $t\in\R$.
Then
\begin{align}
	&\norm*{\exp\left(-it\sum_{j=1}^{L}H_{j}\right)-\left[S_{2}\left(-\frac{it}{r}\right)\right]^r} \nonumber\\
	&\quad \leq\frac{\Lambda^3|t|^3}{r^2}\bigg\{\frac{1}{24}D+\frac{1}{12}T_2+\frac{1}{6}T_3+\frac{1}{8}T_4\bigg\}
	+\frac{4(L\Lambda t)^4}{3r^3}\exp\left(\frac{2L\Lambda |t|}{r}\right).
  \label{eq:second_comm_result}
	\end{align}
\end{theorem}

\begin{proof}
	As in the proof of \thm{first_comm}, we explicitly compute the third-order error and bound the higher-order terms by $\rem_{3}$. First we show that the first-order formula
	\begin{equation}
	\norm*{\exp\left(\lambda\sum_{j=1}^{L}H_{j}\right)-\prod_{j=1}^{L}\exp(\lambda H_{j})}
	\end{equation}
	has a third-order error of at most
	\begin{equation}
	|\lambda|^3\Lambda^3\biggl(\frac{1}{3}\bar D+\frac{2}{3}\bar T_2+\frac{4}{3}\bar T_3+\bar T_4\biggr),
  \label{eq:third_order_err}
	\end{equation}
where the coefficients $\bar D, \bar T_2, \bar T_3, \bar T_4$ are defined as in \eq{countD}--\eq{countT4}, but with respect to the original Hamiltonians $\{H_j\}_{j=1}^{L}$ instead of $\{\tilde H_j\}_{j=1}^{2L}$.

The third-order term in $\exp(\lambda\sum_{j=1}^{L}H_{j})$ is
	\begin{equation}
	\frac{1}{3!}\Biggl(\lambda\sum_{j=1}^{L}H_{j}\Biggr)^3=\frac{\lambda^3}{6}\sum_{i,j,k}H_iH_jH_k,
	\end{equation}
	whereas the third-order term in $\prod_{j=1}^{L}\exp(\lambda H_{j})$ is
	\begin{align}
	&\frac{\lambda^3}{6}\sum_{i}H_i^3+\frac{\lambda^3}{2}\sum_{i<k}H_i^2H_k
	+\frac{\lambda^3}{2}\sum_{i<k}H_iH_k^2+\lambda^3\sum_{i<j<k}H_iH_jH_k \nonumber\\
	&\quad=\frac{\lambda^3}{6}\sum_{i}H_i^3+\frac{\lambda^3}{2}\sum_{i\neq k}H_i^2H_k+\lambda^3\sum_{i<j<k}H_iH_jH_k,
	\end{align}
	where we have used the fact that the square of any Pauli operator is the identity. Taking the difference gives
	\begin{equation}
	\frac{\lambda^3}{6}\sum_{i\neq j}H_i[H_j,H_i]+\frac{\lambda^3}{6}\sum_{\substack{i,j,k\\\text{pairwise different}}}H_iH_jH_k-\lambda^3\sum_{i<j<k}H_iH_jH_k.
	\end{equation}
The norm of the first term is at most
	\begin{equation}
	\label{eq:first_term2}
	\frac{1}{3}|\lambda|^3 \Lambda^3 \bar D,
	\end{equation}
	whereas the last two terms can be written as follows:
	\begin{align}
	&\frac{\lambda^3}{6}\sum_{\substack{i_1,i_2,i_3\\\text{pairwise different}}}H_{i_1}H_{i_2}H_{i_3}-\lambda^3\sum_{i_1<i_2<i_3}H_{i_1}H_{i_2}H_{i_3} \nonumber\\
	&\quad=\frac{\lambda^3}{6}\sum_{i_1<i_2<i_3}\sum_{\sigma\in \Sym{3}}H_{i_{\sigma(1)}}H_{i_{\sigma(2)}}H_{i_{\sigma(3)}}
	-\lambda^3\sum_{i_1<i_2<i_3}H_{i_1}H_{i_2}H_{i_3}\\
	&\quad=\frac{\lambda^3}{6}\sum_{i_1<i_2<i_3}\bigl(1+f(1,2)+f(2,3)+f(1,2)f(1,3)f(2,3)\\
	&\qquad+f(1,3)f(1,2)+f(1,3)f(2,3)\bigr)H_{i_1}H_{i_2}H_{i_3}-\lambda^3\sum_{i_1<i_2<i_3}H_{i_1}H_{i_2}H_{i_3}.
	\end{align}
(Here $\Sym{3}$ denotes the symmetric group on three elements.)
By performing case analysis, we can evaluate the coefficients and upper bound the norm by
	\begin{equation}
	\label{eq:last_term2}
  |\lambda|^3\Lambda^3 \biggl(
	\frac{2}{3}\bar T_2
  + \frac{4}{3}\bar T_3
  + \bar T_4 \biggr).
	\end{equation}
	Combining (\ref{eq:first_term2}) and (\ref{eq:last_term2}), we obtain the claimed upper bound \eq{third_order_err}
	for the third-order error in the first-order formula.

Now we consider the second-order formula. Similarly to the proof of \thm{first_comm}, we begin by proving the bound
	\begin{equation}
	\label{eq:second_pre}
	\norm*{\exp\left(\lambda\sum_{j=1}^{L}H_j\right)-S_2(\lambda)}
	\leq(|\lambda|\Lambda)^3\biggl(\frac{D}{24}+\frac{T_2}{12}+\frac{T_3}{6}+\frac{T_4}{8}\biggr)+\frac{4}{3}(L|\lambda|\Lambda)^4\exp(2L|\lambda|\Lambda),
	\end{equation}
which implies \eq{second_comm_result} by the triangle inequality.

	To establish \eq{second_pre}, we apply \eq{third_order_err} to the augmented Hamiltonian list $\{\tilde{H}_{j}\}_{j=1}^{2L}$ with $\lambda$ replaced by $\frac{\lambda}{2}$. This shows that the third-order error is at most
	\begin{equation}
	\frac{|\lambda|^3\Lambda^3}{8}\biggl(\frac{D}{3}+\frac{2T_2}{3}+\frac{4T_3}{3}+T_4\biggr).
	\end{equation}
	The higher-order errors can be bounded by a routine calculation as
	\begin{equation}
	\rem_{3}\left(\exp\left(\lambda\sum_{j=1}^{L}H_j\right)-S_2(\lambda)\right)
	\leq 2\rem_{3}(\exp(2L|\lambda|\Lambda))
	\leq 2\frac{(2L|\lambda|\Lambda)^4}{4!}\exp(2L|\lambda|\Lambda).
	\end{equation}
	This completes the proof of \eq{second_pre}. The remainder of the proof proceeds similarly to the second half of the proof of \prop{firstana}.
\end{proof}

A similar bound holds for the fourth-order formula, as follows.

\begin{theorem}[Fourth-order commutator bound]
	\label{thm:fourth_comm}
Let $H_1,\ldots,H_L$ be Hermitian operators with norm at most $\Lambda := \max_{j}\norm{H_{j}}$, where each $H_j$ is a tensor product of Pauli operators. Define the augmented set of Hamiltonians
	\begin{equation}
	\tilde{H}_{j}=
	\begin{cases}
	H_{j-2hL},   & 2hL+1\leq j\leq (2h+1)L\\
	H_{2(h+1)L-j+1},       & (2h+1)L+1\leq j \leq 2(h+1)L\\
	\end{cases}
	\qquad h \in \{0,1,2,3,4\}.
	\end{equation}
	Let $f(i,j)=1$ if $\tilde{H}_{i}, \tilde{H}_{j}$ commute and $f(i,j)=-1$ otherwise. Finally, let
	\begin{align}
	\label{eq:comm_pattern}
	N_{a}&:=|\{(i,j) : f(i,j)=a_{1},\,i<j\}|\\
	N_{b_1 b_2 b_3}&:=|\{(i,j,k) : f(i,j)=b_{1},\,f(i,k)=b_2,\,f(j,k)=b_3,\,i<j<k\}|\\
	N_{c_1 \ldots c_6}&:=|\{(i,j,k,l) : f(i,j)=c_{1},\,f(i,k)=c_2,\,f(i,l)=c_3, \nonumber\\
	&\qquad f(j,k)=c_4,\,f(j,l)=c_5,\,f(k,l)=c_6,\,i<j<k<l\}|\\
	N_{d_1 \ldots d_{10}}&:=|\{(i,j,k,l,m) : f(i,j)=d_{1},\,f(i,k)=d_2,\,f(i,l)=d_3, \nonumber\\
	&\qquad f(i,m)=d_4,\,f(j,k)=d_5,\,f(j,l)=d_6,\,f(j,m)=d_7,\,f(k,l)=d_8, \nonumber\\
	&\qquad f(k,m)=d_9,\,f(l,m)=d_{10},\,i<j<k<l<m\}|
  \label{eq:comm_pattern_last}
	\end{align}
  where $i,j,k,l,m \in \{1,\ldots,10L\}$, let $t \in\R$, and let $p:=1/(4-4^{1/3})$. Then
	\begin{align}
	\label{eq:fourth_err}
	\norm*{\exp\Biggl(-it\sum_{j=1}^{L}H_{j}\Biggr)-\left[S_{4}\left(-\frac{it}{r}\right)\right]^r}
	&\leq\left(\frac{4p-1}{2}\Lambda |t|\right)^5\frac{1}{5! r^4}\bigg\{
	\sum_{a \in \pm1} C_{a}N_{a}+\sum_{b_1,b_2,b_3 \in \pm1} C_{b_1 b_2 b_3}N_{b_1 b_2 b_3} \nonumber\\
	&\quad+\sum_{c_1,\ldots,c_6 \in \pm1} C_{c_1 \ldots c_6}N_{c_1 \ldots c_6}
	+\sum_{d_1,\ldots,d_{10}\in \pm1} C_{d_1 \ldots d_{10}}N_{d_1 \ldots d_{10}}
	\bigg\} \nonumber\\
	&\quad+2\frac{(5(4p-1)L\Lambda t)^6}{6!\cdot r^5}\exp\left(\frac{(5(4p-1)L\Lambda |t|}{r}\right)
	\end{align}
	for some real coefficients $C_{a}$, $C_{b_1 b_2 b_3}$, $C_{c_1 \ldots c_6}$, $C_{d_1 \ldots d_{10}}$.
\end{theorem}

We omit the proof, which proceeds along similar lines to that of \thm{second_comm}.  Note that similar bounds also hold for higher-order formulas.

The coefficients $C_{a}$, $C_{b_1 b_2 b_3}$, $C_{c_1 \ldots c_6}$, $C_{d_1 \ldots d_{10}}$ can in principle be determined by a computer program.  The list of coefficients is long, so we omit it here. In \sec{pfcomconcrete}, we give a specialized version of this theorem that can be applied to the Hamiltonian \eq{heisenberg} without explicitly computing all the coefficients.

To give an idea of how the bound is proved, we show how to determine the coefficient $C_{-1}$ in \eq{fourth_err}. Similar arguments can be used to determine all the coefficients in the bound.

First consider the fifth-order terms of the expression
	\begin{equation}
	\label{eq:fifth_err}
	\exp\left(\sum_{j=1}^{L}H_{j}\lambda\right)-\exp(H_{1}\lambda)\cdots \exp(H_{L}\lambda).
	\end{equation}
The coefficient $C_{-1}$ of $N_{-1}$ counts the pairs of non-commuting terms $H_{i}$ and $H_{j}$. The second term in \eq{fifth_err} contributes
	\begin{equation}
	\label{eq:second_term}
	\begin{aligned}
	\frac{\lambda^4}{4!}\lambda\sum_{i<j}\left(H_{i}^{4}H_{j}+H_{i}H_{j}^{4}\right)
	+\frac{\lambda^3}{3!}\frac{\lambda^2}{2!}\sum_{i<j}\left(H_{i}^{3}H_{j}^{2}+H_{i}^{2}H_{j}^{3}\right),
	\end{aligned}
	\end{equation}
	whereas the first term in (\ref{eq:fifth_err}) contributes
	\begin{equation}\label{eq:fifth_err_first}
	\begin{aligned}
	\frac{\lambda^5}{5!}\sum_{i\neq j}&\bigl(H_i^4H_j+H_i^3H_jH_i+H_i^2H_jH_i^2+H_iH_jH_i^3+H_jH_i^4\bigr)\\
	+\frac{\lambda^5}{5!}\sum_{i\neq j}&\bigl(H_i^3H_j^2+H_i^2H_jH_iH_j+H_i^2H_j^2H_i+H_iH_jH_i^2H_j+H_iH_jH_iH_jH_i \\
  &\quad +H_iH_j^2H_i^2+H_jH_i^3H_j+H_jH_i^2H_jH_i+H_jH_iH_jH_i^2+H_j^2H_i^3\bigr).
	\end{aligned}
	\end{equation}
Since we assume the terms of the Hamiltonian are tensor products of Pauli operators, we can interchange the order of multiplication, possibly introducing minus signs. Thus \eq{fifth_err_first} equals
	\begin{align}
	\label{eq:first_term}
	&\frac{\lambda^5}{5!}\sum_{i<j}\left(\bigl(3+2f(i,j)\bigr)H_i^4H_j+\bigl(3+2f(i,j)\bigr)H_iH_j^4\right) \nonumber\\
	&\quad+\frac{\lambda^5}{5!}\sum_{i<j}\left(\bigl(6+4f(i,j)\bigr)H_i^3H_j^2+\bigl(6+4f(i,j)\bigr)H_i^2H_j^3\right).
	\end{align}
	Subtracting (\ref{eq:first_term}) from (\ref{eq:second_term}) and setting $f(i,j)=-1$, we find
	\begin{equation}
	\frac{\lambda^5}{5!}\sum_{i<j}\left(4H_i^4H_j+4H_iH_j^4+8H_i^3H_j^2+8H_i^2H_j^3\right),
	\end{equation}
	whose spectral norm is bounded by
	\begin{equation}
	\frac{|\lambda|^5}{5!}24\Lambda^5N_{-1}.
	\end{equation}
Comparing the result to \eq{fourth_err}, we find that $C_{-1}=24$.

\subsubsection{Concrete commutator bounds}
\label{sec:pfcomconcrete}

To apply the above commutator bounds, we must compute the number of tuples of terms in the Hamiltonian satisfying certain commutation relations (e.g., equations \eq{comm_pattern}--\eq{comm_pattern_last} for the fourth-order bound). While this can be done in polynomial time provided the degree is constant, a direct approach is prohibitive in practice. For example, consider the problem of counting $5$-tuples of $\Theta(n)$ terms, as required for the fourth-order commutator bound. A naive iteration over all 5-tuples takes time $O(n^5)$, and the multiplicative constant is large since our Hamiltonian gives rise to $40n$ terms $\tilde H_i$. Thus a direct computation can be infeasible even for small $n$.

However, for the Hamiltonian \eq{heisenberg}, we show that each number of tuples is given by a low-degree polynomial in $n$. In turn, this means that the lowest-order contribution to the error is also a polynomial in $n$.  Thus, by performing polynomial interpolation on a constant number of numerically-obtained values, we can determine a closed-form expression for general $n$.  A similar result holds for other Hamiltonians with appropriate symmetries, so this approach could be applied more broadly.

We first assume the Hamiltonian acts locally on a line (in particular, we assume open boundary conditions instead of periodic ones). For example, we could consider the Hamiltonian
\begin{align}
	\sum_{j=1}^{n-1} \bigl(\vec \sigma_j \cdot \vec \sigma_{j+1} + h_j \sigma_j^z \bigr),
	\label{eq:heisenberg_line}
\end{align}
which is simply the Heisenberg model \eq{heisenberg} without the boundary term, and without the local magnetic field on the rightmost spin (that term could also be handled, but we exclude it to simplify the presentation). We show that if a counting problem for such a Hamiltonian is \emph{oblivious} (defined below), then the result is a polynomial in $n$. We then show that a similar property holds for Hamiltonians with periodic boundary conditions, giving a method for evaluating the commutator bound for our model system.

We focus on the list of terms in the Hamiltonian to be applied in each segment of a product formula. Specifically, consider $n$ qubits arranged on a line, and suppose we are given Hermitian operators $G_1,\ldots,G_\numcycles$ with values $d_1,\ldots,d_\numcycles \in \{-1,1\}$ indicating shift directions. We assume that each term $G_i$ acts nontrivially only on the first $k$ qubits if $d_i=1$ or on the last $k$ qubits if $d_i=-1$.
This determines an ordered list of length $\numterms := \numcycles (n-k+1)$ including all translations of the given terms, of the form
\begin{align}
\bigl(&G_1,R^{d_1} G_1 R^{-d_1},\ldots,R^{d_1(n-k)} G_1 R^{-d_1(n-k)}, \nonumber\\
&G_2,R^{d_2} G_2 R^{-d_2},\ldots,R^{d_2(n-k)} G_2 R^{-d_2(n-k)}, \nonumber\\
&\vdots \nonumber\\
&G_\numcycles,R^{d_\numcycles} G_\numcycles R^{-d_\numcycles},\ldots,R^{d_\numcycles(n-k)} G_\numcycles R^{-d_\numcycles(n-k)}\bigr),
\label{eq:hamcycles}
\end{align}
where $R$ is the operator that shifts the qubits right by one position.  For each $i \in \range{\mu} := \{1,\ldots,\mu\}$, let $H_i$ denote the $i$th term in this list.

In the example \eq{heisenberg_line}, the $n$-qubit Hamiltonian consists of $4(n-1)$ terms. We adopt a specific ordering of the terms, namely
\begin{align}
  \sigma_1^x \sigma_2^x, \ldots, \sigma_{n-1}^x \sigma_n^x, \;
  \sigma_1^y \sigma_2^y, \ldots, \sigma_{n-1}^y \sigma_n^y, \;
  \sigma_1^z \sigma_2^z, \ldots, \sigma_{n-1}^z \sigma_n^z, \;
  \sigma_1^z, \ldots, \sigma_{n-1}^z.
  \label{eq:opentermorder}
\end{align}
Using the fourth-order product formula, every term of the Hamiltonian repeats $10$ times in a fixed order within each segment. Thus, we obtain an ordered list of $2$-local operators (i.e., $k=2$) with length $\numterms = 40(n-1)$, and we see that it has the form of \eq{hamcycles} with $\numcycles=40$.

We now introduce two mappings, $\support$ and $\interval$. The mapping $\support\colon \range{\numterms}\rightarrow \Z_{2}^{n}$ returns the \emph{support} of a term in the Hamiltonian. The support of $G_i$ is $1$ in the first $k$ positions and $0$ in the rest if $d_i=1$, or $1$ in the last $k$ positions and $0$ in the rest if $d_i=-1$. Similarly, the support of $R^j G_i R^{-j}$ has a $1$ at positions $j+1,j+2,\ldots,j+k$ if $d_i=1$, or a $1$ at positions $n-j-k+1,n-j-k+2,\ldots,n-j$ if $d_i=-1$.  Note that the support may include positions at which the Hamiltonian acts trivially.

To study the action of a $\tuplesize$-tuple of Hamiltonian operators, we take the entrywise $\OR$ of all elements of the tuple, defining
\begin{align}
	\support\colon \mathcal{S}^\numterms_\tuplesize&\rightarrow \Z_{2}^{n}\\
	\support\bigl(
\{i_1,\ldots,i_\tuplesize\}
	\bigr)&=\bigvee_{j=1}^\tuplesize\support(i_j).
\end{align}
Here $\mathcal{S}^\numterms_\tuplesize := \{S \subseteq \range{\numterms} : |S|=\tau\}$ is the set of all subsets of $\range{\numterms}$ with size $\tau$.
Since the support of each term in the Hamiltonian has Hamming weight $k$, the Hamming weight of $\support(\{i_1,\ldots,i_\tuplesize\})$ is at most $\tuplesize k$.

The mapping $\interval\colon \Z_{2}^{n}\rightarrow\Z_{\tuplesize k}^\tuplesize$ returns the \emph{interval pattern} for a given subset of qubits. The interval pattern is a list of the numbers of consecutive $1$s in the string, padded with trailing $0$s if necessary to have length $\tau k$. For example, $\interval([1\ 0\ 1\ 1\ 0])=[1\ 2\ 0\ \ldots\ 0]$.

Recall that, by assumption, every term in the Hamiltonian acts nontrivially on at most $k$ consecutive qubits, and its support is defined to include all $k$ relevant positions (possibly including ones where the Hamiltonian acts trivially). Thus any such term can occupy at most one interval of any given interval pattern.

Our main goal is to count the number of $\tuplesize$-tuples of Hamiltonian operators satisfying a given commutation pattern. We regard this as a special case of a more general task, which we call \emph{oblivious counting}.

Let $f\colon \mathcal{S}^\numterms_\tuplesize\rightarrow \Z_{2}$ be a binary-valued function that represents whether a tuple of operators has some specified property. For example, when computing the fourth-order commutator bound, we set $f(\{i_1, i_2, i_3, i_4, i_5\})=1$ if the tuple $(H_{i_1},\ldots,H_{i_5})$ satisfies a given commutation pattern, and $f(\{i_1, i_2, i_3, i_4, i_5\})=0$ otherwise. Our task now is to count the number of tuples satisfying the property, i.e., to compute
\begin{equation}
	\bigl|f^{-1}(1)\bigr|=\bigl|\bigl\{\{i_1, \ldots, i_\tuplesize\} : f\bigl(\{i_1, \ldots, i_\tuplesize\}\bigr)=1\bigr\}\bigr|.
	\label{eq:countingprob}
\end{equation}
We say the counting is \emph{oblivious} if, for any fixed $z \in \Z_{\tuplesize k}^{\tuplesize}$, the quantity
\begin{equation}
	\bigl|\bigl\{\{i_1, \ldots, i_\tuplesize\} : f\bigl(\{i_1, \ldots, i_\tuplesize\}\bigr)=1,\, \support\bigl(\{i_1, \ldots, i_\tuplesize\}\bigr)=y,\, \interval(y)=z\bigr\}\bigr|
\end{equation}
is identical for all $n$ and all $y\in\Z_2^n$ with $\interval(y)=z$. In other words, an oblivious counting problem depends only on the interval pattern of the $\tuplesize$-tuple, not on its precise support or the size of the underlying system.

Our main result for counting such tuples is as follows.

\begin{theorem}
	\label{thm:poly_rep}
	Let $G_1,\ldots,G_\numcycles$ be Hermitian operators with $d_1,\ldots,d_\numcycles \in \{-1,1\}$. We assume that each term $G_i$ acts as the identity on all but the first (if $d_i=1$) or last (if $d_i=-1$) $k$ qubits.
	Suppose that $f\colon\mathcal{S}^\numterms_\tuplesize\rightarrow \Z_{2}$ induces an oblivious counting problem, as defined above. Then $|f^{-1}(1)|$ is a polynomial in $n$.
\end{theorem}

\begin{proof}
	Our task is to compute
	\begin{align}
	\bigl|f^{-1}(1)\bigr|
	&=\bigl|\bigl\{\{i_1, \ldots, i_\tuplesize\} : f\bigl(\{i_1, \ldots, i_\tuplesize\}\bigr)=1\bigr\}\bigr|\\
	&=\sum_{z\in\Z_{\tuplesize k}^{\tuplesize}}\sum_{y\in\Z_{2}^{n}}C_{zy}(n),
	\end{align}
	where
	\begin{equation}
		C_{zy}(n):=\bigl|\bigl\{\{i_1, \ldots, i_\tuplesize\} : f\bigl(\{i_1, \ldots, i_\tuplesize\}\bigr)=1,\, \support\bigl(\{i_1, \ldots, i_\tuplesize\}\bigr)=y,\, \interval(y)=z\bigr\}\bigr|.
	\end{equation}

We claim that
\begin{enumerate}[topsep=4pt,itemsep=0pt,label=(\roman*)]
		\item a constant number of values $z\in\Z_{\tuplesize k}^{\tuplesize}$ (independent of $n$) contribute to the sum; and \label{item:const}
		\item for fixed $z\in\Z_{\tuplesize k}^{\tuplesize}$, the number of $y\in\Z_{2}^{n}$ such that $\interval(y)=z$ is a polynomial in $n$. \label{item:poly}
\end{enumerate}
Claim \ref{item:const} follows because we assume the size $\tuplesize$ of each tuple and the number $k$ of local operators are both independent of $n$.  To prove claim \ref{item:poly}, fix an interval pattern
	\begin{equation}
		z=
		\begin{bmatrix}
			z_1 & \ldots & z_r & 0 & \ldots & 0
		\end{bmatrix}
		\in\Z_{\tuplesize k}^{\tuplesize}
	\end{equation}
	and define
\begin{align}
  T(n; z_1,\ldots,z_r) := |\{y\in\Z_{2}^{n} : \interval(y)=z \}|.
\end{align}
We show by induction on $r$ that $T(n; z_1,\ldots,z_r)$ is a polynomial of degree at most $r$. The base case requires us to place an interval of length $z_1$ in $n$ slots, and clearly $T(n; z_1) = n-z_1+1$, a polynomial of degree $1$.

For the induction step, suppose the claim is true for any $r-1$ intervals $(z_1,\ldots,z_{r-1},0,\ldots,0)$. Using the hypothesis for the case of $r-1$ intervals, we find
	\begin{align}
	T(n; z_1,\ldots,z_{r})
	&=\sum_{j=1}^{n-z_r}T(n-z_r-j; z_1,\ldots,z_{r-1})\\
	&=\sum_{j=1}^{n-z_r}c_{r-1}(z_1,\ldots,z_{r-1})\,(n-z_r-j)^{r-1} \nonumber\\
  &\quad\qquad+ \cdots +c_{0}(z_1,\ldots,z_{r-1})\,(n-z_r-j)^{0},
	\end{align}
where in the second equality we use the hypothesis to conclude that
\begin{equation}
	T(n; z_1,\ldots,z_{r-1})=c_{r-1}(z_1,\ldots,z_{r-1})n^{r-1}+\cdots+c_{0}(z_1,\ldots,z_{r-1})n^{0}
\end{equation}
is a polynomial of degree at most $r-1$.
It follows from Faulhaber's formula that $T(n; z_1,\ldots,z_{r})$ is a polynomial of degree at most $r$, which completes the inductive step.

To complete the proof of the main result, we observe that
\begin{align}
	\bigl|f^{-1}(1)\bigr|
	&= \sum_{z\in\Z_{\tuplesize k}^{\tuplesize}}\sum_{\substack{y\in\Z_{2}^{n}\\\interval(y)=z}}C_{zy}(n)\\
	&= \sum_{z\in\Z_{\tuplesize k}^{\tuplesize}}p_z(n) C_{z},
	\end{align}
where, for each $z \in \Z_{\tuplesize k}^{\tuplesize}$, $p_z(n)$ is a polynomial in $n$ (by claim \ref{item:poly}).  Here we denote $C_z := C_{zy}(n)$ since it is independent of $n$ and $y$, because we assume the counting problem is oblivious.
\end{proof}

The above theorem relies crucially on the assumption that the terms in the Hamiltonian are arranged on a line.
We now describe a variant of the theorem that handles Hamiltonians with periodic boundary conditions.
To this end, consider $n$ qubits arranged on a ring, and again suppose we are given Hermitian operators $G_1,\ldots,G_\numcycles$ with values $d_1,\ldots,d_\numcycles \in \{-1,1\}$ indicating shift directions. We assume that each term $G_i$ acts as the identity except on the first $k$ qubits if $d_i=1$, or on the first $k-1$ qubits and the last qubit if $d_i=-1$. This determines an ordered list of length $\numterms := \numcycles n$ including all translations of the given terms, of the form
\begin{align}
	\bigl(&G_1,R^{d_1} G_1 R^{-d_1},\ldots,R^{d_1(n-1)} G_1 R^{-d_1(n-1)}, \nonumber\\
	&G_2,R^{d_2} G_2 R^{-d_2},\ldots,R^{d_2(n-1)} G_2 R^{-d_2(n-1)}, \nonumber\\
	&\vdots \nonumber\\
	&G_\numcycles,R^{d_\numcycles} G_\numcycles R^{-d_\numcycles},\ldots,R^{d_\numcycles(n-1)} G_\numcycles R^{-d_\numcycles(n-1)}\bigr),
	\label{eq:hamcycles_ring}
\end{align}
where $R$ is the operator that shifts the qubits right by one position.  For each $i \in \range{\mu} := \{1,\ldots,\mu\}$, let $H_i$ denote the $i$th term in this list.

We define the support similarly as before, except that we take into account the periodic boundary conditions. In particular, the support of $R^j G_i R^{-j}$ with $d_i=1$ has a $1$ at positions $(j\bmod n)+1,((j+1)\bmod n)+1,\ldots,((j+k-1)\bmod n)+1$. Likewise, the support of $R^{-j} G_i R^{j}$ with $d_i=-1$ has a $1$ at positions $((n-j-1)\bmod n)+1,((n-j)\bmod n)+1,\ldots,((n-j+k-2)\bmod n)+1$. Again we extend the definition of support to a $\tuplesize$-tuple of terms by taking the entrywise OR of all elements of the tuple.  The notion of an interval pattern remains unchanged.

Our analysis divides the problem into a part where the boundary is irrelevant (so that we can apply the previous analysis) and a part where the boundary plays a role (which only involves a constant number of terms).  To do this, we divide a support into three parts: the \emph{prefix}, the \emph{suffix}, and the \emph{infix}.

Let $y\in\Z_{2}^{n}$ be the support of a tuple of operators. If $y_1=0$, we define $\prefix(y)\in\Z_{2}^{n}$ to be the zero vector. Otherwise, $\prefix(y)$ is the indicator function of the first interval in $y$ (i.e., it is a $0/1$ vector that is $1$ only on the first interval in $y)$. Likewise, if $y_n=0$, we define $\suffix(y)\in\Z_{2}^{n}$ to be the zero vector. Otherwise, $\suffix(y)$ is the indicator function of the last interval in $y$. Then $\infix(y) \in \Z_{2}^{n}$ is the vector obtained by removing the first and the last interval from $y$, i.e.,
\begin{equation}
	\infix(y) := y\oplus \prefix(y) \oplus \suffix(y).
\end{equation}

Let $z\in\Z_{\tuplesize k}^{\tuplesize}$ be an interval pattern. We define the \emph{boundary} as a mapping $\boundary\colon \Z_{\tuplesize k}^{\tuplesize}\rightarrow \Z_{2}^{n}$. If there exists a $y\in\Z_{2}^{n}$ with $y_1=y_n=1$ such that $\interval(y)=z$, then $\boundary(z)=\prefix(y)\oplus \suffix(y)$. Otherwise, we set $\boundary(z)$ to be the zero vector. The mapping $\boundary$ is a well-defined function of the interval pattern $z$ since it does not depend on the specific choice of support $y$ (provided $y_1=y_n=1$, which is the only case in which we will use this function).

We now consider the counting problem.  As before, we let $f\colon \mathcal{S}^\numterms_\tuplesize\rightarrow \Z_{2}$ be a binary-valued function that represents whether a tuple of operators has some specified property, and our goal is to compute $|f^{-1}(1)|$ as in \eq{countingprob}.
For a Hamiltonian on a ring with periodic boundary conditions, we say that the counting problem is \emph{oblivious} if for any fixed $z\in \Z_{\tuplesize k}^{\tuplesize}$, the following two requirements are satisfied:
\begin{enumerate}[topsep=4pt,itemsep=0pt,label=(O\arabic*)]
	\item the quantity
	\begin{equation}
	\bigl|\bigl\{\{i_1, \ldots, i_\tuplesize\} : f\bigl(\{i_1, \ldots, i_\tuplesize\}\bigr)=1,\, \support\bigl(\{i_1, \ldots, i_\tuplesize\}\bigr)=y,\, \interval(y)=z\bigr\}\bigr|
	\end{equation}
	is identical for all $n$ and all $y\in\Z_2^n$ with $\interval(y)=z$ and $y_1 y_n=0$; and \label{item:nonboundary}
	\item \label{item:boundary_poly} fixing $\omega\in\range{\tau}$ and $\{j_1,\ldots,j_\omega\}\in\mathcal{S}^\numterms_\omega$ with $\support(\{j_1,\ldots,j_\omega\})=\boundary(z)$, the quantity
	\begin{equation}
	\bigl|\bigl\{\{k_1, \ldots, k_{\tuplesize-\omega}\} : f_{j_1,\ldots,j_\omega}\bigl(\{k_1, \ldots, k_{\tuplesize-\omega}\}\bigr)=1,\, \support\bigl(\{k_1, \ldots, k_{\tuplesize-\omega}\}\bigr)=\infix(y),\, \interval(y)=z\bigr\}\bigr|
	\end{equation}
	is identical for all $n$ and all $y\in\Z_2^n$ with $\interval(y)=z$ and $y_1=y_n=1$. Here $f_{j_1,\ldots,j_\omega}$ is a restricted version of $f$ defined by
	\begin{equation}
		f_{j_1,\ldots,j_\omega}\bigl(\{k_1, \ldots, k_{\tuplesize-\omega}\}\bigr)=f\bigl(\{j_1,\ldots,j_\omega,k_1, \ldots, k_{\tuplesize-\omega}\}\bigr).
	\end{equation}
\end{enumerate}
Note that the set of indices $\{j_1,\ldots,j_\omega\}$ is disjoint from $\{k_1, \ldots, k_{\tuplesize-\omega}\}$ as $\boundary(z)$ and $\infix(y)$ do not overlap. Therefore, the function $f_{j_1,\ldots,j_\omega}$ is well-defined. Note also that in the first case, no operator can wrap around the boundary since $y_1=0$ or $y_n= 0$. We can thus impose obliviousness on $n-1$ qubits, where the Hamiltonian is arranged on a line. In the second case, we have $y_1=y_n=1$ so that an operator could wrap around the boundary. However, the above definition of obliviousness allows us to enumerate all possible ways to put tuples on the boundary. The usual requirement of obliviousness can then be imposed on the remaining qubits that are not on the boundary.

We now show that oblivious counting
also results in a polynomial in $n$ in the periodic case. Throughout the remainder of this subsection, we use the notations
\begin{align}
	C_{zy}(n,f)&:=\bigl|\bigl\{\{i_1, \ldots, i_\tuplesize\} : f\bigl(\{i_1, \ldots, i_\tuplesize\}\bigr)=1,\, \support\bigl(\{i_1, \ldots, i_\tuplesize\}\bigr)=y,\, \interval(y)=z\bigr\}\bigr| \\
	C_{zy}(n,f_{j_1,\ldots,j_\omega})&:= \bigl|\bigl\{\{k_1, \ldots, k_{\tuplesize-\omega}\} : \begin{aligned}[t] & f_{j_1,\ldots,j_\omega}\bigl(\{k_1, \ldots, k_{\tuplesize-\omega}\}\bigr)=1,\, \\
	& \support\bigl(\{k_1, \ldots, k_{\tuplesize-\omega}\}\bigr)=\infix(y),\, \interval(y)=z\bigr\}\bigr|.
	\end{aligned}
\end{align}

\begin{theorem}
	\label{thm:poly_rep_ring}
	Let $G_1,\ldots,G_\numcycles$ be Hermitian operators and let $d_1,\ldots,d_\numcycles \in \{-1,1\}$. We assume that each term $G_i$ acts as identity except on the first $k$ qubits if $d_i=1$, or on the first $k-1$ qubits and the last qubit if $d_i=-1$.
	Suppose that $f\colon\mathcal{S}^\numterms_\tuplesize\rightarrow \Z_{2}$ induces an oblivious counting problem, as defined above. Then $|f^{-1}(1)|$ is a polynomial in $n$.
\end{theorem}
\begin{proof}
	We compute $|f^{-1}(1)|$ by expanding it over all interval patterns and supports as
	\begin{align}
		|f^{-1}(1)|
		&=\sum_{z\in\Z_{\tuplesize k}^{\tuplesize}}\sum_{y\in\Z_{2}^{n}}C_{zy}(n,f) \\
		&= \sum_{z\in\Z_{\tuplesize k}^{\tuplesize}}\sum_{\substack{y\in\Z_{2}^{n}\\ y_1 y_n=0}}C_{zy}(n,f)
		+ \sum_{z\in\Z_{\tuplesize k}^{\tuplesize}}\sum_{\substack{y\in\Z_{2}^{n}\\ y_1=y_n=1}}C_{zy}(n,f).
	\end{align}
	For the first term, no operator crosses the boundary. We can thus reduce the problem to counting $\tau$-tuples on $n-1$ qubits with interval pattern $z$ and support given by either the first or last $n-1$ positions of $y$. The Hamiltonian can effectively be regarded as arranged on a line. Invoking \thm{poly_rep}, we see that the first term is a polynomial in $n$.

	To handle the second term, we use a basic counting argument: to count the number of $\tuplesize$-tuples, we add up all the possible cases where $\omega$ of them are on the boundary $\boundary(z)$ and the other $\tuplesize-\omega$ have support $\infix(y)$ away from the boundary. This allows us to rewrite
	\begin{equation}
		\sum_{z\in\Z_{\tuplesize k}^{\tuplesize}}\sum_{\substack{y\in\Z_{2}^{n}\\ y_1=y_n=1}}C_{zy}(n,f)
		=\sum_{z\in\Z_{\tuplesize k}^{\tuplesize}}\sum_{\omega\in\range{\tuplesize}}\sum_{\substack{\{j_1,\ldots,j_\omega\}\in\mathcal{S}^\numterms_\omega\\ \support(\{j_1,\ldots,j_\omega\})=\boundary(z)}}
		\sum_{\substack{y\in\Z_{2}^{n}\\ y_1=y_n=1}}C_{zy}(n,f_{j_1,\ldots,j_\omega}).
	\end{equation}
	By \ref{item:boundary_poly} and a similar argument as above, the result is a sum of constant number of polynomials in $n$, which completes the proof.
\end{proof}

Now we apply \thm{poly_rep_ring} to compute commutator bounds for our model Hamiltonian \eq{heisenberg}.
To construct a product formula, we order the terms of the Hamiltonian similarly to the example in \eq{opentermorder}; in particular, we take them in the order
\begin{align}
  \sigma_1^x \sigma_2^x, \ldots, \sigma_{n-1}^x \sigma_n^x, \sigma_n^x \sigma_1^x, \;
  \sigma_1^y \sigma_2^y, \ldots, \sigma_{n-1}^y \sigma_n^y, \sigma_n^y \sigma_1^y, \;
  \sigma_1^z \sigma_2^z, \ldots, \sigma_{n-1}^z \sigma_n^z, \sigma_n^z \sigma_1^z,\;
  \sigma_1^z, \ldots, \sigma_n^z.
  \label{eq:cycletermorder}
\end{align}

Our approach can be applied to product formulas of any order. However, we focus on the second- and fourth-order cases, for which the list of of operators in the product formulas have lengths $\numterms=8n$ and $\numterms=40n$, respectively. The resulting lists of terms take the form specified in \thm{poly_rep_ring}. Since each term in the Hamiltonian acts on at most two neighboring qubits, we can take the support of each operator to be a 0/1 vector of Hamming weight $2$, with a $1$ in two consecutive positions.

To apply the commutator bounds, we must count $\tuplesize$-tuples of Hamiltonian operators satisfying a certain commutation pattern, where $\tuplesize \in \{2,3\}$ for the second-order bound and $\tuplesize \in \{2,3,4,5\}$ for the fourth-order case. It remains to check that this counting problem is oblivious.

We first verify condition \ref{item:nonboundary}. Fix an interval pattern $z\in \Z_{\tuplesize k}^{\tuplesize}$ and consider supports $y\in\Z_2^n$ and $y' \in \Z_2^{n'}$ with $\interval(y)=\interval(y')=z$ and $y_1 y_n = y'_1 y'_{n'} = 0$. By the latter conditions, we know that the relevant operators $H_{i_1},\ldots,H_{i_\tuplesize}$ do not touch the boundary. Intuitively, since the Hamiltonian includes all possible shifts of any given term, the counting problems for $y$ and $y'$ should be equivalent.

To make this rigorous, we set up a one-to-one correspondence between $\tuplesize$-tuples with support $y$ and those with support $y'$.
First consider the special case in which there is only one interval with the same length for both $y$ and $y'$. Let $H_{j_1},\ldots,H_{j_\omega}$ be operators with support $\support(\{j_1,\ldots,j_\omega\})=y$. We construct a corresponding $\omega$-tuple with support $y'$ as follows. Let $\delta$ be the position difference between the first non-zero element of $y'$ and that of $y$.
For any term in the Hamiltonian of the form $R^{d_i l} G_i R^{-d_il}$, the corresponding operator is $R^{d_il+\delta} G_i R^{-d_il-\delta}$.
More generally, we can apply the same transformation to all intervals in $y$ and $y'$ with lengths $z_1,\ldots,z_\tuplesize$ for any $\omega\in\range{\tau}$ and $\{j_1,\ldots,j_\omega\}\in\mathcal{S}^\numterms_\omega$. The resulting correspondence is easily checked to be invertible, so it specifies a bijection between $\omega$-tuples as claimed.
We have thus established condition \ref{item:nonboundary}.
Condition \ref{item:boundary_poly} follows similarly as $\boundary(z)$ and $\infix(y)$ do not overlap.

Using symbolic polynomial interpolation, we derive the following succinct commutator bounds for the second- and fourth-order formulas.

\begin{theorem}[Second-order commutator bound, succinct form]
	\label{thm:second_succinct}
  Let $H$ be the Hamiltonian \eq{heisenberg}, with terms ordered as in \eq{cycletermorder}.  Then the error in the second-order product formula approximation satisfies
	\begin{equation}
	\begin{aligned}
	\norm*{\exp(-iHt)-\left[S_{2}\left(-{it}/{r}\right)\right]^r}
	\leq\frac{\Lambda^3|t|^3}{r^2}T_2(n)
	+\frac{4(4n\Lambda t)^4}{3r^3}\exp\left(\frac{8n\Lambda |t|}{r}\right),
	\end{aligned}
	\end{equation}
	where
	\begin{equation}
	T_2(n) :=
	\begin{cases}
	194,\ &n=3 \\
	40n^2-58n,\ &n\geq 4.
	\end{cases}
	\end{equation}
\end{theorem}

\begin{theorem}[Fourth-order commutator bound, succinct form]
	\label{thm:fourth_succinct}
  Let $H$ be the Hamiltonian \eq{heisenberg}, with terms ordered as in \eq{cycletermorder}, and let $p:=1/(4-4^{1/3})$.
  Then the error in the fourth-order product formula approximation satisfies
	\begin{align}
	&\norm*{\exp(-iHt)-\left[S_{4}\left(-{it}/{r}\right)\right]^r} \nonumber\\
	&\quad \leq \biggl(\frac{4p-1}{2}\Lambda |t|\biggr)^5\frac{1}{5!\cdot r^4}T_4(n)
	+2\frac{(20(4p-1)n\Lambda t)^6}{6!\cdot r^5}\exp\left(\frac{20(4p-1)n\Lambda |t|}{r}\right),
	\end{align}
  where
	\begin{equation}
	T_4(n) :=
	\begin{cases}
	23073564672,\ &n=3 \\
	94192316416,\ &n=4\\
	278878851840,\ &n=5\\
	1280000000n^4-7701760000n^3+23685120000n^2-30224677632n,\ &n\geq 6.
	\end{cases}
	\end{equation}
\end{theorem}

We conclude this appendix by considering the asymptotic gate complexity of the \PF\ algorithm using our commutator bounds. Take the fourth-order bound as an example. With $\Lambda=\Theta(1)$ and $t=\Theta(n)$, the commutator bound is
\begin{equation}
\norm*{\exp(-iHt)-\left[S_{4}\left(-{it}/{r}\right)\right]^r}
\leq O\left(\frac{n^9}{r^4}+\frac{n^{12}}{r^5}\right).
\end{equation}
To guarantee that the simulation error $\epsilon$ is at most some constant, it suffices to use $r=O(n^{2.4})$ segments. Since the circuit for each segment has size $O(n)$, we have an overall gate complexity of $O(n^{3.4})$. Along similar lines, we find gate complexities of $O(n^4)$ (resp., $O(n^{11/3})$) for the first-order (resp., second-order) commutator bound. These bounds improve the asymptotic gate complexities of the \PF\ algorithm established by the analytic and minimized bounds (as shown in \tab{algsummary}), which give $O(n^5)$ for the first-order formula, $O(n^4)$ for second order, and $O(n^{3.5})$ for fourth order.

We only present concrete commutator bounds for the first-, second-, and fourth-order product formulas. In general, to obtain the $2k$th-order commutator bound, one must count the number of $(2k+1)$-tuples satisfying a certain commutation pattern in a list of operators of length $8\cdot5^{k-1}n$. For $k\geq 3$, computing the exact form of the $(2k)$th order bound seems challenging even with the help of \thm{poly_rep_ring} and polynomial interpolation.

Nevertheless, it is still possible to obtain the asymptotic $n$-dependence of the commutator bound. The key step is to study those $(2k+1)$-tuples for which all pairs of operators commute with each other. The number $N_{\mathrm{comm}}$ of such commuting $(2k+1)$-tuples satifsfies
\begin{equation}
	(8\cdot5^{k-1})^{2k+1}\binom{n-2k}{2k+1}\leq N_{\mathrm{comm}}\leq \binom{8\cdot5^{k-1}n}{2k+1}.
\end{equation}
(Here the lower bound follows by placing each $(2k+1)$-tuple so that its support contains exactly $2k+1$ intervals.)  Invoking \thm{poly_rep_ring}, we see that $N_{\textrm{comm}}$ is a polynomial in $n$ whose leading term is $(8\cdot5^{k-1}n)^{2k+1}$.

When we Taylor expand the evolutions $\exp(\lambda H)$ and $S_{2k}(\lambda)$, those $(2k+1)$-tuples for which all pairs of operators commute with each other cancel. The remaining terms are either $(2k+1)$-tuples where at least one pair of operators do not commute, or $l$-tuples with $l\leq 2k$. Our above discussion shows that there are $O(n^{2k})$ such tuples. Therefore, the $2k$th-order commutator bound takes the form
\begin{equation}
	\norm*{\exp\left(-iHt\right)-\left[S_{2k}\left(-{it}/{r}\right)\right]^r}
	\leq O\left(\frac{|t|^{2k+1}n^{2k}}{r^{2k}}+\frac{(nt)^{2k+2}}{r^{2k+1}}\right)
	=O\left(\frac{n^{4k+1}}{r^{2k}}+\frac{n^{4k+4}}{r^{2k+1}}\right).
\end{equation}
To ensure that the simulation error is $O(1)$ for $t=n$, it suffices to choose $r=O(n^{2+2/(2k+1)})$, which leads to a total gate complexity of $O(n^{3+2/(2k+1)})$.  This improves over the performance of the analytic and minimized bounds, which give complexity $O(n^{3+1/k})$.

\subsection{Empirical bounds}
\label{sec:empiricalbounds}

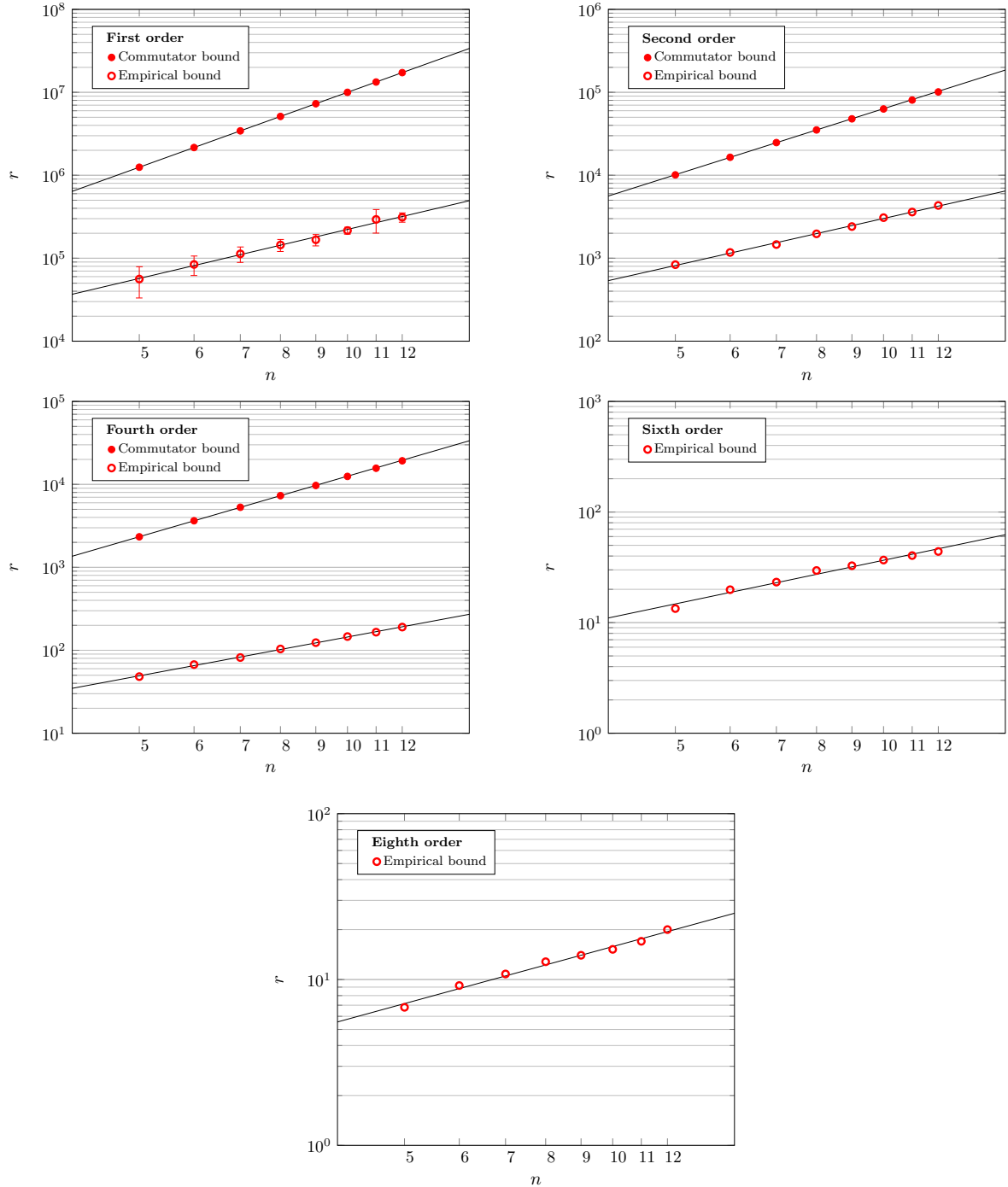
\begin{figure}
  \begin{subfigure}{.5\linewidth}
    \resizebox{.9\textwidth}{!}{
	\begin{tikzpicture}
	\begin{axis}[
	width=10cm,
	log x ticks with fixed point,
	xtick={5,6,7,8,9,10,11,12},
	xmode=log,
	ymode=log,
	xmin = 4,
	xmax = 15,
	ymin = 10^4,
	ymax = 10^8,
	ymajorgrids=true,
	yminorgrids=true,
	legend style={at={(0.05,0.95)},anchor=north west, font=\fontsize{8}{5}\selectfont},
	xlabel={$n$},
	ylabel={$r$},
	every axis legend/.append style={nodes={right},
		every x tick label/.append style={font=\small},
		every y tick label/.append style={font=\small},
	}
	]

  \addlegendimage{empty legend}
  \addlegendentry[yshift=0pt]{\hspace{-.25cm}\textbf{First order}}

	\addplot[only marks, red] coordinates {
		(5,1250268)
		(6,2160462)
		(7,3430733)
		(8,5121094)
		(9,7291555)
		(10,10002133)
		(11,13312839)
		(12,17283687)
	};
	\addlegendentry{Commutator bound}

	\addplot+[only marks,red,mark=o,mark options={fill=white,line width=1.25pt},error bars/.cd,y dir=both,y explicit] coordinates {
		(5,56093.95) +- (0,22918.36)
		(6,84104.95) +- (0,22220.77)
		(7,112854.35) +- (0,23819.89)
		(8,144517.00) +- (0,23849.85)
		(9,166919.60) +- (0,26090.40)
		(10,215675.60) +- (0,20655.86)
		(11,293557.00) +- (0,92756.01)
		(12,311816.20) +- (0,39949.60)
	};
	\addlegendentry{Empirical bound}

	\addplot[
	color = black,
	mark = none
	]	coordinates {
		( 1, 10002.17 )
		( 100, 10002093516.99 )
	};

	\addplot[
	color = black,
	mark = none
	]	coordinates {
		( 1, 2417.22 )
		( 100, 20480580.59 )
	};

	\end{axis}
	\end{tikzpicture}
    }
  \end{subfigure}
  \begin{subfigure}{.5\linewidth}
    \resizebox{.9\textwidth}{!}{
	\begin{tikzpicture}
	\begin{axis}[
	width=10cm,
	log x ticks with fixed point,
	xtick={5,6,7,8,9,10,11,12},
	xmode=log,
	ymode=log,
	xmin = 4,
	xmax = 15,
	ymin = 10^2,
	ymax = 10^6,
	ymajorgrids=true,
	yminorgrids=true,
	legend style={at={(0.05,0.95)},anchor=north west, font=\fontsize{8}{5}\selectfont},
	xlabel={$n$},
	ylabel={$r$},
	every axis legend/.append style={nodes={right},
		every x tick label/.append style={font=\small},
		every y tick label/.append style={font=\small},
	}
	]

  \addlegendimage{empty legend}
  \addlegendentry[yshift=0pt]{\hspace{-.25cm}\textbf{Second order}}

	\addplot[only marks, red] coordinates {
		(5,10110)
		(6,16472)
		(7,24775)
		(8,35190)
		(9,47881)
		(10,62998)
		(11,80689)
		(12,101091)
	};
	\addlegendentry{Commutator bound}

	\addplot[only marks,red,mark=o,mark options={fill=white,line width=1.25pt}] coordinates {
		(5,835.85)
		(6,1173.40)
		(7,1461.80)
		(8,1967.80)
		(9,2408.60)
		(10,3084.00)
		(11,3602.20)
		(12,4322.60)
	};
	\addlegendentry{Empirical bound}

	\addplot[
	color = black,
	mark = none
	]	coordinates {
		( 1, 39.47 )
		( 100, 229905.64 )
	};

	\addplot[
	color = black,
	mark = none
	]	coordinates {
		( 1, 143.54 )
		( 100, 28058448.95 )
	};

	\end{axis}
	\end{tikzpicture}
    }
  \end{subfigure}
  \begin{subfigure}{.5\linewidth}
    \resizebox{.9\textwidth}{!}{
	\begin{tikzpicture}
	\begin{axis}[
	width=10cm,
	log x ticks with fixed point,
	xmode=log,
	ymode=log,
	xmin = 4,
	xmax = 15,
	xtick={5,6,7,8,9,10,11,12},
	ymin = 10,
	ymax = 10^5,
	ymajorgrids=true,
	yminorgrids=true,
	legend style={at={(0.05,0.95)},anchor=north west, font=\fontsize{8}{5}\selectfont},
	xlabel={$n$},
	ylabel={$r$},
	every axis legend/.append style={nodes={right},
		every x tick label/.append style={font=\small},
		every y tick label/.append style={font=\small},
	}
	]

  \addlegendimage{empty legend}
  \addlegendentry[yshift=0pt]{\hspace{-.25cm}\textbf{Fourth order}}

	\addplot[only marks, red] coordinates {
		(5,2334)
		(6,3644)
		(7,5297)
		(8,7313)
		(9,9704)
		(10,12485)
		(11,15664)
		(12,19256)
	};
	\addlegendentry{Commutator bound}

	\addplot[only marks,red,mark=o,mark options={fill=white,line width=1.25pt}] coordinates {
		(5,48.15)
		(6,67.25)
		(7,81.80)
		(8,103.70)
		(9,123.60)
		(10,146.40)
		(11,165.40)
		(12,190.40)
	};
	\addlegendentry{Empirical bound}

	\addplot[
	color = black,
	mark = none
	]	coordinates {
		( 1, 47.25 )
		( 100, 3335517.10 )
	};

	\addplot[
	color = black,
	mark = none
	]	coordinates {
		( 1, 4.03 )
		( 100, 5186.16 )
	};

	\end{axis}
	\end{tikzpicture}
    }
  \end{subfigure}
  \begin{subfigure}{.5\linewidth}
    \resizebox{.9\textwidth}{!}{
	\begin{tikzpicture}
	\begin{axis}[
	width=10cm,
	log x ticks with fixed point,
	xtick={5,6,7,8,9,10,11,12},
	xmode=log,
	ymode=log,
	xmin = 4,
	xmax = 15,
	ymin = 10^0,
	ymax = 10^3,
	ymajorgrids=true,
	yminorgrids=true,
	legend style={at={(0.05,0.95)},anchor=north west, font=\fontsize{8}{5}\selectfont},
	xlabel={$n$},
	ylabel={$r$},
	every axis legend/.append style={nodes={right},
		every x tick label/.append style={font=\small},
		every y tick label/.append style={font=\small},
	}
	]

  \addlegendimage{empty legend}
  \addlegendentry[yshift=0pt]{\hspace{-.25cm}\textbf{Sixth order}}

	\addplot[only marks,red,mark=o,mark options={fill=white,line width=1.25pt}] coordinates {
		(5,13.40)
		(6,19.85)
		(7,23.25)
		(8,29.60)
		(9,32.70)
		(10,36.80)
		(11,40.40)
		(12,44)
	};
	\addlegendentry{Empirical bound}

	\addplot[
	color = black,
	mark = none
	]	coordinates {
		( 1, 1.79 )
		( 100, 750.02 )
	};

	\end{axis}
	\end{tikzpicture}
    }
  \end{subfigure}
  \begin{center}
  \begin{subfigure}{.5\linewidth}
  	\resizebox{.9\textwidth}{!}{
  		\begin{tikzpicture}
  		\begin{axis}[
  		width=10cm,
  		log x ticks with fixed point,
  		xtick={5,6,7,8,9,10,11,12},
  		xmode=log,
  		ymode=log,
  		xmin = 4,
  		xmax = 15,
  		ymin = 10^0,
  		ymax = 10^2,
  		ymajorgrids=true,
  		yminorgrids=true,
  		legend style={at={(0.05,0.95)},anchor=north west, font=\fontsize{8}{5}\selectfont},
  		xlabel={$n$},
  		ylabel={$r$},
  		every axis legend/.append style={nodes={right},
  			every x tick label/.append style={font=\small},
  			every y tick label/.append style={font=\small},
  		}
  		]

  		\addlegendimage{empty legend}
  		\addlegendentry[yshift=0pt]{\hspace{-.25cm}\textbf{Eighth order}}

  		\addplot[only marks,red,mark=o,mark options={fill=white,line width=1.25pt}] coordinates {
  			(5,6.80)
  			(6,9.20)
  			(7,10.80)
  			(8,12.80)
  			(9,14)
  			(10,15.20)
  			(11,17)
  			(12,20)
  		};
  		\addlegendentry{Empirical bound}

  		\addplot[
  		color = black,
  		mark = none
  		]	coordinates {
  			( 1, 1.14 )
  			( 100, 218.79 )
  		};

  		\end{axis}
  		\end{tikzpicture}
  	}
  \end{subfigure}
  \end{center}
\caption{Comparison of the values of $r$ using the commutator and empirical bounds for formulas of order $1$, $2$, and $4$, and values of $r$ for the empirical bound for formulas of order $6$ and $8$. Straight lines show power-law fits to the data.  The error bars for product formulas of order greater than $1$ are negligibly small, so we omit them from the plots.\label{fig:pf_empirical}}
\end{figure}

\comment{
In the previous discussions, we have reviewed the several existing
error bounds for product formulas. We also proposed new error bounds
that utilize the commutation relations between terms of the
Hamiltonian. However, as is pointed out in \cite{Pou15}, these bounds
are likely to become loose in practical applications. Below is a toy
example that illustrates this problem.

\begin{example}[Tightness problem]
  Suppose that we want to simulate $H=X+Z$ with time $t=1$ and
  accuracy $\epsilon_{\text{pred}}=10^{-3}$. If the first-order
  formula is used, the corresponding naive bound and commutator bound
  are
  \begin{equation*}
    \begin{aligned}
      \norm{e^{-i(X+Z)}-\big(e^{\frac{-i}{r}X}e^{\frac{-i}{r}Z}\big)^r}&\leq \frac{4e^{\frac{2}{r}}}{r},\\
      \norm{e^{-i(X+Z)}-\big(e^{\frac{-i}{r}X}e^{\frac{-i}{r}Z}\big)^r}&\leq
      \frac{1}{r}+\frac{8e^{\frac{2}{r}}}{3r^2}.
    \end{aligned}
  \end{equation*}
  To achieve the desired accuracy, the above bounds suggest choosing
  $r_{\text{naive}}=4002$ and $r_{\text{comm}}=1003$,
  respectively. However, a direct calculation of the error shows that
  it suffices to set $r_{\text{actual}}=699$.
\end{example}

In this subsection, we describe a possible solution to the above issue by establishing empirical bounds for product formulas.
} 

While the bounds discussed in \sec{pfanamin} and \sec{pfcom} provide rigorous correctness guarantees, they can be very loose.  To understand the minimum resources that suffice for product formula simulation, we estimate their empirical performance.  Of course, since quantum simulation is computationally challenging, we can only directly compute the actual simulation error for small instances.  Using binary search, we find the value of $r$ (the total number of segments) that just suffices to ensure error $10^{-3}$.  We extrapolate this behavior to produce a non-rigorous estimate of the performance of product formula simulation for instances of arbitrary size. We emphasize that the resulting \emph{empirical bound} does not come with a guarantee of correctness. Nevertheless, we believe it better captures the true performance of product formula simulations and indicates the extent to which our rigorous bounds are loose.

We numerically simulate the product formula algorithm for systems of size $5$ to $12$, determining the value of $r$ required to ensure error $10^{-3}$ as described above, and averaging over five random choices of the local field strengths $h_j$.  We fit these data to power laws, as depicted in \fig{pf_empirical}. We find
\begin{equation}
	r_1=2417 n^{1.964},~~
	r_2=39.47 n^{1.883},~~
	r_4=4.035 n^{1.555},~~
	r_6=1.789 n^{1.311},~~
	r_8=1.144 n^{1.141},
\end{equation}
where $r_i$ is the number of segments for the $i$th-order formula to produce a simulation that is accurate to within $10^{-3}$.  Considering the size of the circuit for each segment, this suggests an asymptotic complexity of roughly $9668 n^{2.964}$ for the first-order formula, $315.8 n^{2.883}$ for second order, $161.4 n^{2.555}$ for fourth order, $357.8 n^{2.311}$ for sixth order, and $1144 n^{2.141}$ for eighth order.

%% file: lcu.tex
In this appendix, we discuss some technical details related to the implementation of the Taylor series (\TS) algorithm, as introduced in \sec{alglcu}.  We present concrete error bounds in \sec{lcuerror} and discuss the failure probability of the approach in \sec{tsfailure}.  Then we describe how to implement $\select(V)$ gates, a major component of the algorithm: in \sec{controlencoding}, we explain how the control registers for these gates are encoded, and in \sec{selectV}, we give an optimized implementation of $\select(V)$ gates.

\subsection{Error analysis}
\label{sec:lcuerror}

In this section, we derive an explicit error bound that quantifies how the truncation order $K$ affects the overall accuracy. We then discuss how to apply this bound.

We begin by bounding the error introduced by truncating the Taylor series.

\begin{lemma}
	\label{lem:trunc}
	With the definitions of $\tilde{U}(t)$ in \eq{utwiddle} and $\tseg$ in \eq{tseg}, we have
	\begin{equation}
	\norm{\tilde{U}(t)-\exp(-iHt)}
  \leq 2\frac{(\ln 2)^{K+1}}{(K+1)!}
	\end{equation}
  for any $t \le \tseg$.
\end{lemma}

\begin{proof}
We have
\begin{align}
	\norm{\tilde{U}(t)-\exp(-iHt)}
	&=\norm*{\sum_{k=K+1}^{\infty}\frac{(-iHt)^k}{k!}} \\
	&\leq\sum_{k=K+1}^{\infty}\frac{(\norm{H}t)^k}{k!} \\
	&\leq\sum_{k=K+1}^{\infty}\frac{((\alpha_1+\cdots+\alpha_L)t)^k}{k!}\\
	&\leq\sum_{k=K+1}^{\infty}\frac{(\ln 2)^k}{k!}\\
  &\leq 2\frac{(\ln 2)^{K+1}}{(K+1)!},
\end{align}
where the final inequality follows from \lem{prodtail}.
\end{proof}

Next we consider the error induced by the isometry $\iso(t)$ defined in \eq{iso}.  It is straightforward to verify that
\begin{equation}
  (\langle 0|\otimes I) \iso(t)
  =\frac{3}{s}\tilde{U}(t)-\frac{4}{s^3}\tilde{U}(t)\tilde{U}(t)^{\dagger}\tilde{U}(t).
\end{equation}
As discussed in \sec{alglcu}, we rotate an ancilla qubit to increase the value of $s$ to be precisely $s=2$.  Then the following bound characterizes the error in the implementation of the Taylor series.

\begin{lemma}
	\label{lem:amp_rem}
	Suppose $\norm{\tilde{U}-U}\leq\delta$ for some unitary operator $U$. Then
	\begin{equation}
	\norm*{\frac{3}{2}\tilde{U}-\frac{1}{2}\tilde{U}\tilde{U}^{\dagger}\tilde{U}-U}
	\leq \frac{\delta^2+3\delta+4}{2}\delta.
	\end{equation}
\end{lemma}

\begin{proof}
	Since $\norm{\tilde{U}-U}\leq\delta$, we have
	\begin{align}
		\norm{\tilde{U}}&\leq\norm{\tilde{U}-U}+\norm{U}\leq 1+\delta
	\end{align}
	and therefore
	\begin{align}
		\norm{\tilde{U}\tilde{U}^{\dagger}-I}
    &\leq\norm{(\tilde{U}-U)\tilde{U}^{\dagger}}+\norm{U(\tilde{U}^{\dagger}-U^{\dagger})}
    \leq \delta(2+\delta).
	\end{align}
	Thus, by the triangle inequality, we have
	\begin{align}
		\norm*{\frac{3}{2}\tilde{U}-\frac{1}{2}\tilde{U}\tilde{U}^{\dagger}\tilde{U}-U}
		&\leq\norm{\tilde{U}-U}
		+\frac{1}{2}\norm{\tilde{U}-\tilde{U}\tilde{U}^{\dagger}\tilde{U}}\\
		&\leq\delta
		+\frac{\delta(2+\delta)(1+\delta)}{2}\\
		&=\frac{\delta^2+3\delta+4}{2}\delta
	\end{align}
	as claimed.
\end{proof}

We use the following basic property of contractions (operators of norm at most $1$), which is easily proved using the triangle inequality.

\begin{lemma}
	\label{lem:telescope}
	Suppose $U_i\colon\mathcal{H}\rightarrow\mathcal{H}$ and $V_i\colon\mathcal{H}\rightarrow\mathcal{H}$ are contractions for all $i \in \{1,\ldots,r\}$. If $\norm{U_i-V_i}\leq \eta$ for all $i$, then
	\begin{equation}
	\norm{U_r\cdots U_2 U_1-V_r\cdots V_2 V_1}\leq r\eta.
	\end{equation}
\end{lemma}

We also use the following lemma, which bounds the error introduced by normalization.

\begin{lemma}
\label{lem:normalization}
Suppose $\norm{|\phi\rangle}=1$, $\norm{|\psi\rangle}\leq 1$, and $\norm{|\psi\rangle-|\phi\rangle} \leq \xi < 1$. Then
\begin{equation}
	\norm*{\frac{|\psi\rangle}{\norm{|\psi\rangle}}-|\phi\rangle} \leq \sqrt{1+\xi}-\sqrt{1-\xi}.
\end{equation}
\end{lemma}

\begin{proof}
Decompose $\ket{\psi}$ as
\begin{equation}
	|\psi\rangle=\alpha|\phi\rangle+\beta|\phi^\bot\rangle
\end{equation}
for some normalized state $\ket{\phi^\perp}$ orthogonal to $\ket{\phi}$.
Clearly $|\alpha|^2+|\beta|^2\leq 1$ since $\norm{|\psi\rangle}\leq 1$. Furthermore, the assumption $\norm{|\psi\rangle-|\phi\rangle}\leq \xi$ implies
\begin{equation}
	  |\alpha-1|^2+|\beta|^2\leq\xi^2,
\end{equation}
so
\begin{equation}
  \label{eq:lcu_condition}
	  |\alpha|^2+|\beta|^2\leq\xi^2+2\Re(\alpha)-1
\end{equation}
with $1-\xi \leq \Re(\alpha)\leq 1+\xi$.

Then we have
\begin{align}
	\norm*{\frac{|\psi\rangle}{\norm{|\psi\rangle}}-|\phi\rangle}
	&=\sqrt{2-\frac{2\Re(\alpha)}{\sqrt{|\alpha|^2+|\beta|^2}}}\\
	&\leq\sqrt{2-\frac{2\Re(\alpha)}{\sqrt{\xi^2+2\Re(\alpha)-1}}}\\
	&\leq\sqrt{2-\frac{2(1-\xi^2)}{\sqrt{\xi^2+2(1-\xi^2)-1}}}\\
  &=\sqrt{1+\xi}-\sqrt{1-\xi}.
\end{align}
Here the first inequality uses (\ref{eq:lcu_condition}) and the fact that $\Re(\alpha) \ge 1-\xi \ge 0$, and the second inequality follows since the function $x/\sqrt{2x-1+\xi^2}$ achieves its minimum at $x=1-\xi^2$ within the interval $1-\xi\leq x\leq 1+\xi$.
\end{proof}

With all the above lemmas in hand, we are ready to prove an explicit error bound for the \TS\ algorithm.

\begin{theorem}
	\label{thm:lcu}
The \TS\ algorithm achieves error at most $\sqrt{1+\xi}-\sqrt{1-\xi}$ with success probability at least $(1-\xi)^2$, where
	\begin{align}
  \label{eq:lcu_parameter}
	\xi = r \frac{\delta^2+3\delta+4}{2}\delta \quad \text{with} \quad
	\delta=2\frac{(\ln2)^{K+1}}{(K+1)!}.
	\end{align}
\end{theorem}

\begin{proof}
For $t \in \{\tseg,\trem\}$, \lem{trunc} shows that
\begin{align}
	\norm{\tilde{U}(t)-\exp(-iHt)}&\leq\delta,
\end{align}
and \lem{amp_rem} shows that
\begin{align}
	\norm{(\langle 0|\otimes I) \iso(t)- \exp(-iHt)}&\leq\xi/r.
\end{align}
Since $\iso(t)$ is an isometry, $(\langle 0|\otimes I) \iso(t)$ is a contraction, so by \lem{telescope}, we have
\begin{equation}
	\norm{(\langle 0|\otimes I) \iso(\trem) \bigl((\langle 0|\otimes I) \iso(\tseg)\bigr)^{r-1}
		-\exp(-itH)}\leq\xi.
\end{equation}
The claim about the success probability follows by applying the triangle inequality, and the accuracy can be established by invoking \lem{normalization}.
\end{proof}

To apply this bound, we must determine the truncation order $K$ that achieves the desired error bound $\epsilon$. Just as for the product formula error bounds presented in \app{pf}, it does not seem possible to compute $K$ in closed form. However, since $K$ can only take integer values, it is straightforward to tabulate the error estimates corresponding to all potentially relevant values of $K$, as shown in \tab{K}. Using the known value of $r$, we can then determine which value of $K$ suffices to ensure small error.

\begin{table}
	\centering
	\begin{tabular}{|c|c|}
		\hline
		$K$ & $\xi/r$ \\ \hline\hline
		1 & 1.3626 \\ \hline
		2 & 0.24118 \\ \hline
		3 & 0.039031 \\ \hline
		4 & 0.0053441 \\ \hline
		5 & 0.00061628 \\ \hline
		6 & $6.10123 \times 10^{-5}$ \\ \hline
		7 & $5.28621 \times 10^{-6}$ \\ \hline
		8 & $4.07124 \times 10^{-7}$ \\ \hline
		9 & $2.821965 \times 10^{-8}$ \\ \hline
		10 & $1.778215 \times 10^{-9}$ \\ \hline
	\end{tabular}
	\caption{Lookup table for the truncation order $K$, with $s$ boosted to be $2$ in each segment.}\label{tab:K}
\end{table}

In \sec{empiricalbounds}, we presented empirical error bounds for simulations based on product formulas. It would be natural to perform a similar analysis of the error in the \TS\ algorithm. Unfortunately, it is intractable to find an empirical bound by direct simulation since the number of ancilla qubits used by the \TS\ algorithm puts it beyond the reach of classical simulation even for very small systems (see \fig{qubitcounts}). A more limited alternative would be to use empirical data to improve the estimated error of the truncated Taylor series. However, since
\begin{equation}
\frac{(\ln 2)^{K+1}}{(K+1)!}\leq\sum_{k=K+1}^{\infty}\frac{(\ln 2)^k}{k!}\leq 2\frac{(\ln 2)^{K+1}}{(K+1)!},
\end{equation}
the estimated error can be improved by a factor of at most $2$, which results in an additive offset of at most $\ln 2$ for the truncation order $K$. Thus we do not consider such a bound in our analysis.

\subsection{Failure probability}
\label{sec:tsfailure}

Note that unlike the \PF\ approach, the \TS\ algorithm does not always succeed. When a measurement of the ancilla register indicates a failure, we must restart the simulation. To make a fair comparison between the \PF\ and \TS\ algorithms, we should take the success probability into account.

Fortunately, it turns out that the resulting overhead is almost negligible. With $\epsilon=10^{-3}$, \thm{lcu} gives
\begin{equation}
\xi\leq\sqrt{\epsilon^2-\frac{\epsilon^4}{4}}<0.001,
\end{equation}
so the success probability of the algorithm is at least
\begin{equation}
(1-\xi)^2\geq 0.998.
\end{equation}
The expected number of times we must repeat the algorithm before succeeding is approximately $1/0.998\approx 1.002$.  Since this factor is very close to $1$, we simply neglect it, and we find it reasonable to directly compare gate counts for the \PF\ algorithm (which has no probability of failure) and the \TS\ approach.

%% file: selectV.tex
\newcommand{\lft}{\textsc{left}}
\newcommand{\rgt}{\textsc{right}}

\subsection{Encoding of the control register}
\label{sec:controlencoding}

A crucial step in the implementation of the \TS\ algorithm (and in the \QSP\ algorithm) is to synthesize the $\select(V)$ gates.  The cost of this implementation depends strongly on the chosen representation for the control register.

Recall from \eq{defnofselectV} that the $\select(V)$ operation has the form
\begin{equation}
\select(V):=\sum_{j=0}^{\numselect-1}|j\rangle\langle j|\otimes V_j.
\label{eq:recallselectv}
\end{equation}
For the \TS\ algorithm, the operators $V_j$ are defined
via \eq{tslcu}.  Here the index $j$ labels a value $k \in \{0,1,\ldots,K\}$ and indices $\ell_1,\ldots,\ell_k\in\{1,\ldots,L\}$. Perhaps the most straightforward approach is to represent the entire control register with a binary encoding using $\log_2(K+1) + K \log_2 L$ bits. However, as pointed out in \cite{BCCKS14}, we can significantly reduce the gate complexity by choosing a different encoding of the control register.

Specifically, we use a unary encoding to label $k$ and a binary encoding for each $\ell_1,\ldots,\ell_k$.  With such an encoding, the instance of $\select(V)$ in the \TS\ algorithm can be represented as the map
\begin{equation}\label{eq:105}
|1^{k}0^{K-k}\rangle|\ell_1,\ldots,\ell_K\rangle|\psi\rangle\mapsto
|1^{k}0^{K-k}\rangle|\ell_1,\ldots,\ell_K\rangle(-i)^{k}H_{\ell_1}\cdots H_{\ell_k}|\psi\rangle.
\end{equation}
We implement this transformation as follows. Conditioned on the $j$th qubit of the unary encoding of $k$ being $1$, and the $j$th coordinate of $\ell_1,\ldots,\ell_K$ being the binary encoding of $\ell_j$, we apply $(-i)H_{\ell_j}$. Compared to an entirely binary encoding, this approach only requires an additional $\lceil K+1-\log_2(K+1) \rceil$ qubits, which is a modest increase since $K$ is typically small (see \tab{K}).  In return, instead of selecting on a large register of $\Theta(K \log L)$ bits, we can perform $K+1$ independent selections on registers of $\log_2 L$ bits, each controlled by a single qubit.

\subsection{Implementation of \texorpdfstring{$\select(V)$}{select(V)}}
\label{sec:selectV}

In this section, we present a circuit for the $\select(V)$ operation $\sum_{\gamma=1}^\numselect \ket{\gamma}\bra{\gamma} \otimes V_\gamma$. We assume that the control register $\gamma \in \{1,\ldots,\numselect\}$ is stored in binary using $\lenstring := \ceil{\log_2{\numselect}}$ bits.  As discussed in \sec{controlencoding}, this procedure is a basic subroutine in our implementation of the $\select(V)$ operation appearing in the \TS\ algorithm, using a mixed unary/binary encoding.  We also use this $\select(V)$ construction in our implementation of the \QSP\ algorithm.

Our goal is to apply a unitary operation $V_\gamma$ conditioned on the control register being in the state $\ket{\gamma}$. To do this, we cycle the value of a designated ancilla qubit through $\numselect$ Boolean products of $\lenstring$ literals, where in each of the products, each of the variables $x_1,\ldots,x_\lenstring$ appears exactly once (either negated or not).  We apply $V_\gamma$ conditioned on the ancilla qubit at the $\gamma$th step of this construction.  To simplify the presentation, we describe how to cycle through all $2^\lenstring$ Boolean products. This procedure can be terminated after producing only $\numselect$ Boolean products to avoid unnecessary cost.

A straightforward way of obtaining the Boolean products is to implement them via multiply-controlled Toffoli gates with appropriate control negations.  Implementations of the multiply-controlled Toffoli gates have been well-studied.  In particular, using the approach of \cite{M16}, and assuming that the negative controls introduce no additional cost,\footnote{While a construction with this property was not explicitly developed in \cite{M16}, it appears that any multiply-controlled Toffoli gate with arbitrary positive and negative controls can be implemented with the same number of gates as the multiply-controlled Toffoli gate of the same size with only positive controls.} we can cycle through all Boolean products using the following resources:
\begin{align*}
  \begin{array}{rl}\label{arr:upb}
4\lenstring2^\lenstring-6\cdot 2^\lenstring & \text{Hadamard gates} \\
8\lenstring2^\lenstring-9\cdot 2^\lenstring & \text{$T$/$T^\dagger$ gates} \\
6\lenstring2^\lenstring-6\cdot 2^\lenstring & \text{$\CNOT$ gates} \\
\ceil{\frac{\lenstring-2}{2}} & \text{ancilla qubits, initialized in and returned to the state $\ket{0}$}.
  \end{array}
\end{align*}

We now present an optimized circuit to cycle through these Boolean products.  In particular, we do this using a reversible circuit consisting of $\NOT$, $\CNOT$, and Toffoli gates (we call a circuit over this gate set an \emph{NCT circuit}).  Of the three basic gates used by the NCT circuits, the Toffoli gate is the most expensive---its best known implementation uses $6$ $\CNOT$ gates, $2$ Hadamard gates, and $7$ $T$/$T^\dagger$ gates \cite{ar:ammr}---so we first focus on optimizing the Toffoli gate count.  We use Toffoli gates with controls of arbitrary polarities.  All such Toffoli gates can be implemented by a Clifford+$T$ circuit using the same resources \cite{ar:ammr}.

To cycle though $2^\lenstring$ Boolean products, clearly we must use at least $2^\lenstring$ gates (note that we do not allow gates with multiple targets).
Indeed, most of these must be Toffoli gates, the most expensive gates in an NCT circuit.

\begin{lemma}[{\cite[Lemma 4]{ar:m}}]
\label{lem:sV1}
An NCT circuit for $\select(V)$ control generation requires at least $2^\lenstring-\lenstring-1$ Toffoli gates.
\end{lemma}

Considering the cost of a single Toffoli gate, a Clifford+$T$ circuit with $2^\lenstring$ Toffoli gates could contain as many as $6\cdot 2^\lenstring$ $\CNOT$ gates and $7\cdot 2^\lenstring$ $T$/$T^\dagger$ gates.  We next construct a quantum circuit that cycles through all $2^\lenstring$ Boolean products, and show that after several optimizations, the gate counts are close to the above optimistic figure, and significantly better than those obtained using multiply-controlled Toffoli gates.  We begin by describing an efficient NCT circuit that cycles through $2^\lenstring$ Boolean products.

\begin{lemma}\label{lem:sV2}
The $\select(V)$ control generation can be implemented by an NCT circuit using at most $2$ $\NOT$ gates, $1.5\cdot 2^\lenstring - 2$ $\CNOT$ gates, and $1.5\cdot 2^\lenstring - 4$ Toffoli gates, with $\lenstring-1$ ancilla qubits initialized in and returned to the state $\ket{0}$.
\end{lemma}

\begin{proof}
\begin{figure}[t]
\centering
\resizebox{.9\textwidth}{!}{
\begin{tikzpicture}[every tree node/.style={inner sep=0pt, minimum size=0pt},
  edge from parent/.append style={very thick},
  level distance=1.5cm,
  level 1/.style={sibling distance=1cm},
  level 2/.style={sibling distance=1cm},
  level 3/.style={sibling distance=1cm},
  level 4/.style={sibling distance=1cm}]
\Tree [.\node(t){};
  [.\node(t0){};
    [.\node(t00){};
      [.\node(t000){};
        [.\node(t0000){};]
        [.\node(t0001){};]
      ]
      [.\node(t001){};
        [.\node(t0010){};]
        [.\node(t0011){};]
      ]
    ]
    [.\node(t01){};
      [.\node(t010){};
        [.\node(t0100){};]
        [.\node(t0101){};]
      ]
      [.\node(t011){};
        [.\node(t0110){};]
        [.\node(t0111){};]
      ]
    ]
  ]
  [.\node(t1){};
    [.\node(t10){};
      [.\node(t100){};
        [.\node(t1000){};]
        [.\node(t1001){};]
      ]
      [.\node(t101){};
        [.\node(t1010){};]
        [.\node(t1011){};]
      ]
    ]
    [.\node(t11){};
      [.\node(t110){};
        [.\node(t1100){};]
        [.\node(t1101){};]
      ]
      [.\node(t111){};
        [.\node(t1110){};]
        [.\node(t1111){};]
      ]
    ]
  ]
]
\tikzset{every path/.style={thick,gray,bend left=15}}
\draw (t) to (t0);
\draw (t) to (t1);
\draw (t0) to (t);
\draw (t1) to (t);
\tikzset{every path/.style={thick,blue,bend left=15}}
\draw (t0) to (t00);
\draw (t0) to (t01);
\draw (t00) to (t0);
\draw (t01) to (t0);
\draw (t1) to (t10);
\draw (t1) to (t11);
\draw (t10) to (t1);
\draw (t11) to (t1);
\draw (t00) to (t000);
\draw (t00) to (t001);
\draw (t000) to (t00);
\draw (t001) to (t00);
\draw (t01) to (t010);
\draw (t01) to (t011);
\draw (t010) to (t01);
\draw (t011) to (t01);
\draw (t10) to (t100);
\draw (t10) to (t101);
\draw (t100) to (t10);
\draw (t101) to (t10);
\draw (t11) to (t110);
\draw (t11) to (t111);
\draw (t110) to (t11);
\draw (t111) to (t11);
\draw (t000) to (t0000);
\draw (t000) to (t0001);
\draw (t0000) to (t000);
\draw (t0001) to (t000);
\draw (t001) to (t0010);
\draw (t001) to (t0011);
\draw (t0010) to (t001);
\draw (t0011) to (t001);
\draw (t010) to (t0100);
\draw (t010) to (t0101);
\draw (t0100) to (t010);
\draw (t0101) to (t010);
\draw (t011) to (t0110);
\draw (t011) to (t0111);
\draw (t0110) to (t011);
\draw (t0111) to (t011);
\draw (t100) to (t1000);
\draw (t100) to (t1001);
\draw (t1000) to (t100);
\draw (t1001) to (t100);
\draw (t101) to (t1010);
\draw (t101) to (t1011);
\draw (t1010) to (t101);
\draw (t1011) to (t101);
\draw (t110) to (t1100);
\draw (t110) to (t1101);
\draw (t1100) to (t110);
\draw (t1101) to (t110);
\draw (t111) to (t1110);
\draw (t111) to (t1111);
\draw (t1110) to (t111);
\draw (t1111) to (t111);
\tikzset{every path/.style={black,dotted}}
\draw let \p{t0000}=(t0000), \p{t1111}=(t1111), \p{t}=(t), \p{t0}=(t0), \p{t00}=(t00), \p{t000}=(t000) in
  ($ (\x{t0000},\y{t}) - (1,0) $) coordinate (l0) -- ($ (\x{t1111},\y{t}) + (1,0) $)
  ($ (\x{t0000},\y{t0}) - (1,0) $) coordinate (l1) -- ($ (\x{t1111},\y{t0}) + (1,0) $)
  ($ (\x{t0000},\y{t00}) - (1,0) $) coordinate (l2) -- ($ (\x{t1111},\y{t00}) + (1,0) $)
  ($ (\x{t0000},\y{t000}) - (1,0) $) coordinate (l3) -- ($ (\x{t1111},\y{t000}) + (1,0) $)
  ($ (\x{t0000},\y{t0000}) - (1,0) $) coordinate (l4) -- ($ (\x{t1111},\y{t0000}) + (1,0) $);
\tikzset{every node/.style={black}}
\node[left=0 cm of l0] {$q_0$};
\node[left=0 cm of l1] {$q_1$};
\node[left=0 cm of l2] {$q_2$};
\node[left=0 cm of l3] {$q_3$};
\node[left=0 cm of l4] {$q_4$};
\node[above=0 cm of t] {\huge *};
\tikzset{every node/.style={black,rotate=-90}}
\node[below=0 cm of t0000, anchor=west] {$x_1 x_2 x_3 x_4$};
\node[below=0 cm of t0001, anchor=west] {$x_1 x_2 x_3 \bar x_4$};
\node[below=0 cm of t0010, anchor=west] {$x_1 x_2 \bar x_3 x_4$};
\node[below=0 cm of t0011, anchor=west] {$x_1 x_2 \bar x_3 \bar x_4$};
\node[below=0 cm of t0100, anchor=west] {$x_1 \bar x_2 x_3 x_4$};
\node[below=0 cm of t0101, anchor=west] {$x_1 \bar x_2 x_3 \bar x_4$};
\node[below=0 cm of t0110, anchor=west] {$x_1 \bar x_2 \bar x_3 x_4$};
\node[below=0 cm of t0111, anchor=west] {$x_1 \bar x_2 \bar x_3 \bar x_4$};
\node[below=0 cm of t1000, anchor=west] {$\bar x_1 x_2 x_3 x_4$};
\node[below=0 cm of t1001, anchor=west] {$\bar x_1 x_2 x_3 \bar x_4$};
\node[below=0 cm of t1010, anchor=west] {$\bar x_1 x_2 \bar x_3 x_4$};
\node[below=0 cm of t1011, anchor=west] {$\bar x_1 x_2 \bar x_3 \bar x_4$};
\node[below=0 cm of t1100, anchor=west] {$\bar x_1 \bar x_2 x_3 x_4$};
\node[below=0 cm of t1101, anchor=west] {$\bar x_1 \bar x_2 x_3 \bar x_4$};
\node[below=0 cm of t1110, anchor=west] {$\bar x_1 \bar x_2 \bar x_3 x_4$};
\node[below=0 cm of t1111, anchor=west] {$\bar x_1 \bar x_2 \bar x_3 \bar x_4$};
\end{tikzpicture}
}
\caption{Binary tree encoding the circuit that implements the control generation for the $\select(V)$ operation.}
\label{fig:circtree1}
\end{figure}
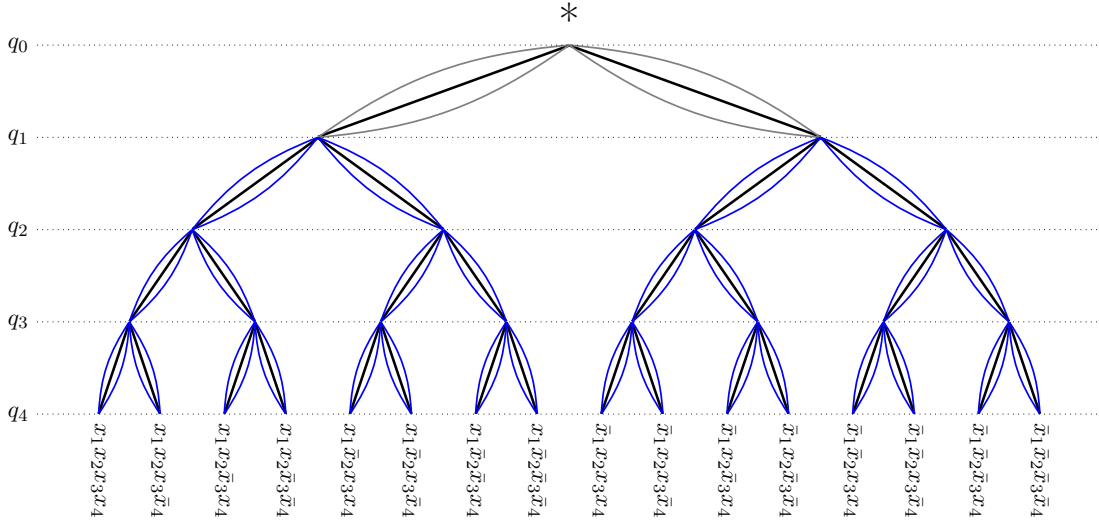

Consider a depth-$\lenstring$ complete binary tree, where vertices are labeled by the Boolean products of the $\lenstring$ variables $x_1,\ldots,x_\lenstring$.  A vertex at depth $d$ is labeled by a 
product of the first $d$ variables, $x_1^{a_1}\ldots x_d^{a_d}$, where $x_i^{a_i}:=x_i$ if $a_i=\lft$ and $x_i^{a_i}:=\overline{x_i}$ if $a_i=\rgt$.  The label $a_d$ indicates how we arrive at $x_1^{a_1}\ldots x_d^{a_d}$ from its parent $x_1^{a_1}\ldots x_{d-1}^{a_{d-1}}$, by taking the left edge if $a_d=\lft$ or the right edge if $a_d=\rgt$.  The product $x_1^{a_1}\ldots x_d^{a_d}$ thus indicates the path to the vertex from the root, as shown in \fig{circtree1}.  Leaves of the tree are labeled by
the Boolean products over all $\lenstring$ variables $x_1,\ldots, x_\lenstring$, in all possible polarities.

In the quantum circuit implementing the control generation for the $\select(V)$ transformation, all vertices at a given depth are represented by a single qubit, labeled $q_d$ for $d \in \{0,1,\ldots,\lenstring\}$.  The qubit $q_k$ carries the product of the first $k$ variables with appropriate polarities that evolve as the computation proceeds. We discuss how to implement the circuit using $\lenstring+1$ ancilla qubits $q_0,q_1,\ldots,q_\lenstring$, although it is straightforward to reduce the number of ancillas to $\lenstring-1$. The root of the tree corresponds to the Boolean constant $1$, so it can be viewed as a virtual qubit.  The two depth-$1$ vertices represent the variable $x_1$ and its negation, values that can be obtained in-place by negating the variable $x_1$ as needed.  Thus it suffices to use $\lenstring-1$ ancilla qubits $q_2,q_3, \ldots, q_\lenstring$.

We compute all desired Boolean products by traversing the tree using the left-hand rule, i.e., always taking the leftmost edge going down if there are vertices not yet explored, and going up if all edges and vertices below were already explored.  The path starts and ends in the root (marked with an asterisk in \fig{circtree1}), and traverses each edge twice, once going down and once coming up, so it has length $4\cdot 2^\lenstring-4$.

When traversing the edge $(x_1^{a_1}x_2^{a_2} \ldots x_{d-1}^{a_{d-1}},x_1^{a_1}x_2^{a_2} \ldots x_{d-1}^{a_{d-1}}x_d^{a_d})$ in either direction, we apply the gate $\TOF(q_{d-1},x_d^{a_d};q_d)$, where $\TOF(c_1,c_2;t)$ denotes a Toffoli gate with controls $c_1$ and $c_2$ and target $t$.  These edges in the walk are colored blue in \fig{circtree1}.  No gate needs be applied for edges connecting to the root, as the variable $x_1$ and its negation can be obtained without using a Toffoli or a $\CNOT$ gate.  Those edges in the walk are colored grey in \fig{circtree1}.  Overall, this circuit uses $4\cdot 2^\lenstring-8$ Toffoli gates and two $\NOT$ gates (one to negate $x_1$ and the other to uncompute this negation). During the walk, each time we reach a leaf, we create a new Boolean product of $\lenstring$ variables with polarities determined by the path to the root.  By visiting all leaves, we create all $2^\lenstring$ desired products.  Completing the path and returning to the root ensures that all ancilla qubits are reset to $0$.  Note that we can return to the root from any given leaf without visiting all leaves, finishing the walk early if only some of the Boolean products are needed (and thereby using fewer gates).

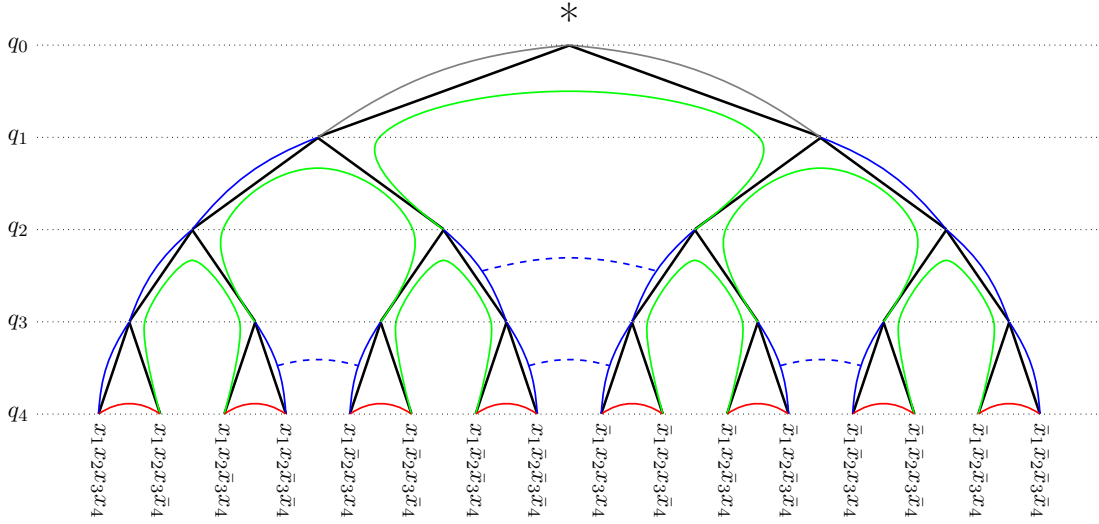
\begin{figure}[t]
\centering
\resizebox{.9\textwidth}{!}{
\begin{tikzpicture}[every tree node/.style={inner sep=0pt, minimum size=0pt},
  edge from parent/.append style={very thick},
  level distance=1.5cm,
  level 1/.style={sibling distance=1cm},
  level 2/.style={sibling distance=1cm},
  level 3/.style={sibling distance=1cm},
  level 4/.style={sibling distance=1cm}]
\Tree [.\node(t){};
  [.\node(t0){};
    [.\node(t00){};
      [.\node(t000){};
        [.\node(t0000){};]
        [.\node(t0001){};]
      ]
      [.\node(t001){};
        [.\node(t0010){};]
        [.\node(t0011){};]
      ]
    ]
    [.\node(t01){};
      [.\node(t010){};
        [.\node(t0100){};]
        [.\node(t0101){};]
      ]
      [.\node(t011){};
        [.\node(t0110){};]
        [.\node(t0111){};]
      ]
    ]
  ]
  [.\node(t1){};
    [.\node(t10){};
      [.\node(t100){};
        [.\node(t1000){};]
        [.\node(t1001){};]
      ]
      [.\node(t101){};
        [.\node(t1010){};]
        [.\node(t1011){};]
      ]
    ]
    [.\node(t11){};
      [.\node(t110){};
        [.\node(t1100){};]
        [.\node(t1101){};]
      ]
      [.\node(t111){};
        [.\node(t1110){};]
        [.\node(t1111){};]
      ]
    ]
  ]
]
\tikzset{every path/.style={thick,gray,bend left=15}}
\draw (t0) to (t);
\draw (t) to (t1);
\tikzset{every path/.style={thick,blue,bend left=15}}
\draw (t0000) to (t000);
\draw (t000) to (t00);
\draw (t00) to (t0);
\draw (t1) to (t11);
\draw (t11) to (t111);
\draw (t111) to (t1111);
\draw (t001) to coordinate[midway](m001) (t0011);
\draw (t0100) to coordinate[midway](m010) (t010);
\draw (t101) to coordinate[midway](m101) (t1011);
\draw (t1100) to coordinate[midway](m110) (t110);
\draw (t01) to coordinate[midway](m01) (t011);
\draw (t011) to coordinate[midway](m011) (t0111);
\draw (t1000) to coordinate[midway](m100) (t100);
\draw (t100) to coordinate[midway](m10) (t10);
\tikzset{every path/.style={thick,blue,dashed,bend left=15}}
\draw (m001) to (m010);
\draw (m011) to (m100);
\draw (m101) to (m110);
\draw (m01) to (m10);
\tikzset{every path/.style={thick,green}}
\draw plot [smooth, tension=.5] coordinates {(t0001) ($ (t000) + (.25,0) $) ($ (t00) - (0,.5) $) ($ (t001) - (.25,0) $) (t0010)};
\draw plot [smooth, tension=.5] coordinates {(t0101) ($ (t010) + (.25,0) $) ($ (t01) - (0,.5) $) ($ (t011) - (.25,0) $) (t0110)};
\draw plot [smooth, tension=.5] coordinates {(t1001) ($ (t100) + (.25,0) $) ($ (t10) - (0,.5) $) ($ (t101) - (.25,0) $) (t1010)};
\draw plot [smooth, tension=.5] coordinates {(t1101) ($ (t110) + (.25,0) $) ($ (t11) - (0,.5) $) ($ (t111) - (.25,0) $) (t1110)};
\draw plot [smooth, tension=.75] coordinates {(t001) ($ (t00) + (.5,0) $) ($ (t0) - (0,.5) $) ($ (t01) - (.5,0) $) (t010)};
\draw plot [smooth, tension=.75] coordinates {(t101) ($ (t10) + (.5,0) $) ($ (t1) - (0,.5) $) ($ (t11) - (.5,0) $) (t110)};
\draw plot [smooth, tension=.75] coordinates {(t01) ($ (t0) + (1,0) $) ($ (t) - (0,.75) $) ($ (t1) - (1,0) $) (t10)};
\tikzset{every path/.style={thick,red,bend left=35}}
\draw (t0000) to (t0001);
\draw (t0010) to (t0011);
\draw (t0100) to (t0101);
\draw (t0110) to (t0111);
\draw (t1000) to (t1001);
\draw (t1010) to (t1011);
\draw (t1100) to (t1101);
\draw (t1110) to (t1111);
\tikzset{every path/.style={black,dotted}}
\draw let \p{t0000}=(t0000), \p{t1111}=(t1111), \p{t}=(t), \p{t0}=(t0), \p{t00}=(t00), \p{t000}=(t000) in
  ($ (\x{t0000},\y{t}) - (1,0) $) coordinate (l0) -- ($ (\x{t1111},\y{t}) + (1,0) $)
  ($ (\x{t0000},\y{t0}) - (1,0) $) coordinate (l1) -- ($ (\x{t1111},\y{t0}) + (1,0) $)
  ($ (\x{t0000},\y{t00}) - (1,0) $) coordinate (l2) -- ($ (\x{t1111},\y{t00}) + (1,0) $)
  ($ (\x{t0000},\y{t000}) - (1,0) $) coordinate (l3) -- ($ (\x{t1111},\y{t000}) + (1,0) $)
  ($ (\x{t0000},\y{t0000}) - (1,0) $) coordinate (l4) -- ($ (\x{t1111},\y{t0000}) + (1,0) $);
\tikzset{every node/.style={black}}
\node[left=0 cm of l0] {$q_0$};
\node[left=0 cm of l1] {$q_1$};
\node[left=0 cm of l2] {$q_2$};
\node[left=0 cm of l3] {$q_3$};
\node[left=0 cm of l4] {$q_4$};
\node[above=0 cm of t] {\huge *};
\tikzset{every node/.style={black,rotate=-90}}
\node[below=0 cm of t0000, anchor=west] {$x_1 x_2 x_3 x_4$};
\node[below=0 cm of t0001, anchor=west] {$x_1 x_2 x_3 \bar x_4$};
\node[below=0 cm of t0010, anchor=west] {$x_1 x_2 \bar x_3 x_4$};
\node[below=0 cm of t0011, anchor=west] {$x_1 x_2 \bar x_3 \bar x_4$};
\node[below=0 cm of t0100, anchor=west] {$x_1 \bar x_2 x_3 x_4$};
\node[below=0 cm of t0101, anchor=west] {$x_1 \bar x_2 x_3 \bar x_4$};
\node[below=0 cm of t0110, anchor=west] {$x_1 \bar x_2 \bar x_3 x_4$};
\node[below=0 cm of t0111, anchor=west] {$x_1 \bar x_2 \bar x_3 \bar x_4$};
\node[below=0 cm of t1000, anchor=west] {$\bar x_1 x_2 x_3 x_4$};
\node[below=0 cm of t1001, anchor=west] {$\bar x_1 x_2 x_3 \bar x_4$};
\node[below=0 cm of t1010, anchor=west] {$\bar x_1 x_2 \bar x_3 x_4$};
\node[below=0 cm of t1011, anchor=west] {$\bar x_1 x_2 \bar x_3 \bar x_4$};
\node[below=0 cm of t1100, anchor=west] {$\bar x_1 \bar x_2 x_3 x_4$};
\node[below=0 cm of t1101, anchor=west] {$\bar x_1 \bar x_2 x_3 \bar x_4$};
\node[below=0 cm of t1110, anchor=west] {$\bar x_1 \bar x_2 \bar x_3 x_4$};
\node[below=0 cm of t1111, anchor=west] {$\bar x_1 \bar x_2 \bar x_3 \bar x_4$};
\end{tikzpicture}
}
\caption{Optimized circuit walk on the binary tree, using grey, blue, red, and green solid edges. Dashed edges indicate reductions in the quantum gate count.}
\label{fig:circtree2}
\end{figure}

Next we show how to reduce the Toffoli gate count.  Observe that the above path makes $2^\lenstring-2$ moves of the form $v_1 \rightarrow v_2 \rightarrow v_3$, where $v_1$ and $v_3$ are the children of $v_2$, which is not the root (for example, consider any path indicated by a red arc in \fig{circtree2}).  Such a path requires applying the gates $\TOF(q_{d-1},x_d^{\lft};q_d)$ and $\TOF(q_{d-1},x_d^{\rgt};q_d)$. However, by the circuit identity
\begin{equation}
\TOF(q_{d-1},x_d^{\lft};q_d) \TOF(q_{d-1},x_d^{\rgt};q_d) = \CNOT(q_{d-1};q_d),
\end{equation}
we can replace two Toffoli gates with one $\CNOT$ gate.  Since $v_2$ is never a leaf, we do not need to explicitly visit it, thereby saving resources.  This modification reduces the Toffoli gate count by $2(2^\lenstring-2)=2 \cdot 2^\lenstring-4$, giving a reduced Toffoli gate count of $2\cdot 2^\lenstring-4$.  However, there are now $2^\lenstring-2$ $\CNOT$s.  The circuit still uses $2$ $\NOT$ gates.

Next, consider a path  $v_1 \rightarrow v_2 \rightarrow v_3 \rightarrow v_4 \rightarrow v_5$, where $v_2$ and $v_4$ are the children of $v_3$, $v_1$ is the right child of $v_2$, and $v_5$ is the left child of $v_4$. Such paths are indicated by green arcs in \fig{circtree2}.  Taking such a path involves applying the circuit
\begin{equation}
\TOF(q_{d-1},\overline{x_d};q_d) \CNOT(q_{d-2};q_{d-1}) \TOF(q_{d-1},x_d;q_d).
\end{equation}
This gate sequence can be simplified to
\begin{equation}
\CNOT(q_{d-1};q_d) \TOF(q_{d-2},x_d;q_d) \CNOT(q_{d-2};q_{d-1}).
\end{equation}
This additional optimization eliminates $\frac{1}{2}\cdot 2^\lenstring-1$ Toffoli gates and introduces $\frac{1}{2}\cdot 2^\lenstring-1$ $\CNOT$s, giving a reduced Toffoli gate count of $2\cdot 2^\lenstring-4-\frac{1}{2}\cdot 2^\lenstring+1 = 1.5 \cdot 2^\lenstring - 3$.

We conclude by carefully counting the resources, while reducing them a little further. Naively, every move along an edge of the tree involves applying a Toffoli gate. As detailed above, each red shortcut lets us traverse two edges using only one $\CNOT$ gate and each green shortcut lets us traverse four edges using one Toffoli and two $\CNOT$s. Overall, the number of Toffoli gates, per the colored tree in \fig{circtree2}, is
\begin{align}
&4\cdot 2^\lenstring-4 & &\text{(length of the entire unreduced path)} \nonumber\\
&\quad- 2^\lenstring & &\text{(red shortcuts reduction)} \nonumber\\
&\quad- 1.5 \cdot 2^\lenstring+3 & &\text{(green shortcuts reduction)} \nonumber\\
&\quad- 2 & &\text{(grey edges connect to a virtual node and require no Toffoli gates)} \nonumber\\
&\quad- 1 & &\text{(we can compute $\bar x_1 x_2$ from $x_1 \bar x_2$ using two $\CNOT$s and no Toffoli gates)} \nonumber\\
&= 1.5\cdot 2^\lenstring - 4.
\end{align}
The $\CNOT$ gate count is
\begin{align}
&0.5\cdot 2^\lenstring && \text{(red)} \nonumber\\
&\quad + 2\cdot (2^{\lenstring-1}-1) && \text{(green)} \nonumber\\
&\quad + 0 && \text{(across blue and grey paths)} \nonumber\\
&= 1.5\cdot 2^\lenstring - 2.
\end{align}
The $\NOT$ gate count is $2$.
\end{proof}

By substituting Toffoli gates with their Clifford+$T$ implementations, the desired control generation circuit may be implemented using $2$ $\NOT$ gates, $3\cdot 2^\lenstring - 8$ Hadamard gates, $10.5\cdot 2^\lenstring - 28$ $T$ gates, and $10.5\cdot 2^\lenstring - 26$ $\CNOT$ gates. In the remainder of this section, we reduce this gate count. To accomplish this, we pair Toffoli gates as indicated in \fig{circtree2} with dashed blue lines.  Formally, we pair every $\TOF(q_{d-1},\overline{x_d};q_d)$ with the $\TOF(q_{d-1},x_d;q_d)$, such that the tree walk path between these two edges does not visit any leaf, and the gates applied between these two Toffoli gates do not affect qubits $x_d$ and $q_d$.  The path beginning and ending with such a pair of Toffoli gates results in the following circuit, where the unitary $R$ implements all gates in between:
\[
\Qcircuit @C=.5em @R=.83em @!R {
\lstick{\ket{...}} 		& {/} \qw	& \qw		& \multigate{1}{R} 	& \qw		& \qw \\
\lstick{\ket{q_{d-1}}} 	&	 \qw	& \ctrl{1}	& \ghost{R} 		& \ctrl{1}	& \qw \\
\lstick{\ket{x_d}} 		&	 \qw	& \ctrlo{1}	& \qw	 			& \ctrl{1}	& \qw \\
\lstick{\ket{q_d}} 		&	 \qw	& \targ		& \qw	 			& \targ		& \qw
}
\]
The two Toffoli gates can be decomposed into Clifford+$T$ gates such that the $\TOF(q_{d-1},\overline{x_d};q_d)$ gate ends with one Hadamard, two $\CNOT$s, and 3 $T$/$T^\dagger$ gates on qubits $x_d$ and $q_d$. These cancel with one Hadamard, two $\CNOT$s, and 3 $T$/$T^\dagger$ gates in a decomposition of $\TOF(q_{d-1},x_d;q_d)$, leaving behind a single Phase gate (due to the Toffoli gate controls appearing with different polarities).  We furthermore choose the minimal resource implementation of $\TOF(q_{d-1},\overline{x_d};q_d)$ such that the last gate applying to qubit $\ket{q_{d-1}}$ is $\CNOT(q_d;q_{d-1})$ and the first gate in the decomposition of $\TOF(q_{d-1},x_d;q_d)$ affecting qubit $\ket{q_{d-1}}$ is also $\CNOT(q_d;q_{d-1})$.  These two $\CNOT$s can be commuted to each other and canceled, since Phase commutes with the control on the qubit $\ket{q_d}$ and $R$ is a reversible circuit with its output $\XOR$ed onto the target qubit $\ket{q_{d-1}}$, implying that $R$ commutes with $\CNOT(q_d;q_{d-1})$.  The reduced circuit, equivalent to the one shown above, is as follows:
\[
\Qcircuit @C=.4em @R=.4em @!R {
\lstick{\ket{...}} 		& {/} \qw	& \qw		& \qw 		& \qw 				& \qw		& \qw 		& \qw 		& \qw 		& \multigate{1}{R}
& \qw 				& \qw 		& \qw 				& \qw		& \qw 				& \qw 		& \qw 		& \qw \\
\lstick{\ket{q_{d-1}}} 	&	 \qw	& \gate{T} 	& \targ 	& \gate{T} 	& \targ		& \gate{T^\dagger} 	& \targ 	& \gate{T^\dagger} 	& \ghost{R}
& \gate{T^\dagger} 			& \targ 	& \gate{T} 	& \targ		& \gate{T^\dagger} 	& \targ 	& \gate{T} 	& \qw \\
\lstick{\ket{x_d}} 		&	 \qw	& \qw		& \ctrl{-1} & \qw  				& \qw		& \qw  		& \ctrl{-1} & \qw  		& \qw
& \qw  				& \ctrl{-1} & \qw  				& \qw		& \qw  				& \ctrl{-1} & \qw 		& \qw \\
\lstick{\ket{q_d}} 		&	 \qw	& \gate{H}	& \qw		& \qw  				& \ctrl{-2}	& \qw  		& \qw		& \qw 		& \gate{P}
& \qw  				& \qw		& \qw  				& \ctrl{-2}	& \qw  				& \qw		& \gate{H} 	& \qw
}
\]
The equivalence of the circuits before and after optimization can be verified by direct computation.

The control generation circuit contains $\frac{1}{2} \cdot 2^\lenstring - \lenstring$ such Toffoli pairs, each reducing the gate count by $2$ Hadamards, $6$ $T$/$T^\dagger$ gates, and $6$ $\CNOT$s while introducing a single Phase gate. Overall, the optimized $\select(V)$ control generation circuit thus uses the following resources:
\begin{align*}
\begin{array}{rl}
2 						& \text{$\NOT$ gates} \\
2 \cdot 2^\lenstring + 2\lenstring - 8 	& \text{Hadamard gates} \\
0.5 \cdot 2^\lenstring - \lenstring 		& \text{Phase gates} \\
7.5 \cdot 2^\lenstring+6\lenstring-28 	& \text{$T$/$T^\dagger$ gates} \\
7.5 \cdot 2^\lenstring+6\lenstring-26 	& \text{$\CNOT$ gates} \\
\lenstring-1 					& \text{ancilla qubits, initialized in and returned to the state $\ket{0}$}.
\end{array}
\end{align*}

Note that the resulting $T$/$T^\dagger$-count is only about $7\%$ higher than the expected $T$/$T^\dagger$-count in a circuit with $2^\lenstring-\lenstring-1$ Toffoli gates (which is the number of Toffoli gates required to generate $2^\lenstring$ Boolean products by a reversible NCT circuit, according to \lem{sV1}).  The overall elementary gate count in our optimized implementation is $17.5 \cdot 2^\lenstring + o(2^\lenstring)$. This compares favorably to the baseline estimation of $18\lenstring2^\lenstring + o(\lenstring2^\lenstring)$ elementary gates, improving the gate count by a factor of $\frac{36\lenstring}{35}$ at leading order.  Recall that in our application to the \TS\ algorithm, we have $\lenstring = \ceil{\log_2{\numselect}}$, where $\numselect$ is the number of selected unitaries,
 as in \eq{recallselectv}.

%% file: low_chuang.tex
We now turn our attention to the quantum signal processing (\QSP) algorithm of Low and Chuang \cite{LC17,LC16}, as introduced in \sec{algqsp}. We estimate the success probability of the \QSP\ algorithm in \sec{lc_failure} and discuss how this can be taken into account to make a fair comparison with other simulation algorithms. In \sec{lc_opt}, we describe optimizations that reduce the gate count of our implementation.   Then, in \sec{lc_phase_seg}, we discuss the difficulty of computing the phases that specify the algorithm and describe a segmented version of the algorithm that mitigates this issue.  Finally, in \sec{lc_bound}, we describe empirical bounds on the error in the truncated Jacobi-Anger expansion and in the overall algorithm.

\subsection{Failure probability}
\label{sec:lc_failure}

Just like the \TS\ algorithm, the \QSP\ algorithm is probabilistic. The success probability is not explicitly evaluated in references \cite{LC17,LC16}. The error analysis presented in \sec{algqsp} guarantees that the post-selected operator $(\langle +|\langle G|\otimes I)V(|+\rangle|G\rangle\otimes I)$ is close to the ideal evolution $e^{-iHt}$ in the sense that
\begin{equation}
	\norm*{\big(\langle +|\langle G|\otimes I\big)V\big(|+\rangle|G\rangle\otimes I\big)-e^{-iHt}}\le \epsilon.
\end{equation}
Therefore, for any given input state $|\psi\rangle$, the postselection succeeds with probability at least
\begin{align}
	&\norm*{\bigl(\langle +|\langle G|\otimes I\bigr)V\bigl(|+\rangle|G\rangle\otimes |\psi\rangle\bigr)}^2 \nonumber\\
	&\quad\geq \bigl(\norm*{e^{-iHt}|\psi\rangle}-\norm*{\bigl(\langle +|\langle G|\otimes I\bigr)V\bigl(|+\rangle|G\rangle\otimes |\psi\rangle\bigr)-e^{-iHt}|\psi\rangle}\bigr)^2 \\
	&\quad\geq \bigl(1-\norm*{\bigl(\langle +|\langle G|\otimes I\bigr)V\bigl(|+\rangle|G\rangle\otimes I\bigr)-e^{-iHt}}\bigr)^2 \\
	&\quad\ge 1-2\epsilon.
\end{align}
With $\epsilon=10^{-3}$, the success probability of the \QSP\ algorithm is at least $1-2\epsilon=0.998$. Reasoning similarly as in \sec{tsfailure}, we find that the expected number of times we must repeat the algorithm before succeeding is approximately $1/0.998\approx 1.002$. This factor is very close to $1$, so as for the \TS\ algorithm, we neglect it when comparing gate counts.

\subsection{Circuit optimizations}
\label{sec:lc_opt}

The $\select(V)$ gate is a major component of the \QSP\ algorithm, so we use the optimized implementation of that subroutine described in \sec{selectV} in our implementation of the \QSP\ algorithm. In this section, we present some further circuit optimizations that also reduce the gate count.

As discussed in \sec{algqsp}, we use the phased iterate $V_\phi$ defined in \eq{qspiter}, whereas Low and Chuang use the operation $V_\phi'$ defined in \eq{lciter}.  It is easy to see that
\begin{align}
  V_\phi = (I \otimes Z_{\pi/2}) V_\phi' (I \otimes Z_{-\pi/2}),
\end{align}
so
\begin{align}
  V_\phi^\dag = (I \otimes Z_{\pi/2}) V_\phi^{\prime\dagger} (I \otimes Z_{-\pi/2})
\end{align}
also involves conjugation by $I \otimes Z_{\pi/2}$. Thus, when the phased iterates are applied in the sequence \eq{pisequence}, the inner partial reflection gates cancel.  Furthermore, $Z_{\varphi}$ simply introduces a relative phase between $|G\rangle$ and its orthogonal subspace, so its action is trivial if the ancilla register is initialized in and postselected on $|G\rangle$. Thus we see that the implementation as defined in \sec{algqsp} has the same effect as in \cite{LC16}.  Each partial reflection is implemented using $O(\log n)$ elementary gates, and there are $O(n^2)$ phased iterates, so our implementation saves $O(n^2\log n)$ gates.

We apply a similar simplification to further reduce the gate count. For every phased iterate $V_\phi$ defined in \eq{qspiter}, we must implement a controlled version of the operator $-iQ= -i\bigl((2|G\rangle\langle G|-I)\otimes I\bigr)\select(H)$. In particular, this requires us to perform a controlled-reflection about $|G\rangle$. We can do this by performing a controlled-$U^\dagger$ that unprepares the state $|G\rangle$ (where $U$ is a unitary operation satisfying $U|0\rangle=|G\rangle$), a controlled reflection about $|0\rangle$, and finally a controlled-$U$ that prepares the state $|G\rangle$. However, observe that we can replace the controlled unitary conjugation by its uncontrolled version without changing the behavior of the circuit. Furthermore,
by grouping neighboring pairs of phased iterates in the sequence of $V_\phi$ and $V_{\phi'}^\dagger$ operations, we can cancel pairs of unitary operators $U$ and $U^\dagger$ for state preparation and unpreparation.

\subsection{Phase computation and segmented algorithm}
\label{sec:lc_phase_seg}

Recall that to specify the \QSP\ algorithm, we must find phases $\phi_1,\ldots,\phi_M$ that realize the truncated Jacobi-Anger expansion.  In principle, these angles can be computed in polynomial time \cite{LYC16}. However, this computation is difficult in practice, so we can only carry it out for very small instances. Specifically, we found the time required to calculate the angles to be prohibitive for values of $M$ greater than about $32$.  For $n=10$ qubits with $t=n$ and $\epsilon=10^{-3}$, the error bound \eq{qsp_errbd} suggests that we should take $M=1100$.  Thus the difficulty of computing the angles prevents us from synthesizing nontrivial instances of the algorithm.  This difficulty arises because the procedure for computing the angles requires us to compute the roots of a high-degree polynomial to high precision.  It is a natural open problem to give a more practical method for computing the angles.

Fortunately, to determine the Clifford+$R_z$ gate count in our implementation of the \QSP\ algorithm, we do not need to know the angles of the phased iterates.  Furthermore, since most $R_z$ gates require approximately the same number of Hadamard and $T$ gates to realize within a given precision, we can get a reasonable estimate of the Clifford+$T$ count by using random angles in place of the true values.  However, we emphasize that this method does not produce a correct quantum simulation circuit, and should only be used as a benchmark of the resource requirements of the \QSP\ algorithm---which is only useful if the true angles can ultimately be computed.

An alternative is to consider what we call a \emph{segmented} version of the algorithm. In this approach, we first divide the evolution time into $r$ segments, each of which is sufficiently short that the angles can readily be computed. Since the optimality of the \QSP\ approach to Hamiltonian simulation relies essentially on simulating the entire evolution as a single segment, the segmented approach has higher asymptotic complexity. However, it allows us to give a complete implementation, and the overhead for moderate values of $n$ is not too great.

To analyze the algorithm with $r$ segments, we apply the error bound \eq{qsp_errbd} with $t$ replaced by $t/r$ and $\epsilon$ replaced by $\epsilon/r$.  This gives the sufficient condition
\begin{equation}
  \label{eq:lc-closed}
  \frac{4({\alpha t}/{r})^{q}}{2^{q} \, q!}\leq \frac{\epsilon}{8r}
\end{equation}
where $q=\frac{M}{2}+1$.
Thus
\begin{align}
  r=O\Bigl(\alpha t \Bigl(\tfrac{\alpha t}{2\epsilon(\frac{M}{2}+1)!}\Bigr)^{2/M}\Bigr)
\end{align}
segments suffice to ensure overall error at most $\epsilon$. With $t=n$, $\alpha=O(n)$, and $M$ a fixed constant, we have $r=O(n^{2+4/M})$ segments. Within each segment, the number of phased iterates is $M$, which is independent of the system size. Each phased iterate has circuit size $O(n)$ using the improved $\select(V)$ implementation described in \sec{selectV}. Therefore, the segmented algorithm has gate complexity $O(n^{3+4/M})$.

In our implementation, we use $M=28$ (i.e., $q=15$). For the instance of quantum simulation considered in this paper, we set $\epsilon=10^{-3}$, $\alpha=4n$, and $t=n$. With these values, \eq{lc-closed} shows that it suffices to use
\begin{equation}
	r\geq\sqrt[14]{\frac{10^3\cdot 2^{20} n^{30}}{15!}}
	= 0.6010 \, n^{15/7}
\end{equation}
segments.

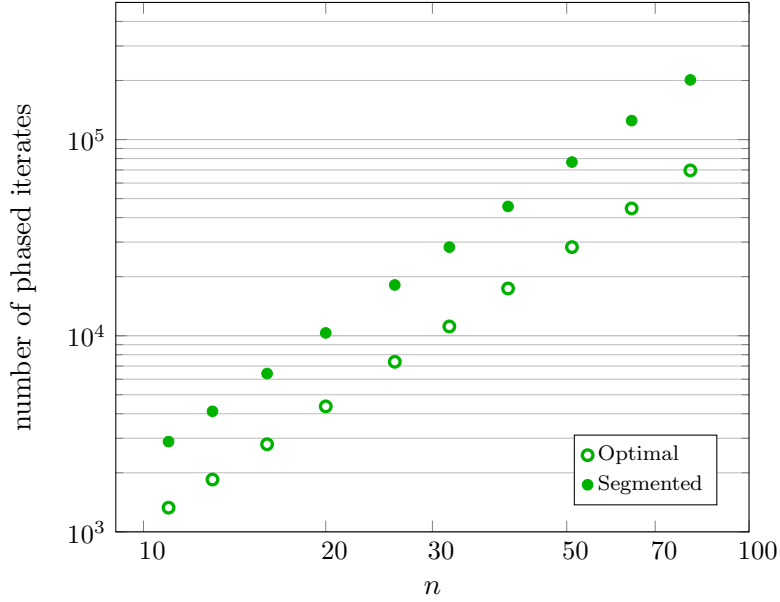
\begin{figure}
	\centering
	\begin{tikzpicture}
	\begin{axis}[
	log x ticks with fixed point,
	xtick={10,20,30,50,70,100},
	xmode=log,
	ymode=log,
	xmin = 9,
	xmax = 100,
	ymin = 1000,
	ymax = 5*10^5,
  width=10cm,
	ymajorgrids=true,
	yminorgrids=true,
	legend style={at={(0.95,0.05)},anchor=south east, font=\fontsize{8}{5}\selectfont},
	xlabel={$n$},
	ylabel={number of phased iterates},
	every axis legend/.append style={nodes={right},
		every x tick label/.append style={font=\small},
		every y tick label/.append style={font=\small},
	}
	]

  \addplot[only marks,black!30!green,mark=o,mark options={fill=white,line width=1.25pt}] coordinates {
  	(11,1328)
		(13,1848)
		(16,2796)
		(20,4360)
		(26,7362)
		(32,11144)
		(40,17406)
		(51,28290)
		(64,44544)
		(80,69598)
		(101,110924)
	};
	\addlegendentry{Optimal}

	\addplot[only marks, black!30!green] coordinates {
		(11,2884)
		(13,4116)
		(16,6412)
		(20,10332)
		(26,18144)
		(32,28280)
		(40,45612)
		(51,76776)
		(64,124880)
		(80,201432)
		(101,331912)
	};
	\addlegendentry{Segmented}

	\end{axis}
	\end{tikzpicture}
	\caption{Comparison of the number of phased iterates using optimal and segmented implementations of the \QSP\ algorithm.}
	\label{fig:lowchuang_seg}
\end{figure}

\fig{lowchuang_seg} compares the total number of phased iterates used in the segmented and optimal implementations.
Over the range of interest, the segmented algorithm is only worse by a factor between $2$ and $3$.

\subsection{Empirical error bounds}
\label{sec:lc_bound}

The error bound \eq{qsp_errbd} uses the closed-form expression \eq{jaerror} for the remainder of the Jacobi-Anger expansion.
While it is a convenient to use such an analytical expression, it is natural to ask how tightly it bounds the complexity of the algorithm.

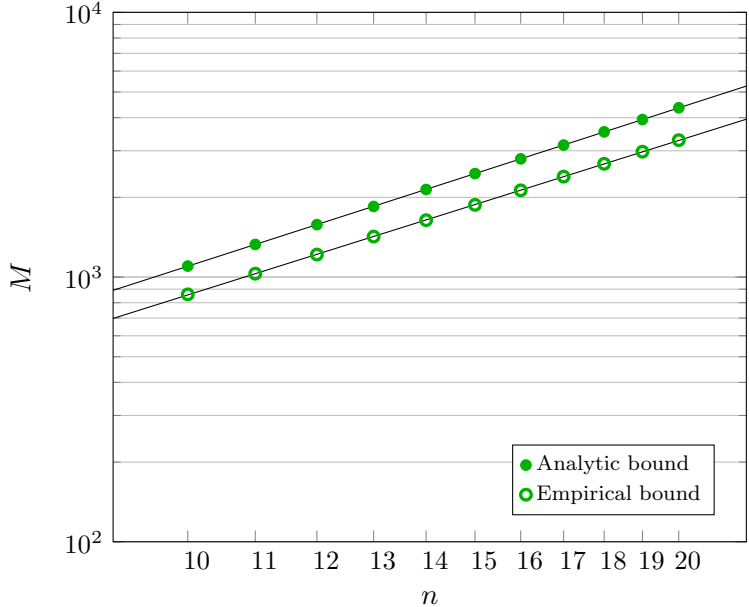
\begin{figure}
	\centering
	\begin{tikzpicture}
	\begin{axis}[
	log x ticks with fixed point,
  xtick={10,11,12,13,14,15,16,17,18,19,20},
	xmode=log,
	ymode=log,
	xmin = 9,
	xmax = 22,
	ymin = 10^2,
	ymax = 10^4,
	width=10cm,
	ymajorgrids=true,
	yminorgrids=true,
	legend style={at={(0.95,0.05)},anchor=south east, font=\fontsize{8}{5}\selectfont},
	xlabel={$n$},
	ylabel={$M$},
	every axis legend/.append style={nodes={right},
		every x tick label/.append style={font=\small},
		every y tick label/.append style={font=\small},
	}
	]

	\addplot[only marks, black!30!green] coordinates {
		(10,1100)
		(11,1328)
		(12,1578)
		(13,1848)
		(14,2142)
		(15,2460)
		(16,2796)
		(17,3152)
		(18,3536)
		(19,3936)
		(20,4360)
	};
	\addlegendentry{Analytic bound}

	\addplot[only marks,black!30!green,mark=o,mark options={fill=white,line width=1.25pt}] coordinates {
		(10,860)
		(11,1030)
		(12,1216)
		(13,1422)
		(14,1640)
		(15,1874)
		(16,2126)
		(17,2396)
		(18,2678)
		(19,2976)
		(20,3290)
	};
	\addlegendentry{Empirical bound}

	\addplot[
	color = black,
	mark = none
	]	coordinates {
		( 1, 11.2987 )
		( 100, 106810 )
	};

	\addplot[
	color = black,
	mark = none
	]	coordinates {
		( 1, 9.8486 )
		( 100, 74393 )
	};

	\end{axis}
	\end{tikzpicture}
	\caption{Comparison of the number of phased iterates using the analytic bound \eq{jaerror} and the empirical bound for the Jacobi-Anger expansion.  Here $M$ is the number of phased iterates and $n$ is the system size.}
	\label{fig:jacobi}
\end{figure}

To address this question, we numerically evaluate the left-hand side of \eq{jaerror} for systems of sizes ranging from $10$ to $20$, as shown in \fig{jacobi}. By extrapolating these data, we estimate the complexity of the \QSP\ algorithm for arbitrary sizes, including those for which classical evaluation of the series is intractable. The empirical bound improves the gate count by a factor between $1.25$ and $1.45$ over the range of interest ($10 \le n \le 100$). More specifically, power law fits to the data give
\begin{equation}
	\label{eq:ja_bound}
	M_{\text{ana}}=11.30 \, n^{1.988},\quad
	M_{\text{emp}}=9.849 \, n^{1.939}
\end{equation}
for the number of phased iterates using either the analytic bound or the empirical bound, respectively. Since each phased iterate has gate complexity $O(n)$ (using the technique from \sec{selectV}), we find that the \QSP\ algorithm has complexity $O(n^{2.988})$ (resp., $O(n^{2.939})$) using the analytic bound (resp., empirical bound).  Note that the former is roughly consistent with the upper bound $O(n^3\log n)$ shown in \tab{algsummary}.

We do not consider the empirical bound for the segmented version of the \QSP\ algorithm, since the savings is small in that case (even less at $M=28$ than at the values shown in \fig{jacobi}), and the main goal of the segmented approach is to have a fully-specified algorithm with rigorous guarantees.  However, we use the empirical bound to estimate resources using the non-segmented \QSP\ algorithm. This produces our most optimistic benchmark for the performance of the \QSP\ algorithm.

\begin{figure}
	\centering
	\begin{tikzpicture}
  \pgfplotsset{
      every non boxed x axis/.style={}
  }
	\begin{groupplot}[
	group style={
		group size=1 by 2,
		xticklabels at=edge bottom,
		vertical sep=0pt
	}
	]

	\nextgroupplot[
		xtick={5,6,7,8,9},
		xmode=log,
		ymode=log,
		xmin = 4.5,
		xmax = 10,
		ymin = 10^(-6),
		ymax = 10^(-4),
		width=10cm,
		height=7cm,
    axis x line=top,
		ymajorgrids=true,
		yminorgrids=true,
		legend style={at={(0.95,0.25)},anchor=south east, font=\fontsize{8}{5}\selectfont},
		y label style={at={(axis description cs:-.06,0.194538)},anchor=mid},
		ylabel={error},
		every axis legend/.append style={nodes={right},
			every x tick label/.append style={font=\small},
			every y tick label/.append style={font=\small},
		},
    scale only axis
	]
  \addplot+[only marks,black!30!green,mark=*,mark options={fill=white,line width=1.25pt},error bars/.cd,y dir=both,y explicit] coordinates {
    (5,0.00003846498613836768) +- (0,0.00002342998792421650)
    (6,0.00005413408105097961) +- (0,0.00002646830715358416)
    (7,0.00005161292779960877) +- (0,0.00001911711152242151)
    (8,0.00005182940875179489) +- (0,0.00002645539014021875)
    (9,0.00005220208967044450) +- (0,0.00001671253354392060)
  };

  \addplot+[only marks,red,mark=*,mark options={fill=white,line width=1.25pt},error bars/.cd,y dir=both,y explicit] coordinates {
    (5,0.00004506908002061410) +- (0,0.00001544416803328523)
    (6,0.00004231281448720436) +- (0,0.00001052993696655874)
    (7,0.00003449385552973019) +- (0,0.000006578262679717169)
    (8,0.00002934885775781500) +- (0,0.000005955828327453132)
    (9,0.00002478130585468500) +- (0,0.000002945891747431949)
  };

  \addplot+[only marks,red,mark=triangle*,mark options={fill=white,line width=1.25pt},error bars/.cd,y dir=both,y explicit] coordinates {
    (5,0.000006957121496613446) +- (0,0.0000006834997689817799)
    (6,0.000005266837501835215) +- (0,0.0000005857497610800569)
    (7,0.000003587903297605539) +- (0,0.0000003989541727554721)
    (8,0.000003283547411227005) +- (0,0.0000005613058081752632)
    (9,0.000002417658902701606) +- (0,0.0000003462214321831288)
  };

	\nextgroupplot[
	xtick={5,6,7,8,9},
  xticklabels={5,6,7,8,9},
  xmode=log,
	ymode=log,
	xmin = 4.5,
	xmax = 10,
	ymin = 3*10^(-11),
	ymax = 5*10^(-10),
  ytick = {0.00000000003,0.0000000001,0.0000000003},
  yticklabels={$3 \times 10^{-11}$,$10^{-10}$,$3 \times 10^{-10}$},
  minor ytick = {0.00000000003,0.00000000004,0.00000000005,0.00000000006,0.00000000007,0.00000000008,0.00000000009,0.0000000002,0.0000000003},
  axis y discontinuity=parallel,
	width=10cm,
	height=4.27647cm,
	ymajorgrids=true,
	yminorgrids=true,
	legend style={at={(0.95,1.2)},anchor=north east, font=\fontsize{8}{5}\selectfont},
	xlabel={$n$},
	axis x line=bottom,
	every axis legend/.append style={nodes={right},
		every x tick label/.append style={font=\small},
		every y tick label/.append style={font=\small},
	},
  scale only axis
	]

  \addplot+[only marks,black!30!green,mark=*,mark options={fill=white,line width=1.25pt},error bars/.cd,y dir=both,y explicit] coordinates {(5,1) +- (0,.1)};
  \addlegendentry{\QSP\ (seg)}

  \addplot+[only marks,red,mark=*,mark options={fill=white,line width=1.25pt},error bars/.cd,y dir=both,y explicit] coordinates {(5,1) +- (0,.1)};
  \addlegendentry{\PF\ (com 1)}

  \addplot+[only marks,red,mark=triangle*,mark options={fill=white,line width=1.25pt},error bars/.cd,y dir=both,y explicit] coordinates {(5,1) +- (0,.1)};
  \addlegendentry{\PF\ (com 2)}

	\addplot+[only marks,red,mark=square*,mark options={fill=white,line width=1.25pt},error bars/.cd,y dir=both,y explicit] coordinates {
		(5,0.0000000001828823539720707) +- (0,0.00000000003794022501383703)
		(6,0.0000000001239206189626176) +- (0,0.00000000001754466436146383)
		(7,0.00000000006139029304243768) +- (0,0.000000000007382059699760361)
		(8,0.00000000005007357115021354) +- (0,0.000000000007252429020121750)
	};
	\addlegendentry{\PF\ (com 4)}

	\end{groupplot}
  \pgfresetboundingbox
  \path[draw=none] (current axis.south west) -- ++(10cm,11.27647cm);
  \path[draw=none] (current axis.south west) -- ++(-1.35cm,-1cm);
	\end{tikzpicture}
	\caption{Empirical error in the segmented \QSP\ algorithm and \PF\ algorithms of orders 1, 2, and 4 (with commutator bound) for small system sizes.}
	\label{fig:empirical_sp}
\end{figure}
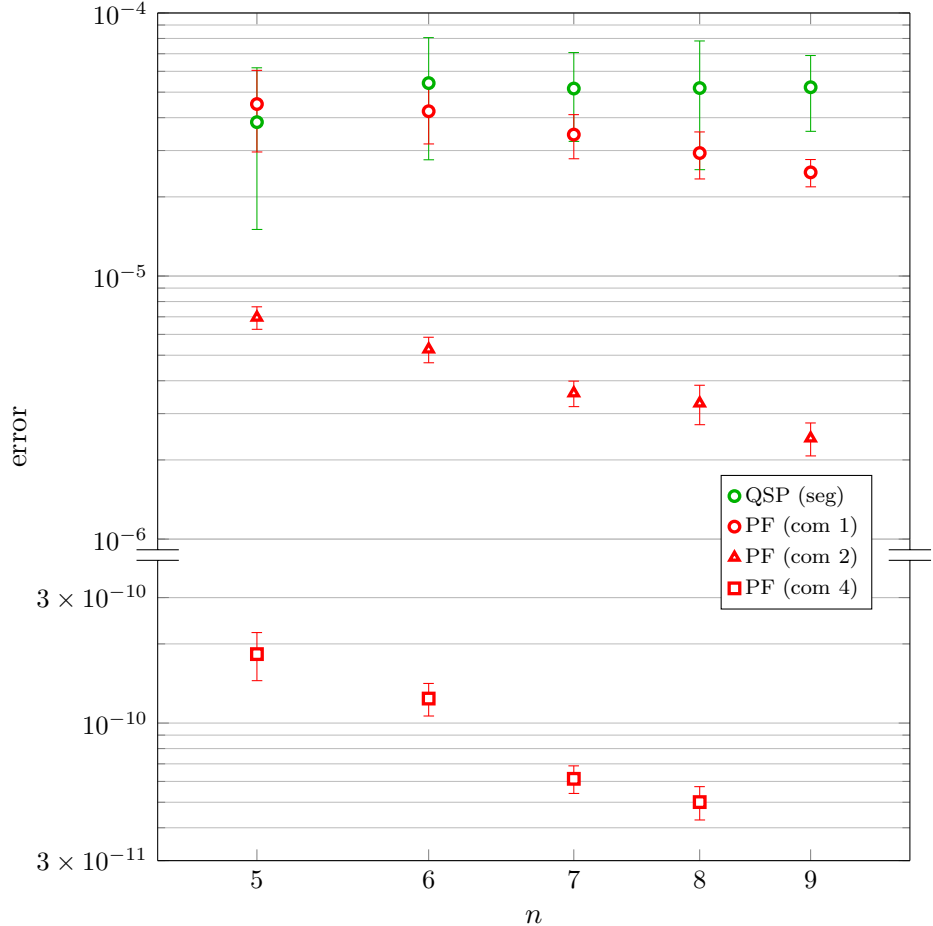

One could also consider a full empirical bound for the \QSP\ algorithm by using direct simulation to determine its true overall error. The need for ancilla qubits makes this challenging: the algorithm uses $n+\lceil\log 4n\rceil+1$ qubits to simulate an $n$-qubit system, as shown in \fig{qubitcounts}.  Fortunately, unlike with the \TS\ algorithm, small instances of the \QSP\ algorithm are just within reach of direct classical simulation.

However, preliminary numerical investigation suggests that the performance of the \QSP\ algorithm cannot be significantly improved using such an empirical bound. \fig{empirical_sp} shows the empirical error in the segmented \QSP\ algorithm for small system sizes, averaging over $10$ random experiments with a fixed target error $\epsilon=10^{-3}$, along with similar data for the \PF\ algorithm using the commutator bound. We observe that for system sizes between $5$ and $9$, the \QSP\ error is consistently around $5\times10^{-5}$, which is not significantly less than the target error of $10^{-3}$. While there was more variation in the error of the \QSP\ algorithm as compared to the \PF\ algorithm, in no case was the \QSP\ error less than $10^{-5}$. In contrast to the \PF\ algorithm, where the error apparently decreases as a power law in $n$, the \QSP\ error shows no indication of decreasing. Furthermore, since the complexity of the \QSP\ algorithm depends logarithmically on the inverse error $1/\epsilon$, even a large reduction in the error may not have a significant effect. For these reasons, we do not consider full empirical error bounds in our resource estimates for the \QSP\ algorithm.

%% file: results.tex
In this appendix, we present our results in more detail, expanding upon the discussion in \sec{results}.

First we briefly describe the variation of gate counts with respect to the random choice of Hamiltonian (namely, the coefficients $h_j \in [-1,1]$ in \eq{heisenberg}).
For each of the system sizes considered below, we generated five random Hamiltonians and determined the gate counts using Quipper. We report the
average gate count over the five Hamiltonians.

In the Clifford+$\Rz$ basis, the gate counts for \PF\ algorithms are actually  independent of the particular choice of coefficients $h_j$: the rotation angles in the circuits differ, but the number of rotations does not. The number of gates in the Clifford+$T$ basis depends on these angles, but only weakly. Across orders and bounds, the $T$ count never varies by more than $1\%$.
The \TS\ algorithm is the most sensitive to the choice of Hamiltonian. For systems of size 13, the Clifford+$\Rz$ gate counts vary by about $10\%$, and similarly for the $T$ counts. However, the variance reduces as the system size increases. For systems of size 100, the variance is about $2\%$ across gate types.
The \QSP\ algorithm (segmented or not) has a mild dependence on the random choice of Hamiltonian. In the Clifford+$\Rz$ basis, the $\CNOT$ and $\Rz$ counts are independent of the choice of Hamiltonian while the Clifford counts vary by less than $1\%$. Because the rotation angles depend on the Hamiltonian, the $T$ count depends on the choice of Hamiltonian, but again the observed variation is always less than $1\%$.

In the remainder of this section, we present our results in detail.  We discuss the \PF\ algorithm first, in \sec{results-pf}, since this case has the most variations to compare.  Then, in \sec{results-rest}, we present results for the \TS\ and \QSP\ algorithms and discuss the relative performance of all three approaches.  In \sec{results-opt}, we explain how the circuit counts are improved using our automated circuit optimizer.

\subsection{Product formulas}
\label{sec:results-pf}

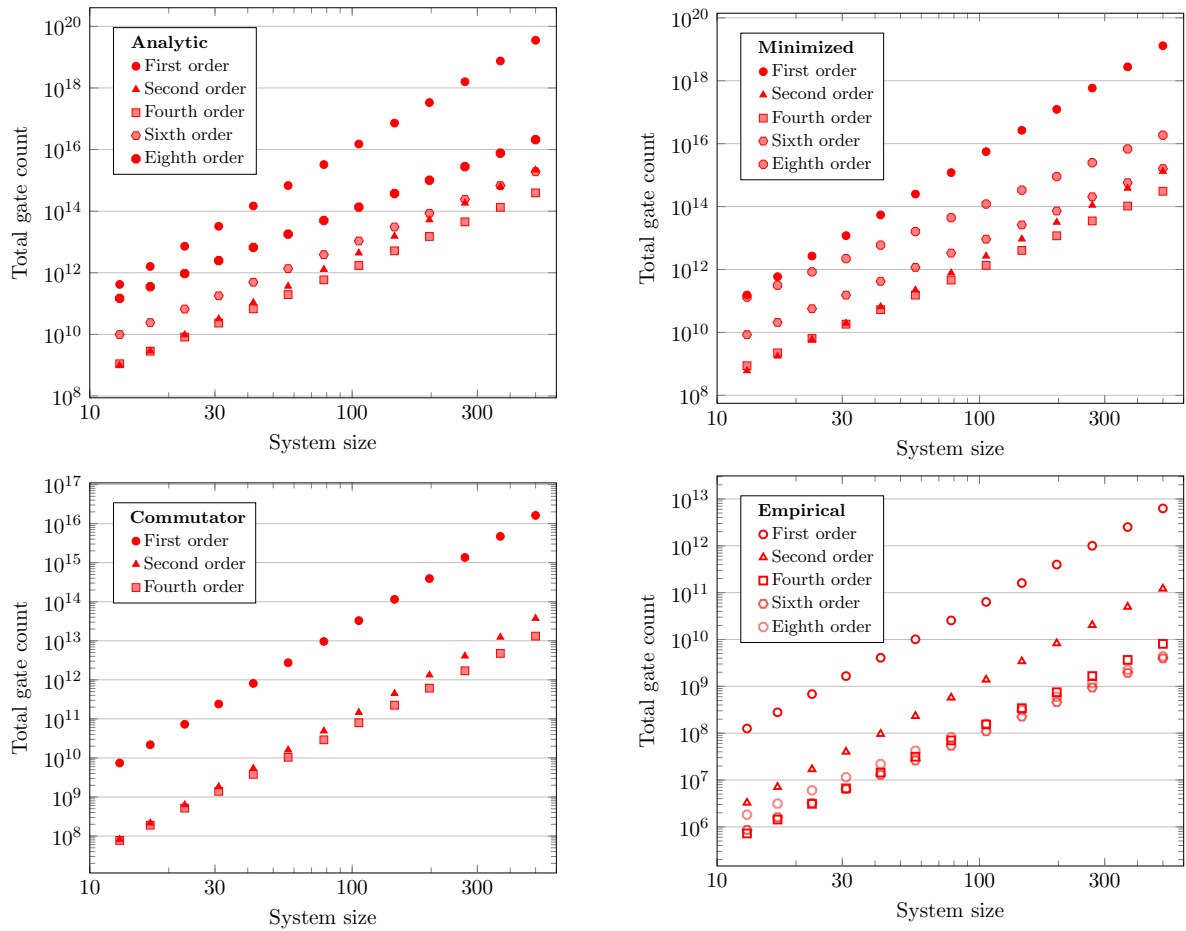
\begin{figure}
  \begin{subfigure}{.5\linewidth}
    \resizebox{.9\textwidth}{!}{\input{ana12468.tex}}
  \end{subfigure}
  \begin{subfigure}{.5\linewidth}
    \resizebox{.9\textwidth}{!}{\input{min12468.tex}}
  \end{subfigure}
  \begin{subfigure}{.5\linewidth}
    \resizebox{.9\textwidth}{!}{\input{com124.tex}}
  \end{subfigure}
  \begin{subfigure}{.5\linewidth}
    \resizebox{.9\textwidth}{!}{\input{fit12468.tex}}
  \end{subfigure}
\caption{Total gate counts in the Clifford+$\Rz$ basis for product formula algorithms using the analytic (top left), minimized (top right), commutator (bottom left), and empirical (bottom right) bounds, for system sizes between 13 and 500.}\label{fig:ana_min_com_CZ}
\end{figure}

The performance of quantum simulation algorithms using product formulas depends strongly on the order of the formula and the method used to bound its error.  We now compare the effect of these choices.

\fig{ana_min_com_CZ} compares the total gate counts of simulation algorithms using product formulas of various orders with each of our four error bounds.  Using rigorous bounds, for system sizes between $13$ and $500$, the fourth-order algorithm outperforms other orders except for very small system sizes.  The second-order algorithm is preferred for the largest range of system sizes with the minimized bound, but even in that case, it only outperforms the fourth-order algorithm for $n \leq 28$. Similar considerations hold for Clifford+$T$ circuits and when we only count gates of a particular type.
Overall, it is clear that the fourth-order algorithm using the commutator bound is preferred for our application if one requires a provable guarantee on the error.
If a heuristic method is acceptable, the empirical bound offers a dramatic reduction in complexity.
In that case, the sixth-order algorithm begins to outperform the fourth-order algorithm at around 30 qubits, although the difference is small for the range of sizes we focus on (at most $100$ qubits).

Of the cases we considered, the first-order algorithm has the worst performance, using from $10^2$ to over $10^4$ times more gates than the fourth- or sixth-order algorithm. The experimental demonstrations of digital quantum simulation that we are aware of \cite{BCC06,Lan11,Bar15} have primarily used the first-order formula, with some limited applications of the second-order formula \cite{BCC06,Lan11}. Our work suggests that higher-order formulas may be practically relevant for surprisingly small instances, and we hope our results motivate practitioners to take advantage of this.

We did not evaluate the commutator bound for orders higher than $4$. This is because the bound is difficult to compute in practice, as discussed in \sec{pfcom}. Since the fourth-order algorithm performs significantly better than the sixth-order one for both the analytic and minimized bounds, we also expect the fourth-order algorithm to outperform the sixth-order algorithm even if we used an analogous sixth-order commutator bound.

\fig{4_ana_min_com_totalCZ} shows the total gate counts for the Clifford+$\Rz$ gate set using the fourth-order algorithm with the analytic, minimized, commutator, and empirical bounds. The commutator bound saves a factor of more than 10 over the minimized bound, a significant improvement in practical terms.  The empirical bound results in even greater savings of a further factor of about 100 (for system sizes in the 10s) to over 500 (for system sizes around 100).  Clearly, the empirical bound is strongly preferred if one can tolerate its lack of a rigorous correctness guarantee.  This suggests that existing rigorous bounds are quite loose, so a natural open problem is to establish stronger rigorous bounds, as mentioned in \sec{discussion}.

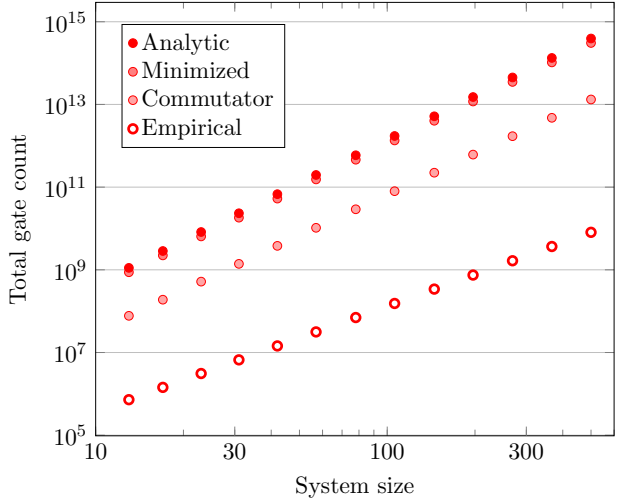
\begin{figure}
  \centering
  \resizebox{.5\linewidth}{!}{\input{4_ana_min_com_totalCZ.tex}}
  \caption{Total gate counts of Clifford+$\Rz$ circuits using the
    fourth-order \PF\ algorithm and varying bounds for system sizes between
    10 and 500.\label{fig:4_ana_min_com_totalCZ}}
\end{figure}

In fact, the commutator and empirical bounds not only reduce the gate count for small instances, but actually improve the dependence on system size $n$.  This is evident in \fig{4_ana_min_com_totalCZ}, where the slope of the data is reduced for the commutator and empirical bounds.  \tab{pfcomplexity} presents the exponents of $n$, both for the proven performance of the commutator bound (as stated in \thm{first_comm}, \thm{second_succinct}, and \thm{fourth_succinct}) and the numerically observed performance of the empirical bound.  Surprisingly, the empirical bound appears to give better asymptotic performance than $O(n^3)$ (the best asymptotic gate complexity listed in \tab{algsummary}): in particular, the eighth-order formula has asymptotic empirical complexity only slightly greater than quadratic in $n$.

While our empirical bound is only directly applicable to the particular system and parameters considered in our study, we expect that qualitatively similar improvements should apply in other cases. In particular, we performed a similar extrapolation for the fourth-order product formula applied to the same system with $h=1/10$ and $h=10$. We found very similar performance estimates for these parameters, with essentially the same exponent of $n$ and slightly different prefactors.

Note that while the constants in high-order commutator bounds are nontrivial to compute, we can more easily determine the asymptotic performance of these bounds, as discussed at the end of \sec{pfcom}. The lowest-order error in the $(2k)$th-order product formula is $O(n^{4k+2}/r^{2k+1})$. Our commutator bound utilizes the commutation relations to reduce the $n$-dependence of this lowest-order error to $O(n^{4k+1}/r^{2k+1})$, while the remaining higher-order terms are handled by standard techniques \cite{BACS05}. The resulting asymptotic performance is $O(n^{3+2/(2k+1)})$, improving over the complexity $O(n^{3+1/k})$ obtained using the standard bound.

\begin{table}
  \begin{subfigure}{.65\linewidth}
	\begin{center}
		\renewcommand{\arraystretch}{1.15}
		\begin{tabular}{r@{\;}l|c|c|c|c|c}
      && \multicolumn{5}{c}{Order} \\
			&Bound & 1 & 2 & 4 & 6 & 8 \\ \hline
      \raisebox{2pt}{\begin{tikzpicture}\draw[mark=*,red] plot coordinates{(0,0)};\end{tikzpicture}}
      &Analytic/Minimized & 5 & 4 & 3.5 & 3.333 & 3.25 \\
      \raisebox{2pt}{\begin{tikzpicture}\draw[mark=*,red,fill opacity=0.35] plot coordinates{(0,0)};\end{tikzpicture}}
      &Commutator & 4 & 3.667 & 3.4 & 3.286 & 3.222 \\
      \raisebox{2pt}{\begin{tikzpicture}\draw[mark=o,red,mark options={fill=white,line width=1.25pt}] plot coordinates{(0,0)};\end{tikzpicture}}
      &Empirical & 2.964 & 2.883 & 2.555 & 2.311 & 2.141
		\end{tabular}
		\renewcommand{\arraystretch}{1}
	\end{center}
  \end{subfigure}
  \begin{subfigure}{.35\linewidth}
    \begin{tikzpicture}
      \begin{axis}[
        width=5cm,
        legend style={at={(1.05,.5)},anchor=west,font=\scriptsize},
        xlabel={Order},
        ylabel={Exponent},
        xlabel near ticks,
        ylabel near ticks,
        extra x ticks={1},
        extra x tick labels={1},
        every axis legend/.append style={nodes={right}},
        every axis/.append style={font=\footnotesize}
        ]
        \pgfplotsset{every tick label/.append style={font=\scriptsize}}

        \addplot[only marks, red] coordinates {
          (1,5)
          (2,4)
          (4,3.5)
          (6,3.333)
          (8,3.25)
        };

        \addplot[only marks, red, fill opacity=0.35] coordinates {
          (1,4)
          (2,3.667)
          (4,3.4)
          (6,3.286)
          (8,3.222)
        };

        \addplot[only marks,red,mark=o,mark options={fill=white,line width=1.25pt}] coordinates {
          (1,2.875)
          (2,2.796)
          (4,2.587)
          (6,2.507)
          (8,2.141)
        };

      \end{axis}
    \end{tikzpicture}
    \end{subfigure}
	\caption{Exponent of $n$ in the asymptotic gate complexity of product formula simulation algorithms using various error bounds.\label{tab:pfcomplexity}}
\end{table}

\subsection{Other algorithms and comparisons}
\label{sec:results-rest}

\fig{opt} compares gate counts for the fourth-order \PF\ algorithm with commutator bound, the better of the fourth- or sixth-order \PF\ algorithm with empirical bound, the \TS\ algorithm, and the \QSP\ algorithm (in both its segmented and non-segmented versions).

Observe that the \TS\ algorithm outperforms the fourth-order \PF\ algorithm with rigorous performance guarantees provided $n$ is at least about $30$.  However, the $\CNOT$ count over Clifford+$R_z$ is only negligibly improved, whereas the $T$ count over Clifford+$T$ is improved by a factor of about 7.  This advantage comes at a significant cost in terms of space, as discussed in \sec{intro}. For system sizes between 20 and 50, the \TS\ algorithm uses between 116 and 171 qubits, whereas the \PF\ algorithm has no space overhead.  Furthermore, the \QSP\ algorithm always outperforms the \TS\ algorithm with respect to both circuit size and number of qubits, so it is clearly preferred.

Recall that we implemented two variants of the \QSP\ algorithm: the full algorithm (with random rotation angles in place of the true, hard-to-compute values) using the empirical Jacobi-Anger bound \eq{ja_bound}, and the segmented version using the rigorous error bound \eq{lc-closed}. Among all algorithms with a complete implementation and a rigorous accuracy guarantee, the segmented \QSP\ algorithm has the lowest gate counts. With respect to the $\CNOT$ count, it outperforms the fourth-order \PF\ algorithm, and even the \TS\ algorithm, by a factor of at least 5. The improvement in the $T$ count is greater. The full \QSP\ algorithm with empirical Jacobi-Anger bound has even lower gate complexity, but this will only be useful in practice if one can overcome the difficulty of classically computing the rotation angles.  Furthermore, if we are willing to accept heuristic error estimates, then the \PF\ algorithm significantly outperforms either \QSP\ algorithm while also requiring no ancilla qubits.

\subsection{Circuit optimization}
\label{sec:results-opt}

To reduce the resource requirements for quantum simulation as much as possible, we post-processed all of our circuits with an automated circuit optimizer \cite{Optimizer} as described in \sec{circuit-opt}.
\fig{opt} compares the gate counts of the \PF, \TS, and \QSP\ algorithms before and after this optimization. The $\CNOT$ and $T$ counts of the fourth-order \PF\ algorithm are both reduced by about $30\%$ throughout the range of interest. While the gate counts of the \TS\ and \QSP\ algorithms are also reduced, the improvement is much less significant.

For the Clifford+$R_z$ resource estimates (see the left plot in \fig{opt}), the \PF\ circuits show a 33\% reduction in $\CNOT$ count, whereas both the \TS\ and \QSP\ circuits admit only about 1\% reduction in their $\CNOT$ counts. We expect to see improvements in the $\CNOT$ count reduction for the \TS\ and \QSP\ circuits when using heavy instead of light optimization, although preliminary results suggest that this reduction will not qualitatively change the relative performance of the different algorithms.
For the Clifford+$T$ resource estimates (see the right plot in \fig{opt}), the $T$ count reduction is approximately $30\%$ for the \PF\ algorithm, $0.5\%$ for \TS, and $1\%$ for \QSP.

With respect to $\CNOT$ count, whereas the \TS\ algorithm outperforms the best \PF\ algorithm with rigorous performance guarantees for all system sizes shown in \fig{opt}, optimization improves the \PF\ algorithm to outperform the \TS\ algorithm for system sizes smaller than about $30$.  While the \TS\ algorithm is dominated by the more efficient \QSP\ approach, this example nevertheless shows that optimization can sometimes affect the relative performance of algorithms.

We also compare the optimization of \PF\ algorithms of different orders, as shown in \fig{opt_pf}.  We plot gate counts for \PF\ circuits with the empirical bound, which offers the best performance. Since the structure of the \PF\ circuit is not affected by choosing a different error bound, essentially the same relative improvements hold for other bounds as well.

The first-order \PF\ circuits do not admit any optimization in the $\CNOT$ counts. The remaining \PF\ circuits admit $\CNOT$ gate count reduction of about 33\%, with marginally more savings observed for higher orders (although this additional gain is too small to see in the scale used in \fig{opt_pf}).
The $R_z$ count of the \PF\ circuits behaves similarly under optimization. For the first-order \PF\ circuits, the $R_z$ count is not reduced.
However, for all higher-order circuits, the optimizer can take advantage of the reflection symmetry in formulas of order $2$ and higher, which reduces the $R_z$ counts in the original Clifford+$R_z$ circuits. The observed reduction is about 29\%.

\comment{
\ys{Discussion of the results:
  \begin{itemize}
    \item Comparison of different algorithms:
    \begin{itemize}
      \item \PF\ optimizes the best (with $xx\%$ reduction)
      \item We also observe a gatecount reduction in the \TS\ and \QSP\ algorithms, although not as significant as that in \PF\
      \item The select(V) subroutine of \TS\ and \QSP\ mainly consists of toffoli gates. These gates are decomposed by our optimizer with the state-of-art approach? Therefore, the $T$ count of \TS\ and \QSP\ is also improved.
    \end{itemize}
    \item Comparison of \PF\ algorithm with different orders:
    \begin{itemize}
      \item The optimizer can take advantage of the symmetry in definition of PF2, PF4 and PF6. So the gate count for these three algorithms is reduced by a larger amount than PF1.
      \item The structure of the \PF\ circuit is not affected by choosing different error bound. For concreteness we only show data for \PF\ with empirical bound. But similar result holds for other bounds.
      \item The percentage of improvement is constant throughout the entire range of interest. This is due to the LCR feature?
    \end{itemize}
  \end{itemize}}
}

\begin{figure}
  \begin{subfigure}{.5\linewidth}
  \centering
  \resizebox{.9\textwidth}{!}{
  \begin{tikzpicture}
    \begin{loglogaxis}[
      width=10cm,
      ymajorgrids=true,
      legend style={at={(0.05,0.95)},anchor=north west,font=\small},
      xlabel={System size},
      ylabel={$\CNOT$ gate count},
      xmin=10,
      xtick={10,100},
      xticklabels={10,100},
      extra x ticks={20,30,50,70},
      extra x tick labels={20,30,50,70},
      every axis legend/.append style={nodes={right}}
      ]

      \addplot[only marks, red, mark=diamond*, forget plot] coordinates {
        \preoptimcznormalcnotfrthcomavg
      };
      \addplot[only marks, red, fill opacity=0.5] coordinates {
        \postoptimcznormalcnotfrthcomavg
      };
      \addlegendentry{\PF~(com 4)}
      \addplot[only marks,red,mark=diamond*,mark options={fill=white,line width=1.25pt}, forget plot] coordinates {
        \preoptimcznormalcnotbestfitavg
      };
      \addplot[only marks,red,mark=o,mark options={fill=white,line width=1.25pt}, fill opacity=0.5] coordinates {
        \postoptimcznormalcnotbestfitavg
      };
      \addlegendentry{\PF~(emp)}
      \addplot[only marks, blue, mark=diamond*, forget plot] coordinates {
        \preoptimcznormalcnotlcuavg
      };
      \addplot[only marks, blue, fill opacity=0.5] coordinates {
        \postoptimcznormalcnotlcuavg
      };
      \addlegendentry{\TS}
      \addplot[only marks, black!30!green, mark=diamond*, forget plot] coordinates {
        \preoptimcznormalcnotspsegmentavg
      };
      \addplot[only marks, black!30!green, fill opacity=0.5] coordinates {
        \postoptimcznormalcnotspsegmentavg
      };
      \addlegendentry{\QSP~(seg)}
      \addplot[only marks,black!30!green,mark=diamond*,mark options={fill=white,line width=1.25pt}, forget plot] coordinates {
        \preoptimcznormalcnotspjaavg
      };
      \addplot[only marks,black!30!green,mark=o,mark options={fill=white,line width=1.25pt}, fill opacity=0.5] coordinates {
        \postoptimcznormalcnotspjaavg
      };
      \addlegendentry{\QSP~(JA emp)}

    \end{loglogaxis}
  \end{tikzpicture}
  }
 \end{subfigure}
  \begin{subfigure}{.5\linewidth}
  \resizebox{.9\textwidth}{!}{
  \begin{tikzpicture}
    \begin{loglogaxis}[
      width=10cm,
      ymajorgrids=true,
      legend style={at={(0.05,0.95)},anchor=north west,font=\small},
      xlabel={System size},
      ylabel={$T$ gate count},
      xmin=10,
      xtick={10,100},
      xticklabels={10,100},
      extra x ticks={20,30,50,70},
      extra x tick labels={20,30,50,70},
      every axis legend/.append style={nodes={right}}
      ]

       \addplot[only marks, red, mark=diamond*, forget plot] coordinates {
       	\preoptimctnormaltfrthcomavg
       };
       \addplot[only marks, red, fill opacity=0.5] coordinates {
       	\postoptimctnormaltfrthcomavg
       };
       \addlegendentry{\PF~(com 4)}

       \addplot[only marks,red,mark=diamond*,mark options={fill=white,line width=1.25pt}, forget plot] coordinates {
       	\preoptimctnormaltbestfitavg
       };
       \addplot[only marks,red,mark=o,mark options={fill=white,line width=1.25pt}, fill opacity=0.5] coordinates {
       	\postoptimctnormaltbestfitavg
       };
       \addlegendentry{\PF~(emp 4)}

       \addplot[only marks, blue, mark=diamond*, forget plot] coordinates {
       	\preoptimctnormaltlcuavg
       };
       \addplot[only marks, blue, fill opacity=0.5] coordinates {
       	\postoptimctnormaltlcuavg
       };
       \addlegendentry{\TS}

      \addplot[only marks, black!30!green, mark=diamond*, forget plot] coordinates {
        \preoptimctnormaltspsegmentavg
      };
      \addplot[only marks, black!30!green, fill opacity=0.5] coordinates {
        \postoptimctnormaltspsegmentavg
      };
      \addlegendentry{\QSP~(seg)}

      \addplot[only marks,black!30!green,mark=diamond*,mark options={fill=white,line width=1.25pt}, forget plot] coordinates {
        \preoptimctnormaltspjaavg
      };
      \addplot[only marks,black!30!green,mark=o,mark options={fill=white,line width=1.25pt}, fill opacity=0.5] coordinates {
        \postoptimctnormaltspjaavg
      };
      \addlegendentry{\QSP~(JA emp)}

    \end{loglogaxis}
  \end{tikzpicture}
  }
\end{subfigure} \caption{Gate counts before (diamonds) and after
  (circles) circuit optimization.  Left: $\CNOT$ gates over
  Clifford+$R_z$.  Right: $T$ gates over
  Clifford+$T$. \label{fig:opt}}
\end{figure}
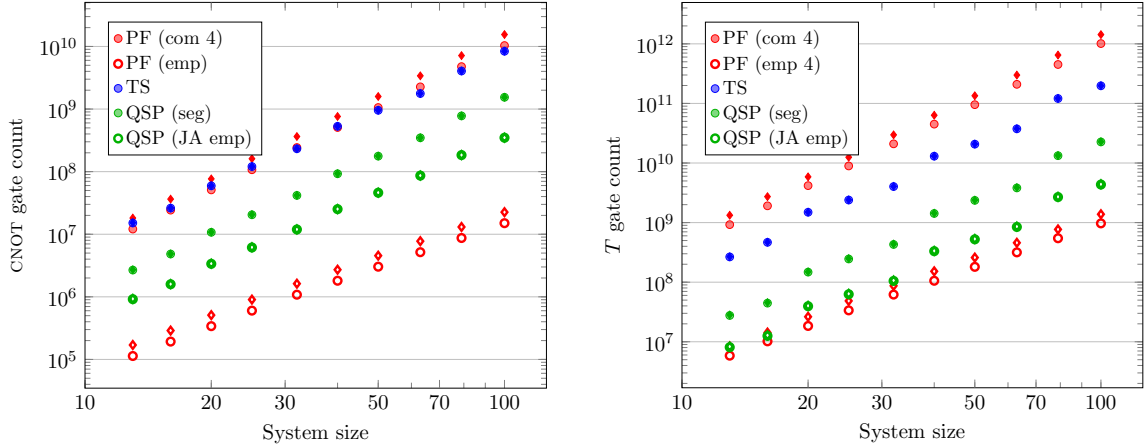

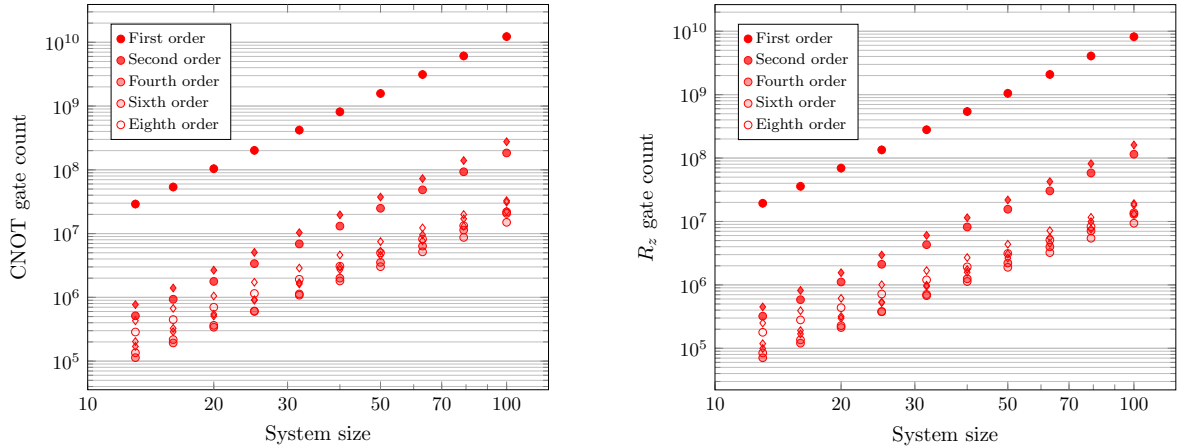
\begin{figure}
	\begin{subfigure}{.5\linewidth}
		\resizebox{.9\textwidth}{!}{
			\begin{tikzpicture}
			\begin{axis}[
			width=10cm,
			log x ticks with fixed point,
			xmin=10,
			xtick={10,100},
			xticklabels={10,100},
			extra x ticks={20,30,50,70},
			extra x tick labels={20,30,50,70},
			xmode=log,
			ymode=log,
			xmin = 10,
			ymajorgrids=true,
			yminorgrids=true,
			legend style={at={(0.05,0.95)},anchor=north west, font=\fontsize{8}{5}\selectfont},
			xlabel={System size},
			ylabel={CNOT gate count},
			every axis legend/.append style={nodes={right},
				every x tick label/.append style={font=\small},
				every y tick label/.append style={font=\small},
			}
			]

			\addplot[only marks, red, mark=diamond*,forget plot] coordinates {
				\preoptimcznormalcnotfstfitavg
			};
			\addplot[only marks, red] coordinates {
				\postoptimcznormalcnotfstfitavg
			};
			\addlegendentry{First order}

			\addplot[only marks, red, mark=diamond*, fill opacity=0.7,forget plot] coordinates {
				\preoptimcznormalcnotsndfitavg
			};
			\addplot[only marks, red, fill opacity=0.7] coordinates {
				\postoptimcznormalcnotsndfitavg
			};
			\addlegendentry{Second order}

			\addplot[only marks, red, mark=diamond*, fill opacity=0.4,forget plot] coordinates {
				\preoptimcznormalcnotfrthfitavg
			};
			\addplot[only marks, red, fill opacity=0.4] coordinates {
				\postoptimcznormalcnotfrthfitavg
			};
			\addlegendentry{Fourth order}

			\addplot[only marks, red, mark=diamond*, fill opacity=.2,forget plot] coordinates {
				\preoptimcznormalcnotsxthfitavg
			};
			\addplot[only marks, red, fill opacity=.2] coordinates {
				\postoptimcznormalcnotsxthfitavg
			};
			\addlegendentry{Sixth order}

      \addplot[only marks, red, mark=diamond*, fill opacity=.05,forget plot] coordinates {
				\preoptimcznormalcnoteigthfitavg
			};
			\addplot[only marks, red, fill opacity=.05] coordinates {
				\postoptimcznormalcnoteigthfitavg
			};
			\addlegendentry{Eighth order}

			\end{axis}
			\end{tikzpicture}
		}
	\end{subfigure}
	\begin{subfigure}{.5\linewidth}
		\resizebox{.9\textwidth}{!}{
			\begin{tikzpicture}
			\begin{axis}[
			width=10cm,
			log x ticks with fixed point,
			xmin=10,
			xtick={10,100},
			xticklabels={10,100},
			extra x ticks={20,30,50,70},
			extra x tick labels={20,30,50,70},
			xmode=log,
			ymode=log,
			xmin = 10,
			ymajorgrids=true,
			yminorgrids=true,
			legend style={at={(0.05,0.95)},anchor=north west, font=\fontsize{8}{5}\selectfont},
			xlabel={System size},
			ylabel={$\Rz$ gate count},
			every axis legend/.append style={nodes={right},
				every x tick label/.append style={font=\small},
				every y tick label/.append style={font=\small},
			}
			]

			\addplot[only marks, red, mark=diamond*,forget plot] coordinates {
				\preoptimcznormalrzfstfitavg
			};
			\addplot[only marks, red] coordinates {
				\postoptimcznormalrzfstfitavg
			};
			\addlegendentry{First order}

			\addplot[only marks, red, mark=diamond*, fill opacity=0.7,forget plot] coordinates {
				\preoptimcznormalrzsndfitavg
			};
			\addplot[only marks, red, fill opacity=0.7] coordinates {
				\postoptimcznormalrzsndfitavg
			};
			\addlegendentry{Second order}

			\addplot[only marks, red, mark=diamond*, fill opacity=0.4,forget plot] coordinates {
				\preoptimcznormalrzfrthfitavg
			};
			\addplot[only marks, red, fill opacity=0.4] coordinates {
				\postoptimcznormalrzfrthfitavg
			};
			\addlegendentry{Fourth order}

			\addplot[only marks, red, mark=diamond*, fill opacity=.2,forget plot] coordinates {
				\preoptimcznormalrzsxthfitavg
			};
			\addplot[only marks, red, fill opacity=.2] coordinates {
				\postoptimcznormalrzsxthfitavg
			};
			\addlegendentry{Sixth order}

      \addplot[only marks, red, mark=diamond*, fill opacity=.05,forget plot] coordinates {
				\preoptimcznormalrzeigthfitavg
			};
			\addplot[only marks, red, fill opacity=.05] coordinates {
				\postoptimcznormalrzeigthfitavg
			};
			\addlegendentry{Eighth order}

			\end{axis}
			\end{tikzpicture}
		}
	\end{subfigure}
	\caption{Gate counts before (diamonds) and after (circles)
          circuit optimization for the first, second, fourth and sixth
          order \PF\ algorithm with empirical bound.  Left: $\CNOT$
          gates over Clifford+$R_z$. Right: $\Rz$ gates over
          Clifford+$R_z$.\label{fig:opt_pf}}
\end{figure}

\comment{
\begin{table}
\makebox[\linewidth][c]{\begin{subfigure}{.6\textwidth}
\centering\small
\postoptimfrthcomcz4060table
\end{subfigure}\begin{subfigure}{.4\textwidth}
\centering\small
\postoptimfrthfitcz4060table
\end{subfigure}}

\bigskip

\makebox[\linewidth][c]{\begin{subfigure}{\linewidth}
\centering\small
\preoptimlcucz4060table
\end{subfigure}}

\bigskip

\makebox[\linewidth][c]{\begin{subfigure}{.5\textwidth}
\centering\small
\preoptimspjacz4060table
\end{subfigure}\begin{subfigure}{.5\textwidth}
\centering\small
\preoptimspsegmentcz4060table
\end{subfigure}}
\caption{\nam{Currently showing preoptim results for non \PF\ results; Julien, do you happen to have the corresponding non \PF\ circuits somewhere in the dropbox that we could optimize?} Detailed gate count comparisons for the fourth-order product formula algorithm (using the commutator and empirical error bounds), the \TS\ algorithm, and the \QSP\ algorithm (using the segmented version with analytic error bound and the non-segmented version with empirical error bound for the Jacobi-Anger expansion) over the Clifford+$\Rz$ basis before and after optimization. \label{tab:optCZ_40/60}}
\end{table}
}

%% file: ana12468.tex
\begin{tikzpicture}
  \begin{loglogaxis}[
    width=10cm,
    ymajorgrids=true,
    legend style={at={(0.05,0.95)},anchor=north west,font=\footnotesize},
    xlabel={System size},
    ylabel={Total gate count},
    xmin=10,
    xmax=600,
    xtick={10,100},
    xticklabels={10,100},
    extra x ticks={30,300},
    extra x tick labels={30,300},
    every axis legend/.append style={nodes={right}}
    ]

    \addlegendimage{empty legend}
    \addlegendentry[yshift=0pt]{\hspace{-.25cm}\textbf{Analytic}}

    \addplot[only marks, red] coordinates {
      \preoptimczlargetotalfstanaavg
    };
    \addlegendentry{First order}

    \addplot[only marks, red, mark=triangle*] coordinates {
      \preoptimczlargetotalsndanaavg
    };
    \addlegendentry{Second order}

    \addplot[only marks, red, fill opacity=.5, mark=square*] coordinates {
      \preoptimczlargetotalfrthanaavg
    };
    \addlegendentry{Fourth order}

    \addplot[only marks, red, fill opacity=.5, mark=hexagon*] coordinates {
      \preoptimczlargetotalsxthanaavg
    };
    \addlegendentry{Sixth order}

    \addplot[only marks, red, mark=octagon*] coordinates {
      \preoptimczlargetotaleigthanaavg
    };
    \addlegendentry{Eighth order}
  \end{loglogaxis}
\end{tikzpicture}

%% file: 4_ana_min_com_totalCZ.tex
\begin{tikzpicture}
  \begin{loglogaxis}[
    width=10cm,
    ymajorgrids=true,
    legend style={at={(0.05,0.95)},anchor=north west},
    xlabel={System size},
    ylabel={Total gate count},
    xmin=10,
    xmax=600,
    xtick={10,100},
    xticklabels={10,100},
    extra x ticks={30,300},
    extra x tick labels={30,300},
    every axis legend/.append style={nodes={right}}
    ]
    
    \addplot[only marks, red] coordinates {
      \preoptimczlargetotalfrthanaavg
    };
    \addlegendentry{Analytic}
    
    \addplot[only marks, red, fill opacity=0.5] coordinates {
      \preoptimczlargetotalfrthminavg
    };
    \addlegendentry{Minimized}
    
    \addplot[only marks, red, fill opacity=0.35] coordinates {
      \preoptimczlargetotalfrthcomavg
    };
    \addlegendentry{Commutator}
    
    \addplot[only marks,red,mark=o,mark options={fill=white,line width=1.25pt}] coordinates {
      \preoptimczlargetotalfrthfitavg
    };
    \addlegendentry{Empirical}
    
  \end{loglogaxis}
\end{tikzpicture}

%% file: bib.tex
\providecommand{\bysame}{\leavevmode\hbox to3em{\hrulefill}\thinspace}